\documentclass[graybox,envcountsame,envcountchap]{svmono}

\usepackage{mathptmx}

\usepackage[mathscr]{euscript}

\usepackage{type1cm}         

\usepackage{makeidx}         % allows index generation
\usepackage{graphicx}        % standard LaTeX graphics tool
\usepackage[bottom]{footmisc}% places footnotes at page bottom

\usepackage{overpic}  % overlay graphics

\usepackage{bbm}
\usepackage{nicefrac}
\usepackage{cite}

\usepackage[ %breaklinks=true,
             colorlinks=true,
             linkcolor = blue,
             urlcolor  = blue,
             citecolor = red,
             anchorcolor = green
]{hyperref}

\usepackage{color}
\usepackage{amsmath}
\usepackage{amsfonts}
\usepackage{amssymb}

\DeclareFontFamily{U}{mathx}{\hyphenchar\font45}
\DeclareFontShape{U}{mathx}{m}{n}{<-> mathx10}{}
\DeclareSymbolFont{mathx}{U}{mathx}{m}{n}
\DeclareMathAccent{\widebar}{0}{mathx}{"73}

\renewenvironment{petit}{}{}
\usepackage[margin=4cm]{geometry}

\newcommand{\cH}{\mathscr{H}}  % Hilbert spaces
\newcommand{\cL}{\mathscr{L}}  % bounded linear operators
\newcommand{\cLdag}{\mathscr{L}^{\dag}}
\newcommand{\cLsub}{\mathscr{L}_{\bullet}}
\newcommand{\cP}{\mathscr{P}}  % psd operators
\newcommand{\cPsub}{\mathscr{P}_{\bullet}}
\newcommand{\cS}{\mathscr{S}}  % states (psd)
\newcommand{\cSsub}{\mathscr{S}_{\bullet}}
\newcommand{\cSnorm}{\mathscr{S}_{\circ}}
\newcommand{\cF}{\mathscr{T}}  % functionals (not necessarily psd)

\newcommand{\cB}{\mathcal{B}}  % for epsilon-ball
\newcommand{\cBp}{\mathcal{B}_{*}} % for pure states
\newcommand{\Fg}{F_{*}} % generalized fidelity

\newcommand{\cp}{\textnormal{CP}}  % completely positive
\newcommand{\cptp}{\textnormal{CPTP}} 
\newcommand{\cpu}{\textnormal{CPU}}
\newcommand{\cptni}{\textnormal{CPTNI}}

\newcommand{\cb}{\textnormal{CB}}   % completely bounded

\newcommand{\sE}{\mathscr{E}}
\newcommand{\sF}{\mathscr{F}}
\newcommand{\sL}{\mathscr{L}}
\newcommand{\sM}{\mathscr{M}} % general measurement map
\newcommand{\sI}{\mathscr{I}} % Identity map (not used)
\newcommand{\sU}{\mathscr{U}} % isometry
\newcommand{\sP}{\mathscr{P}} % Pinching map

\newcommand{\DD}{\mathbb{D}} % general divergence
\newcommand{\Do}{\widebar{D}}    % old Renyi divergence
\newcommand{\Dn}{\widetilde{D}}  % new Renyi divergence
\newcommand{\QQ}{\mathbb{Q}} % trace part of the general divergence
\newcommand{\Qo}{\widebar{Q}}    
\newcommand{\Qn}{\widetilde{Q}}  
\newcommand{\HH}{\mathbb{H}} % trace part of the general divergence
\newcommand{\Ho}{\widebar{H}}    % old Renyi entropy
\newcommand{\Hn}{\widetilde{H}}  % new Renyi entropy
\newcommand{\HHo}{\widebar{\mathbb{H}}} % trace part of the general divergence
\newcommand{\HHn}{\widetilde{\mathbb{H}}} % trace part of the general divergence

\newcommand{\ua}{\scriptscriptstyle \,\uparrow}
\newcommand{\da}{\scriptscriptstyle \,\downarrow}

\newcommand{\bbC}{\mathbb{C}}
\newcommand{\bbR}{\mathbb{R}}
\newcommand{\bbN}{\mathbb{N}}

\newcommand{\norm}[1]{\left\| #1 \right\|}
\newcommand{\unorm}[1]{\left|\!\middle|\!\middle| #1 \middle|\!\middle|\!\right|}
\newcommand{\normb}[1]{\big\| #1 \big\|}
\newcommand{\abs}[1]{\left| #1 \right|}

\newcommand{\ket}[1]{\left| #1 \right\rangle}
\newcommand{\ketn}[1]{| #1 \rangle}
\newcommand{\bra}[1]{\left\langle #1 \right|}
\newcommand{\bran}[1]{\langle #1 |}
\newcommand{\ip}[2]{\left\langle #1 , #2 \right\rangle}
\newcommand{\ipn}[2]{\langle #1 , #2 \rangle}
\newcommand{\braket}[2]{\left\langle #1 \middle| #2 \right\rangle}
\newcommand{\braketn}[2]{\langle #1 | #2 \rangle}
\newcommand{\proj}[1]{\left| #1 \middle\rangle\!\middle\langle #1 \right|}
\newcommand{\projn}[1]{| #1 \rangle\!\langle #1 |}

\newcommand{\bracket}[3]{\left\langle #1 \middle| #2 \middle| #3 \right\rangle}
\newcommand{\bracketn}[3]{\langle #1 | #2 | #3 \rangle}
\newcommand{\bracketb}[3]{\big\langle #1 \big| #2 \big| #3 \big\rangle}
\newcommand{\bracketB}[3]{\Big\langle #1 \Big| #2 \Big| #3 \Big\rangle}

\newcommand{\Choi}{\Gamma}
\newcommand{\choi}{\gamma}

\newcommand{\id}{I}  % identity
\newcommand{\rhoh}{\hat{\rho}}
\newcommand{\rhob}{\bar{\rho}}
\newcommand{\rhot}{\tilde{\rho}}
\newcommand{\sigmah}{\hat{\sigma}}

\newcommand{\tauh}{\hat{\tau}}
\newcommand{\taut}{\tilde{\tau}}
\newcommand{\eps}{\varepsilon}

\newcommand{\Exp}{\mathbb{E}}

\DeclareMathOperator{\tr}{Tr}  % trace
\DeclareMathOperator{\spec}{spec}  % spectrum
\DeclareMathOperator{\rank}{rank}
\smartqed

\begin{document}

\author{Marco Tomamichel}
\title{Quantum Information Processing with Finite Resources}
\subtitle{Mathematical Foundations}
\date{Last update on May 28, 2021\\~\\
 (This is an arXiv version of the book with various bug-fixes compared to the printed version.)}

\hypersetup{pageanchor=false}

\maketitle

\frontmatter%%%%%%%%%%%%%%%%%%%%%%%%%%%%%%%%%%%%%%%%%%%%%%%%%%%%%%

%\include{dedic}
%!TEX root = book.tex

\extrachap{Acknowledgements}

Renato Renner, Mark M.~Wilde, and Andreas Winter encouraged me to write this book. It is my pleasure to thank Christopher T.~Chubb and Mark M.~Wilde for carefully reading the manuscript and spotting many typos. I want to further thank Rupert L.~Frank, Elliott H.~Lieb, Mil\'an Mosonyi, and Renato Renner for many insightful comments and suggestions. While writing I also greatly enjoyed and profited from scientific discussions with Mario Berta, Fr\'ed\'eric Dupuis, Anthony Leverrier, and Volkher~B.\ Scholz about different aspects of this book. 

\medskip

This arXiv version has a slightly reduced number of typos. I removed the ones brought to my attention by Felix Leditzky, David Sutter, Serge Fehr and Roberto Rubboli. I am especially thankful to Mil\'an Mosonyi for pointing out an overly optimistic lemma, which has now been removed, and to Navneeth Ramakrishnan and Roberto Rubboli for discussions that helped to exactly specify the range of applicability of an improved triangle inequality for the purified distance.

\hypersetup{pageanchor=true}

\tableofcontents

\mainmatter%%%%%%%%%%%%%%%%%%%%%%%%%%%%%%%%%%%%%%%%%%%%%%%%%%%%%%%

%!TEX root = book.tex

\chapter{Introduction}
\label{ch:intro} 
% use \chaptermark{}
% to alter or adjust the chapter heading in the running head

\abstract*{This chapter motivates the study of finite resource quantum information theory and the mathematical framework that is required to do so. We will present a motivating example and outline the content of the book.}

As we further miniaturize information processing devices,
the impact of quantum effects will become more and more relevant. Information processing at the microscopic scale  poses challenges but also offers various opportunities: How much information can be transmitted through a physical communication channel if we can encode and decode our information using a quantum computer? How can we take advantage of entanglement, a form of correlation stronger than what is allowed by classical physics? What are the implications of Heisenberg's uncertainty principle of quantum mechanics for cryptographic security?
These are only a few amongst the many questions studied in the emergent field of quantum information theory. 

One of the predominant challenges when engineering future quantum information processors is that large quantum systems are notoriously hard to maintain in a coherent state and difficult to control accurately. Hence, it is prudent to expect that there will be severe limitations on the size of quantum devices for the foreseeable future.
It is therefore of immediate practical relevance to investigate quantum information processing with limited physical resources, for example, to ask:

\begin{center} 
\begin{svgraybox}
How well can we perform information processing tasks if we only have access to a small quantum device? Can we beat fundamental limits imposed on information processing with non-quantum resources? 
\end{svgraybox}
\end{center}

This book will introduce the reader to the mathematical framework required to answer such questions, and many others. 
%Just to name a few areas of application: 
In quantum cryptography we want to show that a key of \emph{finite} length is secret from an adversary, in quantum metrology we want to infer properties of a \emph{small} quantum system from a \emph{finite} sample, and in quantum thermodynamics we explore the thermodynamic properties of \emph{small} quantum systems. What all these applications have in common is that they concern properties of small quantum devices and require precise statements that remain valid outside asymptopia\,---\,the idealized asymptotic regime where the system size is unbounded.

%%%%%

\section{Finite Resource Information Theory}

Through the lens of a physicist it is natural to see Shannon's information theory~\cite{shannon48} as a  resource theory. Data sources and communication channels are traditional examples of resources in information theory, and its goal is to investigate how these resources are interrelated and how they can be transformed into each other. For example, we aim to compress a data source that contains redundancy into one that does not, or to transform a noisy channel into a noiseless one. Information theory quantifies how well this can be done and in particular provides us with fundamental limits on the best possible performance of any transformation. 

Shannon's initial work~\cite{shannon48} already gives definite answers to the above example questions in the asymptotic regime where resources are unbounded. This means that we can use the input resource as many times as we wish and are interested in the \emph{rate} (the fraction of output to input resource) at which transformations can occur. The resulting statements can be seen as a first approximation to a more realistic setting where resources are necessarily finite, and this approximation is indeed often sufficient for practical purposes. 

However, as argued above, specifically when quantum resources are involved we would like to establish more precise statements that remain valid even when the available resources are very limited.
This is the goal of \emph{finite resource information theory}. The added difficulty in the finite setting is that we are often not able to produce the output resource perfectly. The best we can hope for is to find a tradeoff between the transformation rate and the error we allow on the output resource. In the most fundamental \emph{one-shot} setting we only consider a single use of the input resource and are interested in the tradeoff between the amount of output resource we can produce and the incurred error. We can then see the finite resource setting as a special case of the one-shot setting where the input resource has additional structure, for example a source that produces a sequence of {independent and identically distributed (iid)} symbols or a channel that is memoryless or ergodic.

Notably such considerations were part of the development of information theory from the outset. They motivated the study of \emph{error exponents}, for example by Gallager~\cite{gallager68}. Roughly speaking, error exponents approximate how fast the error vanishes for a fixed transformation rate as the number of available resources increases. However, these statements are fundamentally asymptotic in nature and make strong assumptions on the structure of the resources. Beyond that, Han and Verd\'u established the information spectrum method~\cite{verdu93,han02} which allows to consider unstructured resources but is asymptotic in nature.
More recently finite resource information theory has attracted considerable renewed attention, for example due to the works of Hayashi~\cite{hayashi08,hayashi09} and Polyanskiy \emph{et al.}~\cite{polyanskiy10}. The approach in these works\,---\,based on Strassen's techniques~\cite{strassen62}\,---\,is motivated operationally: in many applications we can admit a small, \emph{fixed error} and our goal is to find the maximal possible transformation rate as a function of the error and the amount of available resource.\footnote{The topic has also been reviewed recently by Tan~\cite{tanbook14}.}

In an independent development, approximate or asymptotic statements were also found to be insufficient in the context of cryptography. In particular the advent of quantum cryptography~\cite{bb84,ekert91} motivated a precise information-theoretic treatment of the security of secret keys of finite length~\cite{maurer05,renner05}. In the context of quantum cryptography many of the standard assumptions in information theory are no longer valid if one wants to avoid any assumptions on the eavesdropper's actions. In particular, the common assumption that resources are iid or ergodic is hardly justified. In quantum cryptography we are instead specifically interested in the one-shot setting, where we want to understand how much (almost) secret key can be extracted from a single use of an unstructured resource.

The abstract view of finite resource information theory as a resource theory also reveals why it has found various applications in physical resource theories, most prominently in thermodynamics (see, e.g.,~\cite{delrio11,wehner13,faist15} and references therein).

%%%%%

\subsubsection*{R\'enyi and Smooth Entropies}

The main focus of this book will be on various measures of entropy and information that underly finite resource information theory, in particular R\'enyi and smooth entropies.
The concept of entropy has its origins in physics, in particular in the works of Boltzmann~\cite{boltzmann1872} and Gibbs~\cite{gibbs1876} on thermodynamics. Von Neumann~\cite{vonneumann32} generalized these concepts to quantum systems.
Later Shannon~\cite{shannon48}\,---\,well aware of the origins of entropy in physics\,---\,interpreted entropy as a measure of uncertainty of the outcome of a random experiment. 
He found that entropy, or \emph{Shannon entropy} as it is called now in the context of information theory\footnote{Notwithstanding the historical development, we follow the established tradition and use \emph{Shannon entropy} to refer to entropy. We use \emph{von Neumann entropy} to refer to its quantum generalization.}, characterizes the optimal \emph{asymptotic} rate at which information can be compressed. However, we will soon see that it is necessary to consider alternative information measures if we want to move away from asymptotic statements.

Error exponents can often be expressed in terms of \emph{R\'enyi entropies}~\cite{renyi61} or related information measures, which partly explains the central importance of this one-parameter family of entropies in information theory. R\'enyi entropies share many mathematical properties with the Shannon entropy and are powerful tools in many information-theoretic arguments. A significant part of this book is thus devoted to exploring quantum generalizations of R\'enyi entropies, for example the ones proposed by Petz~\cite{petz86} and a more recent specimen~\cite{lennert13,wilde13} that has already found many applications. 

%We will review the current state of the art of the research on quantum R\'enyi entropies and discuss the properties of these information measures in detail. A particular goal of this endeavor is to give simple and concise proofs of their most important properties.
% (see Chapters~\ref{ch:renyi} and~\ref{ch:cond}).

The particular problems encountered in cryptography led to the development of smooth entropies~\cite{renner04} and their quantum generalizations~\cite{rennerkoenig05,renner05}. Most importantly, the smooth min-entropy captures the amount of uniform randomness that can be extracted from an unstructured source if we allow for a small error. (This example is discussed in detail in Section~\ref{sc:rand-ext}.) The smooth entropies are variants of R\'enyi entropies and inherit many of their properties. They have since found various applications ranging from information theory to quantum thermodynamics and will be the topic of the second part of this book.
%(see Chapter~\ref{ch:calc}). 

We will further motivate the study of these information measures with a simple example in the next section.

\bigskip

Besides their operational significance, there are other reasons why the study of information measures is particularly relevant in quantum information theory.
Many standard arguments in information theory can be formulated in term of entropies, and often this formulation is most amenable to a generalization to the quantum setting. For example, conditional entropies provide us with a measure of the uncertainty inherent in a quantum state from the perspective of an observer with access to side information. This allows us to circumvent the problem that we do not have a suitable notion of conditional probabilities in quantum mechanics. 
As another example, arguments based on typicality and the asymptotic equipartition property can be phrased in terms of smooth entropies which often leads to a more concise and intuitive exposition.
Finally, the study of quantum generalizations of information measures sometimes also gives new insights into the classical quantities. For example, our definitions and discussions of conditional R\'enyi entropy also apply to the classical special case where such definitions have not yet been firmly established.

%%%%%%%

\section[Motivating Example]{Motivating Example: Source Compression}

%In the following we consider a simple example of an information processing task: source compression.

%%%%

%\subsubsection*{Perfect Source Compression}

We are using notation that will be formally introduced in Chapter~\ref{ch:prelim} and concepts that will be expanded on in later chapters (cf.~Table~\ref{tb:links}). A data source is described probabilistically as follows. Let $X$ be a random variable with distribution $\rho_X(x) = \Pr[X = x]$ that models the distribution of the different symbols that the source emits. The number of bits of memory needed to store one symbol produced by this source so that it can be recovered with certainty is given by $\lceil H_0(X)_{\rho} \rceil$, where $H_0(X)_{\rho}$ denotes the \emph{Hartley entropy}~\cite{hartley28} of $X$, defined as
\begin{align}
  H_0(X)_{\rho} = \log_2 \big| \{ x : \rho_X(x) > 0 \} \big| \,.
\end{align}
The Hartley entropy is a limiting case of a {R\'enyi entropy}~\cite{renyi61}
%\footnote{More precisely, $H_0(X)_{\rho} = \lim_{\alpha\to 0} H_{\alpha}(X)_{\rho}$, where $H_{\alpha}(X)_{\rho}$ is the R\'enyi entropy of order $\alpha$.} 
and simply measures the cardinality of the support of $X$. In essence, this means that we can ignore symbols that never occur but otherwise our knowledge of the distribution of the different symbols does not give us any advantage.

\begin{table}[h]
  \begin{center}
  \begin{tabular}{c@{\hspace{0.1cm}}|@{\hspace{0.1cm}}l@{\hspace{0.1cm}}|@{\hspace{0.1cm}}p{0.6\textwidth}}
    & \textbf{Concept} & \textbf{to be discussed further in} \\
    \hline
    \vspace{-9pt} & & \\ % increases spacing after the horizontal line
    $H_{\alpha}$ &  R\'enyi entropy & Chapters~\ref{ch:renyi} and~\ref{ch:cond} \\
    $\Delta(\cdot,\cdot)$ & variational distance & Section~\ref{sc:norms}, as generalized trace distance \\
    $H_{\max}^{\eps}$ & smooth R\'enyi entropy & Chapter~\ref{ch:calc}, as smooth max-entropy$^*$ \\
    & entropic AEP & Section~\ref{sc:qaep}, entropic asymptotic equipartition property
  \end{tabular}
  \end{center}
  \emph{$^*$We will use a different metric for the definition of the smooth max-entropy.}
  \caption{Reference to detailed discussion of the quantities and concepts mentioned in this section.}
  \label{tb:links}
\end{table}

As an example, consider a source that outputs lowercase characters of the English alphabet. If we want to
store a single character produced by this source such that it can be recovered with 
certainty, we clearly need $\lceil \log_2 26 \rceil = 5$ bits of memory as a resource. 

\subsubsection*{Analysis with R\'enyi Entropies}

More interestingly, we may ask how much memory we need to store the output of the source if we allow for a small probability of failure, $\eps \in (0,1)$.
To answer this we investigate encoders that assign codewords of a fixed length $\log_2 m$ (in bits) to the symbols the source produces. These codewords are then stored and a decoder is later used to compute an estimate of~$X$ from the codewords. If the probability that this estimate equals the original symbol produced by the source is at least $1-\eps$, then we call such a scheme an \emph{$(\eps,m)$-code}.
For a source $X$ with probability distribution $\rho_X$, we are thus interested in finding the tradeoff between code length, $\log_2 m$, and the probability of failure, $\eps$, for all $(\eps,m)$-codes. 

Shannon in his seminal work~\cite{shannon48} showed that simply disregarding the most unlikely source events (on average) leads to an arbitrarily small failure probability if the code length is chosen
sufficiently long. In particular, Gallager's proof~\cite{gallager68,gallager79} implies that $(\eps,m)$-codes always exist as long as
\begin{align}
  \log_2 m \geq H_{\alpha}(X)_{\rho} + \frac{\alpha}{1-\alpha} \log_2 \frac{1}{\eps} 
  \qquad \textrm{for some} \quad \alpha \in \Big[\frac12, 1\Big) \,.
  \label{eq:gallager}
\end{align}
Here, $H_{\alpha}(X)_{\rho}$ is the \emph{R\'enyi entropy of order $\alpha$}, defined as
\begin{align}
  H_{\alpha}(X)_{\rho} = \frac{1}{1-\alpha} \log_2 \bigg( \sum_{x} \rho_X(x)^{\alpha} \bigg) \,.
\end{align}
for all $\alpha \in (0, 1) \cup (1,\infty)$ and as the respective limit for $\alpha \in \{0, 1, \infty\}$.
The R\'enyi entropies are monotonically decreasing in $\alpha$.
Clearly the lower bound in~\eqref{eq:gallager} thus constitutes a tradeoff: larger values of the order parameter $\alpha$ lead to a smaller R\'enyi entropy but will increase the penalty term $\frac{\alpha}{1-\alpha} \log_2 \frac{1}{\eps}$. 
Statements about the existence of codes as in~\eqref{eq:gallager} are called \emph{achievability bounds} or direct bounds. 

This analysis can be driven further if we consider sources with structure. In particular, consider a sequence of sources that produce $n \in \bbN$ \emph{independent and identically distributed (iid)} symbols 
$X^n = (Y_1, Y_2, \ldots, Y_n)$, where each $Y_i$ is distributed according to the law $\tau_Y(y)$. We then consider a sequence of ($\eps, 2^{nR}$)-codes for these sources, where the \emph{rate} $R$ indicates the number of memory bits required per symbol the source produces.
For this case~\eqref{eq:gallager} reads
\begin{align}
  R \geq \frac{1}{n} H_{\alpha}(X^n)_{\rho} + \frac{\alpha}{n(1-\alpha)} \log_2 \frac{1}{\eps} =
   H_{\alpha}(Y)_{\tau} + \frac{\alpha}{n(1-\alpha)} \log_2 \frac{1}{\eps}  \label{eq:gallager-iid}
\end{align}
where we used additivity of the R\'enyi entropy to establish the equality. The above inequality implies that such a sequence of $(\eps,2^{nR})$-codes exists for sufficiently large $n$ if $R > H_{\alpha}(Y)_{\tau}$. And finally, since this holds for all $\alpha \in [\frac12,1)$, we may take the limit $\alpha \to 1$ in~\eqref{eq:gallager-iid} to recover Shannon's original result~\cite{shannon48}, which states that such codes exists if
\begin{align} 
  R > H(Y)_{\tau}, \quad \textrm{where} \quad H(Y)_{\tau} = H_1(Y)_{\tau} = - \sum_y \tau_Y(y) \log_2 \tau_Y(y)
\end{align}
is the Shannon entropy of the source. This rate is in fact optimal, meaning that every scheme with $R < H(Y)_{\tau}$ necessary fails with certainty as $n \to \infty$. This is an example of an asymptotic statement (with infinite resources) and such statements can often be expressed in terms of the Shannon entropy or related information measures.

%%%%

\subsubsection*{Analysis with Smooth Entropies}

Another fruitful approach to analyze this problem brings us back to the unstructured, one-shot case. We note that the above analysis can be refined without assuming any structure by \emph{``smoothing''} the entropy. Namely, we construct an $(\eps,m)$ code for the source $\rho_X$ using the following recipe:

\begin{itemize}
  \item Fix $\delta \in (0, \eps)$ and
  let $\tilde{\rho}_X$ be any probability distribution that is $(\eps-\delta)$-close to $\rho_X$ in variational distance.
   Namely we require that $\Delta(\tilde{\rho}_X, \rho_X) \leq \eps-\delta$ where $\Delta(\cdot,\cdot)$ denotes the variational distance.
  \item Then, take a $(\delta,m)$-code for the source $\tilde{\rho}_X$. Instantiating~\eqref{eq:gallager} with $\alpha = \frac12$, we find that there exists such a code as long as $\log_2 m \geq  H_{\nicefrac12}(X)_{\tilde{\rho}} + \log_2 \frac{1}{\delta}$.
  \item Apply this code to a source with the distribution $\rho_X$ instead, incurring a total error of at most 
  $\delta + \Delta(\rho_X,\tilde{\rho}_X) \leq \eps$. (This uses the triangle inequality and the fact that the variational distance contracts when we process information through the encoder and decoder.)
\end{itemize}

\noindent Hence, optimizing this over all such $\tilde{\rho}_X$, we find that there exists a $(\eps,m)$-code if
\begin{align}
  \log_2 m \geq H_{\max}^{\eps-\delta}(X)_{\rho} + \log_2 \frac{1}{\delta} , \quad \textrm{where} \quad H_{\max}^{\eps'}(X)_{\rho} := \min_{ \tilde{\rho}_X : \Delta(\rho_X,\tilde{\rho}_X) \leq \eps'} H_{\nicefrac12}(X)_{\tilde{\rho}} 
\end{align}
is the \emph{$\eps'$-smooth max-entropy}, which is based on the R\'enyi entropy of order $\frac12$.

Furthermore, this bound is approximately optimal in the following sense. It can be shown~\cite{renesrenner10} that all $(\eps,m)$-codes must satisfy
$\log_2 m \geq H_{\max}^{\eps}(X)_{\rho}$. Such bounds that give restrictions valid for all codes are called \emph{converse bounds}. 
Rewriting this, we see that the minimal value of $m$ for a given $\eps$, denoted $m_*(\eps)$, satisfies
\begin{align}
  H_{\max}^{\eps}(X)_{\rho} \leq \log_2 m_*(\eps) \leq \inf_{\delta \in (0,\eps)} \Big\lceil H_{\max}^{\eps-\delta}(X)_{\rho} + \log_2 \frac1{\delta} \Big\rceil \,.  \label{eq:directconverse}
\end{align}

We thus informally say that the memory required for one-shot source compression is \emph{characterized} by the smooth max-R\'enyi entropy.\footnote{The smoothing approach in the classical setting was first formally discussed in~\cite{renner04}. A detailed analysis of one-shot source compression, including quantum side information, can be found in~\cite{renesrenner10}.}

Finally, we again consider the case of an iid source, and as before, we expect that in the limit of large $n$, the optimal compression rate $\frac{1}{n} m_*(\eps)$ should be characterized by the Shannon entropy. This is in fact an expression of an entropic version of the \emph{asymptotic equipartition property}, which states that
\begin{align}
  \lim_{n\to\infty} \frac{1}{n} H_{\max}^{\eps'}(X^n)_{\rho} = H(Y)_{\tau}  \quad \textrm{for all} \quad \eps' \in (0,1) \,.
\end{align}

%%%%

\subsubsection*{Why Shannon Entropy is Inadequate}

To see why the Shannon entropy does not suffice to characterize {one-shot} 
source compression, consider a source that produces the symbol `$\sharp$' with probability
${1}/{2}$ and $k$ other symbols with probability ${1}/{2 k}$ each. On the
one hand, for any fixed failure probability $\eps \ll 1$, the converse bound 
in~\eqref{eq:directconverse} evaluates to approximately $\log_2 k$. 
This implies that we cannot compress this source much beyond its Hartley entropy. On the other hand, the Shannon entropy of this distribution is $\frac{1}{2}(\log_2 k + 2)$ and underestimates the required memory by a factor of two.

%%%%%%%%%

%%%%%%%

% should I discuss Shannon entropy, conditional entropy and divergence here?

%\section{Entropy and Divergence}

%In the previous section we have encountered Shannon and R\'enyi entropy, as well as smooth R\'enyi entropy. The purpose of this book is to generalize these measures to the quantum settings

%%%%%%%%%

\section{Outline of the Book}

The goal of this book is to explore quantum generalizations of the measures encountered in our example, namely the R\'enyi entropies and smooth entropies. Our exposition assumes that the reader is familiar with basic probability theory and linear algebra, but not necessarily with quantum mechanics. For the most part we restrict our attention to physical systems whose observable properties are discrete, e.g.\ spin systems or excitations of particles bound in a potential. This allows us to avoid mathematical subtleties that appear in the study of systems with observable properties that are continuous. We will, however, mention generalizations to continuous systems where applicable and refer the reader to the relevant literature.

\bigskip

\noindent The book is organized as follows:

\begin{description}

\item
\textbf{Chapter~\ref{ch:prelim}}
introduces the notation used throughout the book and presents the mathematical framework underlying quantum theory for general (potentially continuous) systems. Our notation is summarized in Table~\ref{tab:notation} %(Section~\ref{sc:notation}) 
so that the remainder of the chapter can easily be skipped by expert readers. The exposition starts with introducing events as linear operators on a Hilbert space (Section~\ref{sc:linear}) and then introduces states as functionals on events (Section~\ref{sc:func}). Multi-partite systems and entanglement is then discussed using the Hilbert space tensor product (Section~\ref{sc:joint}) and finally quantum channels are introduced as a means to study the evolution of systems in the Schr\"odinger and Heisenberg picture (Section~\ref{sc:channel}). Finally, this chapter assembles the mathematical toolbox required to prove the results in the later chapters, including a discussion of operator monotone, concave and convex functions on positive operators (Section~\ref{sc:functions}). Most results discussed here are well-known and proofs are omitted. We do not attempt to provide an intuition or physical justification for the mathematical models employed, but instead highlight some connections to classical information theory.

\item
\textbf{Chapter~\ref{ch:metrics}} treats norms and metrics on quantum states. First we discuss Schatten norms and a variational characterization of the Schatten norms of positive operators that will be very useful in the remainder of the book (Section~\ref{sc:norms}). We then move on to discuss a natural dual norm for sub-normalized quantum states and the metric it induces, the trace distance (Section~\ref{sc:gtd}). The fidelity is another very prominent measure for the proximity of quantum states, and here we sensibly extend its to definition to cover sub-normalized states (Section~\ref{sc:fid}). 
Finally, based on this generalized fidelity, we introduce a powerful metric for sub-normalized quantum states, the purified distance (Section~\ref{sc:pd}). This metric combines the clear operational interpretation of the trace distance with the desirable mathematical properties of the fidelity.

\item
\textbf{Chapter~\ref{ch:renyi}}
discusses quantum generalizations of the R\'enyi divergence. Divergences (or relative entropies) are measures of distance between quantum states (although they are not metrics) and entropy as well as conditional entropy can conveniently be defined in terms of the divergence. Moreover, the entropies inherit many important properties from corresponding properties of the divergence. In this chapter, we first discuss the classical special case of the R\'enyi divergence (Section~\ref{sc:rclass}). This allows us to point out several properties that we expect a suitable quantum generalization of the R\'enyi divergence to satisfy. Most prominently we expect them to satisfy a data-processing inequality which states that the divergence is contractive under application of quantum channels to both states. Based on this, we then explore quantum generalizations of the R\'enyi divergence and find that there is more than one quantum generalization that satisfies all desired properties (Section~\ref{sc:rclassify}). 

We will mostly focus on two different quantum R\'enyi divergences, called the minimal and Petz quantum R\'enyi divergence (Sections~\ref{sc:rminimal}--\ref{sc:rhc}). The first quantum generalization is called the minimal quantum R\'enyi divergence (because it is the smallest quantum R\'enyi divergence that satisfies a data-processing inequality), and is also known as ``sandwiched'' R\'enyi relative entropy in the literature. It has found operational significance in the strong converse regime of asymmetric binary hypothesis testing.
The second quantum generalization is Petz' quantum R\'enyi relative entropy, which attains operational significance in the quantum generalization of Chernoff's and Hoeffding's bound on the success probability in binary hypothesis testing (cf.\ Section~\ref{sc:app-hypo}).

\item
\textbf{Chapter~\ref{ch:cond}} generalizes conditional R\'enyi entropies (and unconditional entropies as a special case) to the quantum setting. The idea is to define operationally relevant measures of uncertainty about the state of a quantum system from the perspective of an observer with access to some side information stored in another quantum system. As a preparation, we discuss how the conditional Shannon entropy and the conditional von Neumann entropy can be conveniently expressed in terms of relative entropy either directly or using a variational formula (Section~\ref{sc:cvn}). 
Based on the two families of quantum R\'enyi divergences, we then define four families of quantum conditional R\'enyi entropies (Section~\ref{sc:cond}). We then prove various properties of these entropies, including data-processing inequalities that they directly inherit from the underlying divergence.
A genuinely quantum feature of conditional R\'enyi entropies is the duality relation for pure states (Section~\ref{sc:rdual}). These duality relations also show that the four definitions are not independent, and thereby also reveal a connection between the minimal and the Petz quantum R\'enyi divergence. Furthermore, even though the chain rule does not hold with equality for our definitions, we present some inequalities that replace the chain rule (Section~\ref{sc:rchain}).

\item
\textbf{Chapter~\ref{ch:calc}} deals with smooth conditional entropies in the quantum setting. First, we discuss the min-entropy and the max-entropy, two special cases of R\'enyi entropies that underly the definition of the smooth entropy (Section~\ref{sc:minmax}). In particular, we show that they can be expressed as semi-definite programs, which means that they can be approximated efficiently (for small quantum systems) using standard numerical solvers. The idea is that these two entropies serve as representatives for the R\'enyi entropies with large and small $\alpha$, respectively.
We then define the smooth entropies (Section~\ref{sc:smooth}) as optimizations of the min- and max-entropy over a ball of states close in purified distance. We explore some of their properties, including chain rules and duality relations (Section~\ref{sc:smooth-properties}). Finally, the main application of the smooth entropy calculus is an entropic version of the asymptotic equipartition property for conditional entropies, which states that the (regularized) smooth min- and max-entropies converge to the conditional von Neumann entropy for iid product states (Section~\ref{sc:qaep}).

\item
\textbf{Chapter~\ref{ch:app}} concludes the book with a few selected applications of the mathematical concepts surveyed here.
First, we discuss various aspects of binary hypothesis testing, including Stein's lemma, the Chernoff bound and the Hoeffding bound as well as strong converse exponents (Section~\ref{sc:app-hypo}). This provides an operational interpretation of the R\'enyi divergences discussed in Chapter~\ref{ch:renyi}. Next, we discuss how the duality relations and the chain rule for conditional R\'enyi entropies can be used to derive entropic uncertainty relations\,---\,powerful manifestations of the uncertainty principle of quantum mechanics (Section~\ref{sc:app-ucr}). Finally, we discuss randomness extraction against quantum side information, a premier application of the smooth entropy formalism that justifies its central importance in quantum cryptography (Section~\ref{sc:rand-ext}).

\end{description}

\subsubsection*{What This Book Does Not Cover}

It is beyond the scope of this book to provide a comprehensive treatment of the many applications the mathematical framework reviewed here has found. However, in addition to Chapter~\ref{ch:app}, we will mention a few of the most important applications in the background section of each chapter. 
Tsallis entropies~\cite{tsallis88} have found several applications in physics, but they have no solid foundation in information theory and we will not discuss them here. It is worth mentioning, however, that many of the mathematical developments in this book can be applied to quantum Tsallis entropies as well. 
There are alternative frameworks besides the smooth entropy framework that allow to treat unstructured resources, most prominently the information-spectrum method and its quantum generalization due to Nagaoka and Hayashi~\cite{nagaoka07}. These approaches are not covered here since they are asymptotically equivalent to the smooth entropy approach~\cite{dattarenner08,tomamichel12}.
Finally, this book does not cover R\'enyi and smooth versions of mutual information and conditional mutual information. These quantities are a topic of active research.

%!TEX root = book.tex

\chapter{Modeling Quantum Information}
\label{ch:prelim} 
% use \chaptermark{}
% to alter or adjust the chapter heading in the running head

\abstract*{Classical as well as quantum information is stored in physical systems, or ``information is inevitably physical'' as Rolf Landauer famously said. 
These physical systems are ultimately governed by the laws of quantum mechanics.
In this chapter we quickly review the relevant mathematical foundations of quantum theory and introduce notational conventions that will be used throughout the book.}

Classical as well as quantum information is stored in physical systems, or ``information is inevitably physical'' as Rolf Landauer famously said. 
These physical systems are ultimately governed by the laws of quantum mechanics.
In this chapter we quickly review the relevant mathematical foundations of quantum theory and introduce notational conventions that will be used throughout the book.

In particular we will discuss concepts of functional and matrix analysis as well as linear algebra that will be of use later. We consider general separable Hilbert spaces in this chapter, even though in the rest of the book we restrict our attention to the finite-dimensional case. This digression is useful because it motivates the notation we use throughout the book, and it allows us to distinguish between the mathematical structure afforded by quantum theory and the additional structure that is only present in the finite-dimensional case.

Our notation is summarized in Section~\ref{sc:notation} and the remainder of this chapter can safely be skipped by expert readers. The presentation here is compressed and we omit proofs. We instead refer to standard textbooks (see Section~\ref{sc:prelim-bg} for some references) for a more comprehensive treatment.

\section{General Remarks on Notation}
\label{sc:notation}

The notational conventions for this book are summarized in Table~\ref{tab:notation}. The table includes references to the sections where the corresponding concepts are introduced. Throughout this book we are careful to distinguish between linear operators (e.g.\ events and Kraus operators) and functionals on the linear operators (e.g.\ states), which are also represented as linear operators (e.g.\ density operators). This distinction is inspired by the study of infinite-dimensional systems where these objects do not necessarily have the same mathematical structure, but it is also helpful in the finite-dimensional setting.\footnote{For example, it sheds light on the fact that we use the operator norm for ordinary linear operators and its dual norm, the trace norm, for density operators.}

\begin{table}[t]
\begin{center}
  \begin{tabular}{c@{\hspace{0.1cm}}|@{\hspace{0.1cm}}c@{\hspace{0.1cm}}|@{\hspace{0.1cm}}p{0.68\textwidth}@{\hspace{0.1cm}}|@{\hspace{0.1cm}}l}
    \textbf{Symbol} & \textbf{Variants} & \textbf{Description} & \textbf{Section} \\
    \hline
    \vspace{-9pt} & & \\ % increases spacing after the horizontal line
    $\bbR$, $\bbC$ & $\bbR_+$ & real and complex fields (and non-negative reals) \\
    $\bbN$ & & natural numbers \\
    $\log, \exp$ & $\ln, e$ & logarithm (to unspecified basis, but $> 1$), and its inverse, the exponential function (natural logarithm and Euler's constant) \\
    \hline
    \vspace{-9pt} & &  \\ % increases spacing after the horizontal line
    $\cH$ & $\cH_{AB}, \cH_X$ & Hilbert spaces (for joint system $AB$ and system $X$) & \ref{sc:hilbert} \\
    $\bran{\cdot},\, \ketn{\cdot}$ & & bra and ket  \\
    $\tr(\cdot)$ & $\tr_A$ & trace (partial trace) & \ref{sc:trace} \\
    $\otimes$ & $(\cdot)^{\otimes n}$ & tensor product ($n$-fold tensor product) & \ref{sc:tensor} \\
    $\oplus$ & & direct sum for block diagonal operators & \ref{sc:events} \\
    $A \ll B$ & & $A$ is dominated by $B$, i.e.\ kernel of $A$ contains kernel of $B$ \\
    $A \perp B$ & & $A$ and $B$ are orthogonal, i.e.\ $A B = B A = 0$ \\
    \hline
    \vspace{-9pt} & &  \\ % increases spacing after the horizontal line
    $\cL$ & $\cL(A,B)$ & bounded linear operators (from $\cH_{A}$ to $\cH_{B}$) & \ref{sc:linop} \\
    $\cLdag$ & $\cLdag(B)$ & self-adjoint operators (acting on $\cH_B$)  \\
    $\cP$ & $\cP(CD)$ & positive semi-definite operators (acting on $\cH_{CD}$) \\
    $\{ A \geq B \}$ & & projector on subspace where $A-B$ is non-negative \\
    $\| \cdot \|$ & & operator norm & \ref{sc:linop} \\
    $\cLsub$ & $\cLsub(E)$ & contractions in $\cL$ (acting on $\cH_E$) \\
    $\cPsub$ & $\cPsub(A)$ & contractions in $\cP$ (corresponding to events on $A$) & \ref{sc:events} \\
    $\id$ & $\id_Y$ & identity operator (acting on $\cH_Y$) \\
    \hline
    \vspace{-9pt} & &\\ % increases spacing after the horizontal line
    $\ip{\cdot}{\cdot}$ & & Hilbert-Schmidt inner product & \ref{sc:funct} \\
    $\cF$ & $\cF \equiv \cL$~$^{\ddag}$ & trace-class operators representing linear functionals  \\
    $\cS$ & $\cS \equiv \cP$~$^{\ddag}$ & operators representing positive functionals \\
    $\| \cdot \|_*$ & $\tr | \cdot |$ & trace norm on functionals & \ref{sc:funct} \\
    $\cSsub$ & $\cSsub(A)$ & sub-normalized density operators (on $A$) & \ref{sc:states} \\
    $\cSnorm$ & $\cSnorm(B)$ & normalized density operators, or states (on $B$) \\
    $\pi$ & $\pi_A$ & fully mixed state (on $A$), in finite dimensions & \ref{sc:states} \\
    $\psi$ & $\psi_{AB}$ & maximally entangled state (between $A$ and $B$), in finite dimensions & \ref{sc:entangle} \\
    \hline
    \vspace{-9pt} & & \\ % increases spacing after the horizontal line
    $\cb$ & $\cb(A,B)$ & completely bounded maps (from $\cL(A)$ to $\cL(B)$) & \ref{sc:cbm} \\
    %$\cb_*$ & $\cb_*(A,B)$ & completely trace norm bounded maps (from $\cF(A)$ to $\cF(B)$) \\ 
    $\cp$ &  & completely positive maps & \ref{sc:cpm} \\
    $\cptp$ & $\cptni$ & completely positive trace-preserving (trace-non-increasing) map\\
    %$\cptni$ & $\cptni^{\dag}$ & trace non-increasing map (or sub-unital) completely positive map \\
    \hline
    \vspace{-9pt} & & \\ % increases spacing after the horizontal line
    $\| \cdot \|_+$ & $\| \cdot \|_p$ & positive cone dual norm (Schatten $p$-norm) & \ref{sc:norms} \\
    $\Delta(\cdot,\cdot)$ & & generalized trace distance for sub-normalized states & \ref{sc:gtd} \\
    $F(\cdot,\cdot)$ & $F_*(\cdot,\cdot)$ & fidelity (generalized fidelity for sub-normalized states) & \ref{sc:fid} \\
    $P(\cdot, \cdot)$ & & purified distance for sub-normalized states & \ref{sc:pd} \\  
    \hline
  \end{tabular}
  \end{center}
  \emph{$^{\ddag}$This equivalence only holds if the underlying Hilbert space is finite-dimensional.}
  \caption{Overview of Notational Conventions.}
  \label{tab:notation}
\end{table}

We do not specify a particular basis for the logarithm throughout this book, and simply use $\exp$ to denote the inverse of $\log$.\footnote{The reader is invited to think of $\log(x)$ as the binary logarithm of $x$ and, consequently, $\exp(x) = 2^{x}$, as is customary in quantum information theory.} The natural logarithm is denoted by $\ln$.

We label different {physical systems} by capital Latin letters $A$, $B$, $C$, $D$, and $E$, as well as $X$, $Y$, and $Z$ which are specifically reserved for classical systems. The label thus always determines if a system is quantum or classical. We often use these labels as subscripts to guide the reader by indicating which system a mathematical object belongs to. We drop the subscripts when they are evident in the context of an expression (or if we are not talking about a specific system). We also use the capital Latin letters $L$, $K$, $H$, $M$, and $N$ to denote linear operators, where the last two are reserved for positive semi-definite operators. The identity operator is denoted $\id$. Density operators, on the other hand, are denoted by lowercase Greek letters $\rho$, $\tau$, $\sigma$, and $\omega$. We reserve $\pi$ and $\psi$ for the fully mixed state and the maximally entangled state, respectively. Calligraphic letters are used to denote quantum channels and other maps acting on operators.

%%%%

\section{Linear Operators and Events}
\label{sc:linear}

For our purposes, a \emph{physical system} is fully characterized by the set of events that can be observed on it. For classical systems, these events are traditionally modeled as a $\sigma$-algebra of subsets of the sample space, usually the power set in the discrete case. For quantum systems the structure of events is necessarily more complex, even in the discrete case. This is due to the non-commutative nature of quantum theory: the union and intersection of events are generally ill-defined since it matters in which order events are observed. 

Let us first review the mathematical model used to describe events in quantum mechanics (as positive semi-definite operators on a Hilbert space). Once this is done, we discuss physical systems carrying quantum and classical information.

\subsection{Hilbert Spaces and Linear Operators}
\label{sc:hilbert}

For concreteness and to introduce the notation, we consider two physical systems $A$ and $B$ as examples in the following. 
We associate to $A$ a separable \emph{Hilbert space} $\cH_A$ over the field $\bbC$, equipped with an \emph{inner product} $\ip{\cdot}{\cdot}: \cH_A \times \cH_A \to \mathbb{C}$. In the finite-dimensional case, this is simply a complex inner product space, but we will follow a tradition in quantum information theory and call $\cH_A$ a Hilbert space also in this case.
Analogously, we associate the Hilbert space $\cH_B$ to the physical system~$B$.

%%%%

\subsubsection*{Linear Operators}
\label{sc:linop}

Our main object of study are \emph{linear operators} acting on the system's Hilbert space. We consistently use upper-case Latin letters to denote such linear operators. 
%The set of linear operators from $\cH_A$ to $\cH_B$ is denoted $\cLL(A,B)$. 
More precisely, we consider the set of \emph{bounded linear operators} from $\cH_A$ to $\cH_B$, which we denote by $\cL(A, B)$. Bounded here refers to the \emph{operator norm} induced by the Hilbert space's inner product.
\begin{svgraybox}
The \textbf{operator norm} on $\cL(A,B)$ is defined as
\begin{align}
  \| \cdot \| :\  \quad L \mapsto \sup \Big\{ \sqrt{ \ip{Lv}{Lv}_B} \ :\ v \in \cH_A,\  \ip{v}{v}_A \leq 1 \Big\} .
\end{align}
\end{svgraybox}
For all $L \in \cL(A,B)$, we have $\| L \| < \infty$ by definition. A linear operator is continuous if and only if it is bounded.\footnote{
\emph{Relation to Operator Algebras:}
Let us note that $\cL(A,B)$ with the norm $\| \cdot \|$ is a Banach space over $\bbC$. Furthermore, the operator norm satisfies
\begin{align}
  \| L \|^2 = \| L^{\dagger} \|^2 = \| L^{\dagger} L \| \quad \textnormal{and} \quad \| L K \| \leq \|L\| \cdot \|K\| \,.
\end{align}
for any $L \in \cL(A, B)$ and $K \in \cL(B, A)$. The inequality states that the norm is \emph{sub-multiplicative}.

The above properties of the norm imply that the space $\cL(A)$ is (weakly) closed under multiplication and the adjoint operation. In fact, $\cL(A)$ constitutes a (Type I factor) von Neumann algebra or $C^*$ algebra. Alternatively, we could have started our considerations right here by postulating a Type 1 von Neumann algebra as the fundamental object describing individual physical systems, and then deriving the Hilbert space structure as a consequence.
}
Let us now summarize some important concepts and notation that we will frequently use throughout this book.
\begin{itemize}
\item 
The \emph{identity} operator on $\cH_A$ is denoted $\id_A$.

\item
The \emph{adjoint} of a linear operator $L \in \cL(A,B)$ is the unique operator $L^{\dag} \in \cL(B,A)$ that satisfies $\ipn{ {w} }{ L {v} }_B = \ipn{ L^{\dag} {w} }{  {v} }_A$ for all ${v} \in \cH_A$, ${w} \in \cH_B$. Clearly, $(L^{\dag})^{\dag} = L$.

\item For scalars $\alpha \in \bbC$, the adjoint corresponds to the complex conjugate, $\alpha^{\dag} = \overline{\alpha}$.

\item We find $(L K)^{\dag} = K^{\dag} L^{\dag}$ by applying the definition twice.
\item
The \emph{kernel} of a linear operator $L \in \cL(A,B)$ is the subspace of $\cH_A$ spanned by vectors $v \in \cH_A$ satisfying $L v = 0$.
The \emph{support} of $L$ is its orthogonal complement in $\cH_A$ and the \emph{rank} is the cardinality of the support. Finally, the image of $L$ is the subspace of $\cH_B$ spanned by vectors $w \in \cH_B$
such that $w = L v$ for some $v \in \cH_A$.

\item For operators $K, L \in \cL(A)$ we say that $L$ is \emph{dominated} by $K$ if the kernel of $K$ is contained in the kernel of $L$. Namely, we write $L \ll K$ if and only if 
\begin{align}
  K \ket{v}_A = 0 \implies L \ket{v}_A = 0 \qquad \textrm{for all} \quad v \in \cH_A \,.
\end{align}

\item We say $K, L \in \cL(A)$ are \emph{orthogonal} (denoted $K \perp L$) if $KL = LK = 0$.

\item
We call a linear operator $U \in \cL(A,B)$ an \emph{isometry} if it preserves the inner product, namely if $\ip{U v}{U w}_B = \ip{v}{w}_A$ for all $v, w \in \cH_A$. 
This holds if $U^{\dag} U = \id_A$.

\item
An isometry is an example of a \emph{contraction}, i.e.\ an operator $L \in \cL(A,B)$ satisfying $\| L \| \leq 1$. The set of all such contractions is denoted $\cLsub(A,B)$. Here the bullet `$\bullet$' in the subscript of $\cLsub(A,B)$ simply illustrates that we restrict $\cL(A,B)$ to the unit ball for the norm $\| \cdot \|$.

\end{itemize}

For any $L \in \cL(A)$, we denote by $L^{-1}$ its Moore-Penrose \emph{generalized inverse} or pseudoinverse~\cite{penrose55} (which always exists in finite dimensions). In particular, the generalized inverse satisfies $L L^{-1} L = L$ and $L^{-1} L L^{-1} = L^{-1}$.
If $L = L^{\dag}$, the generalized inverse is just the usual inverse evaluated on the operator's support.

%%%%

\subsubsection*{Bras, Kets and Orthonormal Bases}
\label{sc:braket}

We use the \emph{bra-ket notation} throughout this book. For any vector ${v}_A \in \cH_A$, we use its \emph{ket}, denoted $\ketn{v}_A$, to describe the embedding
\begin{align}
  \ket{v}_A :\ \bbC \to \cH_A  ,\quad  \alpha \mapsto \alpha {v}_A \,. \label{eq:ket}
\end{align}
Similarly, we use its \emph{bra}, denoted $\bran{v}_A$, to describe the functional
\begin{align}
  \bra{v}_A :\ \cH_A \to \bbC ,\quad {w}_A \mapsto \ip{{v}}{{w}}_A \,. \label{eq:bra}
\end{align}

It is natural to view kets as linear operators from $\bbC$ to $\cH_A$ and bras as linear operators from $\cH_A$ to $\bbC$.
The above definitions then imply that
\begin{align}
  \ket{L {v}}_A = L \ket{v}_A, \quad \bra{L v}_A = \bra{v}_A L^{\dag}, \quad  \textnormal{and} \quad \bra{v}_A = \ket{v}_A^{\dag} \,.
\end{align}
Moreover, the inner product can equivalently be written as $\ip{{w}}{L {v}}_B = \bran{w}_B\, L \ketn{v}_A$. Conjugate symmetry of the inner product then corresponds to the relation 
\begin{align}
  \overline{\bran{w}_B L \ketn{v}}_A = \bran{v}_A L^{\dag} \ketn{w}_B \,.
\end{align}
As a further example, we note that $\ketn{v}_A$ is an isometry if and only if $\braketn{v}{v}_A = 1$.

In the following we will work exclusively with linear operators (including bras and kets) and we will not use the underlying vectors (the elements of the Hilbert space) or the inner product of the Hilbert space anymore. 
%With a slight abuse of notation, we write $\ketn{v}_A \in \cH_A$ to introduce a ket without referring to the vector.

We now restrict our attention to the space $\cL(A) := \cL(A,A)$ of bounded linear operators acting on $\cH_A$. An operator $U \in \cL(A)$ is \emph{unitary} if $U$ and $U^{\dag}$ are isometries.
An \emph{orthonormal basis} (ONB) of the system $A$ (or the Hilbert space $\cH_A$) is a set of vectors $\{ e_x \}_x$, with $e_x \in \cH_A$, such that
\begin{align}
 \braket{e_x}{e_y}_A = \delta_{x,y} := \begin{cases} 1 & x = y \\ 0 & x \neq y \end{cases}  \quad \textnormal{and} \quad \sum_x \proj{e_x}_A = \id_A \,. \label{eq:id}
\end{align}
We denote the dimension of $\cH_A$ by $d_A$ if it is finite and note that the index $x$ ranges over $d_A$ distinct values. For general separable Hilbert spaces $x$ ranges over any countable set. (We do not usually specify such index sets explicitly.) Various ONBs exist and are related by unitary operators: if $\{ e_x \}_x$ is an ONB then $\{ U e_x \}_x$ is too, and, furthermore, given two ONBs there always exists a unitary operator mapping one basis to the other, and vice versa.

\subsubsection*{Positive Semi-Definite Operators}
\label{sc:self-adj}
\label{sc:psd}

A special role is played by operators that are self-adjoint and positive semi-definite. 
We call an operator $H \in \cL(A)$ \emph{self-adjoint} if it satisfies $H = H^{\dag}$, and the set of all self-adjoint operators in $\cL(A)$ is denoted $\cLdag(A)$.
Such self-adjoint operators have a spectral decomposition, 
\begin{align}
  H = \sum_x \lambda_x \proj{e_x}
\end{align}
where $\{ \lambda_x \}_x \subset \mathbb{R}$ are called \emph{eigenvalues} and $\{ \ket{e_x} \}_x$ is an orthonormal basis with \emph{eigenvectors} $\ket{e_x}$. The set $\{ \lambda_x \}_x$ is also called the \emph{spectrum} of $H$, and it is unique.

Finally we introduce the set $\cP(A)$ of \emph{positive semi-definite} operators in $\cL(A)$. An operator $M \in \cL(A)$ is positive semi-definite if and only if $M = L^\dag L$ for some $L \in \cL(A)$, so in particular such operators are self-adjoint and have non-negative eigenvalues. 
Let us summarize some important concepts and notation concerning self-adjoint and positive semi-definite operators here.
\begin{itemize}
%\item We use $\cPsub(A)$ to denote contractions in $\cP(A)$.

\item We call $P \in \cP(A)$ a \emph{projector} if it satisfies $P^2 = P$, i.e.\ if it has only eigenvalues $0$ and $1$. The identity $\id_A$ is a projector.

\item 
For any $K, L \in \cLdag(A)$, we write $K \geq L$ if $K - L \in \cP(A)$. Thus, the relation `$\geq$' constitutes a partial order on $\cL(A)$.

\item For any $G, H \in \cLdag(A)$, we use $\{ G \geq H \}$ to denote the projector onto the subspace corresponding to non-negative eigenvalues of $G-H$. Analogously, $\{ G < H\} = \id - \{ G \geq H \}$ denotes the projector onto the subspace corresponding to negative eigenvalues of $G - H$.
\end{itemize}

%%%%

\subsubsection*{Matrix Representation and Transpose}
\label{sc:matrix}

Linear operators in $\cL(A,B)$ can be conveniently represented as matrices in $\bbC^{d_A} \times \bbC^{d_B}$.
Namely for any $L \in \cL(A,B)$, we can write
\begin{align}
  L = \sum_{x,y} \projn{f_y}_B\, L \projn{e_x}_A = \sum_{x,y}\, \bran{f_y} L \ketn{e_x} \cdot \ketn{f_y}\!\bran{e_x} ,
\end{align}
where $\{ e_x \}_x$ is an ONB of $A$ and $\{ f_y \}_y$ an ONB of B. This decomposes $L$ into elementary operators $\ketn{f_y}\!\bran{e_x} \in \cLsub(A,B)$ and the matrix with entries $[L]_{yx} = \bran{f_y} L \ketn{e_x}$.

Moreover, there always exists a choice of the two bases such that the resulting matrix is diagonal. For such a choice of bases, we find the \emph{singular value decomposition} $L = \sum_{x} s_x \ketn{f_x}\!\bran{e_x}$, where $\{s_x\}_x$ with $s_x \geq 0$ are called the {singular values} of $L$. In particular, for self-adjoint operators, we can choose $\ketn{f_x} = \ketn{e_x}$ and recover the eigenvalue decomposition with 
$s_x = |\lambda_x|$.

The \emph{transpose} of $L$ with regards to the bases $\{ e_x \}$ and $\{ f_y \}$ is defined as
\begin{align}
  L^T := \sum_{x,y}\, \bran{f_y} L \ketn{e_x} \cdot \ketn{e_x}\!\bran{f_y} , \quad L^T \in \cL(B,A) \label{eq:transp}
\end{align}
Importantly, in contrast to the adjoint, the transpose is only defined with regards to a particular basis.
Also contrast~\eqref{eq:transp} with the matrix representation of $L^{\dag}$,
\begin{align}
  L^{\dag} = \sum_{x,y}\, \big(\bran{f_y} L \ketn{e_x}\big)^{\dag} \cdot \ketn{e_x}\!\bran{f_y} = \sum_{x,y}\, \bran{e_x} L^{\dagger} \ketn{f_y} \cdot \ketn{e_x}\!\bran{f_y} = \overline{L}^{T} \,.
\end{align}
Here, $\overline{L}$ denotes the complex conjugate, which is also basis dependent.

%%%%%%%%%

%It is worth verifying this fundamental property of the trace.
%\begin{petit}
%\begin{proof}
%Let $\{ \ket{x}_A \}_x$ be a basis of $A$ and let $\{ \ket{y}_B \}_y$ be a basis of $B$. Then, using the representation of the identity in~\eqref{eq:id} twice, we find
%\begin{align}
%  \tr ( K L ) &= \sum_x \bra{x} K L \ket{x}_A = \sum_{x} \bra{x}_A K \Big( \sum_y \proj{y}_B \Big) L \ket{x}_A \\
%%  &= \sum_{x,y} \bra{x} K \proj{y} L \ket{x} \\
%%  & = \sum_{x,y} \bra{y} L \proj{x} K \ket{y} \\
%  & = \sum_{y} \bra{y}_B L \Big( \sum_x \proj{x}_A \Big) K \ket{y}_B = \sum_y \bra{y} L K \ket{y}_B = \tr ( L K ) \,.
%\end{align}
%where we changed the order of the sums and the scalars $\bra{x}K\ket{y}$ and $\bra{y}L\ket{x}$.
%\qed
%\end{proof}
%\end{petit}

%%%%

\subsection{Events and Measures}
\label{sc:events}

We are now ready to attach physical meaning to the concepts introduced in the previous section, and apply them to physical systems carrying quantum information.

\begin{svgraybox}
Observable \textbf{events} on a quantum system $A$ correspond to operators in the unit ball of $\cP(A)$, namely the set
\begin{align}
  \cPsub(A) := \{ M \in \cL(A) :\ 0 \leq M \leq \id \} \,.
\end{align}
(The bullet `$\bullet$' indicates that we restrict to the unit ball of the norm $\|\cdot\|$.)
\end{svgraybox}

Two events $M, N \in \cPsub(A)$ are called \emph{exclusive} if $M + N$ is an event in $\cPsub(A)$ as well. In this case, we call $M + N$ the \emph{union} of the events $M$ and $N$. A complete set of mutually exclusive events that sum up to the identity is called a \emph{positive operator valued measure} (POVM). More generally, 
for any measurable space $(\mathcal{X}, \Sigma$) with $\Sigma$ a $\sigma$-algebra, a POVM is a function
\begin{align}
  O_A : \Sigma \to \cPsub(A)  \quad \textrm{with} \quad O_A(\mathcal{X}) = \id_A
\end{align}
that is $\sigma$-additive, meaning that $O_A( \bigcup_i \mathcal{X}_i ) = \sum_i O_A(\mathcal{X}_i)$ for mutually disjoint subsets $\mathcal{X}_i \subset \mathcal{X}$. This definition is too general for our purposes here, and we will restrict our attention to the case where $\mathcal{X}$ is discrete and $\Sigma$ the power set of $\mathcal{X}$. In that case the POVM is fully determined if we associate mutually exclusive events to each $x \in \mathcal{X}$.

\begin{svgraybox}
A function $x \mapsto M_A(x)$ with $M_A(x) \in \cPsub(A)$, $\sum_x M_A(x) = \id_A$ is called a \textbf{positive operator valued measure} (POVM) on $A$. 
\end{svgraybox}
We assume that $x$ ranges over a countable set for this definition, and we will in fact not discuss measurements with continuous outcomes in this book. We call $x \mapsto M_A(x)$ a \emph{projective} measure if all $M_A(x)$ are projectors, and we call it \emph{rank-one} if all $M_A(x)$ have rank one.

\subsubsection*{Structure of Classical Systems}

Classical systems have the distinguishing property that all events commute. 

To model a classical system $X$ in our quantum framework, we restrict $\cPsub(X)$ to a set of events that commute. These
are diagonalized by a common ONB, which we call the \emph{classical basis} of~$X$. For simplicity, the classical basis is denoted $\{ x \}_x$ and the corresponding kets are $\ket{x}_X$. (To avoid confusion, we will call the index $y$ or $z$ instead of $x$ if the systems $Y$ and $Z$ are considered instead.)

\begin{svgraybox}
Every $M \in \cPsub(X)$ on a classical system can be written as
\begin{align}
  M = \sum_x M(x) \proj{x}_X = \bigoplus_x M(x) , \quad \textrm{where} \quad 0 \leq M(x) \leq 1 \,.\label{eq:class-event}
\end{align}
\vspace{-0.5cm}
\end{svgraybox}
Instead of writing down the basis projectors, $\proj{x}$, we sometimes employ the direct sum notation to illustrate the block-diagonal structure of such operators.
In the following, whenever we introduce a classical event $M$ on $X$ we also implicitly introduce the function $M(x)$, and {vice versa}.

This definition of ``classical'' events still goes beyond the usual classical formalism of discrete probability theory. In the usual formalism, $M$ represents a subset of the sample space (an element of its $\sigma$-algebra), and thus corresponds to a projector in our language, with $M(x) \in \{0, 1\}$ indicating if $x$ is in the set. Our formalism, in contrast, allows to model probabilistic events, i.e.\ the event $M$ occurs at most with probability $M(x) \in [0,1]$ even if the state is deterministically $x$.\footnote{This generalization is quite useful as it, for example, allows us to see the optimal (probabilistic) Neyman-Pearson test as an event.}

%%%%%%%%%%%%

\section{Functionals and States}
\label{sc:func}

States of a physical system are functionals on the set of bounded linear operators that map events to 
the probability that the respective event occurs. Continuous linear functionals can be represented as trace-class operators, which leads us to density operators for quantum and classical systems.

\subsection{Trace and Trace-Class Operators}
\label{sc:trace}

The most fundamental linear functional is the \emph{trace}. For any orthonormal basis $\{ e_x \}_x$ of $A$, we define the {trace} over $A$ as 
\begin{align}
  \tr_A(\cdot):\, \cL(A) \to \bbC, \quad L \mapsto \sum_x \bra{e_x} L \ket{e_x}_A \,. \label{eq:trace}
\end{align}
%We will see soon that this definition in fact does not depend on the choice of basis.
Note that $\tr(L)$ is finite if $d_A < \infty$ or more generally if $L$ is \emph{trace-class}.
The trace is cyclic, namely we have 
\begin{align}
\tr_A( K L ) = \tr_B( L K ) 
\end{align}
for any two operators $L \in \cL(A,B)$, $K \in \cL(B,A)$ when $K L$ and $L K$ are trace-class. Thus, in particular, for any $L \in \cL(A)$, we have $\tr_A(L ) = \tr_B ( U L U^{\dagger} )$ for any \emph{isometry} $U \in \cL(A, B)$, which shows that the particular choice of basis used for the definition of the trace in~\eqref{eq:trace} is irrelevant. Finally, we have $\tr(L^{\dag}) = \overline{\tr(L)}$.

\subsubsection*{Trace-Class Operators}
\label{sc:funct}

Using the trace, continuous linear functionals can be conveniently represented as elements of the dual Banach space of $\cL(A)$, namely the space of linear operators on $\cH_A$ with bounded \emph{trace norm}.
\begin{svgraybox}
The \textbf{trace norm} on $\cL(A)$ is defined as
\begin{align}
  \|\cdot\|_*\,:\, \quad \xi \mapsto \tr |\xi| = \tr \left( \sqrt{\xi^{\dag} \xi} \right) \,.
\end{align}
Operators $\xi \in \cL(A)$ with $\| \xi \|_* < \infty$ are called \textbf{trace-class operators}.
\end{svgraybox}
We denote the subspace of $\cL(A)$ consisting of trace-class operators by $\cF(A)$ and
we use lower-case Greek letters to denote elements of $\cF(A)$.
In infinite dimensions $\cF(A)$ is a proper subspace of $\cL(A)$. In finite dimensions $\cL(A)$ and $\cF(A)$ coincide, but we will use this convention to distinguish between linear operators and linear operators representing functionals nonetheless. 

For every trace-class operator $\xi \in \cF(A)$, we define the functional $F_\xi(L) := \ip{\xi}{L}$
using the sesquilinear form
\begin{align}
  \ip{\cdot}{\cdot}:\ \cF(A) \times \cL(A) \to \mathbb{C}, \quad (\xi, L) \mapsto \tr( \xi^{\dag} L) \,. \label{eq:hsip}
\end{align}
This form is continuous in both $\cL(A)$ and $\cF(A)$ with regards to the respective norms on these spaces, which is a direct consequence of H\"older's inequality $| \tr( \xi^{\dag} L ) | \leq \|\xi\|_* \cdot \| L \|$.\footnote{Note also that the norms $\|\cdot\|$ and $\|\cdot\|_*$ are dual with regards to this form, namely we have
\begin{align}
  \| \xi \|_* = \sup \big\{ \left| \ip{\xi}{L} \right| :\, L \in \cLsub(A) \big\} \,.
\end{align}
The trace norm is thus sometimes also called the \emph{dual norm}.}
In finite dimensions it is also tempting to view $\cL(A) = \cF(A)$ as a Hilbert space with $\ipn{\cdot}{\cdot}$ as its inner product, the \emph{Hilbert-Schmidt inner product}. 
Finally, \emph{positive functionals} map $\cP(A)$ onto the positive reals. Since $\tr(\omega M) \geq 0$ for all $M \geq 0$ if and only if $\omega \geq 0$, we find that positive functionals correspond to positive semi-definite operators in $\cF(A)$, and we denote these by~$\cS(A)$. 

%These functionals are also called \emph{states}.
%  \item We use $\cSsub(A)$ to denote the unit ball of the norm $\| \cdot \|_{*}$ in $\cS(A)$, i.e.\ the set $\cSsub(A) := \{ \omega \in \cS(A) : \tr(\omega) \leq 1 \}$.
%  \item Similarly, we use $\cSnorm(A)$ to denote the unit sphere of the norm $\| \cdot \|_{*}$ in $\cS(A)$, i.e.\ the set
%  $\cSnorm(A) := \{ \omega \in \cS(A) : \tr(\omega) = 1 \}$.

%%%%%

\subsection{States and Density Operators}
\label{sc:states}

A \emph{state} of a physical system $A$ is a functional that maps events $M \in \cPsub(A)$ to the respective probability that $M$ is observed. We want the probability of the union of two mutually exclusive events to be additive, and thus such functionals must be linear. Furthermore, we require them to be continuous with regards to small perturbations of the events. Finally, they ought to map events into the interval $[0,1]$, hence they must also be positive and normalized.

Based on the discussion in the previous section, we can conveniently parametrize all functionals corresponding to states as follows. We define the set of
\emph{sub-normalized density operators} as trace-class operators in the unit ball, 
\begin{align}
  \cSsub(A) := \{ \rho_A \in \cF(A) :\ \rho_A \geq 0 \ \land \ \tr(\rho_A) \leq 1 \} \,.
\end{align}
Here the bullet `$\bullet$' refers to the unit ball in the norm $\| \cdot \|_*$. (This norm simply corresponds to the trace for positive semi-definite operators.)
\begin{svgraybox}
For any operator $\rho_A \in \cSsub(A)$, we define the functional
\begin{align}
  \Pr_{\rho}(\cdot) :\ \cPsub(A) \to [0, 1], \quad M \mapsto \ip{\rho_A}{M} = \tr ( \rho_A M ),
\end{align}
which maps events to the probability that the event occurs.
\end{svgraybox}

This is an expression of Born's rule, and often taken as an axiom of quantum mechanics. Here it is just a natural way to map events to probabilities. We call such operators $\rho_A$ density operators. 

It is often prudent to further require that the union of all events in a POVM, namely the event $\id$, has probability $1$. This leads us to normalized density operators:
\begin{svgraybox}
Quantum {states} are represented as \textbf{normalized density operators} in
\begin{align}
  \cSnorm(A) := \{ \rho_A \in \cF(A) :\ \rho_A \geq 0 \ \land \ \tr(\rho_A) = 1 \} \,,
\end{align}
(The circle `$\circ$' indicates that we restrict to the unit sphere of the norm $\|\cdot\|_*$.)
\end{svgraybox}

In the following we will use the expressions state and density operator interchangeably. We also use the set $\cS$ which contains all positive semi-definite operators, if there is no need for normalization.

States form a convex set, and a state is called \emph{pure} if it is extremal, i.e.\ if it cannot be written as a nontrivial convex combination of two distinct states. Otherwise, it is called \emph{mixed}.
The fully mixed state (in finite dimensions) is denoted $\pi_A := \id_A / d_A$. Pure states are represented by density operators with rank one, and can be written as $\phi_A = \projn{\phi}_A$ for some $\phi \in \cH_A$. With a slight abuse of nomenclature, we often call the corresponding ket, $\ketn{\phi}_A$, a state.

\subsubsection*{Probability Mass Functions}

The structure of density operators simplifies considerably for classical systems. We are interested in evaluating the probabilities for events of the form~\eqref{eq:class-event}. 
Hence, for any $\rho_X \in \cSnorm(X)$, we find
\begin{align}
  \Pr_{\rho}(M) = \tr(\rho_X M) = \sum_x M(x) \bra{x} \rho_X \ket{x}_X = \sum_x M(x) \rho(x) ,
\end{align}
where we defined $\rho_X(x) = \bra{x} \rho_X \ket{x}_X$. We thus see that it suffices to consider states of the following form:
\begin{svgraybox}
   States $\rho_X \in \cSnorm(X)$ on a classical system $X$ have the form
\begin{align}
  \rho_X = \sum_x \rho(x) \proj{x}_X, \quad \textrm{where} \quad \rho(x) \geq 0, \quad \sum_x \rho(x) = 1 \,. \label{eq:prob}
\end{align}
  where $\rho(x)$ is called a \textbf{probability mass function}.
\end{svgraybox}

Moreover, if $\rho_X \in \cSsub(X)$ is a sub-normalized density operator, we require that $\sum_x \rho(x) \leq 1$ instead of the equality. Again, whenever we introduce a density operator $\rho_X$ on $X$, we implicitly also introduce the function $\rho(x)$, and {vice versa}.

%%%%%%%%%%

\section{Multi-Partite Systems}
\label{sc:joint}

A joint system $AB$ is modeled using bounded linear operators on a tensor product of Hilbert spaces, $\cH_{AB} := \cH_A \otimes \cH_B$.
The respective set of bounded linear operators is denoted~$\cL(AB)$ and the events on the joint systems are thus the elements of $\cPsub(AB)$. 
Analogously, all the other sets of operators defined in the previous sections are defined analogously for the joint system. 

%%%%%

\subsection{Tensor Product Spaces}
\label{sc:tensor}

For every $v \in \cH_{AB}$ on the joint system $AB$, there exist two ONBs, $\{ e_x \}_x$ on $A$ and $\{ f_y \}_y$ on $B$, as well as a unique set of positive reals, $\{ \lambda_x \}_x$, such that we can write
\begin{align}
  \ket{v}_{AB} = \sum_x \sqrt{\lambda_x}\, \ket{e_x}_A \otimes \ket{f_x}_B \,. \label{eq:schmidt}
\end{align}
This is called the \emph{Schmidt decomposition} of $v$. The convention to use a square root is motivated by the fact that the sequence $\{ \sqrt{\lambda_x} \}_x$ is square summable, i.e.\ $\sum_x \lambda_x < \infty$.
Note also that $\{ e_x \otimes f_y \}_{x,y}$ can be extended to an ONB on the joint system $AB$. 

%%%%

\subsubsection*{Embedding Linear Operators}

We embed the bounded linear operators $\cL(A)$ into $\cL(AB)$ by taking a tensor product with the identity on $B$. We often omit to write this identity explicitly and instead use subscripts to indicate on which system an operator acts. For example, for any $L_A \in \cL(A)$ and $\ketn{v}_{AB} \in \cH_{AB}$ as in~\eqref{eq:schmidt}, we write
\begin{align}
  L_A \ket{v}_{AB} = L_A \otimes \id_B \ket{v}_{AB} = \sum_x \sqrt{\lambda_x}\ L_A \ket{e_x}_A \otimes \ket{f_x}_B
\end{align}
Clearly, $\| L_A \otimes \id_B \| = \| L_A \|$, and in fact, more generally for all $L_A \in \cL(A)$ and $L_B \in \cL(B)$, we have 
\begin{align}
 \| L_A \otimes L_B \| = \| L_A \| \cdot \| L_B \| \,.
\end{align}

We say that two operators $K,L \in \cL(A)$ \emph{commute} if $[K,L] := KL - LK = 0$.
Clearly, elements of $\cL(A)$ and $\cL(B)$ mutually commute as operators in $\cL(AB)$, i.e.\ for all $L_A \in \cL(A)$, $K_B \in \cL(B)$, we have
$[ L_A \otimes \id_B, \id_A \otimes K_B ] = 0$.

Finally, every linear operator $L_{AB} \in \cL(AB)$ has a decomposition
\begin{align}
  L_{AB} = \sum_k L_{A}^k \otimes L_{B}^k, \quad \textrm{where} \quad L_A^{k} \in \cL(A),\, L_B^k \in \cL(B) \label{eq:lin-decomp}
\end{align}
Similarly, every self-adjoint operator $L_{AB} \in \cLdag(AB)$ decomposes in the same way but now $L_A^{k} \in \cLdag(A)$ and $L_B^k \in \cLdag(B)$ can be chosen self-adjoint as well. However, crucially, it is not always possible to decompose a positive semi-definite operator into products of positive semi-definite operators in this way.

\subsubsection*{Representing Traces of Matrix Products Using Tensor Spaces}
\label{sc:mirror}

Let us next consider trace terms of the form $\tr_A(K_A L_A)$ where $K_A, L_A \in \cL(A)$ are general linear operators and $\cH_A$ is finite-dimensional.
It is often convenient to represent such traces as follows. 

First, we introduce an auxiliary system $A'$ such that $\cH_A$ and $\cH_{A'}$ are isomorphic (i.e.\ they have the same dimension). Furthermore, we fix a pair of bases $\{ \ket{e_x}_A \}_x$ of $A$ and $\{ \ket{e_x}_{A'} \}_x$ of $A'$. (We can use the same index set here since these spaces are isomorphic.) Clearly every linear operator on $A$ has a natural embedding into $A'$ given by this isomorphism.
Using these bases, we further define a rank one operator $\Psi \in \cS(AA')$ in its Schmidt decomposition as
\begin{align}
  \ket{\Psi}_{AA'} = \sum_x \ket{x}_A \otimes \ket{x}_{A'} \,.
\end{align}
(Note that this state has norm $ \| \Psi \|_{*} = d_A$, which is why this discussion is restricted to finite dimensions.)
Using the matrix representation of the transpose in~\eqref{eq:transp}, we now observe that $L_A \otimes \id_{A'} \ket{\Psi}_{AA'} = \id_A \otimes L_{A'}^T \ket{\Psi}_{AA'}$ and, therefore,
\begin{align}
\tr(K_A L_A) &= \bra{\Psi} K_A L_{A} \ket{\Psi} = \bra{\Psi}_{AA'} K_A \otimes L_{A'}^T \ket{\Psi}_{AA'} \,.
\end{align}
We will encounter this representation many times and keep $\Psi$ thus reserved for this purpose, 
without going through the construction explicitly every time.\footnote{Note that $\Psi$ is an (unnormalized) maximally entangled state, usually denoted $\psi$.}

%%%%%

\subsubsection*{Marginals of Functionals}
\label{sc:marginal}

Given a bipartite system $AB$ that consists of two sets of operators $\cL(A)$ and $\cL(B)$, we now want to specify how a trace-class operator $\xi_{AB} \in \cF(AB)$ acts on $\cL(A)$.
For any $L_A \in \cL(A)$, we have
\begin{align}
  F_{\xi_{AB}}(L_A) = \ip{\xi_{AB}}{L_A \otimes \id_B} = \tr\big(\xi_{AB}^{\dag}\, L_A \otimes \id_B\big) = \tr_A\big( \tr_B\big(\xi_{AB}^{\dag}\big)\, L_A\big) ,
\end{align}
where we simply used that $\tr_{AB}(\cdot) = \tr_{A}(\tr_B(\cdot))$ where $\tr_B$ as defined in~\eqref{eq:trace} naturally embeds as a linear map from $\cF(AB)$ into $\cF(A)$, i.e.\
\begin{align}
  \tr_B (X_{AB} ) = \sum_x  \big( \id_A \otimes  \bra{e_x}_B \big) X_{AB} \big( \id_A \otimes  \ket{e_x}_B \big) \,.
\end{align}
This is also called the \emph{partial trace} and will be discussed further in the context of completely bounded maps in Section~\ref{sc:cpm}. 

The above discussion allows us to define the \emph{marginal} on $A$ of the trace-class operator $\xi_{AB} \in \cF(AB)$ as follows:
\begin{align}
  \xi_{A} := \tr_B \big(\xi_{AB}\big) \quad \textrm{such that} \quad F_{\xi_{AB}}(L_A) = F_{\xi_A}(L_A) = \ip{\xi_A}{L_A} \,.
\end{align}
We usually do not introduce marginals explicitly. For example, if we introduce a trace-class operator $\xi_{AB}$ then its marginals $\xi_A$ and $\xi_B$ are implicitly defined as well.

%%%%%%

\subsection{Separable States and Entanglement}
\label{sc:entangle}

The occurrence of entangled states on two or more quantum systems is one of the most intriguing features of the formalism of quantum mechanics. 

\begin{svgraybox}
We call a positive operator $M_{AB} \in \cP(AB)$ of a joint quantum system $AB$ \textbf{separable} if it can be 
written in the form
\begin{align}
  M_{AB} = \sum_{k \in \mathcal{K}} L_A(k) \otimes K_B(k), \quad \textrm{where} \quad L_A(k) \in \cP(A),\ 
  K_B(k) \in \cP(B)  \label{eq:sep} \,,
\end{align}
for some index set $\mathcal{K}$.
Otherwise, it is called \textbf{entangled}. 
\end{svgraybox}

The prime example of an entangled state is the \emph{maximally
entangled} state. For two quantum systems $A$ and $B$ of finite dimension, a maximally entangled state is a state of the form
\begin{align}
  \ket{\psi}_{AB} = \frac{1}{\sqrt{d}} \sum_x \ket{e_x}_A \otimes \ket{f_x}_B , \quad d = \min \{ d_A, d_B \}
\end{align}
where $\{ e_x \}_x$ is an ONB of $A$ and $\{ f_x \}_x$ is an ONB of $B$. 

This state
cannot be written in the form~\eqref{eq:sep} as the following argument, due to Peres~\cite{peres96} and 
Horodecki~\cite{horodecki96}, shows.
%\begin{petit}
Consider the operation $(\cdot)^{T_B}$ of taking a \emph{partial transpose} on the system $B$ with regards to to $\{ f_x \}_x$ on $B$. 
%(The trace will be discussed in Section~\ref{sc:matrix}.) .... already defined now
Applied to separable states of the from~\eqref{eq:sep}, this always results in a state, i.e.\
\begin{align}
  \rho_{AB}^{T_B} = \sum_k \sigma_A(k) \otimes \big( \tau_B(k) \big)^{T_B} \geq 0 \,.
\end{align}
is positive semi-definite.
Applied to $\psi_{AB}$, however, we get
\begin{align}
  \psi_{AB}^{T_B} = \frac{1}{d} \sum_{x,x'} \ketn{e_x}\!\bran{e_{x'}} \otimes \big( \ketn{f_x}\!\bran{f_{x'}} \big)^{T_B}
    = \frac{1}{d} \sum_{x,x'} \ketn{e_x}\!\bran{e_{x'}} \otimes \ketn{f_{x'}}\!\bran{f_x} \,.
\end{align}
This operator is not positive semi-definite. For example, we have
\begin{align}
  \bracketb{\phi}{\psi_{AB}^{T_B}}{\phi} = -\frac{2}{d}, \quad \textrm{where} \quad
    \ket{\phi} = \ket{e_1} \otimes \ket{e_2} - \ket{e_2} \otimes \ket{e_1} \,.
\end{align}
%\end{petit}

Generally, we have seen that a bipartite state is separable only if it remains positive semi-definite under the partial transpose. 
The converse is not true in general.

%%%%%

\subsection{Purification}

Consider any state $\rho_{AB} \in \cS(AB)$, and its marginals $\rho_A$ and $\rho_B$. Then we say that $\rho_{AB}$ is an \emph{extension} of $\rho_A$ and $\rho_B$. Moreover, if $\rho_{AB}$ is pure, we call it a \emph{purification} of $\rho_A$ and $\rho_B$.
Moreover, we can always construct a purification of a given state $\rho_A \in \cS(A)$. Let us say that $\rho_A$ has eigenvalue decomposition
\begin{align}
  \rho_A = \sum_x \lambda_x \proj{e_x}_A \,, \quad \textrm{then the state} \quad \ket{\rho}_{AA'} = 
    \sum_x \sqrt{\lambda_x} \ket{e_x}_A \otimes \ket{e_x}_{A'}
\end{align} 
is a \emph{purification} of $\rho_A$. Here, $A'$ is an auxiliary system of the same dimension as $A$ and $\{ \ket{e_x}_{A'} \}_x$ is any ONB of $A'$. Clearly, $\tr_{A'}(\rho_{AA'}) = \rho_A$.

%%%%%%

\subsection{Classical-Quantum Systems}
\label{sc:cq}

An important special case are joint systems where one part consists of a classical system. Events $M \in \cPsub(XA)$ on such joint systems can be decomposed as
\begin{align}
  M_{XA} = \sum_x \proj{x}_X \otimes M_A(x) = \bigoplus_x M_A(x), \quad \textrm{where} \quad M_A(x) \in \cPsub(A) \,.
\end{align}

Moreover, we call states of such systems \emph{classical-quantum} states. For example, consistent with our notation for classical systems in~\eqref{eq:prob}, a state $\rho_{XA} \in \cSsub(XA)$ can be decomposed as 
\begin{align}
  \rho_{XA} = \sum_x \proj{x}_X \otimes \rho_A(x), \quad \textrm{where} \quad \rho_A(x) \geq 0, \quad \sum_x \tr\big(\rho_A(x)\big) \leq 1 \,.
\end{align}
Clearly, $\rho_A(x) \in \cSsub(A)$ is a sub-normalized density operator on $A$. Furthermore, comparing with~\eqref{eq:sep}, it is evident that such states are always separable.

If $\rho_{XA} \in \cSnorm(XA)$, it is sometimes more convenient to instead further decompose 
\begin{align}
\rho_A(x) = \rho(x) \hat{\rho}_A(x) ,
\end{align}
where $\rho(x)$ is a probability mass function and $\hat{\rho}_A(x) \in \cSnorm(A)$ normalized as well.

%%%%%%%%%%%%%%%

\section{Functions on Positive Operators}
\label{sc:functions}

Besides the inverse, we often need to lift other continuous real-valued functions to positive semi-definite operators. For any continuous function $f : \bbR_+ \setminus \{0\} \to \bbR$ and $M \in \cP(A)$, we use the convention 
\begin{align}
  f(M) = \sum_{x : \lambda_x  \neq 0} f(\lambda_x) \proj{e_x} \,.
\end{align}
if the resulting operator is bounded (e.g.\ if the spectrum of $M$ is compact).
That is, as for the generalized inverse, we simply ignore the kernel of $M$.\footnote{This convention is very useful to keep the presentation in the following chapters concise, but some care is required. If $\lim_{\eps \to 0} f(\eps) \neq 0$, then $M \mapsto f(M)$ is not necessarily continuous even if $f$ is continuous on its support.} 
By definition, we thus have
$f(U M U^{\dag}) = U f(M) U^{\dag}$ for any unitary $U$. Moreover, we have
\begin{align}
  L f( L^{\dag} L ) = f( L L^{\dag} ) L , \label{eq:polar-trick}
\end{align}
which can be verified using the \emph{polar decomposition}, stating that we can always write $L = U |L|$ for some unitary operator $U$. An important example is the \emph{logarithm}, defined as $\log M = \sum_{x : \lambda_x \neq 0} \log \lambda_x\, \proj{e_x}$.

Let us in the following restrict our attention to the finite-dimensional case. Notably, trace functionals of the form $M \mapsto \tr(f(M))$ inherit continuity, monotonicity, concavity and convexity from $f$ (see, e.g.,~\cite{carlen09}).
For example, for any monotonically increasing continuous function $f$, we have 
\begin{align}
  \tr(f(M)) \leq \tr(f(N)) \qquad \textrm{for all} \quad M, N \in \cP(A) \quad \textrm{with}\quad M \leq N \,. 
  \label{eq:trace-mono}
\end{align}

%%%

\subsubsection*{Operator Monotone and Concave Functions}
%\label{sc:func-op-mono}

Here we discuss classes of functions that, when lifted to positive semi-definite operators, retain their defining properties.
A function $f: \mathbb{R}_+ \to \mathbb{R}$ is called \emph{operator monotone} if
\begin{align}
  M \leq N \implies f(M) \leq f(N) \quad \textrm{for all} \quad M, N \geq 0 \,.
\end{align}
If $f$ is operator monotone then $-f$ is \emph{operator anti-monotone}.
Furthermore, $f$ is called \emph{operator convex} if
\begin{align}
  \lambda f(M) + (1-\lambda) f(N) \geq f\big( \lambda M + (1-\lambda) N \big) \quad \textrm{for all} \quad M, N \geq 0\,
\end{align}
and $\lambda \in [0, 1]$. If this holds with the inequality reversed, then the function is called \emph{operator concave}. These definitions naturally extend to functions $f: (0,\infty) \to \mathbb{R}$, where we consequently choose $M, N > 0$.

There exists a rich theory concerning such functions and their properties (see, for example, Bhatia's book~\cite{bhatia97}), but we will only mention a few prominent examples in Table~\ref{tb:monotone} that will be of use later. 

\begin{table}
\begin{center}
  \begin{tabular}{c@{\hspace{0.1cm}}|@{\hspace{0.1cm}}c@{\hspace{0.1cm}}|@{\hspace{0.1cm}}c@{\hspace{0.1cm}}|@{\hspace{0.1cm}}c@{\hspace{0.1cm}}|@{\hspace{0.1cm}}c@{\hspace{0.1cm}}|@{\hspace{0.1cm}}c}
    \textbf{function} & \textbf{range} & \textbf{op. monotone} & \textbf{op. anti-monotone} & \textbf{op. convex} & \textbf{op. concave} \\
    \hline
    \vspace{-9pt} & & \\ % increases spacing after the horizontal line
%    $t \mapsto t$ & yes & yes & yes \\
    $\sqrt{t}$ & $[0,\infty)$ & yes & no & no & yes \\
    $t^2$ & $[0,\infty)$ & no & no & yes & no \\
    $\frac{1}{t}$ & $(0,\infty)$ & no & yes & yes & no \\
    $t^{\alpha}$ & & $\alpha \in [0,1]$ & $\alpha \in [-1, 0)$ & $\alpha \in [-1, 0) \cup [1, 2]$ & $\alpha \in (0,1]$ \\
    $\log t$ & $(0,\infty)$ & yes & no & no & yes \\
    $t \log t$ & $[0,\infty)$ & no & no & yes & no \\
    \hline
  \end{tabular}
  \end{center}
  \caption[Examples of Operator Monotone, Concave and Convex Functions]{\textbf{Examples of Operator Monotone, Concave and Convex Functions.} Note in particular that $t^{\alpha}$ is neither operator monotone, convex nor concave for $\alpha < -1$ and $\alpha > 2$.}
  \label{tb:monotone}
\end{table}

We say that a two-parameter function is \emph{jointly concave} (\emph{jointly convex}) if it is concave (convex) when we take convex combinations of input tuples.
Lieb~\cite{lieb73a} and Ando~\cite{ando79} established the following extremely powerful result. The map
\begin{align}
  \cP(A) \times \cP(B) \to \cP(AB), \quad (M_A, N_B) \mapsto f \big(M_A \otimes N_B^{-1} \big) M_A \otimes \id_B
\end{align}
is jointly convex on strictly positive operators if $f: (0, \infty) \to \mathbb{R}$ is operator monotone. This is \emph{Ando's convexity theorem}~\cite{ando79}. 
In particular, we find that the functional
\begin{align}
  (M_A, N_B) \mapsto \bra{\Psi} K \cdot \big(M_A \otimes N_{B'}^{-T}\big)^{\alpha-1} M_A \cdot K^{\dagger} \ket{\Psi}_{BB'} 
  = \tr_A ( M_A^{\alpha} K^{\dagger} N_B^{1-\alpha} K ) \label{eq:lieb-ando}
\end{align}
for any $K \in \cL(A,B)$ is jointly concave for $\alpha \in (0,1)$ and jointly convex for $\alpha \in (1, 2)$.
The former is known as Lieb's concavity theorem. Since this will be used extensively, we include a derivation of this particular result in Appendix~\ref{app:analysis}.

%%%%%%%%%%%%

\section{Quantum Channels}
\label{sc:channel}

Quantum channels are used to model the time evolution of physical systems. There are two equivalent ways to model a quantum channel, and we will see that they are intimately related. In the Schr\"odinger picture, the events are fixed and the state of a system is time dependent. Consequently, we model evolutions as quantum channels acting on the space of density operators. In the Heisenberg picture, the observable events are time dependent and the state of a system is fixed, and we thus model evolutions as adjoint quantum channels acting on events.

\subsection{Completely Bounded Maps}
\label{sc:cbm}

Here, we introduce linear maps between bounded linear operators on different systems, and their adjoints, which map between functionals on different systems. 
For later convenience, we use calligraphic letters to denote the latter maps, for example $\sE$ and $\sF$ and use the adjoint notation for maps between bounded linear operators.
The action of a linear map on an operator in a tensor space is well-defined by linearity via the decomposition in~\eqref{eq:lin-decomp}, and as for linear operators, we usually omit to make this embedding explicit. 

The set of \emph{completely bounded} (CB) linear maps from $\cL(A)$ to $\cL(B)$ is denoted by $\cb(A,B)$. Completely bounded maps $\sE^{\dag} \in \cb(A,B)$ have the defining property that for any operator $L_{AC} \in \cL(AC)$ and any auxiliary system $C$, we have $\| \sE^{\dag}(L_{AC}) \| < \infty$.\footnote{It is noteworthy that the weaker condition that the map be bounded, i.e.\ $\| \sE^{\dag}(L_A) \| < \infty$, is not sufficient here and in particular does not imply that the map is completely bounded. In contrast, bounded linear operators in $\cL(A)$ are in fact also completely bounded in the above sense.}
We then define the linear map $\sE$ from $\cF(A)$ to $\cF(B)$ as the \emph{adjoint map} for some $\sE^{\dag} \in \cb(B,A)$ via the sesquilinear form. Namely, $\sE$ is defined as the
 {unique} linear map satisfying
\begin{align}
  \ipn{\sE(\xi)}{L} = \ipn{ \xi}{ \sE^{\dag}(L) } \qquad \textnormal{for all} 
  \qquad \xi \in \cF(A),\ L \in \cL(B) \,.
\end{align}
Clearly, $\sE$ maps $\cF(A)$ into $\cF(B)$. Moreover, for any $\xi_{AC}$ in $\cF(AC)$, we have
\begin{align}
  \| \sE(\xi_{AC}) \|_* = \sup \left\{ \left| \ip{\xi_{AC}}{\sE^{\dag}(L_{BC})} \right| :\, L_{BC} \in \cLsub(BC) \right\} < \infty \,.
\end{align}
So these maps are in fact completely bounded in the trace norm and we collect them in the set $\cb_*(A,B)$. Again, in finite dimensions $\cb(A,B)$ and $\cb_*(A,B)$ coincide.

%%%

\subsection{Quantum Channels}
\label{sc:cpm}

Physical channels necessarily map positive functionals onto positive functionals. 
A map $\sE \in \cb_*(A,B)$ is called \emph{completely positive} ($\cp$) if it maps $\cS(AC)$ to $\cS(BC)$ for any auxiliary system $C$, namely if
  \begin{align}
    \ipn{ \sE(\omega_{AC}) }{ M_{BC} } \geq 0  \quad \textnormal{for all} \quad \omega \in \cS(AC),\ M \in \cP(BC) \,.
   \end{align}
A map $\sE$ is $\cp$ if and only if $\sE^{\dag}$ is $\cp$, in the respective sense. The set of all $\cp$ maps from $\cF(A)$ to $\cF(B)$ is denoted $\cp(A,B)$.

Physical channels in the Schr\"odinger picture are modeled by completely positive trace-preserving maps, or quantum channels.
\begin{svgraybox}
  A \textbf{quantum channel} is a map $\sE \in \cp(A,B)$ that is \textbf{trace-preserving}, namely a map that satisfies
\begin{align}
  \tr(\sE(\xi)) = \tr(\xi) \quad \textrm{for all} \quad \xi \in \cF(A) \,.
\end{align}
\vspace{-0.6cm}
\end{svgraybox}
 Naturally, such maps take states to states, more precisely, they map $\cSnorm(A)$ to $\cSnorm(B)$ and $\cSsub(A)$ to $\cSsub(B)$.
The corresponding adjoint quantum channel $\sE^{\dag}$ from $\cL(B)$ to $\cL(A)$ in the Heisenberg picture is a completely positive and \emph{unital} map, namely it satisfies $\sE^{\dag}(\id_A) = \id_B$. 
In fact, a map $\sE$ is trace-preserving if and only if $\sE^{\dag}$ is {unital}. Unital maps take $\cPsub(B)$ to $\cPsub(A)$ and thus map events to events.
Clearly, 
\begin{align}
  \Pr_{\sE(\rho)} ( M ) = \ip{\sE(\rho)}{M} = \ip{\rho}{\sE^{\dag}(M)} = \Pr_{\rho} \big( \sE^{\dag}(M) \big) \,.
\end{align}
Let us summarize some further notation:
\begin{itemize}
  \item We denote the set of all
completely positive trace-preserving ($\cptp$) maps from $\cF(A)$ to $\cF(B)$ by
$\cptp(A,B)$.
\item The set of all CP unital maps from $\cL(A)$ to $\cL(B)$ is denoted $\cpu(A,B)$.
\item Finally, a map $\sE \in \cp(A,B)$ is called \emph{trace-non-increasing} if $\tr(\sE(\omega)) \leq \tr(\omega)$ for all $\omega \in \cS(A)$.
A $\cp$ map is trace-non-increasing if and only if its adjoint is \emph{sub-unital}, i.e.\ it satisfies $\sE^{\dag}(\id_B) \leq \id_A$.
\end{itemize}

 %Similarly, we use
%$\cptni$ to denote CP trace-non-increasing maps and $\cpsu$ to denote CP sub-unital maps.

\subsubsection*{Some Examples of Channels}

The simplest example of such a $\cp$ map is the \emph{conjugation} with an
operator $L \in \cL(A,B)$, that is the map $\sL: \xi \mapsto L \xi L^{\dag}$.
We will often use the following basic property of completely positive maps. Let
$\sE \in \cp(A,B)$, then
\begin{align}
  \xi \geq \zeta \implies \sE(\xi) \geq \sE(\zeta) \quad \textnormal{for all} \quad \xi,\zeta \in \cF(A) \,.
\end{align}

As a consequence, we take note of the following property of positive semi-definite operators. 
For any $M \in \cP(A)$, $\xi \in \cS(A)$, we have 
\begin{align}
  \tr(\xi M) = \tr\big(\sqrt{M} \xi \sqrt{M}\big) \geq 0 \,, 
\end{align}
where the last inequality follows from the fact that the conjugation with $\sqrt{M}$ is a completely 
positive map. In particular, if $L, K \in \cL(A)$ satisfy $L \geq K$, we find $\tr(\xi L) 
\geq \tr(\xi K)$.

An instructive example is the embedding map $L_A \mapsto L_A \otimes \id_B$, which is completely bounded, CP and unital. Its adjoint map is the $\cptp$ map $\tr_B$, the partial trace, as we have seen in Section~\ref{sc:marginal}.
Finally, for a POVM $x \mapsto M_A(x)$, we consider the measurement map $\sM \in \cptp(A,X)$ given by
\begin{align}
   \sM : \rho_A \mapsto \sum_x \proj{x} \, \tr(\rho_A M_A(x)) \,. 
\end{align}
This maps a quantum system into a classical system with a state corresponding to the probability mass function $\rho(x) = \tr(\rho_A M_A(x))$ that arises from Born's rule. 
If the events $\{ M_A(x) \}_x$ are rank-one projectors, then this map is also unital.

\subsection{Pinching and Dephasing Channels}
\label{sc:pinching}

Pinching maps (or channels) constitute a particularly important class of quantum channels that we will use extensively in our technical derivations.
A \emph{pinching map} is a channel of the form $\sP: L \mapsto \sum_x P_x\, L\, P_x$ where $\{ P_x \}_x$, $x \in [m]$ are orthogonal projectors that sum up to the identity. Such maps are $\cptp$, unital and equal to their own adjoints. Alternatively, we can see them as \emph{dephasing} operations that remove off-diagonal blocks of a matrix. They have two equivalent representations:
\begin{align}
  \sP(L) = \sum_{x \in [m]} P_x L P_x = \frac{1}{m} \sum_{y \in [m]} U_y L U_y^{\dag} , \quad \textrm{where} \quad U_y = \sum_{x \in [m]} e^{\frac{2\pi i y x}{m}} P_x \label{eq:pinch-dephase}
\end{align}
are unitary operators. Note also that $U_m = \id$.

For any self-adjoint operator $H \in \cLdag(A)$ with eigenvalue decomposition $H = \sum_x \lambda_x \proj{e_x}$, we define 
the set $\spec(H) = \{ \lambda_x \}_x$ and its cardinality, $|\spec(H)|$, is the number of distinct eigenvalues of $H$. For each $\lambda \in \spec(H)$, we also define $P_{\lambda} = \sum_{x: \lambda_x = \lambda} \proj{e_x}$ such that $H = \sum_{\lambda} \lambda P_{\lambda}$ is its \emph{spectral decomposition}.
Then, the \emph{pinching map} for this spectral decomposition is denoted
\begin{align}
  \sP_{H} : L \mapsto \sum_{\lambda \in \spec(H)} P_{\lambda}\, L\, P_{\lambda} \,.
\end{align}
Clearly, $\sP_H(H) = H$, $\sP_H(L)$ commutes with $H$, and $\tr(\sP_H(L) H) = \tr(L H)$.

%  alternative proof - not necessary
%
%Finally, for any two vectors $\ket{\phi}$ and $\ket{\theta}$, we have
%\begin{align}
%   &\bracket{\vartheta}{ \Big( \big|\spec(H)\big| \cdot \sP_H(\proj{\phi}) - \proj{\phi} \Big)}{\vartheta} = \\
%   &\qquad \big|\spec(H)\big| \cdot \sum_{\lambda} \big| \bracket{\vartheta}{P_{\lambda}}{\phi} \big|^2  
%     - \bigg| \sum_{\lambda} \bracket{\vartheta}{P_{\lambda}}{\phi} \bigg|^2  \geq 0
%\end{align}
%by the Cauchy-Schwartz inequality. This means that the operator inequality $\proj{\phi} \leq \big|\spec(H)\big| \cdot \sP_H(\proj{\phi})$ holds, and 
%using the spectral decomposition of $M \in \cP(A)$, 

For any $M \in \cP(A)$, using the second expression in~\eqref{eq:pinch-dephase} and the fact that $U_x M U_x^{\dagger} \geq 0$, we immediately arrive at
\begin{align}
  \sP_{H}(M) = \frac1{|\spec(H)|} \sum_{y \in [m]} U_y M U_y^{\dag} \geq \frac1{|\spec(H)|} M \,  .   \label{eq:pinching}
\end{align}
This is Hayashi's \emph{pinching inequality}~\cite{hayashi02b}.

%%%

Finally, if $f$ is operator concave, then for every pinching $\sP$, we have
\begin{align}
  f(\sP(M)) = f\bigg( \frac{1}{m} \sum_{x \in [m]} U_x M U_x^{\dag} \bigg) 
    &\geq \frac{1}{m} \sum_{x \in [m]} f\big( U_x M U_x^{\dag} \big) \\
    &= \frac{1}{m} \sum_{x \in [m]} U_x f(M) U_x^{\dag} = \sP(f(M)) \,.
\end{align}
This is a special case of the \emph{operator Jensen inequality} established by Hansen and Pedersen~\cite{hansen03}. For all $H \in \cLdag(A)$, 
every operator concave function $f$ defined on the spectrum of $H$, and all unital maps $\sE \in \cpu(A,B)$, we have
\begin{align}
  f(\sE(H) ) \geq \sE ( f(H) ) \label{eq:jensen} \,.
\end{align}
This relation remains through when $\sE$ is sub-unital instead of unital, and $f(0) \geq 0$.

\subsection{Channel Representations}
\label{sc:channel-rep}

The following representations for trace non-increasing and trace preserving $\cp$ maps are of crucial importance in quantum information theory. 

\subsubsection*{Kraus Operators}

Every $\cp$ map can be represented as a sum of conjugations of the input~\cite{kraus69,kraus70}.
More precisely, $\sE \in \cp(A,B)$ if and only if
  there exists a set of linear operators $\{ E_k \}_k$, $E_k \in \cL(A,B)$ such that
  \begin{align}
    \sE(\xi) = \sum_k E_k \xi\, {E_k}^{\!\!\!\dag} \quad \textnormal{for all} \quad \xi \in \cF(A) \,.
  \end{align}
  Furthermore, such a channel is trace-preserving if and 
  only if $\sum_k {E_k}^{\!\!\!\dag} E_k = \id$, and trace-non-increasing if 
  and only if $\sum_k {E_k}^{\!\!\!\dag} E_k \leq \id$. 
The operators $\{ E_k \}$ are called \emph{Kraus operators}. Moreover, the adjoint $\sE^{\dag}$ of
$\sE$ is completely positive and has Kraus operators $\{ {E_k}^{\!\!\!\dag} \}$ since
\begin{align}
  \tr\big(\xi \sE^{\dag}(L) \big) = \tr \big(\sE(\xi) L \big) = \tr \Big( \xi \sum_k {E_k}^{\!\!\!\dag} L\, E_k \Big) \,.
\end{align}

\subsubsection*{Stinespring Dilation}

Moreover, every $\cp$ map can be decomposed into its \emph{Stinespring dilation}~\cite{stinespring54}.
That is, $\sE \in \cp(A,B)$ if and only if
  there exists a system $C$ and an operator $L \in \cL(A, BC)$ such that
  \begin{align}
    \sE(\xi) = \tr_C(L \xi L^{\dag}) \quad \textnormal{for all} \quad \xi \in \cF(A) \,.
  \end{align}
  Moreover, if $\sE$ is trace-preserving then $L = U$, where $U \in \cLsub(A,BC)$ is an 
  isometry. If $\sE$ is trace-non-increasing, then $L = P U$ is an isometry followed by a 
  projection $P \in \cPsub(C)$.

%%%

\subsubsection*{Choi-Jamiolkowski Isomorphism}
\label{sc:choi-jami}

For finite-dimensional Hilbert spaces, the \emph{Choi-Jamiolkowski isomorphism}~\cite{jamiliolkowski72} 
between bounded linear maps from $A$ to $B$ and linear functionals on $A'B$ is given by
\begin{align}
  \Choi:\ \cF(\cF(A),\cF(B)) \to \cF(A'B), \quad \sE \mapsto \choi_{A'B}^{\sE} = \sE \big( \proj{\Psi}_{A'A} \big) , \label{eq:choi-jami}
\end{align}
where the state $\choi_{A'B}^{\sE}$ is called the Choi-Jamiolkowski state of $\sE$. 
The inverse operation, $\Choi^{-1}$, maps linear functionals to bounded linear maps
\begin{align}
  \Choi^{-1}:\ \choi_{A'B} \mapsto \Big\{ \sE^{\gamma}: \rho_A \mapsto \tr_{A'}\big(\choi_{A'B} (\id_{B} \otimes \rho_{A'}^T)\big) \Big\} \label{eq:choi-jami-inv},
\end{align}
where the transpose is taken with regards to the Schmidt basis of $\Psi$.

There are various relations between properties of bounded linear maps and properties of the corresponding Choi-Jamiolkowski functionals, for example:
\begin{align}
  \sE \textrm{ is completely positive} \quad &\iff \quad \choi_{A'B}^{\sE} \geq 0 , \\
  \sE \textrm{ is trace-preserving} \quad&\iff \quad\tr_{B}(\choi_{A'B}^{\sE}) = \id_{A'} \,, \\
  \sE \textrm{ is unital} \quad&\iff \quad \tr_{A'}(\choi_{A'B}^{\sE}) = \id_{B} \,.
\end{align}

%%%%%%%%%%%%%%%%%%%%%%%%%%

\section{Background and Further Reading}
\label{sc:prelim-bg}
 
Nielsen and Chuang's book~\cite{nielsen00} offers a good introduction to the quantum formalism. Hayashi's~\cite{hayashi06} and Wilde's~\cite{wildebook13} books both also carefully treat the concepts relevant for quantum information theory in finite dimensions.
Finally, Holevo's recent book~\cite{holevo12} offers a comprehensive mathematical introduction to quantum information processing in finite and infinite dimensions. 

Operator monotone functions and other aspects of matrix analysis are covered in Bhatia's books~\cite{bhatia97,bhatia07}, and the book by Hiai and Petz~\cite{hiaipetz14}.

%%%%%%%%%%

%\bibliographystyle{spphys}
%\bibliography{library}

%!TEX root = book.tex

\chapter{Norms and Metrics}
\label{ch:metrics} 
% use \chaptermark{}
% to alter or adjust the chapter heading in the running head

\abstract*{In this chapter we equip the space of quantum states with some additional structure by discussing various norms and metrics for quantum states. We discuss Schatten norms and an important variational characterization of these norms, amongst other properties. We go on to discuss the trace norm on positive semi-definite operators and the trace distance associated with it. Uhlmann's fidelity for quantum states is treated next, as well as the purified distance, a useful metric based on the fidelity.}

In this chapter we equip the space of quantum states with some additional structure by discussing various norms and metrics for quantum states. We discuss Schatten norms and an important variational characterization of these norms, amongst other properties. We go on to discuss the trace norm on positive semi-definite operators and the trace distance associated with it. Uhlmann's fidelity for quantum states is treated next, as well as the purified distance, a useful metric based on the fidelity.

Particular emphasis is given to sub-normalized quantum states, and the above quantities are generalized to meaningfully include them.
This will be essential for the definition of the smooth entropies in Chapter~\ref{ch:calc}.

%%%%%%%%%

\section{Norms for Operators and Quantum States}
\label{sc:norms}

We restrict ourselves to finite-dimensional Hilbert spaces hereafter.
We start by giving a formal definition for unitarily invariant norms on linear operators. 
An example of such a norm is the operator norm $\| \cdot \|$ of the previous chapter.

\begin{definition}
\label{def:norm}
\begin{svgraybox}
  A \textbf{norm} for linear operators is a map $\norm{\cdot}: \cL(A) \to [0,\infty)$ which satisfies the following properties, for any $L, K \in \cL(A)$.
\begin{description}
  \item[Positive-definiteness:] $\norm{L} \geq 0$ with equality if and only if $L = 0$.
  \item[Absolute scalability:] $\norm{a L} = |\alpha| \cdot \norm{L}$ for all $a \in \mathbb{C}$.
  \item[Subadditivity:] $\norm{L + K} \leq \norm{L} + \norm{K}$.
\end{description}
  A norm $\unorm{\cdot}$ is called a \textbf{unitarily invariant norm} if it further satisfies
\begin{description}
  \item[Unitary invariance:] $\unorm{ U L V^{\dag} } = \unorm{ L }$ for any isometries $U, V \in \cL(A,B)$.
\end{description}
\end{svgraybox}
\end{definition}

We reserve the notation $\unorm{\cdot}$ for unitarily invariant norms.
Combining subadditivity and scalability, we note that norms are convex:
\begin{align}
  \norm{ \lambda L + (1-\lambda) K } \leq \lambda \norm{ L } + (1-\lambda) \norm{K} \qquad \textrm{for all} \quad \lambda \in [0,1].
\end{align}

\subsection{Schatten Norms}

The \emph{singular values} of a general linear operator $L \in \cL(A)$ are the eigenvalues of its \emph{modulus}, the positive semi-definite operator $|L| := \sqrt{L^{\dag}L}$.
The Schatten $p$-norm of~$L$ is then simply defined as the $p$-norm of its singular values. 

\begin{definition}
\begin{svgraybox}
   For any $L \in \cL(A)$, we define the \textbf{Schatten $p$-norm} of $L$ as
   \begin{align}
      \| L \|_{p} :=  \Big( \tr \big( |L|^p  \big) \Big)^{\frac{1}{p}} \qquad \textrm{for} \quad p \geq 1 \,. \label{eq:def-p-norm}
   \end{align}
   \vspace{-0.5cm}
\end{svgraybox}
\end{definition}

We extend this definition to all $p > 0$, but note that in this case $\| L \|_{p}$ is not a norm.
In particular, $\|L\|_p$ for $p \in [0, 1)$ does not satisfy the subadditivity inequality in Definition~\ref{def:norm}. The operator norm is recovered in the limit $p \to \infty$. We have
\begin{align}
  \| L \|_{\infty} = \| L \| ,\qquad 
  \| L \|_2 = \sqrt{\tr (L^{\dagger} L)} , \qquad 
  \| L \|_1 = \tr | L | = \| L \|_* \,.
\end{align}
The latter two norms are the \emph{Frobenius} or \emph{Hilbert-Schmidt} norm and the \emph{trace norm}.

The Schatten norms are unitarily invariant and subadditive. 
Using this and the representation of pinching channels in~\eqref{eq:pinch-dephase}, we find
\begin{align}
  \unorm{ \sP(L) } = \unorm{ \sum_{x \in [m]} \frac{1}{m} U_x L U_x^{\dagger} } \leq \sum_{x \in [m]} \frac{1}{m} \unorm{ U_x L U_x^{\dagger}} = \unorm{ L } \,. 
  \label{eq:pinch-norm}
\end{align}
This is called the \emph{pinching inequality} for (unitarily invariant) norms.

%%%%%

\subsubsection*{H\"older Inequalities and Variational Characterization of Norms}

Next we introduce the following powerful generalization of the \emph{H\"older and reverse H\"older inequalities} to the trace of linear operators:
\begin{lemma}
\label{lm:hoelder}
\begin{svgraybox}
  Let $L, K \in \cL(A)$, $M, N \in \cP(A)$ and $p, q \in \mathbb{R}$ such that $p > 0$ and $\frac{1}{p} + \frac{1}{q} = 1$. Then, we have
  \begin{align}
     | \tr (L K) | &\leq \tr |L K| \leq \| L \|_p \cdot \| K \|_q \quad  &&\textrm{if}  \quad p > 1 \label{eq:hoelder1} \\
     \tr (M N)  &\geq \| M \|_p \cdot \big\| N^{-1} \big\|_{-q}^{-1}  &&\textrm{if} \quad p \in (0, 1) \ \textrm{and} \ M \ll N \,. \label{eq:hoelder2}
  \end{align}
  Moreover, for every $L$ there exists a $K$ such that equality is achieved in~\eqref{eq:hoelder1}. In particular, for $M, N \in \cP(A)$, equality is achieved in all inequalities if $M^p = a N^q$ for some constant $a \geq 0$.
\end{svgraybox}
\end{lemma}

\begin{petit}
\begin{proof}
      We omit the proof of the first statement (see, e.g.,~Bhatia~\cite[Cor.~IV.2.6]{bhatia97}).
     
     For $p \in (0,1)$, let us first consider the case where $M$ and $N$ commute. Then,~\eqref{eq:hoelder1} yields
     \begin{align}
       \| M \|_p^p = \tr (M^p) = \tr ( M^p N^p N^{-p}) &\leq \| M^p N^p \|_{\frac{1}{p}} \cdot \| N^{-p} \|_{\frac{1}{1-p}} \\
       &= \big( \tr(M N) \big)^p \cdot \Big( \tr \Big(|N|^{-\frac{p}{1-p}} \Big) \Big)^{1-p} \,,
     \end{align}
     which establishes the desired statement. To generalize~\eqref{eq:hoelder2} to non-commuting operators, note that
     the commutative inequality yields
     \begin{align}
       \tr\big( M N \big) = \tr \big( \sP_{N}(M) N \big) \geq \big\| \sP_{N}(M) \big\|_p \cdot \big\| |N|^{-1} \big\|_{-q}^{-1} \,. \label{eq:hoelder-proof1}
     \end{align}
     Moreover, since $t \mapsto t^{p}$ is operator concave, the operator Jensen inequality~\eqref{eq:jensen} establishes that
     \begin{align}
       \big\| \sP_{N}(M) \big\|_p^p = \tr \big( \big( \sP_{N}(M) \big)^{p} \big) 
       \geq \tr \big( \sP_{N} \big( M^{p} \big) \big) = \tr \big( M^p \big) \,.
      \end{align}
     Substituting this into~\eqref{eq:hoelder-proof1} yields the desired statement for general $M$ and $N$. \qed
\end{proof}
\end{petit}

These H\"older inequalities are extremely useful, for example they allow us to derive various variational characterizations of Schatten norms and trace terms.
For $p > 1$, the H\"older inequality implies norm duality, namely~\cite[Sec.~IV.2]{bhatia97}
\begin{align}
  \| L \|_p = \max_{K \in \cL(A) \atop \| K \|_q \leq 1 } \left| \tr\big( L^\dag K \big) \right| \qquad \textrm{for} \quad \frac1{p} + \frac1{q} = 1, \ p, q > 1
   \,. \label{eq:Schatten-dual}
\end{align}
This is a quite useful variational characterization of the Schatten norm, which we extend to $p \in (0, 1)$ using the reverse H\"older inequality. Here we state the resulting variational formula for positive operators.
\begin{lemma}
\label{lm:hoelder-var}
Let $M \in \cP(A)$ and $p > 0$. Then, for $r = 1 - \frac{1}{p}$, we find
\begin{align}
    \| M \|_p &= \max \Big\{  \tr\big( M N^{r} \big) : N \in \cSnorm(A) \Big\} \qquad && \textrm{if } p \geq 1 \\
    \| M \|_p &= \min \Big\{  \tr\big( M N^{r} \big) : N \in \cSnorm(A)\ \land\ M \ll N \Big\}  && \textrm{if } p \in (0,1]  
    \,.
\end{align}
\end{lemma}

Furthermore, as a consequence of the H\"older inequality for $p > 1$ we find
\begin{align}
  \log \tr (M N) &\leq  \frac{1}{p} \log \tr (M^p) + \frac{1}{q} \log \tr(N^q)  \\
  &\leq  \log \Big( \frac{1}{p} \tr (M^p) + \frac{1}{q} \tr (N^q) \Big) \,, 
\end{align}
where the last inequality follows by the concavity of the logarithm. 
Hence, we have
\begin{align}
  \tr (M N) \leq  \frac{1}{p} \tr (M^p) + \frac{1}{q} \tr (N^q) \quad \textrm{with equality iff} \quad M^p = N^q  \label{eq:app-hoelder} \,,
\end{align}
which is a matrix trace version of Young's inequality.
Similarly, the reverse H\"older inequality for $p \in (0,1)$ and $M \ll N$ yields again~\eqref{eq:app-hoelder} with the inequality reversed.

%%%

\subsection{Dual Norm For States}
\label{sc:dualnorm}

We have already encountered the norm $\| \cdot \|_*$, which is the dual norm of the operator norm on linear operators.
Given the operational relation between density operators (positive functionals) and events (positive semi-definite operators), it is natural to consider the following dual norm on positive functionals:
\begin{definition}
\begin{svgraybox}
We define the \textbf{positive cone dual norm} as
\begin{align}
  \| \cdot \|_+ :\, \quad \cF(A) \to \mathbb{R}_+, \quad \omega \mapsto \max_{M \in \cPsub(A)} \left| \tr\big( \omega M \big) \right| . \label{eq:dual-norm}
\end{align}
\vspace{-0.5cm}
\end{svgraybox}
\end{definition}
Here we emphasize that the maximization in the definition of the dual norm is only over events in $\cPsub(A)$.
In fact, optimizing over operators in $\cLsub(A)$ in the above expression yields the Schatten-$1$ norm as we have seen in~\eqref{eq:Schatten-dual}. Thus, we clearly have
%\begin{align}
  $\| \xi \|_+ \leq \| \xi \|_1$.
%\end{align}

Let us verify that this is indeed a norm according to Definition~\ref{def:norm}. (However, it is not unitarily invariant.)
\begin{petit}
\begin{proof}
  From the definition it is evident that $\| \alpha \xi \|_+ = |\alpha| \cdot \| \xi \|_+$ for every scalar $\alpha \in \mathbb{C}$. Furthermore,
  the triangle inequality is a consequence of the fact that
  \begin{align}
    \|\xi + \zeta \|_+ = \max_{M \in \cPsub(A)} \left| \tr\big( (\xi + \zeta) M \big) \right| \leq \max_{M \in \cPsub(A)} \left| \tr\big( \xi M \big) \right|
    + \max_{M \in \cPsub(A)} \left| \tr\big( \zeta M \big) \right| = \| \xi \|_+ + \| \zeta\|_+.
  \end{align}
  for every $\xi,\zeta \in \cL$.
  It remains to show that $\| \xi \|_+ \geq 0$ with equality if and only if $\xi = 0$. This follows from the following lower bound on the dual norm:
  \begin{align}
    \| \xi \|_+ \geq \max_{\ket{v} :\,  \braket{v}{v} = 1} \abs{ \bra{v} \xi \ket{v} } = w(\xi) \geq 0 \quad \textrm{with equality only if $\xi = 0$} . \label{eq:dual-norm-proof1}
  \end{align}
  To arrive at~\eqref{eq:dual-norm-proof1}, we chose $M = \proj{v}$ and let $w(\cdot)$ denote the numerical radius~(see, e.g., Bhatia~\cite[Sec.~I.1]{bhatia97}). The equality condition is thus inherited from the numerical radius. \qed
\end{proof}
\end{petit}

For functionals represented by self-adjoint operators $\xi \in \cF(A)$, we can explicitly find the operator that achieves the maximum in~\eqref{eq:dual-norm} using the spectral decomposition of $\xi$. Specifically, we find that the expression is always maximized by the projector $\{\xi \geq 0\}$ or its complement $\{\xi < 0\}$, namely we want to either sum up all positive or all negative eigenvalues to maximize the absolute value.
The dual norm thus evaluates to 
\begin{align}
  \| \xi \|_+ &= \max \Big\{ \tr \big( \{\xi \geq 0\} \xi \big) ,\ -\tr \big( \{\xi < 0\} \xi \big) \Big\} \,.
\end{align}
This can be further simplified using $\max \{ a, b \} = \frac12 ({a+b} + |a-b|)$, which yields
\begin{align}
   \| \xi \|_+ &= \frac12 \tr \Big( \big( \{\xi \geq 0\} -  \{\xi < 0\} \big) \xi \Big) + \frac12 \Big| \tr \Big( \big( \{\xi \geq 0\} + \{\xi < 0\} \big) \xi \Big) \Big|  \\
   &= \frac12 \tr |\xi| + \frac12  \big| \tr(\xi) \big| = \frac12 \| \xi \|_1 + \frac12 \big| \tr(\xi) \big|  \,. \label{eq:dual-norm-simp}
\end{align}
Finally, for positive functionals this further simplifies to $\| \omega \|_+  = \| \omega \|_1 = \tr(\omega)$.% Hence, the unit ball for positive functionals in this norm simply consists of operators with trace smaller or equal to $1$, further justifying the notation for $\cSsub$. %Similarly, the unit sphere consists of operators with unit trace, further justifying the notation for $\cSnorm$.

%%%%%%%%

\section{Trace Distance}
\label{sc:gtd}

We start by introducing a straightforward generalization of the trace distance to general (not necessarily normalized) states. The definition also makes sense for general trace-class operators, so we will state the results in their most general form.

\begin{definition}
\begin{svgraybox}
  \label{df:gtd}
  For $\xi, \zeta \in \cF(A)$, we define the \textbf{generalized trace distance} between $\xi$ and $\zeta$ as $\Delta(\xi, \zeta) := \| \xi - \zeta \|_+$.
\end{svgraybox}
\end{definition}
This distance is also often called \emph{total variation distance} in the classical literature.
It is a \emph{metric} on $\cF(A)$, an immediate consequence of the fact that $\|\cdot\|_+$ is a norm. 

\begin{definition}
\begin{svgraybox}
 A \textbf{metric} is a functional $\cF(A) \times \cF(A) \to \bbR_+$ with the following properties. For any $\xi,\zeta,\kappa \in \cF(A)$,
 it satisfies
\begin{description}
  \item[Positive-definiteness:] $\Delta(\xi,\zeta) \geq 0$ with equality if and only if $\xi = \zeta$.
  \item[Symmetry:] $\Delta(\xi,\zeta) = \Delta(\zeta,\xi)$.
  \item[Triangle inequality:] $\Delta(\xi,\zeta) \leq \Delta(\xi,\kappa) + \Delta(\kappa,\zeta)$.
\end{description}
\end{svgraybox}
\end{definition}
When used with states, the generalized trace distance can be expressed in terms of the trace norm and the absolute value of the trace using~\eqref{eq:dual-norm-simp}. This yields
\begin{align}
 \Delta(\rho, \tau) 
 = \frac12 \| \rho - \tau \|_1 + \frac12 \left| \tr(\rho-\tau) \right| \,.
\end{align}
Hence the definition reduces the usual trace distance $\Delta(\rho,\tau) = \frac12 \| \rho - \tau \|_1$ in case both density operators have the same trace, for example if $\rho,\tau \in \cSnorm(A)$. More generally, for sub-normalized states in $\cSsub(A)$, we can express the generalized trace distance as
\begin{align}
\Delta(\rho, \tau) = \frac12 \| \rhoh - \tauh \|_1  = \Delta(\rhoh,\tauh) \,,
\end{align}
where $\rhoh = \rho \oplus (1-\tr(\rho))$ and $\tauh = \tau \oplus (1-\tr(\tau))$ are block-diagonal. We will use the hat notation to refer to this construction in the following.
%$\rhoh = \left( \begin{array}{cc} \rho\ & 0 \\ 0\ & 1 - \tr \rho \end{array} \right)$ and $\tauh = \left( \begin{array}{cc} \tau\ & 0 \\ 0\ & 1 - \tr \tau \end{array} \right)$ are block-diagonal. 

\bigskip

For normalized states $\rho, \tau \in \cSnorm(A)$, this definition expresses the \emph{distinguishing advantage} in binary hypothesis testing. Let us consider the task of distinguishing between two hypotheses, $\rho$ and $\tau$, with uniform prior using a single observation. For every event $M \in \cPsub(A)$, we consider the following strategy: we perform the POVM $\{M, \id - M\}$ and select $\rho$ in case we measure $M$ and $\tau$ otherwise. Optimizing over all strategies, the probability of selecting the correct state can be expressed in terms of the distinguishing advantage, $\Delta(\rho,\tau)$, as follows:
\begin{align}
  p_\textrm{corr}(\rho, \tau) &:= \max_{M \in \cPsub(A)} \left( \frac12 \tr (\rho M) + \frac12 \tr (\tau (\id-M)) \right) = \frac{1}{2} \big( 1 + \Delta(\rho, \tau) 
    \big) \label{eq:dist-adv}.
\end{align}

Like any metric based on a norm, the generalized trace distance is also jointly convex. For all $\lambda \in [0,1]$, we have
\begin{align}
  \Delta(\lambda \rho_1 + (1-\lambda) \rho_2, \lambda \tau_1 + (1-\lambda) \tau_2) \leq \lambda \Delta(\rho_1, \tau_1) + (1-\lambda) \Delta(\rho_2, \tau_2) \,.
\end{align}
Moreover, the generalized trace distance contracts when we apply a quantum channel (or any trace-non-increasing completely positive map) on both states.
\begin{proposition}
\begin{svgraybox}
  Let $\xi, \zeta \in \cF(A)$, and let $\sF \in \cptni(A,B)$ be a trace-non-increasing $\cp$ map. Then, $\Delta(\sF(\xi),\sF(\zeta)) \leq \Delta(\xi,\zeta)$.
\end{svgraybox}
\end{proposition}
\begin{petit}
\begin{proof}
Note that if $\sF \in \cp(A,B)$ is trace non-increasing, then $\sF^{\dag} \in \cp(B,A)$ is sub-unital. In particular, $\sF^{\dag}$ maps $\cPsub(B)$ into $\cPsub(A)$. Then,
\begin{align}
  \Delta(\sF(\xi),\sF(\zeta)) 
  &= \max_{M \in \cPsub(B)} \big| \tr(M \sF(\xi - \zeta)) \big| 
  = \max_{M \in \cPsub(B)} \big| \tr(\sF^{\dag}(M) (\xi - \zeta)) \big| \\
  &\leq \max_{M \in \cPsub(A)} \big| \tr(M (\xi - \zeta)) \big| = \Delta(\xi,\zeta) \,. \ 
  \end{align}
where we used the definition of the norm in~\eqref{eq:dual-norm} twice.
\qed
\end{proof}
\end{petit}
As a special case when we take the map to be a partial trace, this relation yields
\begin{align}
  \Delta(\rho_A, \tau_A) \leq \min_{\rho_{AB},\tau_{AB}} \Delta(\rho_{AB},\tau_{AB}) \label{eq:trace-purified-ineq}
\end{align}
where $\rho_{AB}$ and $\tau_{AB}$ are extensions (e.g. purifications) of $\rho_A$ and $\tau_A$, respectively. 

Can we always find two purifications such that~\eqref{eq:trace-purified-ineq} becomes an equality? To see that this is in fact not true, consider the following example. If
$\rho$ is fully mixed on a qubit and $\tau$ is pure, then, $\Delta(\rho,\tau) = \frac12$, but $\Delta(\psi, \vartheta) \geq \frac{1}{\sqrt{2}}$ for all maximally entangled states $\psi$ that purify $\rho$ and product states $\vartheta$ that purify $\tau$.

%%%%%%%%%

\section{Fidelity}
\label{sc:fid}

The last observation motivates us to look at other measures of distance between states. Uhlmann's fidelity~\cite{uhlmann85} is ubiquitous in quantum information theory and we define it here for general states.

\begin{definition}
\begin{svgraybox}
  For any $\rho, \sigma \in \cS(A)$, we define the \textbf{fidelity} of $\rho$ and $\tau$ as
  \begin{align}
    F(\rho,\tau) := \Big( \tr \big| \sqrt{\rho} \sqrt{\tau} \big| \Big)^2 \,.
  \end{align}
\vspace{-0.5cm}
\end{svgraybox}
\end{definition}

Next we will discuss a few basic properties of the fidelity, and we will provide further details when we discuss the minimal quantum R\'enyi divergence in Section~\ref{sc:rminimal}. In fact, the analysis in Section~\ref{sc:rminimal} will reveal that $(\rho,\tau) \mapsto \sqrt{F(\rho,\tau)}$ is jointly concave and non-decreasing when we apply a \cptp{} map to both states. The latter property thus also holds for the fidelity itself.

Beyond that, Uhlmann's theorem~\cite{uhlmann85} states that there always exist purifications with the same fidelity as their marginals.

\begin{theorem}
\begin{svgraybox}
\label{lm:uhlmann}
  For any states $\rho_A, \tau_A \in \cS(A)$ and any purification $\rho_{AB} \in \cS(AB)$ of $\rho_A$ with $d_B \geq d_A$, there exists a purification $\tau_{AB} \in \cS(AB)$ of $\tau_A$ such that $F(\rho_A, \tau_A) = F(\rho_{AB}, \tau_{AB})$.
\end{svgraybox}
\end{theorem}
In particular, combining this with the fact that the fidelity cannot decrease when we take a partial trace, we can write
  \begin{align}
    F(\rho_A, \tau_A) = \max_{\tau_{AB} \in \cS(AB)} F(\rho_{AB}, \tau_{AB}) = \max_{\phi_{AB},\vartheta_{AB} \in \cS(AB)} \big| \! \braket{\phi_{AB}}{\vartheta_{AB}} \! \big|^2 \,, \label{eq:uhlmann-ex}
  \end{align}
  where $\tau_{AB}$ is any extension of $\tau_A$. The latter optimization is over all purifications $\ket{\phi_{AB}}$ of $\rho_{A}$ and $\ket{\vartheta_{AB}}$ of $\tau_{A}$, respectively, and assumes that $d_B \geq d_A$.
  
Uhlmann's theorem has many immediate consequences. For example, for any linear operator $L \in \cL(A)$, we see that
\begin{align}
  F( L \rho L^{\dag}, \tau) = F( \rho, L^{\dag} \tau L )
\end{align}
by using the latter expression in~\eqref{eq:uhlmann-ex}. 

Finally, we find that the fidelity is concave in each of its arguments.
\begin{lemma}
\begin{svgraybox}
\label{lm:fid-conc}
  The functionals $\rho \mapsto F(\rho, \tau)$ and $\tau \mapsto F(\rho, \tau)$ are concave. 
\end{svgraybox}
\end{lemma}

\begin{petit}
\begin{proof}
By symmetry it suffices to show concavity of $\rho \mapsto F(\rho, \tau)$. Let $\rho_A^1, \rho_A^2 \in \cSnorm(A)$ and $\lambda \in (0, 1)$ such that $\lambda\rho_A^1 + (1-\lambda)\rho_A^2 = \rho_A$. Moreover, let $\tau_{AA'} \in \cSnorm(AA')$ be a fixed purification of $\tau_A$.
Then, due to Uhlmann's theorem there exist purifications $\rho_{AA'}^1$ and $\rho_{AA'}^2$ of $\rho_A^1$ and $\rho_A^2$, respectively, such that the following chain of inequalities holds:
\begin{align}
  \lambda F(\rho_A^1, \tau_A) + (1-\lambda) F(\rho_A^2, \tau_A) &=
  \lambda \big| \braketn{\tau_{AA'}}{\rho_{AA'}^1} \big|^2 + (1-\lambda) \big| \braketn{\tau_{AA'}}{\rho_{AA'}^2} \big|^2 \\
  &= \bran{\tau_{AA'}} \big( \lambda \projn{\rho_{AA'}^1} + 
      (1-\lambda) \projn{\rho_{AA'}^2} \big) \ketn{\tau_{AA'}} \\
  &= F\big(\tau_{AA'}, \lambda \projn{\rho_{AA'}^1} + (1-\lambda) \projn{\rho_{AA'}^2} \big) \\
  &\leq F(\tau_A, \lambda \rho_A^1 + (1-\lambda) \rho_A^2) \,.
\end{align}
The final inequality follows since the fidelity is non-decreasing when we apply a partial trace.
\qed
\end{proof}
\end{petit}

%%%%

\subsection{Generalized Fidelity}
\label{sc:gfid}

Before we commence, we define a very useful generalization of the fidelity to sub-normalized density operators, which we call the generalized fidelity.

\begin{definition}
\begin{svgraybox}
  \label{df:gf}
  For $\rho, \tau \in \cSsub(A)$, we define the \textbf{generalized fidelity} between $\rho$ and $\tau$ as
  \begin{align}
    \Fg(\rho, \tau) := \left( \tr \left| \sqrt{\rho} \sqrt{\tau} \right| + \sqrt{(1-\tr\rho)(1-\tr\tau)} \right)^2 . \label{eq:gen-fid}
  \end{align}
\vspace{-0.5cm}
\end{svgraybox}
\end{definition}

Uhlmann's theorem (Theorem~\ref{lm:uhlmann}) adapted to the generalized fidelity states that
\begin{align}
  \Fg(\rho,\tau) &= \max_{\varphi,\vartheta} \Fg(\varphi,\vartheta) = \max_{\vartheta} \Fg(\phi,\vartheta), \label{eq:uhlmann0}
  \quad \textrm{where} \quad \\
  \sqrt{\Fg(\varphi,\vartheta)} &= \left| \braket{\varphi}{\vartheta} \right| + \sqrt{(1-\tr\varphi)(1-\tr\vartheta)}, \label{eq:uhlmann}
\end{align} 
  and $\varphi$ and $\vartheta$ range over all purifications of $\rho$ and $\tau$, respectively, and $\phi$ is a fixed purification of $\rho$. Moreover, using the operators $\rhoh$ and $\tauh$ defined in the preceding section, we can write
  \begin{align}
    \Fg(\rho, \tau) &= \Fg(\rhoh,\tauh) =  \left( \tr \left| \sqrt{\rhoh} \sqrt{\tauh} \right| \right)^2 \,. \label{eq:fidelity-hat}
  \end{align}
  
%%%%%  

From this representation also follows that the square root of the generalized fidelity is jointly concave on $\cSsub(A) \times \cSsub(A)$, inheriting this property from the fidelity. Moreover, the generalized fidelity itself is concave in each of its arguments separately due to Lemma~\ref{lm:fid-conc}.
%For $\rho_1, \rho_2, \tau_1, \tau_2 \in \cSsub(A)$ and $\lambda \in [0, 1]$, we have
%\begin{align}
%  \sqrt{\Fg\big( \lambda \rho_1 + (1-\lambda) \rho_2, \lambda \tau_1 + (1-\lambda) \tau_2 \big)} \geq \lambda \sqrt{\Fg(\rho_1, \tau_1)} + (1-\lambda) \sqrt{\Fg(\rho_2, \tau_2)} \,.
%\end{align}

%%%%%

 The extension to sub-normalized states in~Definition~\ref{df:gf} is chosen diligently so that the generalized fidelity 
 is non-decreasing when we apply a quantum channel, or more generally a trace non-increasing $\cp$ map.
 \begin{proposition}
 \begin{svgraybox}
   \label{pr:gfid-monotone}
   Let $\rho,\tau \in \cSsub(A)$, and let $\sE$ be a trace non-increasing $\cp$ map. Then,
   $\Fg(\sE(\rho), \sE(\tau)) \geq \Fg(\rho,\tau)$.
 \end{svgraybox}
 \end{proposition}
 
 \begin{petit}
 \begin{proof} 
   Recall that a trace non-increasing map $\sF \in \cp(A, B)$ can be decomposed into an isometry   
  $\sU \in \cp(A, BC)$ followed by a projection $\Pi \in \cP(BC)$ and a 
  partial trace over $C$ according to the Stinespring dilation representation.
  
  Let us first restrict our attention to $\cptp$ maps $\sE$ where $\Pi = \id$.
  We write $\rho_B' = \sE[\rho_A]$ and $\tau_B' = \sE[\tau_A]$.
  From the representation of the fidelity in~\eqref{eq:uhlmann0} we can immediately deduce that 
  \begin{align}
    \Fg(\rho_A,\tau_A) &= \max_{\varphi_{AD},\vartheta_{AD}} \Fg(\varphi_{AD},\, \vartheta_{AD})
    = \max_{\varphi_{AD},\vartheta_{AD}} \Fg(\sU(\varphi_{AD}),\, \sU(\vartheta_{AD})) \\
    &\leq \max_{\varphi_{BCD}',\vartheta_{BCD}'} \Fg(\varphi_{BCD}', \vartheta_{BCD}') = \Fg(\rho_B', \tau_B') \,.
  \end{align}
  The maximizations above are restricted to purifications of $\rho_A$ and $\tau_A$, respectively.
  The sole inequality follows since $\sU(\varphi_{AD})$ and $\sU(\vartheta_{AD})$ are particular purifications of $\rho_B'$ and $\tau_B'$ in $\cSsub(BCD)$.
  
  Next, consider a projection $\Pi \in \cP(BC)$ and the $\cptp$ map $\sE$ acting as
  \begin{align}
     \sE:\, \left( \begin{array}{cc} \rho & c \\ d & t \end{array} \right)  \mapsto &\left( \begin{array}{cc} \Pi \rho \Pi & 0 \\ 0 & \tr(\Pi^{\perp} \rho) + t \end{array} \right) 
     \qquad \textrm{with} \quad \Pi^{\perp} = \id - \Pi \,.
  \end{align}
  We then have $\sE(\hat{\rho}) =  \Pi \rho \Pi \oplus (1-\tr(\Pi \rho))$ and $\sE(\hat{\tau}) =  \Pi \tau \Pi \oplus (1-\tr(\Pi \tau))$.
  Applying the inequality for $\cptp$ maps to $\sE$, we find
  \begin{align}
    \sqrt{\Fg(\rho, \tau)} &= \sqrt{\Fg\big(\hat{\rho},\hat{\tau}\big)} \\
    &\leq \left\| \sqrt{\Pi\rho\Pi} \sqrt{\Pi\tau\Pi} \right\|_1 + \sqrt{ (1 - \tr(\Pi\rho) )(1- \tr(\Pi\tau) ) },
  \end{align}
  but the latter sum is exactly the definition of $\sqrt{\Fg(\Pi\rho\Pi, \Pi\tau\Pi)}$. \qed
 \end{proof}
 \end{petit}

%%%%%
  
The main strength of the generalized fidelity compared to the trace distance lies in the following property, which tells us that the inequality in Proposition~\ref{pr:gfid-monotone} is tight if the map is a partial trace. Given two marginal states and an extension of one of these states, we can always find an extension of the other state such that the generalized fidelity is preserved by the partial trace. This is a simple corollary of Uhlmann's theorem.
\begin{corollary}
\begin{svgraybox}
  Let $\rho_{AB} \in \cSsub(AB)$ and $\tau_A \in \cSsub(A)$. Then, there exists an extension $\tau_{AB}$ such that
     $\Fg(\rho_{AB},\tau_{AB}) = \Fg(\rho_A,\tau_A)$. Moreover, if $\rho_{AB}$ is pure and $d_B \geq d_A$, then $\tau_{AB}$ can be chosen pure as well. \label{cor:extension}
\end{svgraybox}  
\end{corollary}  
  
\begin{petit}
\begin{proof}
  Clearly $\Fg(\rho_A,\tau_A) \geq \Fg(\rho_{AB},\tau_{AB})$ by Proposition~\ref{pr:gfid-monotone} for any choice of $\tau_{AB}$. Let us first treat the case where $\rho_{AB}$ is pure. Using Uhlmann's theorem in~\eqref{eq:uhlmann}, we can write
  \begin{align}
     \Fg(\rho_{A}, \tau_A) = \max_{\vartheta_{AB}} \Fg(\phi_{AB}, \vartheta_{AB}), \qquad \textrm{where} \qquad \phi_{AB} = \rho_{AB} \,.
  \end{align}
  We then take $\tau_{AB}$ to be any maximizer.
  For the general case, consider a purification $\rho_{ABC}$ of $\rho_{AB}$. Then, by the above argument there exists a state $\tau_{ABC}$ with $\Fg(\rho_{ABC}, \tau_{ABC}) = \Fg(\rho_A, \tau_A)$.
  Moreover, by Proposition~\ref{pr:gfid-monotone}, we have $\Fg(\rho_{ABC}, \tau_{ABC}) \leq \Fg(\rho_{AB}, \tau_{AB}) \leq \Fg(\rho_A, \tau_A)$. Hence, all inequalities must be equalities,
  which concludes the proof. \qed
\end{proof}
\end{petit}

%%%%%%%%%

\section{Purified Distance}
\label{sc:pd}

  The fidelity is not a metric itself, but for example the \emph{angular distance}~\cite{nielsen00} and the 
  \emph{Bures metric}~\cite{bures69} are metrics. They are respectively defined as
  \begin{align}
    A(\rho,\tau) := \arccos \sqrt{F(\rho, \tau)} \quad \textrm{and} \quad B(\rho,\tau) := \sqrt{2 \left( 1- \sqrt{F(\rho,\tau)}\right)} \,.    
   \end{align}  
   We will now discuss another metric, which we find particularly convenient since it is related to the minimal trace distance of purifications~\cite{nielsen04,rastegin02,tomamichel09}.

\begin{definition}
\begin{svgraybox}
  \label{df:pd}
  For $\rho, \tau \in \cSsub(A)$, we define the \textbf{purified distance} between $\rho$ and $\tau$ as $P(\rho, \tau) := \sqrt{1 - \Fg(\rho, \tau)}$.
\end{svgraybox}
\end{definition}

Then, for quantum states $\rho,\tau \in \cSnorm(A)$, using Uhlmann's theorem we find
\begin{align}
  P(\rho, \tau) &= \sqrt{1 - {\Fg(\rho, \tau)}} 
  = \sqrt{1 - \max_{\varphi, \vartheta} \abs{\braket{\varphi}{\vartheta}}^2} 
 % = \min_{\varphi, \vartheta} \sqrt{1 - \abs{\braket{\varphi}{\vartheta}}^2}
  =  \min_{\varphi, \vartheta} \Delta(\varphi, \vartheta) \, .
\end{align}
Here, $\ketn{\varphi}$ and $\ketn{\vartheta}$ are purifications of $\rho$ and $\tau$, respectively.

As it is defined in terms of the generalized fidelity, the purified distance inherits many of its properties. For example, for trace non-increasing $\cp$ maps $\sF$, we find
\begin{align}
  P(\sF(\rho), \sF(\tau)) \leq P(\rho,\tau) \,.
\end{align}
Moreover, the purified distance is a metric on the set of sub-normalized states.

\begin{proposition}
\begin{svgraybox}
  \label{pr:pd-metric}
  The purified distance is a metric on $\cSsub(A)$. Moreover, for any three states $\rho, \tau, \sigma \in \cSsub(A)$ such that $P(\rho, \sigma)^2 + P(\sigma, \tau)^2 \leq 1$, we can tighten the triangle inequality to
  \begin{align}
  	P(\rho, \tau) \leq P(\rho, \sigma) \sqrt{\Fg(\sigma, \tau)} + P(\sigma, \tau) \sqrt{\Fg(\rho, \sigma)} \,. \label{eq:pd/tight-triangle}
  \end{align}
\end{svgraybox}
\end{proposition}

\begin{petit}
\begin{proof}
  Let $\rho, \tau, \sigma \in \cSsub(A)$.
  The condition $P(\rho, \tau) = 0$ if and only if $\rho = \tau$ can be verified
  by inspection, and symmetry $P(\rho, \tau) = P(\tau, \rho)$ follows
  from the symmetry of the fidelity.
  If $P(\rho, \sigma) + P(\sigma, \tau) \geq 1$ the triangle inequality holds trivially. It thus remains to show~\eqref{eq:pd/tight-triangle}, which implies the triangle inequality if $P(\rho, \sigma)^2 + P(\sigma, \tau)^2 \leq 1$, and thus also if $P(\rho, \sigma) + P(\sigma, \tau) \leq 1$. 
    
  Using~\eqref{eq:fidelity-hat}, the generalized fidelities between $\rho$, $\tau$ and $\sigma$ can be
  expressed as fidelities between the corresponding extensions
  $\rhoh$, $\tauh$ and $\sigmah$. We employ the \emph{angular distance}, which can be expressed in terms of the purified 
  distance as $A(\rhoh, \tauh) = \arccos \sqrt{\Fg(\rho, \tau)} = \arcsin P(\rho, \tau)$. Eq.~\eqref{eq:pd/tight-triangle} can thus be restated as
  \begin{align}
  	\sin A(\rhoh, \tauh) &\leq \sin A(\rhoh, \sigmah) \cos A(\sigmah,\tauh) + \sin A(\sigmah,\tauh) \cos(\rhoh,\sigmah) \\
	&= \sin \left( A(\rhoh, \sigmah) +  A(\sigmah,\tauh) \right) \,,
  \end{align}
  where we employed the trigonometric addition formula. Since the sine is monotonically increasing in $[-\pi/2, \pi/2]$, this inequality follows directly from the triangle inequality for the angular distance as long as $A(\rhoh, \sigmah) +  A(\sigmah,\tauh) \leq \frac{\pi}{2}$. We thus need to verify this condition.
  
  For this purpose we note that for $x, y \in [0,1]$, the condition $\arcsin(x) + \arcsin(y) \leq \frac{\pi}{2}$ is equivalent to the condition $x^2 + y^2 \leq 1$. And thus, we see that with $x = P(\rhoh, \sigmah)$ and $y = P(\sigmah, \tauh)$, our assumption $P(\rho, \sigma)^2 + P(\sigma, \tau)^2 \leq 1$ indeed implies $A(\rhoh, \sigmah) +  A(\sigmah,\tauh) \leq \frac{\pi}{2}$.
  \qed
\end{proof}
\end{petit}

Note that the purified distance is not an intrinsic metric. Given two states $\rho$, $\tau$ with $P(\rho, \tau) \leq \eps$ it is in general not possible to find intermediate states $\sigma^\lambda$ with $P(\rho, \sigma^\lambda) = \lambda\eps$ and $P(\sigma^\lambda, \tau) = (1-\lambda) \eps$. In this sense, the above triangle inequality is not tight. It is thus sometimes
useful to employ the upper bound in~\eqref{eq:pd/tight-triangle} instead. For example, we find that 
$P(\rho, \sigma) \leq \sin(\varphi)$ and $P(\sigma, \tau) \leq \sin(\vartheta)$ implies
\begin{align}
  P(\rho, \tau) \leq \sin(\varphi+\vartheta) < \sin(\varphi) + \sin(\vartheta)
      \label{eq:pd-triangle-eps}
\end{align}
if $\varphi, \vartheta > 0$ and $\varphi + \vartheta \leq \frac{\pi}{2}$. 

The purified distance is jointly quasi-convex since it is an anti-monotone function of the square root of the generalized fidelity, which is jointly concave.
Formally, for any $\rho_1, \rho_2, \tau_1, \tau_2 \in \cSsub(A)$ and $\lambda \in [0,1]$, we have
\begin{align}
  P\big( \lambda \rho_1 + (1-\lambda) \rho_2, \lambda \tau_1 + (1-\lambda) \tau_2 \big) \leq \max_{i \in \{1, 2\}} P(\rho_i, \tau_i) \,.
\end{align}

The purified distance has simple upper and lower bounds in terms of the generalized trace distance. This results from a simple reformulation of the Fuchs--van de Graaf inequalities~\cite{fuchs99} between the trace distance and the fidelity.
\begin{lemma}
%\begin{svgraybox}
  \label{lm:pd-gtd-bounds}
  Let $\rho, \tau \in \cSsub(A)$. Then, the following inequalities hold:
  \begin{align}
    \Delta(\rho, \tau) \leq P(\rho, \tau) \leq \sqrt{2 \Delta(\rho, \tau) - \Delta(\rho, \tau)^2} \leq \sqrt{2 \Delta(\rho, \tau)} \,.
  \end{align}
%\end{svgraybox}
\end{lemma}
\begin{petit}
\begin{proof}
  We first express the quantities using the normalized density operators $\rhoh$
  and $\tauh$, i.e.\ $P(\rho,\tau) = P(\rhoh,\tauh)$ and $\Delta(\rho,\tau) = \Delta(\rhoh,\tauh)$.
  Then, the result follows from the inequalities 
  \begin{align}
   1 - \sqrt{F(\rhoh, \tauh)} \leq D(\rhoh, \tauh) \leq \sqrt{1 - F(\rhoh,  \tauh)}
  \end{align}
   between the trace distance and fidelity, which were first shown by Fuchs and van de Graaf~\cite{fuchs99}. \qed
\end{proof}
\end{petit}

\section{Background and Further Reading}

We defer to Bhatia's book~\cite[Ch.~IV]{bhatia97} for a comprehensive introduction to matrix norms. Fuchs' thesis~\cite{fuchs96} gives a useful overview over distance measures in quantum information.
The fidelity was first investigated by Uhlmann~\cite{uhlmann85} and popularized in quantum information theory by Jozsa~\cite{jozsa94} who also gave it its name. Some recent literature (most prominently Nielsen and Chuang's standard textbook~\cite{nielsen00}) defines the fidelity as $\sqrt{F(\cdot,\cdot)}$, also called the square root fidelity. Here we adopted the historical definition.

The discussion on generalized fidelity and purified distance is based on~\cite{mythesis} and~\cite{tomamichel09}.
The purified distance was initially proposed by Gilchrist \emph{et al.}~\cite{nielsen04} and Rastegin~\cite{rastegin02,rastegin06}, where it is called `sine distance'. However, in these papers the discussion is restricted to normalized states. The name `purified distance' was coined in~\cite{tomamichel09}, where the generalization to sub-normalized states was first investigated.

%%%%%%%%%%%%%%%

%!TEX root = book.tex

\chapter{Quantum R\'enyi Divergence}
\label{ch:renyi} 
% use \chaptermark{}
% to alter or adjust the chapter heading in the running head

\abstract*{Shannon entropy as well as conditional entropy and mutual information can be compactly expressed in terms of the relative entropy, or Kullback-Leibler divergence. In this sense, the divergence can be seen as a parent quantity to entropy, conditional entropy and mutual information, and many properties of the latter quantities can be derived from properties of the divergence. Similarly, we will define R\'enyi entropy, conditional entropy and mutual information in terms of a parent quantity, the R\'enyi divergence. We will see in the following chapters that this approach is very natural and leads to operationally significant measures that have powerful mathematical properties.
However, for now this observation allows us to first focus our attention on quantum generalizations of the Kullback-Leibler and R\'enyi divergence and explore their properties, which is the topic of this chapter.}

Shannon entropy as well as conditional entropy and mutual information can be compactly expressed in terms of the relative entropy, or Kullback-Leibler divergence. In this sense, the divergence can be seen as a parent quantity to entropy, conditional entropy and mutual information, and many properties of the latter quantities can be derived from properties of the divergence. Similarly, we will define R\'enyi entropy, conditional entropy and mutual information in terms of a parent quantity, the R\'enyi divergence. We will see in the following chapters that this approach is very natural and leads to operationally significant measures that have powerful mathematical properties.
This observation allows us to first focus our attention on quantum generalizations of the Kullback-Leibler and R\'enyi divergence and explore their properties, which is the topic of this chapter.

There exist various quantum generalizations of the classical R\'enyi divergence due to the non-commutative nature of quantum physics.\footnote{In fact, uncountably infinite quantum generalizations with interesting mathematical properties can easily be constructed (see, e.g.\,~\cite{audenaert13}).} Thus, it is prudent to restrict our attention to quantum generalizations that attain operational significance in quantum information theory. 
A natural application of classical R\'enyi divergence is in hypothesis testing, where error and strong converse exponents are naturally expressed in terms of the R\'enyi divergence. 
%Hypothesis testing has a natural quantum generalization.
In this chapter we focus on two variants of the quantum R\'enyi divergence that both attain operational significance in quantum hypothesis testing. Here we explore their mathematical properties, whereas their application to hypothesis testing will be reviewed in Chapter~\ref{ch:app}.

%%%%%%%

\section{Classical R\'enyi Divergence}
\label{sc:rclass}

Before we tackle quantum R\'enyi divergences, let us first recapitulate some properties of the classical R\'enyi divergence they are supposed to generalize. We formulate these properties in the quantum language, and we will later see that most of them are also satisfied by some quantum R\'enyi divergences.

%\subsubsection*{An Axiomatic Approach}
\subsection{An Axiomatic Approach}
\label{sc:raxiom}

Alfr\'ed R\'enyi, in his seminal 1961 paper~\cite{renyi61}
investigated an axiomatic approach to derive the Shannon entropy~\cite{shannon48}. He found that five natural requirements for functionals on a probability space single out the Shannon entropy, and by relaxing one of these requirements, he found a family of entropies now named after him. 

The requirements can be readily translated to the quantum language.
Here we consider general functionals $\DD(\cdot\|\cdot)$ that map a pair of operators $\rho,\sigma \in \cS(A)$ with $\rho \neq 0$, $\sigma \gg \rho$ onto the real line.
R\'enyi's six axioms naturally translate as follows:

\begin{enumerate}
\item[(I)] \textbf{Continuity:} $\DD(\rho\|\sigma)$ is continuous in $\rho,\sigma \in \cS(A)$, wherever $\rho \neq 0$ and $\sigma \gg \rho$.
\item[(II)] \textbf{Unitary invariance:} $\DD(\rho\|\sigma) = \DD(U\rho U^{\dagger}\|U \sigma U^{\dagger})$ for any unitary $U$.
\item[(III)] \textbf{Normalization:} $D(1\|\frac12) = \log(2)$.
\item[(IV)] \textbf{Order:} If $\rho \geq \sigma$, then $\DD(\rho\|\sigma) \geq 0$. And, if $\rho \leq \sigma$, then $\DD(\rho\|\sigma) \leq 0$.
\item[(V)] \textbf{Additivity:} $\DD(\rho \otimes \tau\|\sigma \otimes \omega) = \DD(\rho\|\sigma) + \DD(\tau\|\omega)$ for all $\rho,\sigma \in \cS(A)$, $\tau,\omega \in \cS(B)$ with $\rho \neq 0$, $\tau \neq 0$.
\item[(VI)] \textbf{General mean:} There exists a continuous and strictly monotonic function ${g}$ such that $\QQ(\cdot\|\cdot) := {g}(\DD(\cdot\|\cdot))$ satisfies the following.
For $\rho,\sigma \in \cS(A)$, $\tau,\omega \in \cS(B)$,
\begin{align}
  \QQ(\rho \oplus \tau \| \sigma \oplus \omega) = \frac{\tr(\rho)}{\tr(\rho+\tau)} \cdot \QQ( \rho \|\sigma) + \frac{\tr(\tau)}{\tr(\rho+\tau)} \cdot \QQ( \tau \| \omega) \, .
\end{align}
\end{enumerate}

R\'enyi~\cite{renyi61} first shows that (I)--(V) imply $\DD(\lambda\|\mu) = \log {\lambda} - \log {\mu}$ for two scalars $\lambda, \mu > 0$, a quantity that is often referred to as the \emph{log-likelihood ratio}. In fact, the axioms imply the following constraint, which will be useful later since it allows us to restrict our attention to normalized states.
\begin{enumerate}
\item[(III+)] \textbf{Normalization:} $\DD( a \rho \| b \sigma) = \DD(\rho\|\sigma) + \log a - \log b$ for $a, b > 0$.
\end{enumerate}
We also remark that invariance under unitaries~(II) is implied by a slightly stronger property, invariance under isometries. 
\begin{enumerate}
\item[(II+)] \textbf{Isometric Invariance:} $\DD(\rho \| \sigma ) = \DD\big( V \rho V^{\dag} \big\| V \sigma V^{\dag})$ for $\rho, \sigma \in \cS(A)$ and any isometry $V$ from $A$ to $B$.
\end{enumerate}

R\'enyi then considers general continuous and strictly monotonic functions to define a mean in~(VI), such that the resulting quantity is still compatible with (I)--(V). 
Under the assumption that the states $\rho_X$ and $\sigma_X$ are classical, he then establishes that Properties~(I)--(VI) are satisfied only by the \emph{Kullback-Leibler divergence}~\cite{kullback51} and the \emph{R\'enyi divergence} for $\alpha \in (0, 1) \cup (1, \infty)$, which are respectively given as 
\begin{align}
  D(\rho_X\|\sigma_X) &=  \frac{\sum_x \rho(x) (\log \rho(x) - \log \sigma(x))}{\sum_x \rho(x)} \ 
     &&\textrm{with} \quad g: t \mapsto t ,  \label{eq:klrenyi} \\ 
  D_{\alpha}(\rho_X\|\sigma_X) &= \frac{1}{\alpha-1} \log \frac{ \sum_x \rho(x)^{\alpha} \sigma(x)^{1-\alpha} }{\sum_x \rho(x) }  \ 
     &&\textrm{with} \quad g_{\alpha}: t \mapsto \exp\big( (\alpha\!-\!1) t \big) \label{eq:crenyi} \,. 
\end{align}
These quantities are well-defined if $\rho_X$ and $\sigma_X$ have full support and otherwise we use the convention 
that $0 \log 0 = 0$ and $\frac{0}{0} = 1$, which ensures that the divergences are indeed continuous whenever $\rho_X \neq 0$ and $\sigma_X \gg \rho_X$. Finally, note that both quantities diverge to $+\infty$ if the latter condition is not satisfied and $\alpha > 1$.

%\subsubsection*{Positive Definiteness and Data-Processing}
\subsection{Positive Definiteness and Data-Processing}
\label{sc:rprop}

%In the following sections, we will be concerned with generalizing these expressions to non-commuting operators\,---\,but it turns out that 
Unlike in the classical case, the above axioms do not uniquely determine a quantum generalization of these divergences. Hence, we first list some additional properties we would like a quantum generalization of the R\'enyi divergence to have. These are operationally significant, but mathematically more involved than the axioms used by R\'enyi. The classical R\'enyi divergences satisfy all these properties.

The two most significant properties from an operational point of view are positive definiteness and the data-processing inequality. First, \emph{positive definiteness} ensures that the divergence is positive for normalized states and vanishes only if both arguments are equal. This allows us to use the divergence as a measure of distinguishability in place of a metric in some cases, even though it is not symmetric and does not satisfy a triangle inequality.  
 
\begin{enumerate}
  \item[(VII)] \textbf{Positive definiteness:} If $\rho,\sigma \in \cSnorm(A)$, then $\DD(\rho\|\sigma) \geq 0$ with equality iff $\rho = \sigma$.
\end{enumerate}

The \emph{data-processing inequality} (DPI) ensures the divergence never increases when we apply a quantum channel to both states. This strengthens the interpretation of the divergence as a measure of distinguishability\,---\,the outputs of a channel are at least as hard to distinguish as the inputs.

\begin{enumerate}
  \item[(VIII)] \textbf{Data-processing inequality:} For any $\sE \in \cptp(A,B)$ and $\rho, \sigma \in \cS(A)$, we have 
  \begin{align}
    \DD(\rho\|\sigma) \geq \DD(\sE(\rho)\|\sE(\sigma)) \,. \label{eq:dpi}
  \end{align}
  %The statement remains true if $\sE$ is trace non-increasing and $\tr(\sE(\rho)) = \tr(\rho)$.
\end{enumerate}

%Note that if $\sE$ is trace non-increasing then (for example consulting the Kraus decomposition) we find a complementary map $\bar{\sE}$ such that $\rho \mapsto \sE(\rho) \oplus \bar{\sE}(\rho)$ is $\cptp$. Hence, the statement for trace non-increasing maps follows directly from~\eqref{eq:dpi} and the general mean property (VI).

Finally, the following mathematical properties will prove extremely useful. (Note that we expect that either (IXa) or (IXb) holds, but not both.)
\begin{enumerate}
  \item[(IXa)] \textbf{Joint convexity} (applies only to R\'enyi divergence with $\alpha > 1$)\textbf{:} For sets of normalized states $\{ \rho_i \}_i, \{\sigma_i \}_i \subset \cSnorm(A)$ and a probability mass function $\{ \lambda_i \}_i$ such that $\lambda_i \geq 0$ and $\sum_i \lambda_i = 1$, we have
  \begin{align}
    \sum_i \lambda_i \QQ( \rho_i \| \sigma_i) \geq \QQ\left( \sum_i \lambda_i \rho_i \middle\| \sum_i \lambda_i \sigma_i \right) . \label{eq:joint-convex}
  \end{align}
  Consequently, $(\rho,\sigma) \mapsto \DD(\rho\|\sigma)$ is jointly quasi-convex, namely
  \begin{align}
    \DD\left( \sum_i \lambda_i \rho_i \middle\| \sum_i \lambda_i \sigma_i \right) \leq \max_i \DD( \rho_i \| \sigma_i) \, .
  \end{align}
  \item[(IXb)] \textbf{Joint concavity} (applies only to R\'enyi divergence with $\alpha \leq 1$)\textbf{:} The inequality~\eqref{eq:joint-convex} holds in the opposite direction, i.e.\ $(\rho,\sigma) \mapsto \QQ(\rho\|\sigma)$ is jointly concave. Moreover, $(\rho,\sigma) \mapsto \DD(\rho\|\sigma)$ is jointly convex.
\end{enumerate}

These properties are interrelated. For example, we clearly have $\DD(\rho\|\sigma) \geq 0$ in (VII) if data-processing holds, since $\DD(\rho\|\sigma) \geq \DD(\tr(\rho) \| \tr(\sigma) ) = \DD(1\|1) = 0$. Furthermore, $\DD(\rho\|\rho) = 0$ follows from~(IV). To establish positive definiteness (VII) it in fact suffices to show 
\begin{enumerate}
  \item[(VII-)] \textbf{Definiteness:} For $\rho,\sigma \in \cSnorm$, we have $\DD(\rho\|\sigma) = 0 \implies \rho = \sigma$.
\end{enumerate} 
when (IV) and (VIII) hold. The most important connection is drawn in Proposition~\ref{pr:dp-jc} in Section~\ref{sc:rclassify}, and establishes that data-processing holds if and only if joint convexity resp.\ concavity holds (depending on the value of $\alpha$) for all quantum R\'enyi divergences.
The last property generalizes the order property (IV) as follows.

\begin{enumerate}
  \item[(X)] \textbf{Dominance:} For states $\rho, \sigma, \sigma' \in \cS(A)$ with $\sigma \leq \sigma'$, we have $\DD(\rho\|\sigma) \geq \DD(\rho\|\sigma')$.
\end{enumerate}  
 Clearly, dominance (X) and positive definiteness (VII) imply order (IV). 
 \bigskip 
  
 In the following we will show that these properties hold for the classical R\'enyi divergence, i.e.\ for the case when the states $\rho$ and $\sigma$ commute. As we have argued above (and will show in Proposition~\ref{pr:dp-jc}), to establish data-processing, it suffices to prove that the KL divergence in~\eqref{eq:klrenyi} and the classical R\'enyi divergences~\eqref{eq:crenyi} satisfy joint convexity resp.\ concavity as in~(IXa) and (IXb). For this purpose we will need the following elementary lemma:

\begin{lemma}
  \label{lm:logsum}
  If $f$ is convex on positive reals, then $F: (p,q) \mapsto q f\big(\frac{p}{q}\big)$ is jointly convex. Moreover, if $f$ is strictly convex, then $F$ is strictly convex in $p$ and in $q$.
\end{lemma}
\begin{petit}
\begin{proof}
  Let $\{ \lambda_i \}_i$, $\{ p_i \}_i$, $\{ q_i \}_i$ be positive reals such that $\sum_i \lambda_i p_i = p$ and $\sum_i \lambda_i q_i = q$. Then, employing Jensen's inequality, we find
  \begin{align}
    \sum_i \lambda_i q_i f\left(\frac{p_i}{q_i}\right) = q \sum_i \frac{\lambda_i q_i}{q} f\left(\frac{ p_i}{ q_i}\right) \geq
    q f \left( \sum_i \frac{\lambda_i q_i}{q} \frac{ p_i}{ q_i}\right) = q f\left( \frac{p}{q} \right) \,.
  \end{align}
  The second statement is evident if we fix either $p_i = p$ or $q_i = q$.
  \qed
\end{proof}
\end{petit}

This lemma is a generalization of the famous \emph{log sum inequality}, which we recover using the convex function $f: t \mapsto t \log t$. 

%We will not state these results formally since we will later see that they are corollaries of the corresponding statements for quantum states.

Let us then recall that for normalized $\rho_X, \sigma_X \in \cSnorm(X)$, we have
\begin{align}
  Q_{\alpha}(\rho_X\|\sigma_X) := g_{\alpha}\big( D_{\alpha}(\rho_X\|\sigma_X) \big) = \sum_x \sigma(x) \left( \frac{\rho(x)}{\sigma(x)} \right)^{\alpha} \,.
\end{align}
First, note that $Q_{\alpha}$ has the form of a Csisz\'ar-Morimoto $f$-divergence~\cite{csiszar63,morimoto63}, where $f_{\alpha}: t \mapsto t^{\alpha}$ is concave for $\alpha \in (0, 1)$ and convex for $\alpha > 1$. Joint convexity resp.\ concavity of $Q_{\alpha}$ is then a direct consequence of Lemma~\ref{lm:logsum}, which we apply for each summand of the sum over $x$ individually. 
By the same argument applied for $f: t \mapsto t \log t$ (i.e.\ the log sum inequality), we also find that
\begin{align}
  D(\rho_X\|\sigma_X) = \sum_x \sigma(x) f\left( \frac{\rho(x)}{\sigma(x)} \right) 
\end{align}
is jointly convex.

The R\'enyi divergences satisfy the data-processing inequality (VIII), i.e.\ $D_{\alpha}$ is contractive under application of classical channels to both arguments. This can be shown directly, but since we have established joint convexity resp.\ concavity, it also follows from (a classical adaptation of) Proposition~\ref{pr:dp-jc} below and we thus omit the proof here.

Dominance (X) is evident from the definition. It remains to show definiteness~(VII-) and thus (VII). This is a consequence of the fact that $Q$ and $Q_{\alpha}$ are strictly convex resp.\ concave in the second argument due to Lemma~\ref{lm:logsum}. Namely, let us assume for the sake of contradiction that $D(\rho_X\|\rho_X) = D(\rho_X\|\sigma_X) = 0$. Then we get that $D(\rho_X\| \frac12 \rho_X + \frac12 \sigma_X) < 0$ if $\rho_X \neq \sigma_X$, which contradicts positivity. A similar argument applies to $Q_{\alpha}$, and we are done. 

\begin{center}
\begin{svgraybox}
  The {Kullback-Leibler divergence} and the classical {R\'enyi divergence} as defined in~\eqref{eq:klrenyi} and~\eqref{eq:crenyi} satisfy Properties (I)--(X).
\end{svgraybox}
\end{center}

%%%

\subsection{Monotonicity\texorpdfstring{ in $\alpha$}{} and Limits}
\label{sc:rmono}

Due to the parametrization in terms of the parameter $\alpha$, we also find the following relation between different R\'enyi divergences.

\begin{proposition}
\label{pr:cconc}
%\begin{svgraybox}
  The function $(0,1) \cup (1,\infty) \ni \alpha \mapsto \log Q_{\alpha}(\rho_X\|\sigma_X)$ is convex for all $\rho_X, \sigma_X \in \cS(X)$ with
  $\rho_X \neq 0$ and $\sigma_X \gg \rho_X$. Moreover, it is strictly convex unless $\rho_X = a \sigma_X$ for some $a > 0$.
%\end{svgraybox}
\end{proposition}

\begin{petit}
\begin{proof}
  It is sufficient to show this property for $\rho_X, \sigma_X \in \cSnorm(X)$ due to~(III+). We may also fix the logarithm to be the natural logarithm here.
  We then evaluate the second derivative of this function, which is
  \begin{align}
    F'' = \frac{Q_{\alpha}''(\rho_X\|\sigma_X) Q_{\alpha}(\rho_X\|\sigma_X) - Q_{\alpha}'(\rho_X\|\sigma_X)^2 }{Q_{\alpha}(\rho_X\|\sigma_X)^2}
  \end{align}
  where 
  \begin{align}
    Q_{\alpha}'(\rho_X\|\sigma_X) &=  \sum_x \rho(x)^{\alpha} \sigma(x)^{1-\alpha} \big( \ln \rho(x) - \ln \sigma(x) \big) , \quad \textrm{and} \\
    Q_{\alpha}''(\rho_X\|\sigma_X) &= \sum_x \rho(x)^{\alpha} \sigma(x)^{1-\alpha} \big( \ln \rho(x) - \ln \sigma(x) \big)^2 .
  \end{align}
  Note that $P(x) = \rho(x)^{\alpha} \sigma(x)^{1-\alpha} / Q_{\alpha}(\rho_X\|\sigma_X)$ is a probability mass function.
  Using this, the above expression can be simplified to
  \begin{align}
    F'' =  \sum_x P(x) \big( \ln \rho(x) - \ln \sigma(x) \big)^2 - \bigg( \sum_x P(x) \big( \ln \rho(x) - \ln \sigma(x) \big) \bigg)^2 .
  \end{align}
  Hence, $F'' \geq 0$ by Jensen's inequality and the strict convexity of the function $t \mapsto t^2$, with equality if and only if $\rho(x) = a \sigma(x)$ for all $x$. 
\qed
\end{proof}
\end{petit}

As a corollary, we find that the R\'enyi divergences are monotone functions of $\alpha$.
\begin{corollary}
\label{co:rmono}
\begin{svgraybox}
  The function $\alpha \mapsto D_{\alpha}(\rho_X\|\sigma_X)$ is monotonically increasing. Moreover, it is strictly increasing unless $\rho_X = a \sigma_X$ for some $a > 0$.
\end{svgraybox}
\end{corollary}

\begin{petit}
\begin{proof}
    We set $Q_{\alpha} \equiv Q_{\alpha}(\rho_X\|\sigma_X)$ to simplify notation and note that $\log Q_1 = 0$. Let us assume that $\alpha > \beta > 1$ and set $\lambda = \frac{\beta-1}{\alpha-1} \in (0,1)$. Then, by convexity of $\alpha \to \log Q_{\alpha}$, we have
  \begin{align}
     \log Q_{\beta} = \log Q_{\lambda \alpha + (1-\lambda)} \leq \lambda \log Q_{\alpha} + (1-\lambda) \log Q_{1} = \frac{\beta-1}{\alpha-1} \log Q_{\alpha} \,.
  \end{align}
  This establishes that $D_{\alpha}(\rho_X\|\sigma_X) \geq D_{\beta}(\rho_X\|\sigma_X)$, as desired. The inequality is strict unless $\rho_X = a \sigma_X$, as we have seen in Proposition~\ref{pr:cconc}.
  
  For $1 > \alpha \geq \beta$, an analogous argument with $\lambda = \frac{1-\alpha}{1-\beta}$ establishes that $\log Q_{\alpha} \leq \frac{1-\alpha}{1-\beta} \log Q_{\beta}$, which again yields $D_{\alpha}(\rho_X\|\sigma_X) \geq D_{\beta}(\rho_X\|\sigma_X)$ taking into account the sign of the prefactor.
\qed
\end{proof}
\end{petit}

Since we have now established that $D_{\alpha}$ is continuous in $\alpha$ for $\alpha \in (0,1) \cup (1,\infty)$, it will be interesting to take a look at the limits as $\alpha$ approaches $0$, $1$ and $\infty$. First, a direct application of l'H\^{o}pital's rule yields
\begin{align}
  \lim_{\alpha \searrow 1} D_{\alpha}(\rho_X\|\sigma_X) = \lim_{\alpha \nearrow 1} D_{\alpha}(\rho_X\|\sigma_X) = D(\rho_X\|\sigma_X) \,.
\end{align}
So in fact the KL divergence is a limiting case of the R\'enyi divergences and we consequently define $D_1(\rho_X\|\sigma_X) := D(\rho_X\|\sigma_X)$. In the limit $\alpha \to \infty$, we find
\begin{align}
  D_{\infty}(\rho_X\|\sigma_X) := \lim_{\alpha \to \infty} D_{\alpha}(\rho_X\|\sigma_X) = \max_x \log \frac{\rho(x)}{\sigma(x)} \,,
\end{align}
which is the maximum log-likelihood ratio. We call this the \emph{max-divergence}, and note that it satisfies all the properties except the general mean property~(VI). However, the max-divergence instead satisfies 
\begin{align}
  \DD(\rho \oplus \tau \| \sigma \oplus \omega) = \max \big\{ \DD(\rho \| \sigma ),\, \DD( \tau \| \omega) \big\} \,.
\end{align}
The limit $\alpha \to 0$ is less interesting because it leads to the expression
\begin{align}
  D_0(\rho_X\|\sigma_X) := \lim_{\alpha \to 0} D_{\alpha}(\rho_X\|\sigma_X) = - \log \sum_{x: \rho(x) > 0} \sigma(x) \,,
\end{align}
which is discontinuous in $\rho_X$ and thus does not satisfy (I). Hence, we hereafter consider $D_{\alpha}$ with $\alpha > 0$ as a single continuous one-parameter family of divergences.

Monotonicity of $D_{\alpha}$ is not the only byproduct of the convexity of $\log Q_{\alpha}$. For example, we also find that
\begin{align}
  \lambda D_{1+\lambda}(\rho\|\sigma) + (1-\lambda) D_{\infty}(\rho\|\sigma) \geq D_2(\rho\|\sigma) \,. \label{eq:conv-ex}
\end{align}
for $\lambda \in [0, 1]$ and various similar relations.

%%%%%%%%%%%%

\section{Classifying Quantum R\'enyi Divergences}
\label{sc:rclassify}

Clearly, we expect suitable quantum R\'enyi divergences to have the properties discussed in the previous section.

\begin{definition}
\label{def:qrenyi}
\begin{svgraybox}
  A \textbf{quantum R\'enyi divergence} is a quantity $\DD(\cdot\|\cdot)$ that satisfies Properties~(I)--(X) in Sections~\ref{sc:raxiom}. (It either satisfies IXa or IXb.)
  
  A \textbf{family of quantum R\'enyi divergences} is a one-parameter family $\alpha \mapsto \DD_{\alpha}(\cdot\|\cdot)$ of quantum R\'enyi divergences such that
  Corollary~\ref{co:rmono} in Section~\ref{sc:rmono} holds on some open interval containing $1$.
\end{svgraybox}
\end{definition}

Before we discuss two specific families of R\'enyi divergences in Sections~\ref{sc:rminimal} and~\ref{sc:rhc}, let us first make a few observations that apply more generally to all quantum R\'enyi divergences.

\subsection{Joint Concavity and Data-Processing}

  First, the following observation relates joint convexity resp.\ concavity and data-processing for all quantum R\'enyi divergences. It establishes that for functionals satisfying (I)--(VI), these properties are equivalent.

\begin{proposition}
\begin{svgraybox}
\label{pr:dp-jc}
  Let $\DD$ be a functional satisfying (I)--(VI) and let ${g}$ and $\QQ$ be defined as in (VI). Then, the following two statements are equivalent.
  \begin{enumerate}
    \item[(1)] $\QQ$ is jointly convex (IXa) if ${g}$ is monotonically increasing, or jointly concave (IXb) if ${g}$ is monotonically decreasing.
    \item[(2)] $\DD$ satisfies the data-processing inequality (VIII).
  \end{enumerate}
\end{svgraybox}
\end{proposition}

\begin{petit}
\begin{proof}  
   First, we show $(1) \!\implies\! (2)$. Note that the axioms enforce that $\QQ$ is invariant under isometries and consulting the Stinespring dilation, it thus remains to show that the data-processing inequality is satisfied for the partial trace operation. For the case where $\QQ$ is jointly convex, we thus need to show that $\QQ(\rho_{AB}\|\sigma_{AB}) \geq \QQ(\rho_A\|\sigma_A)$ for $\rho_{AB}, \sigma_{AB} \in \cSnorm(AB)$ and $A$ and $B$ are arbitrary quantum systems. 
   
   To show this, consider a unitary basis of $\cL(B)$, for example the generalized Pauli operators $\{X_B^l Z_B^m\}_{l,m}$, where $l, m \in [d_B]$. These act on the computational basis as 
   \begin{align}
    X_B \ket{k} = \ket{k + 1 \!\!\!\!\!\mod\, d_B} \qquad \textrm{and} \qquad Z_B \ket{k} = e^{ \frac{2 \pi i k }{d_B} } \ket{k} \,.
    \end{align}
   (If we only consider classical distributions, we can set $Z_B = \id_B$.)
   Then, after collecting these operators in a set $\{ U_i = X_B^l Z_B^m \}_i$ with a single index $i = (l,m)$, a short calculation reveals that
   \begin{align}
     \sum_{i} \frac{1}{d_B^2} \big(\id_A \otimes U_i\big) \xi_{AB} \big(\id_A \otimes U_i \big)^{\dag} = \xi_A \otimes \pi_B
   \end{align}
   for any $\xi_{AB} \in \cF(AB)$. Consequently, unitary invariance and joint convexity yield
   \begin{align}
     \QQ(\rho_{AB} \| \sigma_{AB}) &= \sum_{i} \frac{1}{d_B^2} \QQ\big( U_i \rho_{AB} {U_i}^{\dag} \big\| U_i \sigma_{AB} {U_i}^{\dag} \big) \\
     &\geq \QQ \bigg( \sum_{i} \frac{1}{d_B^2} U_i \rho_{AB} {U_i}^{\dag} \bigg\| \sum_{i} \frac{1}{d_B^2} U_i \sigma_{AB} {U_i}^{\dag} \bigg) 
     = \QQ(\rho_A \otimes \pi_B \| \sigma_A \otimes \pi_B) \,.
   \end{align}
   Finally, $\QQ(\rho_A \otimes \pi_B \| \sigma_A \otimes \pi_B) = \QQ(\rho_A\|\sigma_A)$ by Properties~(IV) and (V).
   Analogously, joint concavity of $\QQ$ implies data-processing for $-\QQ$, and thus $\DD$.

   Next, we show that $(2) \!\implies\! (1)$.
      Consider $\rho, \sigma, \tau, \omega \in \cSnorm$ and $\lambda \in (0, 1)$. Then, the data-processing inequality implies that
   \begin{align}
   \DD\big( \lambda \rho + (1-\lambda) \tau \big\| \lambda \sigma + (1-\lambda)\omega \big)
   \leq \DD\big( \lambda \rho \oplus (1-\lambda) \tau \big\| \lambda \sigma \oplus (1-\lambda)\omega \big) \,.
   \end{align}
   If $g$ is monotonically increasing, we find that
  \begin{align}
     &g\big(\DD\big( \lambda \rho + (1-\lambda) \tau \big\| \lambda \sigma + (1-\lambda)\omega \big)\big)\\
   &\qquad \leq g\big(\DD\big( \lambda \rho \oplus (1-\lambda) \tau \big\| \lambda \sigma \oplus (1-\lambda)\omega\big) \big)  \\
   &\qquad = \lambda {g}\big(\DD(\lambda \rho\|\lambda \sigma)\big) + (1-\lambda) {g}\big(\DD( (1-\lambda) \tau \| (1-\lambda)\omega )\big) \\
   &\qquad = \lambda {g}\big(\DD(\rho\|\sigma)\big) + (1-\lambda) {g}\big(\DD(\tau \| \omega)\big) \,,
  \end{align}   
  where we used property~(VI) for the first equality and (V) and (IV) for the last. It follows that $\QQ(\cdot\|\cdot)$ is jointly convex. An analogous argument yields joint concavity if $g$ is decreasing.
\end{proof}
\end{petit}

\subsection{Minimal Quantum R\'enyi Divergence}
\label{sc:qmin-div}

Let us assume a quantum R\'enyi divergence $\DD_{\alpha}$ satisfies additivity (V) and the data-processing inequality (VIII). Then, for any pair of states $\rho$ and $\sigma$ and their $n$-fold products, $\rho^{\otimes n}$ and $\sigma^{\otimes n}$, we have
\begin{align}
  \DD_{\alpha}(\rho\|\sigma) = \frac{1}{n} \DD_{\alpha}(\rho^{\otimes n}\|\sigma^{\otimes n}) \geq \frac{1}{n} \DD_{\alpha}\big(\sP_{\sigma^{\otimes n}}(\rho^{\otimes n}) \big\| \sigma^{\otimes n} \big) \,,
\end{align}
where $\sP_{\sigma}(\cdot)$ is the pinching channel discussed in Section~\ref{sc:pinching} and the quantity on the right-hand side is evaluated for two commuting and hence classical states.

So, in particular, a quantum R\'enyi divergence $\DD_{\alpha}$ with property (V) and (VIII) that generalizes $D_{\alpha}$ must satisfy
\begin{align}
  \DD_{\alpha}(\rho\|\sigma) &\geq \lim_{n \to \infty} \frac{1}{n} D_{\alpha}\big(\sP_{\sigma^{\otimes n}}(\rho^{\otimes n}) \big\| \sigma^{\otimes n} \big) \label{eq:qrd-lb0}\\
  &= \frac{1}{\alpha-1} \log \tr \Big( \big( \sigma^{\frac{1-\alpha}{2\alpha}} \rho \sigma^{\frac{1-\alpha}{2\alpha}} \big)^{\alpha} \Big) \,. 
  \label{eq:qrd-lb}
 \end{align}
The proof of the last equality is non-trivial and will be the topic of Section~\ref{sc:pinching-a-sandwich}.

Conversely, this inequality is a necessary but not a sufficient condition for additivity and data-processing. Potentially tighter lower bounds are possible, for example by maximizing over all possible measurement maps on $n$ systems on the right-hand side.
However, we will see in the next section that the \emph{minimal quantum R\'enyi divergence} (also known as \emph{sandwiched R\'enyi divergence}), defined as the expression in~\eqref{eq:qrd-lb}, has all the desired properties of a quantum R\'enyi divergence for a large range of~$\alpha$.

\subsection{Maximal Quantum R\'enyi Divergence}
\label{sc:rmax}

A general upper bound can be found by considering a preparation map, using Matsumoto's elegant construction~\cite{matsumoto14}. For two fixed states $\rho$ and $\sigma$, consider the operator $\Delta = \sigma^{-1/2} \rho \sigma^{-1/2}$ with spectral decomposition
\begin{align}
  \Delta = \sum_x \lambda_x \Pi_x, \quad \textrm{as well as} \quad
  q(x) = \tr(\sigma \Pi_x), \quad p(x) = \lambda_x\, q(x) \,.
\end{align}
Then, the $\cptp$ map $\Lambda(\cdot) = \sum_x \bra{x} \cdot \ket{x} \frac{1}{q(x)} \sqrt{\sigma} \Pi_x \sqrt{\sigma}$ satisfies
\begin{align}
  \Lambda(p) = \sum_x \frac{p(x)}{q(x)} \sqrt{\sigma} \Pi_x \sqrt{\sigma} = \rho, \quad 
  \Lambda(q) = \sum_x \frac{q(x)}{q(x)} \sqrt{\sigma} \Pi_x \sqrt{\sigma} = \sigma \,.
\end{align}
Hence, any quantum generalization of the R\'enyi divergence $\DD_{\alpha}$ with data-processing (VIII) must satisfy
\begin{align}
  \DD_{\alpha}(\rho\|\sigma) 
  %= \DD_{\alpha}(\Lambda(p)\|\Lambda(q)) 
  \leq D_{\alpha}(p\|q) = \frac{1}{\alpha-1} \log \tr\Big( \sigma^{\frac12} \big(\sigma^{-\frac12} \rho \sigma^{-\frac12} \big)^{\alpha} \sigma^{\frac12} \Big)
   \,.
  \label{eq:matsu}
\end{align}
We call the quantity on the right-hand side of~\eqref{eq:matsu} the \emph{maximal quantum R\'enyi divergence}.
For $\alpha \in (0, 1)$, the term in the trace evaluates to a mean~\cite{kubo80}. Specifically, for $\alpha = \frac12$ the right-hand side of~\eqref{eq:matsu} evaluates to $- 2 \log \tr(\rho \# \sigma)$, where `$\#$' denotes the \emph{geometric mean}. These means are jointly concave and thus we also satisfy a data-processing inequality. Furthermore, $D_2(p\|q) = \log \tr( \rho^{2} \sigma^{-1})$ is an upper bound on $\DD_{2}(\rho\|\sigma)$. 
and in the limit $\alpha \to 1$ we find that
\begin{align}
  \DD_1(\rho\|\sigma) \leq \tr\Big(\sigma^{\frac12} \rho \sigma^{-\frac12} \log \big( \sigma^{-\frac12} \rho \sigma^{-\frac12} \big) \Big) = \tr\Big( \rho \log \big( \rho^{\frac12} \sigma^{-1} \rho^{\frac12} \big) \Big)\,.
\end{align}
The last equality follows from~\eqref{eq:polar-trick} and the expression on the right is the Belavkin-Staszewski relative entropy~\cite{belavkin82}.
In spite of its appealing form, the maximal quantum R\'enyi divergence has not found many applications yet, and we will not consider it further in this text. 

The minimal and maximal R\'enyi divergences are compared in Figure~\ref{fig:renyi-orgy}.

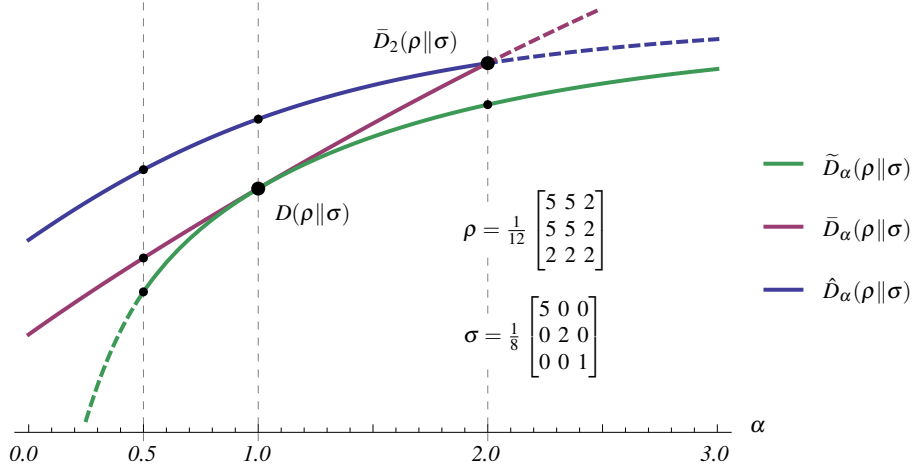
\begin{figure}[t]
The minimal, Petz, and maximal quantum R\'enyi divergences are given by the relation
$D_{\alpha}(\rho\|\sigma) = \frac{1}{\alpha-1} \log Q_{\alpha}$ with the respective functionals
\begin{align*}
  \Qn_{\alpha} = \tr \left( \left( \sigma^{\frac{1-\alpha}{2\alpha}} \rho  \sigma^{\frac{1-\alpha}{2\alpha}} \right)^{\alpha} \right),
  \quad \Qo_{\alpha} = \tr \left(\rho^{\alpha} \sigma^{1-\alpha} \right), \quad \textrm{and} \quad
  \hat{Q}_{\alpha} = \tr \left( \sigma \left( \sigma^{-\frac12} \rho \sigma^{-\frac12} \right)^{\alpha} \right) \,.
\end{align*}
\begin{flushleft}
  \hspace{1cm}
  \begin{overpic}[width =.8\columnwidth]{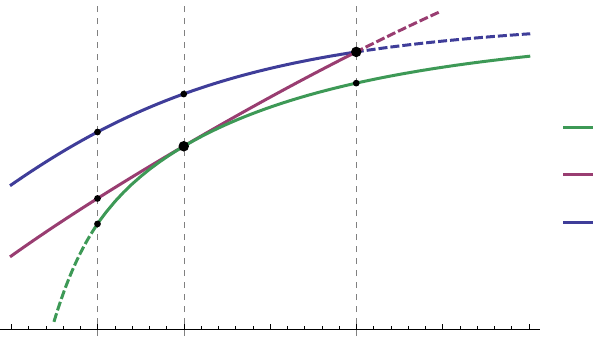}
    % legend
    \put(99,35.5){$\Dn_{\alpha}(\rho\|\sigma)$}
    \put(99,28){$\Do_{\alpha}(\rho\|\sigma)$}
    \put(99,20.5){$\hat{D}_{\alpha}(\rho\|\sigma)$}
    % points
    \put(32,30){$D(\rho\|\sigma)$}
    \put(44,51){$\Do_2(\rho\|\sigma)$}
    %\put(17,44){$\hat{D}(\rho\|\sigma)$}
    %\put(60,40){$\Dn_2(\rho\|\sigma)$}
    %\put(17,17){$\Dn_{1/2}(\rho\|\sigma)$}
    % axis label
     \put(90,4){$\alpha$}
    % ticks
    \put(-0.5,0.5){\it 0.0}
    \put(14,0.5){\it 0.5}
    \put(28,0.5){\it 1.0}
    \put(56,0.5){\it 2.0}
    \put(84,0.5){\it 3.0}
    % matrices
    \put(58,29){\color{white}\circle*{10}}  % erase tick line...
    \put(58,28){\color{white}\circle*{10}}  
    \put(55,28){$\rho = \frac{1}{12} \begin{bmatrix} 5 & 5 & 2 \\ 5 & 5 & 2 \\ 2 & 2 & 2  \end{bmatrix}$}
    \put(58,16){\color{white}\circle*{10}}  % erase tick line...
    \put(58,15){\color{white}\circle*{10}}  
    \put(55,15){$\sigma = \frac{1}{8} \begin{bmatrix} 5 & 0 & 0 \\ 0 & 2 & 0 \\ 0 & 0 & 1 \end{bmatrix}$}    
  \end{overpic}
\end{flushleft}
\caption{Minimal, Petz and maximal quantum R\'enyi entropy (for small $\alpha$). These divergences are discussed in Section~\ref{sc:rminimal}, Section~\ref{sc:rhc}, and Section~\ref{sc:rmax}, respectively. Solid lines are used to indicate that the quantity satisfies the data-processing inequality in this range of $\alpha$.}
\label{fig:renyi-orgy}
\end{figure}

\subsection{Quantum Max-Divergence}

The bounds in the previous subsection are not sufficient to single out a unique quantum generalization of the R\'enyi divergence for general $\alpha$ (and neither are the other desirable properties discussed above), except in the limit $\alpha \to \infty$, where the lower bound in~\eqref{eq:qrd-lb} and upper bound in~\eqref{eq:matsu} converge. Hence, the max-divergence has a unique quantum generalization.

Let us verify this now.
First note that for $\alpha \to \infty$ Eq.~\eqref{eq:matsu} yields
\begin{align}
  \DD_{\infty}(\rho\|\sigma) \leq D_{\infty}(p\|q) = \max_{x} \log \lambda_x = \log \| \Delta \|_{\infty} = \inf \{ \lambda : \rho \leq \exp(\lambda) \sigma \} \,.
\end{align}
So let us thus define the \emph{quantum max-divergence} as follows~\cite{datta08,renner05}:
\begin{definition}
\label{df:max-div}
\begin{svgraybox}
  For any $\rho, \sigma \in \cP(A)$, we define the \textbf{quantum max-divergence} as
\begin{align}
  D_{\infty}(\rho\|\sigma) := \inf \{ \lambda : \rho \leq \exp(\lambda) \sigma \} \,, \label{eq:max-div}
\end{align}
  where we follow the usual convention that $\inf\, \emptyset = \infty$.
\end{svgraybox}
\end{definition}

Using the pinching inequality~\eqref{eq:pinching}, we find that
\begin{align}
  \rho \leq \exp(\lambda)\sigma &\implies \sP_{\sigma}(\rho) \leq \exp(\lambda)\sigma \,, \\
  \sP_{\sigma}(\rho) \leq \exp(\lambda)\sigma &\implies  \rho \leq |\spec(\sigma)| \exp(\lambda)\sigma \,,
\end{align}
and, thus, the quantum max-divergence satisfies
\begin{align}
  D_{\infty}(\sP_{\sigma}(\rho) \| \sigma) \leq D_{\infty}(\rho\|\sigma) \leq D_{\infty}(\sP_{\sigma}(\rho) \| \sigma) + \log \big|\spec(\sigma)\big| \,. \label{eq:pinching-max}
\end{align}
We now apply this to $n$-fold product states $\rho^{\otimes n}$ and $\sigma^{\otimes n}$ and use the fact that $|\spec(\sigma^{\otimes n})| \leq (n+1)^{d_A-1}$
grows at most polynomially in $n$, such that
\begin{align}
  0 \leq \lim_{n \to \infty} \frac{1}{n} \log \big|\spec(\sigma^{\otimes n})\big| \leq \lim_{n \to \infty} \frac{d_A - 1}{n} \log (n+1) = 0 \,. \label{eq:spec-limit}
\end{align} 
The term thus vanishes asymptotically as $n \to \infty$, which means that 
\begin{align}
  \frac{1}{n} D_{\infty}(\rho^{\otimes n} \| \sigma^{\otimes n}) \quad \textrm{and} \quad \frac{1}{n} D_{\infty}\big(\sP_{\sigma^{\otimes n}}(\rho^{\otimes n}) \big\| \sigma^{\otimes n} \big)
\end{align} 
are asymptotically equivalent. Further using that $D_{\infty}$ is additive, we establish that
\begin{align}
  D_{\infty}(\rho\|\sigma) &= \lim_{n \to \infty} \frac{1}{n} D_{\infty}\big( \rho^{\otimes n} \big\| \sigma^{\otimes n} \big) = \lim_{n \to \infty} \frac{1}{n} D_{\infty}\big(\sP_{\sigma^{\otimes n}}(\rho^{\otimes n}) \big\| \sigma^{\otimes n} \big) \,. \label{eq:max-pinching}
\end{align}
This argument is in fact a special case of the discussion that we will follow in Section~\ref{sc:pinching-a-sandwich} for general R\'enyi divergences.

Hence, Eq.~\eqref{eq:qrd-lb} yields that $\DD_{\infty}(\rho\|\sigma) \geq D_{\infty}(\rho\|\sigma)$ for any quantum generalization of the max-divergence satisfying data-processing and additivity. We summarize these findings as follows:
\begin{proposition}
  $D_{\infty}$ is the unique quantum generalization of the max-divergence that satisfies additivity (V) and data-processing (VIII).
\end{proposition}
We leave it as an exercise for the reader to verify that that the quantum max-divergence also satisfies Properties (I)--(X).

%%%%%%%%%%%

\section{Minimal Quantum R\'enyi Divergence}
\label{sc:rminimal}

In this section we further discuss the minimal quantum R\'enyi divergence mentioned in Section~\ref{sc:qmin-div}. In particular, we will see that the following closed formula for the minimal quantum R\'enyi divergence corresponds to the limit in~\eqref{eq:qrd-lb0} for all $\alpha$.
\begin{definition}
\begin{svgraybox}
Let $\alpha \in (0, 1) \cup (1, \infty)$, and $\rho, \sigma \in \cS(A)$ with $\rho \neq 0$. Then we define the \textbf{minimal quantum R\'enyi divergence} of $\sigma$ with $\rho$ as
\begin{align}
  \Dn_{\alpha}(\rho\|\sigma) := \begin{cases} \frac{1}{\alpha-1} \log \frac{ \normb{\sigma^{\frac{1-\alpha}{2\alpha}} \rho \sigma^{\frac{1-\alpha}{2\alpha}} }_{\alpha}^{\alpha}}{\tr(\rho)} & 
  \textrm{if } (\alpha < 1 \land \rho \not\perp\sigma)  \lor \rho \ll \sigma \\
  + \infty & \textrm{else} 
  \end{cases} \,.
\end{align}
  Moreover, $\Dn_0$, $\Dn_1$ and $\Dn_{\infty}$ are defined as limits of $\Dn_{\alpha}$ for $\alpha \to \{0,1,\infty\}$.
\end{svgraybox}
\end{definition}
In Section~\ref{sc:qmin-limits} we will see that $\Dn_{\infty}(\rho\|\sigma) = D_{\infty}(\rho\|\sigma)$ (cf. Definition~\ref{df:max-div}).

The minimal quantum R\'enyi divergence is also called `quantum R\'enyi divergence'~\cite{lennert13} and `sandwiched quantum R\'enyi relative entropy'~\cite{wilde13} in the literature, but we propose here to call it \emph{minimal} quantum R\'enyi divergence since it is the smallest quantum R\'enyi divergence that still satisfies the crucial data-processing inequality as seen in~\eqref{eq:qrd-lb}. Thus, it is the minimal quantum R\'enyi divergence for which we can expect operational significance.

By inspection, it is evident that this quantity satisfies isometric invariance (II+), normalization (III+), additivity (V), and general mean (VI). Continuity (I) also holds, but one has to be a bit more careful since we are employing the generalized inverse in the definition. (See~\cite{lennert13} for a proof of continuity when the rank of $\rho$ or $\sigma$ changes.)

\subsection{Pinching Inequalities}
\label{sc:pinching-a-sandwich}

The goal of this section is to establish that $\Dn_{\alpha}$ is contractive under pinching maps and can be asymptotically achieved by the respective pinched quantity.
For this purpose, let us investigate some properties of
\begin{align}
  \Qn_{\alpha}(\rho\|\sigma) := \norm{\sigma^{\frac{1-\alpha}{2\alpha}} \rho \sigma^{\frac{1-\alpha}{2\alpha}} }_{\alpha}^{\alpha} 
  &= \tr \left( \Big( \sigma^{\frac{1-\alpha}{2\alpha}} \rho \sigma^{\frac{1-\alpha}{2\alpha}} \Big)^{\alpha} \right) \\
  &= \tr \left( \Big( \rho^{\frac12} \sigma^{\frac{1-\alpha}{\alpha}} \rho^{\frac12} \Big)^{\alpha} \right)  \,.
\end{align}
for $\rho, \sigma \in \cSnorm(A)$ with $\rho \ll \sigma$. First, we find that it is monotone under the pinching channel~\cite{lennert13}. 
\begin{lemma}
  \label{lm:pinch1}
  For $\alpha > 1$, we have
  \begin{align}
    \Qn_{\alpha}(\rho\|\sigma) &\geq \Qn_{\alpha}\big(\sP_{\sigma}(\rho) \big\|\sigma \big)
  \end{align}
  and the opposite inequality holds for $\alpha \in (0, 1)$.
\end{lemma}
\newcommand{\aaa}{\frac{1-\alpha}{2\alpha}}
\begin{petit}
\begin{proof}
  We have $\sigma^{\aaa} \sP_{\sigma}(\rho) \sigma^{\aaa} = \sP_{\sigma}\big(\sigma^{\aaa} \rho \sigma^{\aaa} \big)$ since the pinching projectors commute with $\sigma$. For $\alpha > 1$, we find
  \begin{align}
    \Qn_{\alpha}(\sP_{\sigma}(\rho) \| \sigma) = \Big\| \sP_{\sigma}\big(\sigma^{\aaa} \rho \sigma^{\aaa} \big) \Big\|_{\alpha}^{\alpha} \leq 
    \Big\| \sigma^{\aaa} \rho \sigma^{\aaa} \Big\|_{\alpha}^{\alpha} = \Qn_{\alpha}(\rho\|\sigma) \,,
  \end{align}
  where the inequality follows from the pinching inequality for norms~\eqref{eq:pinch-norm}. For $\alpha < 1$, the operator Jensen inequality~\eqref{eq:jensen} establishes that $\big(\sP_{\sigma}\big(\sigma^{\aaa} \rho \sigma^{\aaa} \big) \big)^{\alpha} \geq \sP_{\sigma} \big( \big(\sigma^{\aaa} \rho \sigma^{\aaa} \big)^{\alpha} \big)$. Thus,
  \begin{align}
    \Qn_{\alpha}(\sP_{\sigma}(\rho) \| \sigma) &\geq \tr \Big( \sP_{\sigma} \Big( \big(\sigma^{\aaa} \rho \sigma^{\aaa} \big)^{\alpha} \Big)\Big) = \Qn_{\alpha}(\rho\|\sigma) \,.
  \end{align}
\end{proof}
\end{petit}

The following general purpose inequalities will turn out to be very useful:

\begin{lemma}
\label{lm:pinch2a}
  For any $\rho \leq \rho'$, we have $\Qn_{\alpha}(\rho\|\sigma) \leq \Qn_{\alpha}(\rho'\|\sigma)$. Furthermore, if $\sigma \leq \sigma'$ and $\alpha > 1$, we have
  \begin{align}
    \Qn_{\alpha}(\rho\|\sigma) &\geq \Qn_{\alpha}(\rho\|\sigma')
   \end{align}
  and the opposite inequality holds for $\alpha \in [\frac12, 1)$.
\end{lemma}

\begin{petit}
\begin{proof}
  Set $c = \frac{1-\alpha}{\alpha}$.
  If $\rho \leq \rho'$, then $\sigma^{\frac{c}{2}} \rho \sigma^{\frac{c}{2}} \leq \sigma^{\frac{c}{2}} \rho' \sigma^{\frac{c}{2}}$
  and the first statement follows from the monotonicity of the trace of monotone functions~\eqref{eq:trace-mono}.
  
  To prove the second statement for $\alpha \in [\frac12, 1)$, we note that $t \mapsto t^{c}$ is operator monotone. Hence,
  \begin{align}
    \rho^{\frac12} \sigma^{\frac{1-\alpha}{\alpha}} \rho^{\frac12} \leq \rho^{\frac12} \sigma'^{\frac{1-\alpha}{\alpha}} \rho^{\frac12}
  \end{align}
  and the statement again follows by~\eqref{eq:trace-mono}. Analogously, for $\alpha > 1$ we find that $t \mapsto t^{-c}$ is operator monotone and the inequality goes in the opposite direction.
 \qed
\end{proof}
\end{petit}

In particular, the second statement establishes the dominance property (X).
On the other hand, we can employ the first inequality to get a very general pinching inequality. For any $\cp$ maps $\sE$ and $\sF$, and any $\alpha > 0$, we have
  \begin{align}
    \Qn_{\alpha}\big(\sE(\rho) \big\| \sF(\sigma) \big) &\leq |\spec(\sigma)|^{\alpha}\ \Qn_{\alpha}\big(\sE(\sP_{\sigma}(\rho))\big\| \sF(\sigma) \big) \,. \label{eq:gen-pinching}
  \end{align}
A more delicate analysis is possible for the pinching case when $\alpha \in (0, 2]$. We establish the 
following stronger bounds~\cite{hayashitomamichel14}:
\begin{lemma}
  \label{lm:pinch2}
  For $\alpha \in [1, 2]$, we have
  \begin{align}
    \Qn_{\alpha}(\rho\|\sigma) &\leq |\spec(\sigma)|^{\alpha - 1}\ \Qn_{\alpha}\big(\sP_{\sigma}(\rho)\big\|\sigma\big) 
     \label{eq:pinch2-1}
  \end{align}
  and the opposite inequality holds for $\alpha \in (0, 1]$.
%    \Qn_{\alpha}(\rho\|\sigma) &\geq |\spec(\sigma)|^{\alpha - 1}\ \Qn_{\alpha}\big(\sP_{\sigma}(\rho)\big\|\sigma\big) 
%    \qquad \textrm{for} \quad \alpha \in (0, 1] \,. \label{eq:pinch2-2}
\end{lemma}
\begin{petit}
\begin{proof}  
  By the pinching inequality, we have $\rho \leq |\spec(\sigma)|\, \sP_{\sigma}(\rho)$. Then, we write
  \begin{align}
    \Qn_{\alpha}(\rho\|\sigma) = \tr \Big( \big( \sigma^{\aaa} \rho \sigma^{\aaa} \big)^{\alpha-1} \sigma^{\aaa} \rho \sigma^{\aaa} \Big)
  \end{align}
  Then, for $\alpha \in (1, 2]$, we use the fact that $t \mapsto t^{\alpha-1}$ is operator monotone, such that the pinching inequality yields
  the following bound:
  \begin{align}
    \Qn_{\alpha}(\rho\|\sigma) &\leq |\spec(\sigma)|^{\alpha-1} \tr \Big( \big( \sigma^{\aaa} \sP_{\sigma}(\rho) \sigma^{\aaa} \big)^{\alpha-1} \sigma^{\aaa} \rho \sigma^{\aaa} \Big) \,. \label{eq:pinching-proof1}
  \end{align}
  Now, note that the pinching projectors commute with all operators except for the single $\rho$ in the term that we pulled out initially, and hence we can pinch this operator ``for free''. This yields
  \begin{align}
   \tr \Big( \big( \sigma^{\aaa} \sP_{\sigma}(\rho) \sigma^{\aaa} \big)^{\alpha-1} \sigma^{\aaa} \rho \sigma^{\aaa} \Big) 
%   = \tr \Big( \big( \sigma^{\aaa} \sP_{\sigma}(\rho) \sigma^{\aaa} \big)^{\alpha-1} \sigma^{\aaa} \sP_{\sigma}(\rho) \sigma^{\aaa} \Big)
   =\Qn_{\alpha}(\sP_{\sigma}(\rho)\|\sigma)
  \end{align}
  and we have established Eq.~\eqref{eq:pinch2-1}. Similarly, we proceed for $\alpha \in (0, 1)$, where the pinching inequality again yields
  $\sigma^{\aaa} \rho \sigma^{\aaa} \leq |\spec(\sigma)| \sigma^{\aaa} \sP_{\sigma}(\rho) \sigma^{\aaa}$, and thus we have
  \begin{align}
    \big( \sigma^{\aaa} \rho \sigma^{\aaa} \big)^{\alpha-1} &\geq |\spec(\sigma)|^{\alpha-1} \big( \sigma^{\aaa} \sP_{\sigma}(\rho) \sigma^{\aaa} \big)^{\alpha-1} \,
  \end{align}
  on the support of $\sigma^{\aaa} \rho \sigma^{\aaa}$. Combining this with the development leading to~\eqref{eq:pinch2-1} yields the desired bound.
  \qed
\end{proof}
\end{petit}

A combination of the above Lemmas yields an alternative characterization of the minimal quantum R\'enyi divergence in terms of an asymptotic limit 
of classical R\'enyi divergences, as desired.

\begin{proposition}
\label{pr:pinching-a-sandwich}
\begin{svgraybox}
  For $\rho, \sigma \in \cS(A)$ with $\rho \neq 0$, $\rho \ll \sigma$, and $\alpha \geq 0$, we have
  \begin{align*}
     \Dn_{\alpha}(\rho\|\sigma) = \lim_{n \to \infty} \frac{1}{n} D_{\alpha}\big( \sP_{\sigma^{\otimes n}}(\rho^{\otimes n}) \big\| \sigma^{\otimes n} \big) \,.
  \end{align*}
\vspace{-0.5cm}
\end{svgraybox}
\end{proposition}

\begin{petit}
\begin{proof}
  It suffices to show the statement for $\rho,\sigma \in \cSnorm(A)$.
  Summarizing Lemmas~\ref{lm:pinch1}--\ref{lm:pinch2} yields
  \begin{align}
      \Dn_{\alpha}\big(\sP_{\sigma}(\rho) \big\| \sigma \big) &\leq \Dn_{\alpha}\big(\rho \big\|\sigma\big) \\
      &\leq \Dn_{\alpha}\big(\sP_{\sigma}(\rho) \big\| \sigma \big) + 
      \begin{cases}
        \log |\spec(\sigma)| & \textrm{for } \alpha \in (0,1) \cup (1,2] \\
        \frac{\alpha}{\alpha-1} \log |\spec(\sigma)| & \textrm{for } \alpha > 2
      \end{cases} \,.
  \end{align}
  Since $\frac{\alpha}{\alpha-1} < 2$ for $\alpha > 2$, we can replace the correction term on the right-hand side by $2 \log |\spec(\sigma)|$, which has the nice feature that it is independent of $\alpha$.
  Hence, for $n$-fold product states, we have
  \begin{align}
    \left| \frac{1}{n} \Dn_{\alpha}\big(\sP_{\sigma^{\otimes n}}(\rho^{\otimes n}) \big\| \sigma^{\otimes n} \big) - \Dn_{\alpha}(\rho\|\sigma) \right| \leq \frac{2}{n} \log \big|\spec(\sigma^{\otimes n})\big|  \,.
  \end{align}
  The result then follows by employing~\eqref{eq:spec-limit} in the limit $n \to \infty$.
 
  Finally, we note that the convergence is uniform in $\alpha$ (as well as $\rho$ and $\sigma$), and thus the equality also holds for the limiting cases $\Dn_{0}$, $\Dn_{1}$ and $\Dn_{\infty}$. 
  \qed
\end{proof}
\end{petit}

The strength of this result lies in the fact that we immediately inherit some properties of the classical R\'enyi divergence. More precisely, $\alpha \mapsto \log \Qn_{\alpha}(\rho\|\sigma)$ is the point-wise limit of a sequence of convex functions, and thus also convex.

\begin{corollary}
\begin{svgraybox}
  The function $\alpha \mapsto \log \Qn_{\alpha}(\rho\|\sigma)$ is convex, and $\alpha \mapsto \Dn_{\alpha}(\rho\|\sigma)$ is monotonically increasing.
\end{svgraybox}
\end{corollary}

\subsection{Limits and Special Cases}
\label{sc:qmin-limits}

Instead of evaluating the limits for $\alpha \to \infty$ explicitly as in~\cite{lennert13}, we can take advantage of the fact that Proposition~\ref{pr:pinching-a-sandwich} already gives an alternative characterization of the limiting quantity in terms of the pinched divergence. Hence, as Eq.~\eqref{eq:max-pinching} reveals, the limit is the quantum max-divergence of Definition~\ref{df:max-div} as claimed earlier.

In the limit $\alpha \to 1$, we expect to find the `ordinary' quantum relative entropy or quantum divergence, first studied by Umegaki~\cite{umegaki62}.

\begin{definition}
\begin{svgraybox}
  \label{df:rel}
   For any state $\rho \in \cS(A)$ with $\rho \neq 0$ and any $\sigma \in \cS(A)$, we define the \textbf{quantum divergence} of $\sigma$ with $\rho$ as
   \begin{align}
      D(\rho\|\sigma) := \begin{cases} 
      \frac{\tr \big(\rho ( \log \rho - \log \sigma ) \big)}{\tr(\rho)} & \textrm{if } \rho \ll \sigma \\
      +\infty & \textrm{else}
    \end{cases} \,.
   \end{align}
\vspace{-0.5cm}
\end{svgraybox}
\end{definition}

This reduces to the Kullback-Leibler (KL) divergence~\cite{kullback51} if $\rho$ and $\sigma$ are classical (commuting) operators.
We now prove that $\Dn_1(\rho\|\sigma) = D(\rho\|\sigma)$.
\begin{proposition}
\label{pr:relative-limit1}
\begin{svgraybox}
  For $\rho,\sigma \in \cP(A)$ with $\rho \neq 0$, we find that $\Dn_1(\rho\|\sigma)$ equals 
  \begin{align}
     \lim_{\alpha \searrow 1} \Dn_{\alpha}(\rho\|\sigma) = \lim_{\alpha \nearrow 1} \Dn_{\alpha}(\rho\|\sigma) = D(\rho\|\sigma) \,.
  \end{align}
  \vspace{-0.5cm}
\end{svgraybox}
\end{proposition}
The proof proceeds by finding an explicit expression for the limiting divergence~\cite{lennert13,wilde13}. (Alternatively one could show that the quantum relative entropy is achieved by pinching, as is done in~\cite{hayashi97}.) We follow~\cite{wilde13} here:
\begin{petit}
\begin{proof}
  Since the proposed limit satisfies the normalization property~(III+), it is sufficient to evaluate the limit for $\rho,\sigma \in \cSnorm(A)$. Furthermore, we restrict our attention to the case $\rho \ll \sigma$.
  By l'H\^opital's rule and the fact that $\Qn_1(\rho\|\sigma) = 1$, we have
  \begin{align}
    \lim_{\alpha \searrow 1} \Dn_{\alpha}(\rho\|\sigma) = \lim_{\alpha \nearrow 1} \Dn_{\alpha}(\rho\|\sigma) 
    = \log(e) \cdot \frac{\rm{d}}{\rm{d} \alpha} \Qn_{\alpha}(\rho\|\sigma) \bigg|_{\alpha = 1} \,.
  \end{align}
  To evaluate this derivative, it is convenient to introduce a continuously differentiable two-parameter function (for fixed $\rho$ and $\sigma$) as follows:
  \begin{align}
     q\big(r,z\big) = \tr \Big( \big( \sigma^{\frac{r}{2}} \rho \sigma^{\frac{r}{2}}  \big)^z \Big) \quad \textrm{with} \quad r(\alpha) = \frac{1-\alpha}{\alpha} \quad \textrm{and} \quad z(\alpha) = \alpha  
   \end{align}
   such that $\frac{\partial r}{\partial \alpha} = -\frac{1}{\alpha^2}$ and $\frac{\partial z}{\partial \alpha} = 1$ and therefore
   \begin{align}
    \frac{\rm{d}}{\rm{d} \alpha} \Qn_{\alpha}(\rho\|\sigma) \bigg|_{\alpha = 1} 
    &= - \frac{1}{\alpha^2} \frac{\partial}{\partial r} q(r,z) \bigg|_{\alpha=1} + \frac{\partial} {\partial z} q(r,z) \bigg|_{\alpha=1} \\
    &= - \frac{\partial}{\partial r} \tr\big( \sigma^{r} \rho \big) \bigg|_{r=0} + \frac{\partial} {\partial z} \tr \big( \rho^z \big) \bigg|_{z=1} 
    = \tr\big( \rho (\ln \rho - \ln \sigma ) \big) \,.
  \end{align}
  In the penultimate step we exchanged the limits with the differentiation and in the last step we simply used the fact that the derivate commutes with the trace and that $\frac{\rm{d}}{\rm{d} z} \rho^z = \ln (\rho) \rho^z$.
\qed
\end{proof}
\end{petit}

%Then, as a corollary of Propositions~\ref{pr:pinching-a-sandwich} and~\ref{pr:relative-limit1} we find that
%\begin{align}
%  \Dn_1(\rho\|\sigma) = \lim_{n \to \infty} \frac{1}{n} D_1(\big( \sP_{\sigma^{\otimes n}}(\rho^{\otimes n}) \big\| \sigma^{\otimes n} \big) \,.
%\end{align}

Let us have a look at two other special cases that are important for applications. First, at $\alpha = \{ \frac12, 2\}$, we find the negative logarithm of the quantum fidelity and the \emph{collision relative entropy}~\cite{renner05}, respectively. For $\rho, \sigma \in \cS(A)$, we have
\begin{align}
  \Dn_{\nicefrac12}(\rho\|\sigma) = - \log F(\rho,\sigma) \,, \qquad \Dn_2(\rho\|\sigma) = \log \tr \big( \rho \sigma^{-\frac12} \rho \sigma^{-\frac12} \big) \,.
\end{align}

%%%%

\subsection{Data-Processing Inequality}

Here we show that $\Dn_{\alpha}$ satisfies the data-processing inequality for $\alpha \geq \frac12$. First, we show that our pinching inequalities in fact already imply the data-processing inequality for $\alpha > 1$, following an instructive argument due to Mosonyi and Ogawa in~\cite{mosonyiogawa13}. (For $\alpha \in [\frac12, 1)$ we will need a completely different argument.)

\subsubsection*{From Pinching to Measuring and Data-Processing}
\label{sc:data-pinching}

First, we restrict our attention to $\alpha > 1$. According to~\eqref{eq:gen-pinching}, for any measurement map $\sM \in \cptp(A, X)$ with POVM elements $\{ M_x \}_x$, we find
\begin{align}
   \frac{ \Qn_{\alpha}\big(\sM(\rho) \big\| \sM(\sigma) \big) }{|\spec(\sigma)|^{\alpha}} &\leq \ \Qn_{\alpha}\big(\sM(\sP_{\sigma}(\rho))\big\| \sM(\sigma) \big) \\
   &= \sum_x \Big( \tr \big(M_x \sP_{\sigma}(\rho) \big) \Big)^{\alpha} \Big( \tr (M_x \sigma ) \Big)^{1-\alpha} \\
   &= \sum_x \Big( \tr \big(\sP_{\sigma}(M_x) \sP_{\sigma}(\rho) \big) \Big)^{\alpha} \Big( \tr (\sP_{\sigma}(M_x) \sigma ) \Big)^{1-\alpha}  \,.
\end{align}
Now, note that $W(x|a) = \bra{a} \sP_{\sigma}(M_x) \ket{a}$ is a classical channel for states that are diagonal in the eigenbasis $\{ \ket{a} \}_a$ of $\sigma$. Hence the classical data-processing inequality together with Lemma~\ref{lm:pinch1} yields
\begin{align}
  |\spec(\sigma)|^{-\alpha}\, \Qn_{\alpha}\big(\sM(\rho) \big\| \sM(\sigma) \big) &\leq \Qn_{\alpha}(\sP_{\sigma}(\rho) \| \sigma) \leq \Qn_{\alpha}(\rho\|\sigma) \,.
\end{align}
Using a by now standard argument, we consider $n$-fold product states $\rho^{\otimes n}$ and $\sigma^{\otimes n}$ and a product measurement $\sM^{\otimes n}$ in order to get rid of the spectral term in the limit as $n \to \infty$.  Additivity then yields
\begin{align}
  \Dn_{\alpha}(\rho\|\sigma) \geq \Dn_{\alpha}(\sM(\rho)\|\sM(\sigma))
\end{align}
for all measurement maps $\sM$.

Combining this with Proposition~\ref{pr:pinching-a-sandwich} and interpreting the pinching map as a measurement in the eigenbasis of $\sigma$, we have established that, for $\alpha > 1$, the minimal quantum R\'enyi divergence is \emph{asymptotically achievable by a measurement}:
\begin{align}
  \Dn_{\alpha}(\rho \| \sigma) = \lim_{n \to \infty} \frac{1}{n} 
   \max \Big\{ D_{\alpha} \big(\sM_n(\rho^{\otimes n}) \big\| \sM_n(\sigma^{\otimes n}) \big) :\ \sM_n \in \cptp(A^n, X) \Big\} \,. \label{eq:minimal-meas}
\end{align}
We will discuss this further below.
Using the representation in~\eqref{eq:minimal-meas} we can derive the data-processing inequality using a very general argument.
\begin{proposition}
  \label{pr:minimal-meas}
  Let $\DD_{\alpha}$ be a quantum R\'enyi divergence satisfying~\eqref{eq:minimal-meas}. Then, it also satisfies data-processing (VIII).
\end{proposition}

\begin{petit}
\begin{proof}
We show that $\DD_{\alpha}(\rho\|\sigma) \geq \DD_{\alpha}(\sE(\rho)\|\sE(\sigma))$ for all $\sE \in \cptp(A,B)$ and $\rho, \sigma \in \cS(A)$.

First note that since $\sE$ is trace-preserving, $\sE^{\dag}$ is unital. For every measurement map $\sM \in \cptp(B,X)$ consisting of POVM elements $\{ M_x \}_x$, we define the measurement map $\sM^{\sE} \in \cptp(A,X)$ that consists of the POVM elements $\{ \sE^{\dag}(M_x) \}_x$.
Then, using~\eqref{eq:minimal-meas} twice, we find
\begin{align}
  &\DD_{\alpha}( \sE(\rho) \| \sE(\sigma) ) \nonumber\\
  &\ = \lim_{n \to \infty} \frac{1}{n} 
   \sup \Big\{ D_{\alpha} \big(\sM_n(\sE(\rho)^{\otimes n}) \big\| \sM_n(\sE(\sigma)^{\otimes n}) \big) :\ \sM_n \in \cptp(B^n, X) \Big\} \\
   &\ = \lim_{n \to \infty} \frac{1}{n}\sup \Big\{ D_{\alpha} \big( \sM_n^{\sE^{\otimes n}} (\rho^{\otimes n}) \big\| \sM_n^{\sE^{\otimes n}}(\sigma^{\otimes n}) \big) :\ \sM_n \in \cptp(B^n, X) \Big\} \label{eq:crucial1} \\
   &\ \leq \lim_{n \to \infty} \frac{1}{n}\sup \Big\{ D_{\alpha} \big( \sM_n (\rho^{\otimes n}) \big\| \sM_n(\sigma^{\otimes n}) \big) :\ \sM_n \in \cptp(A^n, X) \Big\} \\
  &\ = \DD_{\alpha}( \rho \| \sigma) \,.
\end{align}
This concludes the proof. \qed
\end{proof}
\end{petit}

\subsubsection*{Data-Processing via Joint Concavity}

Unfortunately, the first part of the above argument leading to~\eqref{eq:minimal-meas} only goes through for $\alpha > 1$ (and consequently in the limits $\alpha \to 1$ and $\alpha \to \infty$).
However, the data-processing inequality holds more generally for all $\alpha \geq \frac12$, as was shown by Frank and Lieb~\cite{frank13}. 

It thus remains to show data-processing for $\alpha \in [\frac12, 1)$. Here we show the following equivalent statement (cf.~Proposition~\ref{pr:dp-jc}):
\begin{proposition}
  The map $(\rho,\sigma) \mapsto \Qn_{\alpha}(\rho\|\sigma)$ is jointly concave for $\alpha \in [\frac12, 1)$.
\end{proposition}

\begin{petit}
\begin{proof}
  First, we express $\Qn_{\alpha}$ as a minimization problem. To do this, we use~\eqref{eq:app-hoelder} and set $c = \frac{1-\alpha}{\alpha} \in (0, 1]$, $M = \sigma^{\frac{c}2} \rho \sigma^{\frac{c}2}$, and $N = \sigma^{-\frac{c}2} H \sigma^{-\frac{c}2}$ to find
  \begin{align}
    \Qn_{\alpha}(\rho\|\sigma) = 
    \tr\Big( \big( \sigma^{\frac{c}2} \rho \sigma^{\frac{c}2} \big)^{\alpha} \Big) \leq \alpha \tr(H \rho) + (1-\alpha) \tr\Big( \big(\sigma^{-\frac{c}2} H \sigma^{-\frac{c}2} \big)^{-\frac{1}{c}} \Big) \,.
  \end{align} 
  for all $H \geq 0$ with $H \gg \rho$ and equality can be achieved. Thus, we can write
  \begin{align}
    \Qn_{\alpha}(\rho\|\sigma) = \min \bigg\{ \alpha \tr(H \rho) + (1-\alpha) \tr\Big( \big( H^{-\frac12} \sigma^{c}  H^{-\frac12}\big)^{\frac{1}{c}} \Big) :\ H \geq 0,\, H \gg \rho \bigg\} \,.
  \end{align}
  
  This nicely splits the contributions of $\rho$ and $\sigma$ and we can deal with them separately. The term $\tr(H \rho)$ is linear and thus concave in $\rho$. 
  Next, we want to show that the second term is concave in $\sigma$. To do this, we further decompose it as follows, using essentially the same ideas that we used above. First, using~\eqref{eq:app-hoelder}, we find
  \begin{align}
     \tr \big( H^{-\frac12} \sigma^{c}  H^{-\frac12} X^{1-c} \big) \leq c \tr\Big( \big( H^{-\frac12} \sigma^{c}  H^{-\frac12}\big)^{\frac{1}{c}} \Big) + (1-c) \tr(X) \, ,
  \end{align}
  which allows us to write
  \begin{align}
    \tr\Big( \big( H^{-\frac12} \sigma^{c}  H^{-\frac12}\big)^{\frac{1}{c}} \Big) = \max \bigg\{ \frac{1}{c} \tr \big( H^{-\frac12} \sigma^{c}  H^{-\frac12} X^{1-c} \big) - \frac{1-c}{c} \tr(X) :\ X \geq 0 \bigg\} \,.
  \end{align}
  Since $c \in (0, 1)$, Lieb's concavity theorem~\eqref{eq:lieb-ando} reveals that the function we maximize over is jointly concave in $\sigma$ and $X$. Note that generally the maximum of concave functions is not necessarily concave, but joint concavity in $\sigma$ and $X$ is sufficient to ensure that the 
  maximum is concave in $\sigma$. %(we leave it as an exercise to verify this).
  Hence, $\Qn_{\alpha}(\rho\|\sigma)$ is the minimum of a jointly concave function, and thus jointly concave. \qed
\end{proof}
\end{petit}

The same proof strategy can be used to show that $\Qn_{\alpha}(\rho\|\sigma)$ is jointly convex for $\alpha > 1$, but we already know that this holds due to our previous argument in Section~\ref{sc:data-pinching} that established the data-processing inequality directly.

\subsubsection*{Summary and Remarks}

Let us now summarize the results of this subsection in the following theorem.

\begin{theorem}
\begin{svgraybox}
  Let $\alpha \geq \frac12$ and $\rho, \sigma \in \cS(A)$ with $\rho \neq 0$. The minimal quantum R\'enyi divergence has the following properties:
  \begin{itemize}
    \item The functional $(\rho,\sigma) \mapsto \Qn_{\alpha}(\rho\|\sigma)$ is jointly concave for $\alpha \in (\frac12,1)$ and jointly convex for $\alpha \in (1,\infty)$.
    \item The functional 
    $(\rho,\sigma) \mapsto \Dn_{\alpha}(\rho\|\sigma)$ is jointly convex for $\alpha \in (\frac12,1]$.
    \item For every $\sE \in \cptp(A,B)$, the data-processing inequality holds, i.e.
    \begin{align}
      \Dn_{\alpha}(\rho\|\sigma) \geq \Dn_{\alpha}\big( \sE(\rho) \big\| \sE(\sigma) \big) \,.
    \end{align}
    \item It is asymptotically achievable by a measurement, i.e. 
    \begin{align}
      \Dn_{\alpha}(\rho\|\sigma) =  \lim_{n \to \infty} \frac{1}{n} 
   \max \bigg\{ D_{\alpha} \big(\sM_n(\rho^{\otimes n}) \big\| \sM_n(\sigma^{\otimes n}) \big) :\ \sM_n \in \cptp(A^n, X) \bigg\} \,. \label{eq:meas-rep}
    \end{align}
  \end{itemize}
\vspace{-0.5cm}
\end{svgraybox}
\end{theorem}

A few remarks are in order here. First, note that one could potentially hope that the limit $n \to \infty$ in~\eqref{eq:meas-rep} is not necessary. However, except for the two boundary points $\alpha = \frac12$ and $\alpha = \infty$, it is generally not sufficient to just consider measurements on a single system. (This effect is also called ``information locking''.) 

For $\alpha \in \{ \frac12, \infty \}$, we have in fact (without proof)
\begin{align}
  \Dn_{\alpha}(\rho\|\sigma) = \max \big\{ \Dn_{\alpha}(\sM(\rho) \| \sM(\rho) ) : \sM \in \cptp(A, X) \big\} ,
\end{align}
which has an interesting consequence. Namely, if we go through the proof of Proposition~\ref{pr:minimal-meas} we realize that we never use the fact that $\sE$ is completely positive, and in fact the data-processing inequality holds for all positive trace-preserving maps. Generally, for all $\alpha$, the data-processing inequality holds if $\sE^{\otimes n}$ is positive for all $n$, which is also strictly weaker than complete positivity.

The data-processing inequality together with definiteness of the classical R\'enyi divergence also establishes definiteness (VII-) of the minimal quantum R\'enyi divergence for $\alpha \geq \frac12$, and thus of all quantum R\'enyi divergences. Namely, if $\rho \neq \sigma$, then there exists a measurements (for example an informationally complete measurement) $\sM$ such that $\sM(\rho) \neq \sM(\sigma)$, and thus 
\begin{align}
  \Dn_{\alpha}(\rho\|\sigma) \geq D_{\alpha}(\sM(\rho)\|\sM(\sigma)) > 0 \,. 
\end{align}
This completes the discussion of the minimal quantum R\'enyi divergence.

\begin{center}
\begin{svgraybox}
  The minimal quantum R\'enyi divergences satisfy Properties~(I)--(X) for $\alpha \geq \frac12$, and thus constitute a family of R\'enyi divergences according to Definition~\ref{def:qrenyi}.
\end{svgraybox}
\end{center}

%%%%%%%%%%%%%%%%%%%%%%%%%%%

\section{Petz Quantum R\'enyi Divergence}
\label{sc:rhc}

A straight-forward generalization of the classical expression to quantum states is given by the following expression, which was originally investigated by Petz~\cite{petz86}.
\begin{definition}
\begin{svgraybox}
Let $\alpha \in (0, 1) \cup (1, \infty)$, and $\rho, \sigma \in \cS(A)$ with $\rho \neq 0$. Then we define the \textbf{Petz quantum R\'enyi divergence} of $\sigma$ with $\rho$ as
\begin{align}
  \Do_{\alpha}(\rho\|\sigma) := \begin{cases} \frac{1}{\alpha-1} \log \frac{\tr\big( \rho^{\alpha} \sigma^{1-\alpha} \big)}{\tr(\rho)} & 
  \textrm{if } (\alpha < 1 \land \rho \not\perp\sigma)  \lor \rho \ll \sigma \\
  + \infty & \textrm{else} 
  \end{cases} \,.
\end{align}
  Moreover, $\Do_0$ and $\Do_1$ are defined as the respective limits of $\Do_{\alpha}$ for $\alpha \to \{0, 1\}$.
\end{svgraybox}
\end{definition}

This quantity turns out to have a clear operational interpretation in binary hypothesis testing, where it appears in the quantum generalization of the Chernoff and Hoeffding bounds. More surprisingly, it is also connected to the minimal quantum R\'enyi divergence via duality relations for conditional entropies, as we will see in the next chapter.

We could as well have restricted the definition to $\alpha \in [0, 2]$ since the quantity appears not to be useful outside this range. For $\alpha = 2$ it matches the maximal quantum R\'enyi divergence (cf.\ Figure~\ref{fig:renyi-orgy}) and it is also evident that
\begin{align}
  \Qo_{\alpha}(\rho\|\sigma) := \tr (\rho^{\alpha} \sigma^{1-\alpha})  
\end{align}
is not convex in $\rho$ (for general $\sigma$) since $\rho^{\alpha}$ is not operator convex for $\alpha > 2$.

\subsection{Data-Processing Inequality}

As a direct consequence of the Lieb concavity theorem and the Ando convexity theorem in~\eqref{eq:lieb-ando}, we find the following.
\begin{proposition}
  The functional $\Qo_{\alpha}(\rho\|\sigma)$ is jointly concave for $\alpha \in (0, 1)$ and jointly convex for $\alpha \in (1,2]$.
\end{proposition}

In particular, the Petz quantum R\'enyi divergence $\Do_{\alpha}$ thus satisfies the data-processing inequality. As such, we must also have
\begin{align}
  \Do_{\alpha}(\rho\|\sigma) \geq \Dn_{\alpha}(\rho\|\sigma)  \label{eq:alt-d}
\end{align}
since the latter quantity is the smallest quantity that satisfies data-processing. This inequality is in fact also a direct consequence of the Araki--Lieb--Thirring trace inequalities~\cite{araki90,liebthirring05}, which we will not discuss further here.

Alternatively, the function $\Qo_{\alpha}$ can be seen as a Petz quasi-entropy~\cite{petz86} (see also~\cite{hiai10}). For this purpose, using the notation of Section~\ref{sc:mirror}, let us write
\begin{align}
  \Qo_{\alpha}(\rho\|\sigma) = \tr(\rho^{\alpha} \sigma^{1-\alpha}) = \bra{\Psi} \sigma^{\frac12} f_{\alpha}( \sigma^{-1} \otimes \rho^T \big) \sigma^{\frac12} \ket{\Psi}
\end{align}
where $f_{\alpha} : t \mapsto t^{\alpha}$ is operator concave or convex for $\alpha \in (0,1)$ and $\alpha \in (1, 2]$. Petz used a variation of this representation to show the data-processing inequality.

We leave it as an exercise to verify the remaining properties mentioned in Secs.~\ref{sc:raxiom} and~\ref{sc:rprop} for the Petz R\'enyi divergence.
\begin{center}
\begin{svgraybox}
  The Petz quantum R\'enyi divergences satisfy Properties~(I)--(X) for $\alpha \in (0, 2]$.
\end{svgraybox}
\end{center}

\subsection{\texorpdfstring{Nussbaum--Szko{\l}a}{Nussbaum--Szkola} Distributions}

The following representation due to Nussbaum and Szko{\l}a~\cite{nussbaum09} turns out to be quite useful in applications, and also allows us to further investigate the divergence. Let us fix $\rho, \sigma \in \cSnorm(A)$ and write their eigenvalue decomposition as
\begin{align}
   \rho = \sum_x \lambda_x \proj{e_x}_A  \quad \textrm{and} \quad \sigma = \sum_y \mu_y \proj{f_y} \,.
\end{align}
Then, the two probability mass functions 
\begin{align}
  P_{XY}^{[\rho,\sigma]}(x, y) = \lambda_x \big| \braket{e_x}{f_y} \big|^2 \quad \textrm{and} \quad
  Q_{XY}^{[\rho,\sigma]}(x, y) = \mu_y \big| \braket{e_x}{f_y} \big|^2
\end{align}
mimic the Petz quantum divergence of the quantum states $\rho$ and $\sigma$. Namely, they satisfy 
\begin{align}
\Do_{\alpha}(\rho\|\sigma) = D_{\alpha} \Big(P^{[\rho,\sigma]}_{XY} \Big\| Q^{[\rho,\sigma]}_{XY} \Big) \quad \textrm{for all} \quad \alpha \geq 0 \,. 
\end{align}
Moreover, these distributions inherit some important properties of $\rho$ and $\sigma$. For example, $\rho \ll \sigma \iff P^{[\rho,\sigma]} \ll Q^{[\rho,\sigma]}$ and for product states we have 
\begin{align}
  P^{[\rho \otimes \tau,\sigma \otimes \omega]} = P^{[\rho,\sigma]} \otimes P^{[\tau,\omega]} \,.
\end{align}

Last but not least, since this representation is independent of $\alpha$, we are able to lift the convexity, monotonicity and limiting properties of $\alpha \mapsto D_{\alpha}$ to the quantum regime\,---\,as a corollary of the respective classical properties.

\begin{corollary}
\begin{svgraybox}
  The function $\alpha \mapsto \log \Qo_{\alpha}(\rho\|\sigma)$ is convex, $\alpha \mapsto \Do_{\alpha}(\rho\|\sigma)$ is monotonically increasing, and
 \begin{align} 
   \Do_{1}(\rho\|\sigma) = \frac{\tr \big( \rho (\log \rho - \log \sigma) \big)}{\tr(\rho)} \,.
 \end{align}
 \vspace{-0.5cm}
\end{svgraybox}
\end{corollary}

So, in particular, $\Do_{1}(\rho\|\sigma) = \Dn_{1}(\rho\|\sigma)$. This means that these two curves are tangential at this point and their first derivatives agree (cf.  Figure~\ref{fig:tangent}).

%%%

\subsubsection*{First Derivative at $\alpha = 1$}

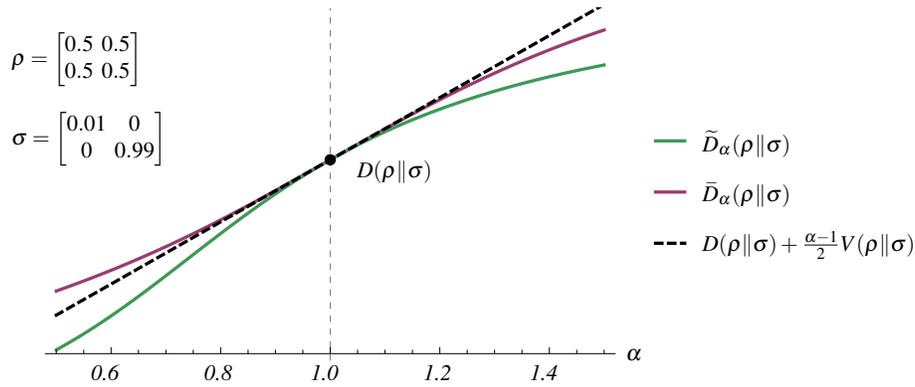
\begin{figure}[t]
\begin{flushleft}
  \hspace{1cm}
  \begin{overpic}[width =.65\columnwidth]{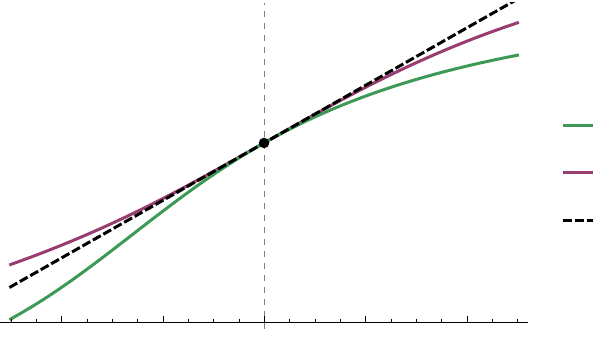}
    % legend
    \put(99,35){$\Dn_{\alpha}(\rho\|\sigma)$}
    \put(99,27.5){$\Do_{\alpha}(\rho\|\sigma)$}
    \put(99,20){$D(\rho\|\sigma) + \frac{\alpha-1}{2} V(\rho\|\sigma)$}
    % points
    \put(47,31){$D(\rho\|\sigma)$}
    % axis label
    \put(87.5,3.5){$\alpha$}
    % matrices
    \put(-5,48){$\rho = \begin{bmatrix} 0.5 & 0.5 \\ 0.5 & 0.5 \end{bmatrix}$}
    \put(-5,36){$\sigma = \begin{bmatrix} 0.01 & 0 \\ 0 & 0.99 \end{bmatrix}$}
    % ticks
    \put(40,0.5){\it 1.0}
    \put(24,0.5){\it 0.8}
    \put(7,0.5){\it 0.6}
    \put(57,0.5){\it 1.2}
    \put(73,0.5){\it 1.4}
  \end{overpic}
\end{flushleft}
\caption{Minimal and Petz quantum R\'enyi entropy around $\alpha = 1$.}
\label{fig:tangent}
\end{figure}

In fact, the Nussbaum--Szko{\l}a representation gives us a simple means to evaluate the first derivative of $\alpha \mapsto \Do_{\alpha}(\rho\|\sigma)$ and $\alpha \mapsto \Dn_{\alpha}(\rho\|\sigma)$ at $\alpha = 1$, which will turn out to be useful later. 

In order to do this, let us first take a step back and evaluate the derivative for classical probability mass functions $\rho_X, \sigma_X \in \cSnorm(X)$. Substituting $\alpha = 1 + \nu$
and introducing the \emph{log-likelihood ratio} as a random variable $Z(X) = \ln ({\rho(X)}/{\sigma(X)})$, where $X$ is distributed according to the law $X \leftarrow \rho_X$, we find
\begin{align}
  D_{1+\nu}(\rho_X\|\sigma_X) = \frac{1}{\nu} \log \sum_x \rho(x) \left( \frac{\rho(x)}{\sigma(x)} \right)^{\nu} 
  = \frac{\log \Exp \left( e^{\nu Z} \right)}{\nu} = \log(e) \frac{G(\nu)}{\nu} \,,
\end{align}
where $G(\nu)$ is the cumulant generating function of $Z$.

Clearly, $G(0) = 0$. Moreover, using l'H\^opital's rule, its first derivative at $\nu = 0$ is
\begin{align}
  \lim_{\nu \to 0} \bigg( \frac{\rm d}{\rm d \nu} \frac{G(\nu)}{\nu} \bigg) 
  &= \lim_{\nu \to 0}  \frac{\nu G'(\nu) - G(\nu)}{\nu^2} \\
  &= \lim_{\nu \to 0}  \frac{G'(\nu) + \nu G''(\nu) - G'(\nu)}{2 \nu} = \frac{G''(0)}{2} \,,
\end{align}
which is one half of the second cumulant of $Z$. The second cumulant simply equals the second central moment, or variance, of the log-likelihood ratio $Z$.
\begin{align}
  G''(0) &= \Exp \big( (Z - \Exp(Z) )^2 \big) = \Exp(Z^2) - \Exp(Z)^2 \\
   &= \sum_x \rho(x) \bigg( \ln \frac{\rho(x)}{\sigma(x)} - \sum_x \rho(x) \ln \frac{\rho(x)}{\sigma(x)} \bigg)^2 
   =: \frac{ V(\rho_X\|\sigma_X) }{ \log(e)^2 }\,.
\end{align}
Combining these steps, we have established that
\begin{align}
  \frac{\rm d}{\rm d \alpha} D_{\alpha}(\rho_X\|\sigma_X) \Big|_{\alpha=1} = \frac{1}{2 \log(e)}\, V(\rho_x\|\sigma_x) \,. 
\end{align}

Now we can simply substitute the Nussbaum--Szko{\l}a distributions to lift this result to the Petz quantum R\'enyi divergence, and thus also the minimal quantum R\'enyi divergence.
We recover the following result~\cite{lintomamichel14}:
\begin{proposition}
\begin{svgraybox}
  Let $\rho,\sigma \in \cSnorm(A)$ with $\rho \ll \sigma$. Then the functions $\alpha \mapsto \Dn_{\alpha}(\rho\|\sigma)$ and $\alpha \mapsto \Do_{\alpha}(\rho\|\sigma)$ are continuously differentiable at $\alpha = 1$ and
  \begin{align}
    \frac{\rm d}{\rm d \alpha} \Dn_{\alpha}(\rho\|\sigma) \Big|_{\alpha=1} = \frac{\rm d}{\rm d \alpha} \Do_{\alpha}(\rho\|\sigma) \Big|_{\alpha=1} = \frac{1}{2\log(e)} V(\rho\|\sigma) \, ,
  \end{align}
  where $V(\rho\|\sigma) := \tr \Big( \rho \big( \log \rho - \log \sigma - D(\rho\|\sigma) \big)^2 \Big)$.
\end{svgraybox}
\end{proposition}

The minimal and Petz quantum R\'enyi divergences are thus differentiable at $\alpha = 1$ and in fact infinitely differentiable.
Hence, by Taylor's theorem, for every interval $[a, b]$ containing $1$, there exist constants $K \in \mathbb{R}_+$ such that, for all $\alpha \in [a,b]$, we have
\begin{align}
\bigg| \Do_{\alpha}(\rho\|\sigma) - D(\rho\|\sigma) - (\alpha-1) \frac{1}{2\log(e)} V(\rho\|\sigma) \bigg| \leq K (\alpha-1)^2 \,.
\label{eq:taylor}
\end{align}
The same statement naturally also holds if we replace $\Do_{\alpha}$ with $\Dn_{\alpha}$.
An example of the first-order Taylor series approximation is plotted in Figure~\ref{fig:tangent}.

%%%%%%%%%%%%%%%%%%%%%%%%%%%

\section{Background and Further Reading}

Shannon was first to derive the definition of entropy axiomatically~\cite{shannon48} and many have followed his footsteps since. We exclusively consider R\'enyi's approach~\cite{renyi61} here, but a recent overview of different axiomatizations can be found in~\cite{csiszar08}.

The Belavkin-Staszewski relative entropy~\cite{belavkin82} was considered a reasonable alternative to Umegaki's relative entropy~\cite{umegaki62} until Hiai and Petz~\cite{hiai91} established the operational interpretation of Umegaki's definition in quantum hypothesis testing.  The proof that joint convexity implies data-processing is rather standard and mimics a development for the relative entropy that is due to Uhlmann~\cite{uhlmann73,uhlmann77} and Lindblad~\cite{lindblad74,lindblad75}.
The data-processing inequality for the quantum relative entropy has been shown in these works, building on previous work by Lieb and Ruskai~\cite{lieb73} that established it for the partial trace.
The data-processing inequality can be strengthened by including a remainder term that characterizes how well the channel can be recovered. This has been shown by Fawzi and Renner~\cite{fawzirenner14} for the partial trace (see also~\cite{brandao14,berta15} for refinements and simplifications of the proof). Recently these results were extended to general channels in~\cite{wilde15} (see also~\cite{berta14b}) and further refined in~\cite{sutter15}.

The max-divergence was first formally introduced by Datta~\cite{datta08}, based on Renner's work~\cite{renner05} treating conditional entropy. However, the idea to define a quantum relative entropy via an operator inequality appears implicitly in earlier literature, for example in the work of Jain, Radhakrishnan, and Sen~\cite{jain02}.
The minimal (or sandwiched) quantum R\'enyi divergence was formally introduced independently in~\cite{lennert13} and~\cite{wilde13}. Some ideas resulting in the former work were already presented publicly in~\cite{mytutorial12} and~\cite{fehrtalk13}, and partial results were published in~\cite{martinthesis} and~\cite[Th. 21]{dupuis13}. The initial works only proved a few properties of the divergence and left others as conjectures. Various other authors then contributed by showing data-processing for certain ranges of $\alpha$ concurrently with Frank and Lieb~\cite{frank13}. Notably, M\"uller-Lennert \emph{et al.}~\cite{lennert13} already establishes data-processing for $\alpha \in (1, 2]$ and conjectured it for all $\alpha \geq \frac12$. Concurrently with~\cite{frank13}, Beigi~\cite{beigi13} provided a proof for data-processing for $\alpha > 1$ and Mosonyi and Ogawa~\cite{mosonyiogawa13} provided the proof discussed above, which is also only valid for $\alpha > 1$. Their proof in turn uses some of Hayashi's ideas~\cite{hayashi06}.

The minimal, maximal and Petz quantum R\'enyi divergence are by no means the only quantum generalizations of the R\'enyi divergence. For example, a two-parameter family of R\'enyi divergences proposed by Jaksic \emph{et al.}~\cite{jaksic10} and further investigated by Audenaert and Datta~\cite{audenaert13} (see also~\cite{hiai12} and \cite{carlen15}) captures both the minimal and Petz quantum R\'enyi divergence. 

Both quantum R\'enyi divergences discussed in this work have found applications beyond binary quantum hypothesis testing. In particular, the minimal quantum R\'enyi divergence has turned out to be a very useful tool in order to establish the strong converse property for various information theoretic tasks. Most prominently it led to a strong converse for classical communication over entanglement-breaking channels~\cite{wilde13}, the entanglement-assisted capacity~\cite{gupta13}, and the quantum capacity of dephasing channels~\cite{tomamichelww14}. Furthermore, the strong converse exponents for coding over classical-quantum channels can be expressed in terms of the minimal quantum R\'enyi divergence~\cite{mosonyi14-2}. The minimal quantum R\'enyi divergence of order $2$ can also be used to derive various achievability results~\cite{beigi13b}.
Besides this, the quantum R\'enyi divergences have also found applications in quantum thermodynamics, e.g.\ in the study of the second law of thermodynamics~\cite{wehner13}, and in quantum cryptography, e.g.\ in~\cite{miller13}.

Finally, we note that many of the definitions discussed here are perfectly sensible for infinite-dimensional quantum systems. However, some of the proofs we presented here do not directly generalize to this setting. Ohya and Petz's book~\cite{ohya93} treats quantum entropies in the even more general algebraic setting. However, a comprehensive investigation of the minimal quantum R\'enyi divergence in the infinite-dimensional or algebraic setting is missing.

%%%%%%%%%

%!TEX root = book.tex

\chapter{Conditional R\'enyi Entropy}
\label{ch:cond} 
% use \chaptermark{}
% to alter or adjust the chapter heading in the running head

\abstract*{Conditional Entropies are measures of the uncertainty inherent in a system from the perspective of an observer who is given side information on the system. The system as well as the side information can be either classical or a quantum. The goal in this chapter is to define conditional R\'enyi entropies that are operationally significant measures of this uncertainty, and to explore their properties. Unconditional entropies are then simply a special case of conditional entropies where the side information is uncorrelated with the system under observation.}

Conditional Entropies are measures of the uncertainty inherent in a system from the perspective of an observer who is given side information on the system. The system as well as the side information can be either classical or a quantum. The goal in this chapter is to define conditional R\'enyi entropies that are operationally significant measures of this uncertainty, and to explore their properties. Unconditional entropies are then simply a special case of conditional entropies where the side information is uncorrelated with the system under observation. 

We want the conditional R\'enyi entropies to retain most of the properties of the conditional von Neumann entropy, which is by now well established in quantum information theory. Most prominently, we expect that they satisfy a data-processing inequality: we require that the uncertainty of the system never decreases when the quantum system containing side information undergoes a physical evolution. This can be ensured by defining R\'enyi entropies in terms of the R\'enyi divergence, in analogy with the case of conditional von Neumann entropy.

%%%%%%%%%%%%%%%

\section{Conditional Entropy from Divergence}
\label{sc:cvn}

Let us first recall Shannon's definition of conditional entropy. For a joint probability mass function $\rho(x,y)$ with marginals $\rho(x)$ and $\rho(y)$, the conditional Shannon entropy is given as
\begin{align}
  H(X|Y)_{\rho} 
  &= \sum_y \rho(y)\, H(X|Y\!=\!y)_\rho \\
  &= \sum_y \rho(y)\, \sum_x \rho(x|y) \log \frac{1}{\rho(x|y)} \\
  &= \sum_{x,y} \rho(x,y) \log \frac{\rho(y)}{\rho(x,y)} \\
  &= H(XY)_\rho - H(Y)_\rho \label{eq:right-one} \,,
\end{align}
where we used the conditional probability distribution $\rho(x|y) = \rho(x,y) / \rho(y)$, and the corresponding Shannon entropy, $H(X|Y\!=\!y)_\rho$. Such conditional distributions are ubiquitous in classical information theory, but it is not immediate how to generalize this concept to quantum information. Instead, we avoid this issue altogether by generalizing the expression in~\eqref{eq:right-one}, which is also called the \emph{chain rule} of the Shannon entropy. This yields the following definition for the quantum conditional entropy.

\begin{definition}
\begin{svgraybox}
  \label{df:cond-entropy}
   For any bipartite state $\rho_{AB} \in \cSnorm(AB)$, we define the \textbf{conditional von Neumann entropy} of $A$ given $B$ for the state $\rho_{AB}$ as
   \begin{align}
      H(A|B)_{\rho} := H(AB)_{\rho} - H(B)_{\rho}  \,, \quad \textrm{where} \quad H(A)_{\rho} := -\tr (\rho_A \log \rho_A) \,.
   \end{align}
   \vspace{-0.5cm}
\end{svgraybox}
\end{definition}

Here, $H(A)_{\rho}$ is the von Neumann entropy~\cite{vonneumann32} and simply corresponds to the Shannon entropy of the state's eigenvalues.
One of the most remarkable properties of the von Neumann entropy is \emph{strong subadditivity}. It states that for any
tripartite state $\rho_{ABC} \in \cSnorm(ABC)$, we have
\begin{align}
  H(ABC)_{\rho} + H(B)_{\rho} \leq H(AB)_{\rho} + H(BC)_{\rho} \label{eq:strong-sub-add}
\end{align}
or, equivalently $H(A|BC)_{\rho} \leq H(A|B)_{\rho}$. The latter is is an expression of another principle, the \emph{data-processing inequality}. It states that any processing of the side information system, in this case taking a partial trace, can at most increase the uncertainty of $A$. Formally, for any $\sE \in \cptp(B,B')$ map we have
\begin{align}
  H(A|B)_{\rho} \leq H(A|B')_{\tau}, \quad \textrm{where} \quad \tau_{AB'} = \sE(\rho_{AB}) \,.  \label{eq:cond-dp}
\end{align}
This property of the von Neumann entropy was first proven by Lieb and Ruskai~\cite{lieb73}. It implies weak subadditivity, and the relation~\cite{araki70}
\begin{align}
  |H(A)_{\rho} - H(B)_{\rho}| \leq H(AB)_{\rho} \leq H(A)_{\rho} + H(B)_{\rho} \,.
\end{align}

The conditional entropy can be conveniently expressed in terms of Umegaki's relative entropy, namely
\begin{align}
  H(A|B)_{\rho} 
  &= H(AB)_{\rho} - H(B)_{\rho} \label{eq:vn} \\
  &= -\tr \left(\rho_{AB} \log \rho_{AB} \right) + \tr\left( \rho_{B} \log \rho_{B} \right) \\
  &= -\tr \left( \rho_{AB} \big( \log \rho_{AB} - \log (\id_A \otimes \rho_B) \big) \right) \label{eq:put_id_in}\\
  &= - D(\rho_{AB} \| \id_A \otimes \rho_B) \label{eq:vn1}.
\end{align}
Here, we used that $\log (\id_A \otimes \rho_B) = \id_A \otimes \log \rho_B$ to establish~\eqref{eq:put_id_in}. 
Sometimes it is useful to rephrase this expression as an optimization problem. Based on~\eqref{eq:put_id_in} we can introduce an auxiliary state $\sigma_B \in \cSnorm(B)$ and write
\begin{align}
  H(A|B)_{\rho} 
  &= -\tr \left( \rho_{AB} \big( \log \rho_{AB} - \id_A \otimes \log \sigma_B \big) \right) 
     + \tr \big(\rho_B (\log \rho_B - \log \sigma_B) \big) \\
  &= - D(\rho_{AB}\|\id_A \otimes \sigma_B) + D(\rho_B \| \sigma_B) \,.
\end{align}
Since the latter divergence is always non-negative and equals zero if and only if $\sigma_B = \rho_B$, this yields the following expression for the conditional entropy:
\begin{align}
  H(A|B)_{\rho} = \max_{\sigma_B \in \cSnorm(B)} - D(\rho_{AB} \| \id_A \otimes \sigma_B) . \label{eq:vn2}
\end{align}

\section{Definitions and Properties}
\label{sc:cond}

In the case of quantum R\'enyi entropies, it is not immediate which of the relations~\eqref{eq:vn},~\eqref{eq:vn1} or \eqref{eq:vn2} should be used to define the conditional R\'enyi entropies. It has been found in the study of the classical special case (see, e.g.~\cite{fehr14,iwamoto13}) that generalizations based on~\eqref{eq:vn} have severe limitations, for example they generally do not satisfy a data-processing inequality.
On the other hand, definitions based on the underlying divergence, as in~\eqref{eq:vn1} or~\eqref{eq:vn2}, have proven to be very fruitful and lead to quantities with operational significance and useful mathematical properties.

Together with the two proposed quantum generalizations of the R\'enyi divergence, $\Dn_{\alpha}$ and $\Do_{\alpha}$,
this leads to a total of four different candidates for conditional R\'enyi entropies~\cite{tomamichel08,lennert13,tomamichel13}. 

\begin{definition}
\begin{svgraybox}
For $\alpha \geq 0$ and $\rho_{AB} \in \cSnorm(AB)$, we define the following \textbf{quantum conditional R\'enyi entropies} of $A$ given $B$ of the state $\rho_{AB}$:
\begin{align}
  \Ho_{\alpha}^{\da}(A|B)_{\rho} &:= - \Do_{\alpha}(\rho_{AB}\|\id_A \otimes \rho_B), \label{eq:had} \\
  \Ho_{\alpha}^{\ua}(A|B)_{\rho} &:= \sup_{\sigma_B \in \cSnorm(B)} - \Do_{\alpha}(\rho_{AB}\|\id_A \otimes \sigma_B), \label{eq:hau} \\
    \Hn_{\alpha}^{\da}(A|B)_{\rho} &:= - \Dn_{\alpha}(\rho_{AB}\|\id_A \otimes \rho_B), \qquad \qquad \qquad \textrm{and} \label{eq:hatd} \\
  \Hn_{\alpha}^{\ua}(A|B)_{\rho} &:= \sup_{\sigma_B \in \cSnorm(B)} - \Dn_{\alpha}(\rho_{AB}\|\id_A \otimes \sigma_B) \label{eq:hatu}.
\end{align}
   \vspace{-0.5cm}
\end{svgraybox}
\end{definition}

Note that for $\alpha > 1$ the optimization over $\sigma_B$ can always be restricted to $\sigma_B$ with support equal to the support of $\rho_B$. Moreover, since small eigenvalues of $\sigma_B$ lead to a large divergence, we can further restrict $\sigma_B$ to a compact set of states with eigenvalues bounded away from $0$. Since we are thus optimizing a continuous function over a compact set, we are justified in writing a maximum in the above definitions. Furthermore, pulling the optimization inside the logarithm, we see that these optimization problems are either convex (for $\alpha > 1$) or concave (for $\alpha < 1$).

Consistent with the notation of the proceeding chapter, we also use $\HH_{\alpha}$ to refer to any of the four entropies and $H_{\alpha}$ to refer to the respective classical quantities. More precisely, we use $\HH_{\alpha}$ only to refer to quantum conditional R\'enyi entropies that satisfy data-processing, which\,---\,as we will see in Sec.~\ref{sc:cond-dp}\,---\,means that $\HH_{\alpha}$ encompasses $\Ho_{\alpha}$ for $\alpha \in [0,2]$ and $\Hn_{\alpha}$ for $\alpha \in [\frac12, \infty]$.

For a trivial system $B$, we find that 
\begin{align}
  \HH_{\alpha}(A)_{\rho} = -D_{\alpha}(\rho_A\|\id_A) = \frac{\alpha}{1-\alpha} \log \| \rho_A \|_{\alpha} \,.
\end{align}
reduces to the classical R\'enyi entropy of the eigenvalues of $\rho_A$. In particular, if $\alpha = 1$, we always recover the von Neumann entropy.

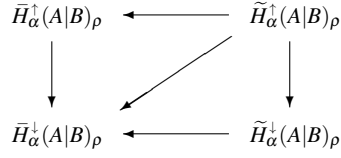
\begin{figure}[t]
\begin{center}
\begin{picture}(130,50)
\put(0,45){\mbox{$\Ho_{\alpha}^{\ua}(A|B)_{\rho}$}}
\put(0,0){\mbox{$\Ho_{\alpha}^{\da}(A|B)_{\rho}$}}
\put(90,45){\mbox{$\Hn_{\alpha}^{\ua}(A|B)_{\rho}$}}
\put(90,0){\mbox{$\Hn_{\alpha}^{\da}(A|B)_{\rho}$}}
\put(15,38){\vector(0,-1){25}}
\put(105,38){\vector(0,-1){25}}
\put(82,48){\vector(-1,0){40}}
\put(82,3){\vector(-1,0){40}}
\put(83,40){\vector(-3,-2){42}}
\end{picture}
\end{center}
\caption{Overview of the different conditional entropies used in this paper. Arrows indicate that one entropy is larger or equal to the other for all states $\rho_{AB} \in \cSnorm(AB)$ and all $\alpha \geq 0$.}
\label{fig:overview}
\end{figure}

Finally, note that we use the symbols `$\uparrow$' and `$\downarrow$' to express the observation that 
\begin{align}
  \widebar{H}_{\alpha}^{\ua}(A|B)_{\rho} \geq \widebar{H}_{\alpha}^{\da}(A|B)_{\rho} \quad \textrm{and} \quad \widetilde{H}_{\alpha}^{\ua}(A|B)_{\rho} \geq \widetilde{H}_{\alpha}^{\da}(A|B)_{\rho}
\end{align}
which follows trivially from the respective definitions. Furthermore, the Araki-Lieb-Thirring inequality in~\eqref{eq:alt-d} yields the relations 
\begin{align}
\widetilde{H}_{\alpha}^{\ua}(A|B)_{\rho} \geq \widebar{H}_{\alpha}^{\ua}(A|B)_{\rho} \quad \textrm{and} \quad \widetilde{H}_{\alpha}^{\da}(A|B)_{\rho} \geq \widebar{H}_{\alpha}^{\da}(A|B)_{\rho} \,.
\end{align}
These relations are summarized in Fig.~\ref{fig:overview}. 

%%%%

\subsubsection*{Limits and Special Cases}

Inheriting these properties from the corresponding divergences, all entropies are monotonically decreasing functions of $\alpha$, and we recover many interesting special cases in the limits $\alpha \to \{0, 1, \infty\}$.

For $\alpha = 1$, all definitions coincide with the usual von Neumann conditional entropy~\eqref{eq:vn1}. For $\alpha = \infty$, two quantum generalizations of the conditional min-entropy emerge, both of which have been studied by Renner~\cite{renner05}. Namely,
\begin{align}
  \Hn_{\infty}^{\da}(A|B)_{\rho} &= \sup \big\{ \lambda \in \mathbb{R} : \rho_{AB} \leq 2^{-\lambda} \id_A \otimes \rho_B \big\} \label{eq:min1} \qquad \textrm{and} \\
  \Hn_{\infty}^{\ua}(A|B)_{\rho} &= \sup \big\{ \lambda \in \mathbb{R} : \exists\, \sigma_B \in \cSnorm(B) \textnormal{ such that } \rho_{AB} \leq 2^{-\lambda} \id_A \otimes \sigma_B  \big\} \label{eq:min2} .
\end{align}
For $\alpha= \frac12$, we find the conditional max-entropy studied by K\"onig \emph{et al.}~\cite{koenig08},\footnote{The notation $H_{\min}(A|B)_{\rho|\rho} \equiv \Hn_{\infty}^{\da}(A|B)_{\rho}$ and $H_{\min}(A|B)_{\rho} \equiv \Hn_{\infty}^{\ua}(A|B)_{\rho}$ is widely used. The alternative notation $H_{\max}(A|B)_{\rho} \equiv \Hn_{\nicefrac{1}{2}}^{\ua}(A|B)_{\rho}$ is often used too, for example in Chapter~\ref{ch:calc}.}
\begin{align}
\Hn_{\nicefrac12}^{\ua}(A|B)_{\rho}=\sup_{\sigma_B \in \cSnorm(B)}  \log F(\rho_{AB}, \id_{A}\otimes\sigma_{B})\,.\label{eq:hmax}
\end{align}

For $\alpha=2$, we find a quantum conditional collision entropy~\cite{renner05}:
\begin{align}
\Hn_{2}^{\da}(A|B)_{\rho}=-\log\tr\Big(\rho_{AB} \Big(\id_{A}\otimes\rho_{B}^{-\frac12} \Big) \rho_{AB} \Big(\id_{A}\otimes\rho_{B}^{-\frac12} \Big)\Big). \label{eq:h2}
\end{align}
For $\alpha = 0$, we find a generalization of the Hartley entropy~\cite{hartley28}, proposed in~\cite{renner05}:
\begin{align}
\Ho_0^{\ua}(A|B)_{\rho} = \sup_{\sigma_B \in \cSnorm(B)} \log \tr \big( \{\rho_{AB} > 0\}\, \id_A \otimes \sigma_B  \big)\,.\label{eq:h0}
\end{align}

%%%%

\subsection{Alternative Expression\texorpdfstring{ for $\Ho_{\alpha}^{\ua}$}{}}

For the quantity $\Ho_{\alpha}^{\ua}$ we find a closed-form expression for the optimal (minimal or maximal) $\sigma_B$. This yields an alternative expression for $\Ho_{\alpha}^{\ua}$ as follows~\cite{sharma13,tomamichel13}.

\begin{lemma}
  \label{lm:dau-new}
  Let $\alpha \in (0, 1) \cup (1, \infty)$ and $\rho_{AB} \in \cS(AB)$. Then,
  \begin{align}
    \Ho_{\alpha}^{\ua}(A|B)_{\rho} = \frac{\alpha}{1-\alpha} \log \tr\Big( \big( \tr_A ( \rho_{AB}^{\alpha} ) \big)^{\frac{1}{\alpha}}  \Big) . \label{eq:dau-new}
  \end{align}
\end{lemma}

\begin{petit}
\begin{proof}
  Recall the definition
  \begin{align}
    H_{\alpha}^{\ua}(A|B)_{\rho} &= \sup_{\sigma_B \in \cSnorm(B)} \frac{1}{1-\alpha} \log \tr\big( \rho_{AB}^{\alpha}\, \sigma_B^{1-\alpha} \big) =
    \sup_{\sigma_B \in \cSnorm(B)} \frac{1}{1-\alpha} \log \tr\big( \tr_A ( \rho_{AB}^{\alpha}) \sigma_B^{1-\alpha} \big).
  \end{align}
  This can immediately be lower bounded by the expression in~\eqref{eq:dau-new} by substituting
  \begin{align}
    \sigma_B^* = \frac{\big( \tr_A ( \rho_{AB}^{\alpha} ) \big)^{\frac{1}{\alpha}}}{\tr \Big( \big( \tr_A ( \rho_{AB}^{\alpha} ) \big)^{\frac{1}{\alpha}} \Big)} \label{eq:opt-sigma}
    \end{align}
    for $\sigma_B$. It remains to show that this choice is optimal. For $\alpha < 1$, we employ the H\"older inequality in~\eqref{eq:hoelder1} for $p = \frac{1}{\alpha}$, $q = \frac{1}{1-\alpha}$, $L = \tr_A(\rho_{AB}^{\alpha} )$ and $K = \sigma_B^{1-\alpha}$ to find
    \begin{align}
      \tr\big( \tr_A ( \rho_{AB}^{\alpha}) \sigma_B^{1-\alpha} \big)
      \leq \bigg( \tr\Big( \big( \tr_A ( \rho_{AB}^{\alpha}) \big)^{\frac{1}{\alpha}} \Big) \bigg)^{\alpha} \big( \tr(\sigma_B) \big)^{1-\alpha},
    \end{align}
    which yields the desired upper bound since $\tr(\sigma_B) = 1$.
    For $\alpha > 1$, we instead use the reverse H\"older inequality~\eqref{eq:hoelder2}. 
    This leads us to~\eqref{eq:dau-new} upon the same substitutions. \qed
\end{proof}
\end{petit}

In particular, note that~\eqref{eq:opt-sigma} gives an explicit expression for the optimal $\sigma_B$ in the definition of $\Ho_{\alpha}^{\ua}$.
A similar closed-form expression for the optimal $\sigma_B$ in the definition of $\Hn_{\alpha}^{\ua}$ is however not known.

%%%

%%

\subsection{Conditioning on Classical Information}

We now analyze the behavior of $\DD_{\alpha}$ and $\HH_{\alpha}$ when applied to partly classical states. Formally, consider normalized classical-quantum states of the form $\rho_{XA} = \sum_x \rho(x) \proj{x} \otimes \rhoh_A(x)$ and $\sigma_{XA} = \sum_x \sigma(x) \proj{x} \otimes \sigmah_A(x)$. A straightforward calculation using Property~(VI) shows that for two such states,
\begin{align}
\DD_{\alpha}(\rho_{XA}\|\sigma_{XA}) = \frac{1}{\alpha-1} \log \bigg(\sum_x \rho(x)^\alpha \sigma(x)^{1-\alpha} \exp\Bigl((\alpha-1) \DD_{\alpha}\big(\rhoh_A(x)\big\|\sigmah_A(x)\big)\Bigr) \bigg)\, . \label{eq:qc-div}
\end{align}
In other words, the divergence $\DD_{\alpha}(\rho_{XA}\|\sigma_{XA})$ decomposes into the divergences $\DD_{\alpha}(\rhoh_A(x)\|\sigmah_A(x))$ of the `conditional' states. This leads to the following relations for conditional R\'enyi entropies.

\begin{proposition}
\begin{svgraybox}
  \label{pr:rcond-class}
  Let $\rho_{ABY} = \sum_y \rho(y) \rhoh_{AB}(y) \otimes \proj{y} \in \cSnorm(ABY)$ and $\alpha \in (0,1) \cup (1, \infty)$. Then,
  the conditional entropies satisfy
  \begin{align}
 \HH_{\alpha}^{\da}(A|BY)_{\rho}
&= \frac{1}{1-\alpha} \log \bigg( \sum_y \rho(y)\, \exp\Big( (1-\alpha) \HH_{\alpha}^{\da}(A|B)_{\rhoh(y)} \Big)  \bigg) \,,\\
  \HH_{\alpha}^{\ua}(A|BY)_{\rho} 
&= \frac{\alpha}{1-\alpha}\log \bigg( \sum_y \rho(y)\, \exp\bigg( \frac{1-\alpha}{\alpha} \,\HH_{\alpha}^{\ua}(A|B)_{\rhoh(y)}\bigg) \bigg) \,.
   \end{align}
   (Here, $\HH_{\alpha}$ is a substitute for $\Hn_{\alpha}$ or $\Ho_{\alpha}$.)
\end{svgraybox}
\end{proposition}

\begin{petit}
\begin{proof}
  The first statement follows directly from~\eqref{eq:qc-div} and the definition of the `$\downarrow$'-entropy.
To show the second statement, recall that by definition, 
\begin{align}
  \HH_{\alpha}^{\ua}(A|BY)_{\rho} = 
  \max_{\sigma_{BY} \in \cSnorm(BY)}  - \DD_{\alpha}(\rho_{ABY} \| \id_A \otimes\, \sigma_{BY})
\end{align}
where the infimum is over all (normalized) states $\sigma_{BY}$, but due to data processing (we can measure the $Y$-register, which does not affect $\rho_{ABY}$), we can restrict to states $\sigma_{BY}$ with classical~$Y$, i.e.\ $\sigma_{BY} = \sum_y \sigma(y) \proj{y} \otimes \sigmah_B(y)$. Using the decomposition of $\DD_{\alpha}$ in~\eqref{eq:qc-div}, we then obtain
\begin{align}
\HH_{\alpha}^{\ua}(A|BY)_{\rho} &= \max_{\sigma_{BY}} - \frac{1}{\alpha-1} \log \bigg( \sum_y \rho(y)^\alpha \sigma(y)^{1-\alpha} \exp\Bigl( (\alpha-1) \DD_{\alpha}\big(\rhoh_{AB}(y) \|\id_A \otimes \,\sigmah_B(y)\big)\Bigr) \bigg) \nonumber\\
&= \max_{ \{ \sigma(y) \}_y } \frac{1}{1-\alpha} \log \bigg( \sum_y \rho(y)^{\alpha} \sigma(y)^{1-\alpha} \exp\Big( (1-\alpha) \HH_{\alpha}^{\ua}(A|B)_{\rhoh(y)} \Big) \bigg) \,.
\end{align}
Writing $r_y = \rho(y) \exp\bigl(\frac{1-\alpha}{\alpha} \HH_{\alpha}^{\ua}(A|B)_{\rhoh(y)}\bigr)$, and using straightforward Lagrange multiplier technique, one can show that the infimum is attained by the distribution $\sigma(y) = r_y/\sum_z r_z$. Substituting this into the above equation leads to the desired relation. \qed 
\end{proof}
\end{petit}

In particular, considering a state $\rho_{XY} = \sum_{x,y} \rho(x,y) \proj{x} \otimes \proj{y}$, we recover two notions of classical conditional R\'enyi entropy 
\begin{align}
H_{\alpha}^{\da}(X|Y)_{\rho} &= \frac{1}{1-\alpha}\log \bigg( \sum_y \sum_x \rho(y) \rho(x|y)^\alpha \bigg) \,, \\
H_{\alpha}^{\ua}(X|Y)_{\rho} &= \frac{\alpha}{1-\alpha}\log \bigg( \sum_y \rho(y) \biggl(\sum_x \rho(x|y)^\alpha \biggr)^{\frac{1}{\alpha}} \bigg) \,,
\end{align}
where the latter was originally suggested by Arimoto~\cite{arimoto75}.

%%%%%%%%%%

\subsection{Data-Processing Inequalities and Concavity}
\label{sc:cond-dp}

Let us first discuss some important properties that immediately follow from the respective properties of the underlying divergence.
First, the conditional R\'enyi entropies satisfy a data-processing inequality. 
\begin{corollary}
\begin{svgraybox}
For any channel $\sE \in \cptp(B,B')$ with $\tau_{AB'} = \sE(\rho_{AB})$ for any state $\rho_{AB} \in \cSnorm(AB)$, we have 
\begin{align}
  \HHo_{\alpha}(A|B)_{\rho} \leq \HHo_{\alpha}(A|B')_{\tau} \quad &\textrm{for} \quad \alpha \in [0, 2] \\
  \HHn_{\alpha}(A|B)_{\rho} \leq \HHn_{\alpha}(A|B')_{\tau} \quad &\textrm{for} \quad \alpha \geq \frac12 \,.
\end{align}
(Here, $\HHo_{\alpha}$ is a substitute for either $\Ho_{\alpha}^{\ua}$ or $\Ho_{\alpha}^{\da}$, and the same for $\HHn_{\alpha}$.)
\end{svgraybox}
\end{corollary}

\noindent In particular, these entropies thus satisfy strong subadditivity in the form
\begin{align}
  \HH_{\alpha}(A|BC)_{\rho} \leq \HH_{\alpha}(A|B)_{\rho} 
\end{align}
for the respective ranges of $\alpha$.

Furthermore, it is easy to verify that these entropies are invariant under applications of local isometries on either the $A$ or $B$ systems. Moreover, for any sub-unital map $\sF \in \cptp(A,A')$ and $\tau_{A'B} = \sF(\rho_{AB})$, we get
\begin{align}
  \Ho_{\alpha}^{\da}(A'|B)_{\tau} = - \Do(\tau_{A'B} \| \id_{A'} \otimes \tau_B) &\geq - \Do(\tau_{A'B} \| \sF(\id_A) \otimes \tau_B) \\
  &\geq - \Do(\rho_{AB} \| \id_A \otimes \rho_B) = \Ho_{\alpha}^{\da}(A|B)_{\rho} \,.
\end{align}
and an analogous argument for the other entropies reveals $\HH_{\alpha}(A'|B)_{\tau} \geq \HH_{\alpha}(A|B)_{\rho}$ for all entropies with data-processing. Hence, sub-unital maps on $A$ do not decrease the uncertainty about $A$. However, note that the condition that the map be sub-unital is crucial, and counter-examples are abound if it is not.

Finally, as for the divergence itself, the above data-processing inequalities remain valid if the maps $\sE$ and $\sF$ are trace non-increasing and $\tr(\sE(\rho)) = \tr(\rho)$ and $\tr(\sF(\rho)) = \tr(\rho)$, respectively.

As another consequence of the joint concavity of $\Qo_{\alpha}$ for $\alpha < 1$, we find that $\rho \mapsto \HHo_{\alpha}(A|B)_{\rho}$ is concave for all $\alpha \in [0,1]$. Moreover it is quasi-concave for $\alpha \in [1, 2]$. Similarly $\rho \mapsto \HHn_{\alpha}(A|B)_{\rho}$ is concave for all $\alpha \in [\frac12, 1]$ and quasi-concave for $\alpha > 1$.
%(Note that the maximum of a concave function is not necessarily concave, but joint concavity suffices to establish concavity after optimizing over one variable.)

%%%

\section{Duality Relations and their Applications}
\label{sc:rdual}

We have now introduced four different quantum conditional R\'enyi entropies. Here we show that these definitions are in fact related and complement each other via duality relations.
It is well known that, for any tripartite pure state $\rho_{ABC}$, the relation
\begin{align}
  H(A|B)_{\rho} + H(A|C)_{\rho} = 0 \label{eq:dual-vn}
\end{align}
holds. We call this a \emph{duality relation} for the conditional entropy. To see this, simply write $H(A|B)_{\rho} = H(\rho_{AB}) - H(\rho_B)$ and $H(A|C)_{\rho} = H(\rho_{AC}) - H(\rho_C)$ and verify consulting the Schmidt decomposition that the spectra of $\rho_{AB}$ and $\rho_C$ as well as the spectra of $\rho_{B}$ and $\rho_{AC}$ agree. The significance of this relation is manyfold\,---\,for example it turns out to be useful in cryptography where the information an adversarial party, let us say $C$, has about a quantum system $A$, can be estimated using local state tomography by two honest parties, $A$ and $B$.

In the following, we are interested to see if such relations hold more generally for conditional R\'enyi entropies.

\subsection{Duality Relation\texorpdfstring{ for $\Ho_{\alpha}^{\da}$}{ 1}}

It was shown in~\cite{tomamichel08} that $\Ho_{\alpha}^{\da}$ indeed satisfies a duality relation.
\begin{proposition}
\label{pr:dual-old}
\begin{svgraybox}
  For any pure state $\rho_{ABC} \in \cSnorm(ABC)$, we have
\begin{align}
  \Ho_{\alpha}^{\da}(A|B)_{\rho} + \Ho_{\beta}^{\da}(A|C)_{\rho} = 0 \qquad \textrm{when}\quad \alpha + \beta = 2,\ \alpha, \beta \in [0, 2] \,.
\end{align}
\end{svgraybox}
\end{proposition}
\begin{petit}
\begin{proof}
  By definition, we have $\Ho_{\alpha}^{\da}(A|B)_{\rho} = \frac{1}{1-\alpha} \log \Qo_{\alpha}(\rho_{AB}\| \id_A \otimes \rho_B)$. Now, note that
  \begin{align}
    \Qo_{\alpha}(\rho_{AB}\| \id_A \otimes \rho_B) = \tr (\rho_{AB}^{\alpha} \rho_B^{1-\alpha} ) 
    &= \tr \big(\rho_{AB}^{\alpha -1} \proj{\rho}_{ABC} \rho_B^{1 - \alpha} \big) \\
    &= \tr \big( \rho_C^{\alpha-1} \proj{\rho}_{ABC} \rho_{AC}^{1-\alpha} \big) = \tr (\rho_C^{\alpha-1} \rho_{AC}^{2-\alpha} ) \,.
  \end{align}
  The result then follows by substituting $\alpha = 2 - \beta$. \qed
\end{proof}
\end{petit}

Note that the map $\alpha \mapsto \beta = 2 - \alpha$ maps the interval $[0, 2]$, where data-processing holds, onto itself. This is not surprising. Indeed, consider the Stinespring dilation $\sU \in \cptp(B,B'B'')$ of a quantum channel $\sE \in \cptp(B,B')$. Then, for $\rho_{ABC}$ pure, $\tau_{AB'B''C} = \sU(\rho_{ABC})$ is also pure and the above duality relation implies that
\begin{align}
  H_{\alpha}^{\da}(A|B)_{\rho} \leq H_{\alpha}^{\da}(A|B')_{\tau} \iff H_{\beta}^{\da}(A|C)_{\rho} \geq H_{\beta}^{\da}(A|B''C)_{\tau} .
\end{align}
Hence, data-processing for $\alpha$ holds if and only if data-processing for $\beta$ holds.

\subsection{Duality Relation\texorpdfstring{ for $\Hn_{\alpha}^{\ua}$}{ 2}}

It was shown in~\cite{lennert13,beigi13} that a similar relation holds for $\Hn_{\alpha}^{\ua}$, generalizing a well-known relation between the min- and max-entropies~\cite{koenig08}.
\begin{proposition}
\begin{svgraybox}
  \label{pr:dual-new}
  For any pure state $\rho_{ABC} \in \cSnorm(ABC)$, we have
\begin{align}
  \Hn_{\alpha}^{\ua}(A|B)_{\rho} + \Hn_{\beta}^{\ua}(A|C)_{\rho} = 0 \qquad \textrm{when}\quad \frac{1}{\alpha} + \frac{1}{\beta} = 2,\ \alpha, \beta \in \Big[\frac12,\infty\Big] \,.
\end{align}
\end{svgraybox}
\end{proposition}
\begin{petit}
\begin{proof}
   Without loss of generality, we assume that $\alpha > 1$ and $\beta < 1$.
   Since $(0, 1) \ni \alpha' := \frac{\alpha-1}{\alpha} = - \frac{\beta-1}{\beta} =: -\beta'$, it suffices to show that
   \begin{align}
     \min_{\sigma_B \in \cSnorm(B)} \Big( \Qn_{\alpha}(\rho_{AB}\|\id_A \otimes \sigma_B) \Big)^{\frac{1}{\alpha}} = \max_{\sigma_B \in \cSnorm(B)} \Big( \Qn_{\beta}(\rho_{AB}\|\id_A \otimes \sigma_B) \Big)^{\frac{1}{\beta}} \,,
   \end{align}
   or, equivalently,
   %\begin{align}
   $\min_{\sigma_B \in \cSnorm(B)} \big\| \rho_{AB}^{\nicefrac12} \sigma_B^{-{\alpha'}}\rho_{AB}^{\nicefrac12} \big\|_{\alpha} = \max_{\tau_C \in \cSnorm(C)} \big\| \rho_{AC}^{\nicefrac12} \tau_C^{-{\beta'}} \rho_{AC}^{\nicefrac12} \big\|_{\beta} $.
   %\end{align}
   Now, leveraging the H\"older and reverse H\"older inequalities in Lemma~\ref{lm:hoelder}, we find for any $M \in \cP(A)$,
   \begin{align}
      \| M \|_{\alpha} &= \max \Big\{ \tr (M N) : N \geq 0, \| N \|_{1/\alpha'} \leq 1 \Big\} 
      = \max_{\tau \in \cSnorm(A)} \tr \big(M \tau^{{\alpha'}}\big) , \quad \textrm{and} \\
      \| M \|_{\beta} &= \min \Big\{ \tr (MN) : N \geq 0, N \gg M, \| N^{-1} \|_{-1/\beta'} \leq 1 \Big\} 
      = \min_{\sigma \in \cSnorm(A) \atop
      \sigma \gg M} \tr\big(M \sigma^{{\beta'}}\big)  \,.
   \end{align}
   In the last expression we can safely ignore operators $\sigma \not\gg M$ since those will certainly not achieve the minimum. Substituting this into the above expressions,
   we find
   \begin{align}
      \Big\| \rho_{AB}^{\nicefrac12} \sigma_B^{-{\alpha'}} \rho_{AB}^{\nicefrac12} \Big\|_{\alpha} =
     \max_{\tau_{AB} \in \cSnorm(AB)} \tr \Big( \rho_{AB}^{\nicefrac12} \sigma_B^{-{\alpha'}} \rho_{AB}^{\nicefrac12} \tau_{AB}^{{\alpha'}} \Big)
   \end{align}
   and, furthermore, choosing $\ket{\Psi} \in \cP(ABC)$ to be the unnormalized maximally entangled state with regards to the Schmidt bases of $\ket{\rho}_{ABC}$ in the decomposition $AB:C$, we find
   \begin{align}
     \max_{\tau_{AB} \in \cSnorm(AB)} \tr \Big( \rho_{AB}^{\nicefrac12} \sigma_B^{-{\alpha'}} \rho_{AB}^{\nicefrac12} \tau_{AB}^{{\alpha'}} \Big) 
     &= \max_{\tau_{C} \in \cSnorm(C)} \bracketB{\Psi}{\rho_{AB}^{\nicefrac12} \sigma_B^{-{\alpha'}} \rho_{AB}^{\nicefrac12} \otimes \tau_{C}^{{\alpha'}}}{\Psi}_{ABC} \\
     &= \max_{\tau_{C} \in \cSnorm(C)} \bracketB{\rho}{\sigma_B^{-{\alpha'}} \otimes \tau_{C}^{{\alpha'}}}{\rho}_{ABC} \,.
   \end{align}
   An analogous argument also reveals that
   \begin{align}
     \Big\| \rho_{AC}^{\nicefrac12} \tau_C^{-{\beta'}} \rho_{AC}^{\nicefrac12} \Big\|_{\beta} = \min_{\sigma_{B} \in \cSnorm(B)} \bracketB{\rho}{\sigma_B^{{\beta'}} \otimes \tau_{C}^{-{\beta'}}}{\rho}_{ABC} = \min_{\sigma_{B} \in \cSnorm(B)} \bracketB{\rho}{\sigma_B^{-{\alpha'}} \otimes \tau_{C}^{{\alpha'}}}{\rho}_{ABC} \,.
   \end{align}
   At this points it only remains to show that the minimum over $\sigma_B$ and the maximum over $\tau_C$ can be interchanged. This can be verified using Sion's minimax theorem~\cite{sion58}, noting that $\bracketn{\rho}{\sigma_B^{-{\alpha'}} \otimes \tau_{C}^{{\alpha'}}}{\rho}_{ABC}$ is convex in $\sigma_B$ and concave in $\tau_C$, and we are optimizing over a compact convex space. \qed
\end{proof}
\end{petit}

We again note that the map $\alpha \mapsto \beta = \frac{\alpha}{2\alpha-1}$ maps $\big[\frac12,\infty]$ onto itself.

\subsection{Duality Relation\texorpdfstring{ for $\Ho_{\alpha}^{\ua}$ and $\Hn_{\alpha}^{\da}$}{ 3}}

The alternative expression in Lemma~\ref{lm:dau-new} leads us to the final duality relation, which establishes a surprising connection between two quantum R\'enyi entropies~\cite{tomamichel13}.

\begin{proposition}
\label{pr:dual-both}
\begin{svgraybox}
  For any pure state $\rho_{ABC} \in \cSnorm(ABC)$, we have
\begin{align}
  \Ho_{\alpha}^{\ua}(A|B)_{\rho} + \Hn_{\beta}^{\da}(A|C)_{\rho} = 0 \qquad \textrm{when}\quad \alpha \beta = 1,\ \alpha, \beta \in [0,\infty] \,.
\end{align}
\end{svgraybox}
\end{proposition}

\begin{petit}
\begin{proof}
  First we note that $\beta = \frac{1}{\alpha}$ and $\frac{\alpha}{1-\alpha} = - \frac{1}{1-\beta}$. Then, using the expression in Lemma~\ref{lm:dau-new},
  it remains to show that 
  \begin{align}
     \tr \Big( \big( \tr_A ( \rho_{AB}^{\alpha} ) \big)^{\frac{1}{\alpha}}  \Big) = \tr \Big( \Big( \rho_C^{\alpha'} \rho_{AC} \rho_C^{\alpha'} \Big)^{\frac{1}{\alpha}}  \Big) , \quad \textrm{where} \quad \alpha' = \frac{\alpha-1}{2} \, .
  \end{align}
    In the following we show something stronger, namely that the operators 
  \begin{equation}
   \tr_A ( \rho_{AB}^{\alpha} ) \qquad \textrm{and} \qquad
   \rho_C^{\alpha'} \rho_{AC} \rho_C^{\alpha'} \label{eq:marginals}
   \end{equation}
   are unitarily equivalent. This is true since both of these operators are marginals\,---\,on $B$ and $AC$\,---\,of the same tripartite rank-$1$ operator, 
   %\begin{align}
     $\rho_{C}^{\alpha'}  \rho_{ABC}  \rho_{C}^{\alpha'}$.
   %\end{align}
    To see that this is indeed true, note the first operator in~\eqref{eq:marginals} can be rewritten as
   \begin{align}
     \tr_A ( \rho_{AB}^{\alpha} ) &= \tr_A \big( \rho_{AB}^{\alpha'} \rho_{AB}\, \rho_{AB}^{\alpha'} \big) 
     = \tr_{AC} \big( \rho_{AB}^{\alpha'} \rho_{ABC} \rho_{AB}^{\alpha'}  \big) 
     = \tr_{AC} \big( \rho_{C}^{\alpha'} \rho_{ABC} \rho_{C}^{\alpha'} \big) \, .
   \end{align}
   The last equality can be verified using the Schmidt decomposition of $\rho_{ABC}$ with regards to the partition $AB$:$C$. \qed
\end{proof}
\end{petit}

Again, note that the transformation $\alpha \mapsto \beta = \frac{1}{\alpha}$ maps the interval $[0, 2]$ where data-processing holds for $\HHo_{\alpha}$ to the interval $[\frac{1}{2}, \infty]$ where data-processing holds for $\HHn_{\beta}$, and vice versa.

%%%

\subsection{Additivity for Tensor Product States}

One implication of the duality relation for $\Hn_{\alpha}^{\ua}$ is that it allows us to show additivity for this quantity. Namely, we can use it to show the following corollary.

\begin{corollary}
\begin{svgraybox}
  For any product state $\rho_{AB} \otimes \tau_{A'B'}$ and $\alpha \in [\frac12,\infty)$, we have
  \begin{align}
     \Hn_{\alpha}^{\ua}(AA'|BB')_{\rho\otimes\tau} = \Hn_{\alpha}^{\ua}(A|B)_{\rho} + \Hn_{\alpha}^{\ua}(A'|B')_{\tau} \,.
  \end{align}
\end{svgraybox}
\end{corollary}

\begin{petit}
\begin{proof}
  By definition of $\Hn_{\alpha}^{\ua}(AA'|BB')_{\rho\otimes\tau}$ we immediately find the following chain of inequalities:
  \begin{align}
    \Hn_{\alpha}^{\ua}(AA'|BB')_{\rho\otimes\tau} 
    &= - \min_{\sigma_{BB'} \in \cS(BB')} \label{eq:addchain1}
      \Dn_{\alpha}\big( \rho_{AB} \otimes \tau_{A'B'} \big\| \id_{AA'} \otimes \sigma_{BB'} \big) \\
    &\geq - \min_{\sigma_{B} \in \cS(B), \atop \omega_{B'} \in \cS(B')} 
      \Dn_{\alpha}\big( \rho_{AB} \otimes \tau_{A'B'} \big\| \id_{A} \otimes \sigma_{B} \otimes \id_{A'} \otimes \omega_{B'} \big) \\
    &= \Hn_{\alpha}^{\ua}(A|B)_{\rho} + \Hn_{\alpha}^{\ua}(A'|B')_{\tau} \,. \label{eq:addchain2}
  \end{align}
  
  To establish the opposite inequality we introduce purifications $\rho_{ABC}$ of $\rho_{AB}$ and $\tau_{A'B'C'}$ of $\tau_{A'B'}$ and choose $\beta$ such that 
  $\frac1{\alpha} + \frac1{\beta} = 2$. Then, an instance of the above inequality~\eqref{eq:addchain1}--\eqref{eq:addchain2} reads
  \begin{align}
        \Hn_{\beta}^{\ua}(AA'|CC')_{\rho\otimes\tau} &\geq \Hn_{\beta}^{\ua}(A|C)_{\rho} + \Hn_{\beta}^{\ua}(A'|C')_{\tau} \,.
  \end{align}
  The duality relation in Prop.~\ref{pr:dual-new} then yields
  %\begin{align}
    $\Hn_{\alpha}^{\ua}(AA'|BB')_{\rho\otimes\tau} \leq \Hn_{\alpha}^{\ua}(A|B)_{\rho} + \Hn_{\alpha}^{\ua}(A'|B')_{\tau}$,
  %\end{align}
  concluding the proof.
\qed
\end{proof}
\end{petit}

Finally, note that the corresponding additivity relations for $\Hn_{\alpha}^{\da}$ and $\Ho_{\alpha}^{\da}$ are evident from the respective definition. Additivity for $\Ho_{\alpha}^{\ua}$ in turn follows directly from the explicit expression established in Lemma~\ref{lm:dau-new}.

%%%

\subsection{Lower and Upper Bounds on Quantum R\'enyi Entropy}

The above duality relations also yield relations between different conditional R\'enyi entropies for arbitrary mixed states~\cite{tomamichel13}.
\begin{corollary}
  \label{cor:dual-ineq}
\begin{svgraybox}
  Let $\rho_{AB} \in \cSnorm(AB)$. Then, the following holds for $\alpha \in \left[\frac{1}{2}, \infty\right]$:
  \begin{align}
    %\Ho_{\alpha}^{\ua}(A|B)_{\rho} \leq 
    \Hn_{\alpha}^{\ua}(A|B)_{\rho} &\leq \Ho_{2 - \frac{1}{\alpha}}^{\ua}(A|B)_{\rho}\,,\label{eq:ineq1} \qquad
    %\Ho_{\alpha}^{\da}(A|B)_{\rho} \leq 
    &\Ho_{\alpha}^{\ua}(A|B)_{\rho} &\leq \Ho_{2 - \frac{1}{\alpha}}^{\da}(A|B)_{\rho}\,,\\
    %\Hn_{\alpha}^{\da}(A|B)_{\rho} \leq 
    \Hn_{\alpha}^{\ua}(A|B)_{\rho} &\leq \Hn_{2-\frac{1}{\alpha}}^{\da}(A|B)_{\rho}\,,\label{eq:ineq3}\qquad
    %\Ho_{\alpha}^{\da}(A|B)_{\rho} \leq 
    &\Hn_{\alpha}^{\da}(A|B)_{\rho} &\leq \Ho_{2-\frac{1}{\alpha}}^{\da}(A|B)_{\rho} \,. 
  \end{align}
\end{svgraybox}
\end{corollary}

\begin{petit}
\begin{proof}
  Consider an arbitrary purification $\rho_{ABC} \in \cS(ABC)$ of $\rho_{AB}$. The relations of Fig.~\ref{fig:overview}, for any $\gamma \geq 0$, applied to the marginal $\rho_{AC}$ are given as
  \begin{align}
    &\Hn_{\gamma}^{\ua}(A|C)_{\rho} \geq \Hn_{\gamma}^{\da}(A|C)_{\rho} \geq \Ho_{\gamma}^{\da}(A|C)_{\rho}\,, \qquad \textrm{and} \\
    &\Hn_{\gamma}^{\ua}(A|C)_{\rho} \geq \Ho_{\gamma}^{\ua}(A|C)_{\rho} \geq \Ho_{\gamma}^{\da}(A|C)_{\rho}\,.
  \end{align}
  We then substitute the corresponding dual entropies according to the duality relations in Sec.~\ref{sc:rdual}, which yields the desired inequalities upon appropriate new parametrization. \qed
\end{proof}
\end{petit}

Some special cases of these inequalities are well known and have operational significance. For example,~\eqref{eq:ineq3} for $\alpha = \infty$ states that $\widetilde{H}_{\infty}^{\ua}(A|B)_{\rho} \leq \widetilde{H}_{2}^{\da}(A|B)_{\rho}$, which relates the conditional min-entropy in~\eqref{eq:min2} to the conditional collision entropy in~\eqref{eq:h2}. To understand this inequality more operationally we rewrite the conditional min-entropy as its dual semi-definite program~\cite{koenig08} (see also Chatper~\ref{ch:calc}),
\begin{align}
\widetilde{H}_{\infty}^{\ua}(A|B)_{\rho} = \min_{\sE \in \cptp(B,A')}-\log\big(d_A \, F(\psi_{AA'},\sE(\rho_{AB})\big)\,,
\end{align}
where $A'$ is a copy of $A$ and $\psi_{AA'}$ is the maximally entangled state on $A:A'$. Now, the above inequality becomes apparent since the conditional collision entropy can be written as~\cite{berta13}
\begin{align}
\Hn_{2}^{\da}(A|B)_{\rho} = -\log\big(d_A\, F(\phi_{AA'},\sE^{\mathrm{pg}}(\rho_{AB})\big)\,,
\end{align}
where $\sE^{\mathrm{pg}}$ denotes the pretty good recovery map of Barnum and Knill~\cite{barnum02}.

Finally,~\eqref{eq:ineq1} for $\alpha = \frac12$ yields $\Hn_{\nicefrac{1}{2}}^{\ua}(A|B)_{\rho} \leq \Ho_0^{\ua}(A|B)_{\rho}$, which relates the quantum conditional max-entropy in~\eqref{eq:hmax} to the quantum conditional generalization of the Hartley entropy in~\eqref{eq:h0}.

%%%%%%%%%%%

\subsubsection*{Dimension Bounds}

First, note two particular inequalities from Corollary~\ref{cor:dual-ineq}: 
\begin{align}
 \Hn_{\infty}^{\da}(A|B)_{\rho} \leq \Ho_{2}^{\da}(A|B)_{\rho} \quad \textrm{and} \quad 
 \Hn_{\nicefrac12}^{\ua}(A|B)_{\rho} \leq \Ho_{0}^{\ua}(A|B)_{\rho} \,.
\end{align} 
From this and the monotonicity in $\alpha$, we find that all conditional entropies (that satisfy the data-processing inequality) can be upper and lower bounded as follows.
\begin{align}
  \Hn_{\infty}^{\da}(A|B)_{\rho} \leq \HH_{\alpha}(A|B)_{\rho} \leq \Ho_0^{\ua}(A|B)_{\rho} \,.
\end{align}
Thus, in order to find upper and lower bounds on quantum R\'enyi entropies it suffices to investigate these two quantities.

\begin{lemma}
\begin{svgraybox}
\label{lm:lubounds}
  Let $\rho_{AB} \in \cSnorm(AB)$. Then the following holds:
  \begin{align}
     -\log \min \{ \rank(\rho_A), \rank(\rho_B) \} \leq \HH_{\alpha}(A|B)_{\rho} \leq \log \rank(\rho_A) \,.
  \end{align}
  Moreover, $\HH_{\alpha}(A|B)_{\rho} \geq 0$ if $\rho_{AB}$ is separable.
\end{svgraybox}
\end{lemma}

\begin{petit}
\begin{proof}
  Without loss of generality (due to invariance under local isometries) we assume that $\rho_A$ and $\rho_B$ have full rank.
  The upper bound follows since $\Ho_0^{\ua}(A|B)_{\rho} \leq H_0(A)_{\rho} = \log d_A$. Similarly, we find $H_{\infty}^{\da}(A|B)_{\rho} = - \Ho_0^{\ua}(A|C)_{\rho} \geq - H_0(A)_{\rho} = -\log d_A$ by taking into account an arbitrary purification $\rho_{ABC}$ of $\rho_{AB}$. 
  On the other hand, for any decomposition $\rho_{AB} = \sum_i \lambda_i \proj{\phi_i}$ into pure states, quasi-concavity of $\HH_{\alpha}$ (which is a direct consequence of the quasi-convexity of $\DD_{\alpha}$) yields
  \begin{align}
    H_{\infty}^{\da}(A|B)_{\rho} \geq \min_i H_{\infty}^{\da}(A|B)_{\phi_i} = \min_i - H_0(A)_{\phi_i} \geq - \log d_B \,.
  \end{align}
  This concludes the proof of the first statement.
  
  For separable states, we may write
  \begin{align}
    \rho_{AB} = \sum_k p_k\, \sigma_A^k \otimes \tau_B^k \leq \sum_k p_k\, \id_A \otimes \tau_B^k = \id_A \otimes \rho_B \,,
  \end{align}
  and, hence, $H_{\infty}^{\da}(A|B)_{\rho} = \sup \{ \lambda \in \mathbb{R} : \rho_{AB} \leq \exp(-\lambda) \id_A \otimes \rho_B \} \geq 0$.
\qed
\end{proof}
\end{petit}

%%%%%%%%%%%%%%

\section{Chain Rules}
\label{sc:rchain}

The chain rule, $H(AB|C) = H(A|BC) + H(B|C)$, is fundamentally important in many applications because it allows us to see the entropy of a system as the sum of the entropies of its parts.
However, $\HH_{\alpha}(AB|C) = \HH_{\alpha}(A|BC) + \HH_{\alpha}(B|C)$, generally does not hold for $\alpha \neq 1$. Nonetheless, there exist weaker statements that we can prove.

For a first such statement, we note that for any $\rho_{ABC} \in \cSnorm(ABC)$, the inequality
\begin{align}
  \rho_{BC} \leq \exp \big( - \Hn_{\infty}^{\da}(B|C)_{\rho} \big) \, \id_B \otimes \rho_{C} 
\end{align}
holds by definition of $\Hn_{\infty}^{\da}$. Hence, using the dominance relation of the R\'enyi divergence, we find
\begin{align}
  \Ho_{\alpha}^{\da}(A|BC)_{\rho} &= - \Do_{\alpha}(\rho_{ABC} \| \id_A \otimes \rho_{BC} ) \\
  &\leq - \Do_{\alpha}(\rho_{ABC} \| \id_{AB} \otimes \rho_C) - \widetilde{H}_{\infty}^{\da}(B|C)_{\rho} ,
\end{align}
or, equivalently
%\begin{align}
  $\Ho_{\alpha}^{\da}(AB|C)_{\rho} \geq \Ho_{\alpha}^{\da}(A|BC)_{\rho} + \Hn_{\infty}^{\da}(B|C)_{\rho}$.
%\end{align}
Using an analogous argument we get the same statement also for $\Hn_{\alpha}$.
\begin{proposition}
\begin{svgraybox}
  For any state $\rho_{ABC} \in \cSnorm(ABC)$, we have
  \begin{align}
    \HH_{\alpha}^{\da}(AB|C)_{\rho} &\geq \HH_{\alpha}^{\da}(A|BC)_{\rho} + \Hn_{\infty}^{\da}(B|C)_{\rho} \,.
  \end{align}
\end{svgraybox}
\end{proposition}
Several other variations of the chain rule can now be established using the duality relations, for example
\begin{align}
  \Ho_{\alpha}^{\ua}(AB|C)_{\rho} \leq \Ho_0^{\ua}(A|BC)_{\rho} + \Ho_{\alpha}^{\ua}(B|C)_{\rho} \,.
\end{align}

Next, let us try to find a chain rule that only involves entropies of the `$\uparrow$' type. For this purpose, we follow the above argument but
start with the fact that
\begin{align}
  \rho_{BC} \leq \exp \big( - \Hn_{\infty}^{\ua}(B|C)_{\rho} \big) \, \id_B \otimes \sigma_{C} 
\end{align}
for some $\sigma_C \in \cSnorm(C)$. This yields the relation
\begin{align}
  \Hn_{\alpha}^{\ua}(AB|C)_{\rho} \geq \Hn_{\alpha}^{\da}(A|BC)_{\rho} + \Hn_{\infty}^{\ua}(B|C)_{\rho}
\end{align}
and we can use the inequality in~\eqref{eq:ineq3} to remove the remaining `$\downarrow$'. This leads to 
\begin{align}
  \Hn_{\alpha}^{\ua}(AB|C)_{\rho} \geq \Hn_{\beta}^{\ua}(A|BC)_{\rho} + \Hn_{\infty}^{\ua}(B|C)_{\rho} , \quad \alpha = 2-\frac{1}{\beta} \,.
\end{align}
This result is a special case of a beautiful set of chain rules for $\Hn_{\alpha}^{\ua}$ that were recently established by Dupuis~\cite{dupuis14}.

\begin{theorem}
\begin{svgraybox}
\label{th:chains}
  Let $\rho_{ABC} \in \cS(ABC)$ and $\alpha,\beta,\gamma \in \big(\frac12, 1\big) \cup (1,\infty)$ such that $\frac{\alpha}{\alpha-1} = \frac{\beta}{\beta-1} + \frac{\gamma}{\gamma-1}$. 
  Then, if $(\alpha-1)(\beta-1)(\gamma-1) > 0$,
  \begin{align}
    H_{\alpha}^{\ua}(AB|C)_{\rho} \geq H_{\beta}^{\ua}(A|BC)_{\rho} + H_{\gamma}^{\ua}(B|C)_{\rho} \,,
  \end{align}
  and the inequality is reversed if $(\alpha-1)(\beta-1)(\gamma-1) < 0$.
\end{svgraybox}
\end{theorem}

The proof in~\cite{dupuis14} is outside the scope of this book (see also Beigi~\cite{beigi13}). 
The chain rules for the von Neumann entropy follow as a limit of the above relation. For example, if we choose $\beta = \gamma = 1 + 2\eps$ so that 
$\alpha = \frac{1+2\eps}{1+\eps}$ for a small parameter $\eps \to 0$, we recover the relation
\begin{align}
  H(AB|C)_{\rho} \geq H(A|BC)_{\rho} + H(B|C)_{\rho} \,.
\end{align}
The opposite inequality follows by choosing $\beta = \gamma = 1 - 2\eps$.

Finally, we want to stress that slightly stronger chain rules are sometimes possible when the underlying state has structure.

%%%

\subsubsection*{Entropy of Classical Information}

We explore this with the example of classical and coherent-classical quantum states, which arise when we purify classical systems. 
For concreteness, consider a state $\rho \in \cSsub(XAB)$ that is classical on $X$, and a purification of the form
\begin{align}
 \rho_{XX'ABC} := \sum_{x,x'} \ket{x'}\!\bra{x}_X \otimes \ket{x'}\!\bra{x}_{X'} \otimes \ket{\rho(x')}\!\bra{\rho(x)}_{ABC} , 
\end{align}
where $\rho_{ABC}(x)$ is a purification of $\rho_{AB}(x)$.
We say that $\rho_{XX'ABC}$ is \emph{coherent-classical} between $X$ and $X'$: if one of these systems is traced out the remaining states are isomorphic and classical on $X$ or $X'$, respectively.

\begin{lemma}
  Let $\rho \in \cSsub(XX'AB)$ be coherent-classical between $X$ and $X'$. Then, 
  \begin{align}
    \HH_{\alpha}^{\ua}(XA|X'B)_{\rho} \leq \HH_{\alpha}^{\ua}(A|XX'B)_{\rho} \quad \textrm{and} \quad \HHn_{\alpha}(XA|B)_{\rho} \geq \HHn_{\alpha}(A|B)_{\rho} \,.
  \end{align}
\end{lemma}
The second statement reveals that classical information has non-negative entropy, regardless of the nature of the state on $AB$. (Note that Lemma~\ref{lm:lubounds} already established this fact for the case where $A$ is trivial.)
\begin{petit}
\begin{proof}
  We will establish the first inequality for all conditional R\'enyi entropies of the type `$\uparrow$'. The second inequality then follows by the respective duality relations, and a relabelling $B \leftrightarrow C$.
  
  We consider the case $\alpha \in [\frac12, 1)$ such that $\zeta = \frac{1-\alpha}{\alpha} \in (0,1]$, and the entropy $\widetilde{H}_{\alpha}$. We find
  \begin{align}
    \widetilde{Q}_{\alpha}(P \rho P \| \sigma) %&= \tr \left( \left( \left( P \rho P \right)^{\frac12} \sigma^{\frac{1-\alpha}{\alpha}} \left( P \rho P \right)^{\frac12} \right)^{\alpha} \right) 
    &= \tr \left( \left( \left( P \rho P \right)^{\frac12} P \sigma^{\frac{1-\alpha}{\alpha}} P \left( P \rho P \right)^{\frac12} \right)^{\alpha} \right) \\
    &\leq \tr \left( \left( \left( P \rho P \right)^{\frac12} ( P \sigma P)^{\frac{1-\alpha}{\alpha}} \left( P \rho P \right)^{\frac12} \right)^{\alpha} \right) = \widetilde{Q}_{\alpha}(P \rho P \| P \sigma P ) ,
  \end{align}
where the inequality follows by the operator Jensen inequality in~\eqref{eq:jensen} for sub-unital maps.
  By a similar argument, one can verify that $\widebar{Q}_{\alpha}(P \rho P \| \sigma) \leq \widebar{Q}_{\alpha}(P \rho P \| P \sigma P )$ for $\alpha \in [0, 1)$. 

  Now define the projector $\Pi_{XX'} = \sum_x \proj{x}_X \otimes \proj{x}_{X'}$ such that $\rho_{XX'AB} = \Pi_{XX'} \rho_{XX'AB} \Pi_{XX'}$. For any $\sigma \in \cSnorm(X'B)$, it holds that
   \begin{align}
     \QQ_{\alpha}(\rho_{XX'AB} \| \id_{XA} \otimes \sigma_{X'B}) &\leq \QQ_{\alpha}(\rho_{XX'AB} \| \id_{A} \otimes \Pi_{XX'} (\id_{X'} \otimes \sigma_{X'B}) \Pi_{XX'}) \\
 %    &= \QQ_{\alpha}\Big(\rho_{XX'AB} \,\Big\|\, \id_{A} \otimes \sum_x \proj{x}_{X} \otimes \bra{x} \sigma_{X'B} \ket{x}_{X'}\Big)
      &\leq \max_{\sigma \in \cSnorm(XX'B)} \QQ_{\alpha}(\rho_{XX'AB} \| \id_{A} \otimes \sigma_{XX'B}) \, ,
   \end{align}
   where we used that $\tr(\Pi_{XX'} (\id_{X'} \otimes \sigma_{X'B}) \Pi_{XX'}) = \tr(\sigma_{X'B}) = 1$. 
   From this we conclude that the desired statement holds for $\alpha < 1$. 
   
   Analogous arguments with inequalities in the opposite direction apply for $\alpha > 1$, although some additional care has to be taken due to the discontinuity at $0$ of the inverse function.
\qed
\end{proof}
\end{petit}

Finally, the following result gives dimension-dependent bounds on how much information a classical register can contain.
\begin{lemma}
  Let $\rho \in \cSsub(XAB)$ be classical on $X$. Then,
  \begin{align}
    &\HH_{\alpha}^{\ua}(XA|B)_{\rho} \leq \HH_{\alpha}^{\ua}(A|XB) + \log d_X \,.
   \end{align}
\end{lemma}
\begin{petit}
\begin{proof}
  Simply note that for any $\sigma_B \in \cSnorm(B)$, we have
  \begin{align}
    \DD_{\alpha}(\rho_{XAB} \| \id_{XA} \otimes \sigma_{B} ) &= \DD_{\alpha}(\rho_{AXB} \| \id_{A} \otimes (\pi_X \otimes \sigma_{B)} ) - \log d_X \\
    &\geq \min_{\sigma_{XB} \in \cSnorm(XB)} \DD_{\alpha}(\rho_{AXB} \| \id_{A} \otimes \sigma_{XB} ) - \log d_X\,.
  \end{align}
\qed
\end{proof}
\end{petit}
For example, combining the above two lemmas, we find that
\begin{align}
  \Hn_{\alpha}^{\ua}(A|B)_{\rho} \leq \Hn_{\alpha}^{\ua}(AX|B)_{\rho} \leq \Hn_{\alpha}^{\ua}(A|BX)_{\rho} + \log d_X \,.
  \label{eq:ua-class-ineq}
\end{align}

%%%%%%%%%%%%

\section{Background and Further Reading}

Strong subadditivity~\eqref{eq:strong-sub-add} was first conjectured by Lanford and Robinson in~\cite{lanford68}. Its first proof by Lieb and Ruskai~\cite{lieb73} is one of the most celebrated results in quantum information theory. The original proof
is based on Lieb's theorem~\cite{lieb73a}. Simpler proofs were subsequently presented by Nielsen and Petz~\cite{nielsenpetz04} and Ruskai~\cite{ruskai07}, amongst others.
In this book we proved this statement indirectly via the data-processing inequality for the relative entropy, which in turns follows by continuity from the data-processing inequality for the R\'enyi divergence in Chapter~\ref{ch:renyi}. We also provide an elementary proof in Appendix~\ref{app:analysis}.

The classical version of $H_{\alpha}^{\ua}$ was introduced by Arimoto for an evaluation of the guessing probability~\cite{arimoto75}. Gallager used $H_{\alpha}^{\ua}$ to upper bound the decoding error probability of a random coding scheme for data compression with side-information~\cite{gallager79}.
More recently, the classical and the classical-quantum special cases of $\Ho_{\alpha}^{\ua}$ were investigated by Hayashi (see, for example,~\cite{hayashi12}).

The quantum conditional R\'enyi entropy $\Ho_{\alpha}^{\da}$ was first studied in~\cite{tomamichel08}. We note that the expression for $\Ho_{\alpha}^{\ua}$ in Lemma~\ref{lm:dau-new} can be derived using a quantum Sibson's identity, first proposed by Sharma and Warsi~\cite{sharma13}.
On the other hand, the quantum R\'enyi entropy $\Hn_{\alpha}^{\ua}$ was proposed in~\cite{mytutorial12} and investigated in~\cite{lennert13}, whereas $\Hn_{\alpha}^{\da}$ is first considered in~\cite{tomamichel13}.

It is an open question whether the inequalities in Corollary~\ref{cor:dual-ineq} also hold for the R\'enyi divergences themselves.
Relatedly, Mosonyi~\cite{mosonyi13} used a converse of the Araki-Lieb-Thirring trace inequality due to Audenaert~\cite{Audenaert08} to find a converse to the ordering $\Do_{\alpha}(\rho\|\sigma) \geq \Dn_{\alpha}(\rho\|\sigma)$, namely
\begin{align}
\Dn_{\alpha}(\rho\|\sigma) \geq \alpha\, \Do_{\alpha}(\rho\|\sigma) + \log \tr(\rho) - \log \tr \big(\rho^{\alpha} \big) +(\alpha - 1)\log\|\sigma\| \,.
\end{align}

In this book we focus our attention on conditional R\'enyi entropies, but similar techniques can also be used to explore R\'enyi generalizations of the mutual information~\cite{gupta13,hayashitomamichel14} and conditional mutual information~\cite{bertawilde14}.

%%%%%%%%%

%!TEX root = book.tex

\newcommand{\fid}[2]{\left\|{\sqrt{#1}\sqrt{#2}}\right\|}
\newcommand{\fidn}[2]{\|{\sqrt{#1}\sqrt{#2}}\|}
\newcommand{\fidb}[2]{\big\|{\sqrt{#1}\sqrt{#2}}\big\|}

\chapter{Smooth Entropy Calculus}
\label{ch:calc} 
% use \chaptermark{}
% to alter or adjust the chapter heading in the running head

%%%%%%%%%%%%%%%%%%%%%%%%%%%

\abstract*{Smooth R\'enyi entropies are defined as optimizations (either minimizations or maximization) of R\'enyi entropies over a set of close states. For many applications it suffices to consider just two smooth R\'enyi entropies: the smooth min-entropy acts as a representative of all conditional R\'enyi entropies with $\alpha > 1$, whereas the smooth max-entropy acts as a representative for all R\'enyi entropies with $\alpha < 1$. These two entropies have particularly nice properties and can be expressed in various different ways, for example as semi-definite optimization problems. 
Most importantly, they give rise to an entropic (and fully quantum) version of the asymptotic equipartition property, which states that both the (regularized) smooth min- and max-entropies converge to the conditional von Neumann entropy for iid product states. This is because smoothing implicitly allows us to restrict our attention to a typical subspace where all conditional R\'enyi entropies coincide with the von Neumann entropy. Furthermore, we will see that the smooth entropies inherit many properties of the underlying R\'enyi entropies.}

Smooth R\'enyi entropies are defined as optimizations (either minimizations or maximization) of R\'enyi entropies over a set of close states. For many applications it suffices to consider just two smooth R\'enyi entropies: the smooth min-entropy acts as a representative of all conditional R\'enyi entropies with $\alpha > 1$, whereas the smooth max-entropy acts as a representative for all R\'enyi entropies with $\alpha < 1$. These two entropies have particularly nice properties and can be expressed in various different ways, for example as semi-definite optimization problems. 
Most importantly, they give rise to an entropic (and fully quantum) version of the asymptotic equipartition property, which states that both the (regularized) smooth min- and max-entropies converge to the conditional von Neumann entropy for iid product states. This is because smoothing implicitly allows us to restrict our attention to a typical subspace where all conditional R\'enyi entropies coincide with the von Neumann entropy. Furthermore, we will see that the smooth entropies inherit many properties of the underlying R\'enyi entropies.

%%%%%%%%%%%%%%%%%%%%%%%%%%%

\section{Min- and Max-Entropy}
\label{sc:minmax}

This section develops a variety of useful alternative expressions for the min- and max-entropies, $\Hn_{\infty}^{\ua}$ and $\Hn_{\nicefrac12}^{\ua}$. In particular, we express both the min- and the max-entropy in terms of semi-definite programs.

\subsection{Semi-Definite Programs}
\label{sc:sdp}

Optimization problems that can be formulated as semi-definite programs are particularly interesting because they have a rich structure and efficient numerical solvers. Here we present a formulation of semi-definite programs that has a very symmetric structure, following Watrous' lecture notes~\cite{watrous-ln08}. 

\begin{definition}
\begin{svgraybox}
A semi-definite program (SDP) is a triple $\{K, L, \sE \}$, where $K \in \cL^{\dag}(A)$, $L \in \cL^{\dag}(B)$ and $\sE \in \cL(\cL(A),\cL(B))$ is a super-operator from $A$ to $B$ that preserves self-adjointness. The following two optimization problems are associated with the semi-definite program:
\begin{align}
  \begin{array}{rlcrl}
    \multicolumn{2}{c}{\underline{\textnormal{primal problem}}} & \qquad \quad  & 
      \multicolumn{2}{c}{\underline{\textnormal{dual problem}}} \vspace{0.2cm} \\
    \textrm{minimize}:\ & \tr(K X) & & \textrm{maximize}:\ & \tr(L Y) \\
    \textrm{subject to}:\ & \sE(X) \geq L & & \textrm{subject to}:\ & \sE^{\dag}(Y) \leq K \\
    & X \in \cP(A) &&& Y \in \cP(B)
  \end{array} 
\end{align}
\end{svgraybox}
\end{definition}

We call an operator $X \in \cP(A)$ primal feasible if it satisfies $\sE(X) \geq L$. Similarly, we say that $Y \in \cP(B)$ is dual feasible if $\sE^{\dag}(Y) \leq K$. Moreover, we denote the optimal solution of the primal problem by $a$ and the optimal solution of the dual problem by $b$. Formally, we define
\begin{align}
  a &= \inf \big\{ \tr(K X) : X \in \cP(A),\ \sE(X) \geq L \big\} \label{eq:sdp/alphabeta1} \\
  b &= \sup \big\{ \tr(L Y) : Y \in \cP(B),\ \sE^{\dag}(Y) \leq K \big\} \label{eq:sdp/alphabeta2}.  
\end{align}

The following two statements are true for any SDP and provide a relation between the primal and dual problem. The first fact is called \emph{weak duality}, and the second statement is also known as Slater's condition for \emph{strong duality}.

%\begin{lemma}
%  \label{lm:sdp-dual}
  %Let $\{A, B, \sE\}$ be a SDP and let $a$, $b$ be defined as in~\eqref{eq:sdp/alphabeta1} and~\eqref{eq:sdp/alphabeta2}, respectively. Then, the following holds.
  \begin{description}
    \item[Weak Duality:] We have $a \geq b$. 
    \item[Strong Duality:] If $a$ is finite and there exists an operator $Y > 0$ such that $\sE^{\dag}(Y) < K$, then $a = b$ 
      and there exists a primal feasible $X$ such that $\tr(K X) = a$.
  \end{description}
%\end{lemma}

For a proof we defer to~\cite{watrous-ln08}.
As an immediate consequence, this implies that every dual feasible operator $Y$ provides a lower bound of 
$\tr(L Y)$ on $a$ and every primal feasible operator $X$ provides an upper bound of 
$\tr(K X)$ on $b$.

\subsection{The Min-Entropy}

We first recall the expression for $\Hn_{\infty}^{\ua}$ in~\eqref{eq:min2}, which we will simply call \emph{min-entropy} in this chapter.
We extend the definition to include sub-normalized states~\cite{renner05}.

\begin{definition}
  \label{df:min-entropy}
\begin{svgraybox}
  Let $\rho_{AB} \in \cSsub(AB)$. The \textbf{min-entropy} of $A$ conditioned on $B$ of the 
  state $\rho_{AB}$ is
  \begin{align}
    H_{\min}(A|B)_{\rho} =
      \sup_{\sigma_B \in \cSsub(B)}\, \sup \big\{ \lambda \in \bbR : \rho_{AB} \leq \exp(-\lambda) \id_A \otimes \sigma_B \big\} \label{eq:min/def} \,.
  \end{align}
  \vspace{-0.5cm}
\end{svgraybox}
\end{definition}

Let us take a closer look at the inner supremum first.
First, note that there exists a feasible $\lambda$ if and only if $\sigma_B \gg \rho_B$. However, if this 
condition on the support is satisfied, then using the generalized inverse, we find that
\begin{align}
  \lambda_* = -\log \norm{ {\sigma_B}^{-\frac12} \rho_{AB} {\sigma_B}^{-\frac12} }_{\infty}
\end{align}
is feasible and achieves the maximum. The min-entropy can thus alternatively be 
written as
\begin{align}
  H_{\min}(A|B)_{\rho} = \max_{\sigma_B} - \log \norm{{\sigma_B}^{-\frac12} \rho_{AB} {\sigma_B}^{-\frac12}}_{\infty} \label{eq:min/inv},
\end{align}
where we use the generalized inverse and the maximum is taken over all $\sigma_B \in \cSsub(B)$ 
with $\sigma_B \gg \rho_B$. 
\textnormal
We can also reformulate~\eqref{eq:min/def} as a semi-definite program. 

For this purpose, we include the factor $\exp(-\lambda)$ in $\sigma_B$ and allow $\sigma_B$ to be an arbitrary positive semi-definite operator. The min-entropy can then be written as
\begin{align}
  H_{\min}(A|B)_{\rho} = -\log\, \min \big\{ \tr(\sigma_B) : 
    \sigma_B \in \cP(B)\ \wedge\ \rho_{AB} \leq \id_A \otimes \sigma_B \big\} \,.
\end{align}

In particular, we consider the following semi-definite optimization problem for the expression $\exp(-H_{\min}(A|B)_{\rho})$,
which has an efficient numerical solver.
\begin{lemma}
\begin{svgraybox}
Let $\rho_{AB} \in \cSsub(AB)$.
Then, the following two optimization problems satisfy strong duality and both evaluate to $\exp(-H_{\min}(A|B)_{\rho})$.
\begin{align}
  \label{eq:min/sdp}
  \begin{array}{rlcrl}
    \multicolumn{2}{c}{\underline{\textnormal{primal problem}}} & \qquad \quad & 
      \multicolumn{2}{c}{\underline{\textnormal{dual problem}}} \vspace{0.2cm} \\
    \textnormal{minimize}: & \tr(\sigma_B) & & \textnormal{maximize}: & \tr(\rho_{AB} X_{AB}) \\
    \textnormal{subject to}: & \id_A \otimes \sigma_B \geq \rho_{AB} & & \textnormal{subject to}: & \tr_A[X_{AB}] \leq \id_B \\
    & \sigma_B \geq 0 &&& X_{AB} \geq 0
  \end{array}   &
\end{align} 
  \vspace{-0.5cm}
\end{svgraybox}
\end{lemma}

\begin{petit}
\begin{proof}
Clearly, the dual problem has a finite solution; in fact, we always have 
$\tr[\rho_{AB} X_{AB}] \leq \tr{X_{AB}} \leq d_B$. Furthermore, there exists a 
$\sigma_B > 0$ with $\id_A \otimes \sigma_B > \rho_{AB}$. Hence, strong duality 
applies and the values of the primal and dual problems are equal. \qed
\end{proof}
\end{petit}

Let us investigate the dual problem next. We can replace the inequality in the condition $X_{B} \leq \id_B$ by an equality since adding a positive part to $X_{AB}$ only increases $\tr(\rho_{AB} X_{AB})$. Hence, $X_{AB}$ can be interpreted as a Choi-Jamiolkowski state of a unital $\cp$ map~(cf.\ Sec.~\ref{sc:choi-jami}) from 
$\cH_{A'}$ to $\cH_B$. Let $\sE^{\dag}$ be that map, then
\begin{align}
  \exp\big(-H_{\min}(A|B)_{\rho}\big) = \max_{\sE^{\dag}} \tr\big(\rho_{AB} \sE^{\dag}(\Psi_{AA'})\big) = d_A \max_{\sE} 
    \tr\big(\sE[\rho_{AB}] \psi_{AA'}\big) \,,
\end{align}
where the second maximization is over all $\sE  \in \cptp(B, A')$, i.e.\ all maps whose adjoint is completely positive and unital from $A'$ to $B$. 
The fully entangled state $\psi_{AA'} = \Psi_{AA'}/d_A$ is pure and normalized and if $\rho_{AB} \in \cSnorm(AB)$ is normalized as well, we can rewrite the above expression in terms of the fidelity~\cite{koenig08}
\begin{align}
  H_{\min}(A|B)_{\rho} = - \log \bigg( d_A \max_{\sE \in \cptp(B,A')} F \big( \sE(\rho_{AB}), \psi_{AA'} \big) \bigg) \geq - \log d_A \label{eq:min/koenig} \,.
\end{align}
(Note that $\psi$ is defined as the fully entangled in an arbitrary but fixed basis of $\cH_A$ and $\cH_{A'}$. The expression is invariant under the choice of basis, since the fully entangled states can be converted into each other by an isometry appended to $\sE$.) 

Alternatively, we can interpret $X_{AB}$ as the Choi-Jamiolkowski state of a TP-CPM map from $\cH_{B'}$ to $\cH_{A}$, leading to
\begin{align}
  H_{\min}(A|B)_{\rho} &= - \log \bigg( d_B \max_{\sE \in \cptp(B',A)} \tr\big(\rho_{AB}  \sE(\psi_{BB'})\big) \bigg) \geq - \log d_B\,.
\end{align}

%We may now decompose the TP-CPMs of~\eqref[min/koenig] into their Stinespring dilation: an isometry $U : \cH_B \to \cH_{B'B''}$ followed by a partial trace over $\cH_{B'}$. Uhlmann's theorem now implies that there exists an extension 
%of $\psi_{AB'}$ to $B''$ such that $\Fg( U \rho_{AB} U^{\dag} ,\, \psi_{AB'B''} ) = \Fg( \sE[\rho_{AB}] ,\, \psi_{AB'} )$. Since such extensions of a pure state are necessarily of the form $\psi_{AB'B''} = \psi_{AB'} \otimes \tau_{B''}$, we recover the following expression for the min-entropy
%%
%\begin{align}
%  H_{\min}(A|B)_{\rho} = - \log\, d_A\! \max_{B \to B'B''} \max_{\tau} \, F^2 \big( \rho_{AB'B''} ,\, 
%    \psi_{AB'} \otimes \tau_{B''}  \big) \label{eq:min-omni},
%\end{align}
%%
%where the maximization is over all isometries from $B$ to $B'B''$ and 
%states $\tau \in \cSnorm{\cH_{B''}}$.
%
%Using the expression in~\eqref[min-omni], the min-entropy can be interpreted as a 
%measure of distance to a state describing an observer $B$ that is ☼[omniscient]{observer!omniscient} about $A$.
%Such an observer must necessarily hold a state $\psi$ that is fully entangled with $A$ and
%may, in addition, hold an arbitrary state $\tau$ that is uncorrelated with $A$. The
%min-entropy now evaluates the distance (in terms of the fidelity) of $\rho$ to the closest such 
%state.

\subsection{The Max-Entropy}

We use the following definition of the max-entropy, which coincides with $\Hn_{\nicefrac12}^{\uparrow}$ in the case where $\rho_{AB}$ is normalized.

\begin{definition}
\begin{svgraybox}
  \label{df:max-entropy}
  Let $\rho_{AB} \in \cSsub(AB)$. The max-entropy of $A$ conditioned on $B$ of the 
  state $\rho_{AB}$ is
  \begin{align}
    H_{\max}(A|B)_{\rho} := \max_{\sigma_B \in \cSsub(B)}\, \log\, F(\rho_{AB}, \id_A \otimes \sigma_B ) \label{eq:max/def} \,.
  \end{align}
  \vspace{-0.5cm}
\end{svgraybox}
\end{definition}

Clearly, the maximum is taken for a normalized state in $\cSnorm(B)$. %Pulling out the dimension factor, we may rewrite the max-entropy as
%\begin{align}
%  H_{\max}(A|B)_{\rho} = \log \bigg( d_A \max_{\sigma_B \in \cSnorm(B)} F \big( \rho_{AB}, \pi_A \otimes \sigma_B \big) \bigg) .
%\end{align}
However, note that the fidelity term is not linear in $\sigma_B$, and thus this cannot directly be interpreted as an SDP.
This can be overcome by introducing an arbitrary purification $\rho_{ABC}$ of $\rho_{AB}$ and applying Uhlmann's theorem, which yields
\begin{align}
  \exp \big( H_{\max}(A|B)_{\rho} \big) =  d_A \max_{\tau_{ABC} \in \cSsub(ABC)}\ \bracket{ \rho_{ABC} }{ \tau_{ABC} }{ \rho_{ABC} } \,,
\end{align}
where $\tau_{ABC}$ has marginal $\tau_{AB} = \pi_A \otimes \sigma_B$ for some $\sigma_B \in \cSsub(B)$. This is the dual problem of a semi-definite program.
\begin{lemma}
%\begin{svgraybox}
Let $\rho_{AB} \in \cSsub(AB)$.
Then, the following two optimization problems satisfy strong duality and both evaluate to $\exp(H_{\max}(A|B)_{\rho})$.
\begin{align}
  \begin{array}{rlcrl}
    \multicolumn{2}{c}{\underline{\textnormal{primal problem}}} && 
      \multicolumn{2}{c}{\underline{\textnormal{dual problem}}} \vspace{0.2cm} \\
    \textnormal{minimize}: & \mu & & 
      \quad\textnormal{maximize}: & \tr(\rho_{ABC} Y_{ABC}) \\
    \textnormal{subject to}: & \mu \id_B \geq \tr_{A}(Z_{AB}) & & 
      \textnormal{subject to}: & \tr_{C}(Y_{ABC}) \leq \id_A \otimes \sigma_B \\
    & Z_{AB} \otimes \id_C \geq \rho_{ABC} & & & \tr(\sigma_B) \leq 1 \\
    & Z_{AB} \geq 0,\, \mu \geq 0 &&& Y_{ABC} \geq 0,\, \sigma_B \geq 0  \,.
  \end{array} 
\end{align}
%\end{svgraybox}
\end{lemma}

\begin{petit}
\begin{proof}
The dual problem has a finite solution,
$\tr(Y_{ABC}) \leq d_A$, and hence the maximum cannot exceed $d_A$. There are also primal feasible points with $Z_{AB} \otimes \id_C > \rho_{ABC}$ and $\mu \id_B > Z_B$. \qed
\end{proof}
\end{petit}

The primal problem can be rewritten by noting that the optimization over $\mu$ corresponds to evaluating the operator norm of $Z_B$.
\begin{align}
  H_{\max}(A|B)_{\rho} = \log \min \Big\{ \norm{Z_B}_{\infty} : Z_{AB} \otimes \id_C
    \geq \rho_{ABC},\, Z_{AB} \in \cP(AB) \Big\} \label{eq:max/minimize}\,.
\end{align}
To arrive at this SDP we introduced a purification of $\rho_{AB}$, and consequently~\eqref{eq:max/minimize} depends on $\rho_{ABC}$ as well. 
This can be avoided by choosing a different SDP for the fidelity.
\begin{lemma}
\begin{svgraybox}
  \label{lm:max-entropy-alt}
  For all $\rho_{AB} \in \cSsub(AB)$, we have
  \begin{align}
    \exp\big( H_{\max}(A|B)_{\rho} \big) = \inf_{ Y_{AB} > 0 } \tr \big( \rho_{AB} Y_{AB}^{-1} \big) \| Y_B \|_{\infty} \,.
  \end{align}
  \vspace{-0.5cm}
\end{svgraybox}
\end{lemma}
This can be interpreted as the Alberti form~\cite{alberti83} of the max-entropy. Its proof is based on an SDP formulation of the fidelity due to Watrous~\cite{watrous12} and Killoran~\cite{killoranthesis}.
\begin{petit}
\begin{proof}
From~\cite{watrous12,killoranthesis} we learn that $\max_{\sigma_B \in \cS(B)} \sqrt{F(\rho_{AB}, \id_A \otimes \sigma_B)}$ equals the dual problem of the following SDP:
\begin{align}
  \begin{array}{rlcrl}
    \multicolumn{2}{c}{\underline{\textnormal{primal problem}}} &\quad& 
      \multicolumn{2}{c}{\underline{\textnormal{dual problem}}} \vspace{0.2cm} \\
    \textnormal{minimize}: & \tr(\rho_{AB} Y_{AB}) + \gamma & & 
      \quad\textnormal{maximize}: & \frac12 \big( \tr X_{12} + \tr X_{21} \big)  \\
    \textnormal{subject to}: & \gamma \id_B \geq \tr_A(Y_{22}) & & 
      \textnormal{subject to}: & X_{11} \leq \rho_{AB} \\
    & \left(\begin{matrix} Y_{11} & 0 \\ 0 & Y_{22} \end{matrix}\right) \geq \frac12 \left(\begin{matrix} 0 & \id \\ \id & 0 \end{matrix}\right) &&& \begin{array}{l} X_{22} \leq \id_A \otimes \sigma_B \\ \tr(\sigma_B) \leq 1 \end{array} \\
    & Y_{11} \geq 0,\, Y_{22} \geq 0,\, \gamma \geq 0 &&& \left(\begin{matrix} X_{11} & X_{12} \\ X_{21} & X_{22} \end{matrix}\right) \geq 0, \, \sigma_B \geq 0 \,.
 \end{array}
\end{align}
  Strong duality holds. The primal program can be simplified by noting that $\left(\begin{matrix} Y_{11} & 0 \\ 0 & Y_{22} \end{matrix}\right) \geq \left(\begin{matrix} 0 & \id \\ \id & 0 \end{matrix}\right)$ holds if and only if $\sqrt{Y_{22}} Y_{11} \sqrt{ Y_{22} } \geq \id$. This allows us to simplify the primal problem and we find
  \begin{align}
    \max_{\sigma_B \in \cS(B)} \sqrt{F(\rho_{AB}, \id_A \otimes \sigma_B)} = \inf_{Y_{AB} > 0} \frac12 \tr\big(\rho_{AB} Y_{AB}^{-1}\big) + \frac12 \| Y_B \|_{\infty} \,.
  \end{align}
  Now, by the arithmetic geometric mean inequality, we have
  \begin{align}
    \frac12 \tr(\rho_{AB} Y_{AB}^{-1}) + \frac12 \| Y_B \|_{\infty} &\geq \sqrt{\tr\big(\rho_{AB} Y_{AB}^{-1}\big) \| Y_B \|_{\infty}} = \frac12 \tr(\rho_{AB} (c Y_{AB})^{-1}) + \frac12 \| c Y_B \|_{\infty} \\
    &\geq \inf_{Y_{AB} > 0} \frac12 \tr(\rho_{AB} Y_{AB}^{-1}) + \frac12 \| Y_B \|_{\infty} \,.
  \end{align}
  Here, $c$ is chosen such that $\frac{1}{c} \tr\big(\rho_{AB} Y_{AB}^{-1}\big) = c \| Y_B \|_{\infty}$, such that the arithmetic geometric mean inequality becomes an equality. Therefore we have
  \begin{align}
     \max_{\sigma_B \in \cS(B)} \sqrt{F(\rho_{AB}, \id_A \otimes \sigma_B)} = \inf_{Y_{AB} > 0} \sqrt{\tr\big(\rho_{AB} Y_{AB}^{-1}\big) \| Y_B \|_{\infty}}
  \end{align}
  and the desired equality follows.
\qed
\end{proof}
\end{petit}

This can be used to prove upper bounds on the max-entropy. For example, the quantity $\Ho_{0}^{\uparrow}(A|B)_{\rho}$\,|\,which is sometimes used instead of the max-entropy~\cite{renner05}\,|\,is an upper bound on $H_{\max}(A|B)_{\rho}$.
\begin{align}
  \Ho_{0}^{\uparrow}(A|B)_{\rho} = \log \max_{\sigma_B \in \cSsub(B)} \tr\big( \{\rho_{AB} > 0 \} \id_A \otimes \sigma_B \big) \geq H_{\max}(A|B)_{\rho} \,.
\end{align}
This follows from Lemma~\ref{lm:max-entropy-alt} by the choice $Y_{AB} = \{\rho_{AB} > 0\} + \eps \id_{AB}$ with $\eps \to 0$, which yields the projector onto the support of $\rho_{AB}$. Furthermore, we have
\begin{align}
  \big\| \tr_A\big(\{\rho_{AB} > 0\} \big) \big\|_{\infty} = \max_{\sigma_B \in \cSsub(B)} \tr\big( \{\rho_{AB} > 0\} \id_A \otimes \sigma_B \big) \,.
\end{align}

\subsubsection*{Min- and Max-Entropy Duality}

Finally, the max-entropy can be expressed as a min-entropy of the purified state using the duality relation in Proposition~\ref{pr:dual-new},
which for this special case was first established by K\"onig \emph{et al.}~\cite{koenig08}.
\begin{lemma}
  \label{lm:min-max/dual}
\begin{svgraybox}
  Let $\rho \in \cSsub(ABC)$ be pure. Then, $H_{\max}(A|B)_{\rho} = - H_{\min}(A|C)_{\rho}$.
\end{svgraybox}
\end{lemma}

\begin{petit}
\begin{proof}
  We have already seen in Proposition~\ref{pr:dual-new} that this relation holds for normalized states. The lemma thus follows from the observation that
  \begin{align}
    H_{\min}(A|B)_{\rho} = H_{\min}(A|B)_{\rhot} - \log t , \quad \textrm{and} \quad H_{\max}(A|B)_{\rho} = H_{\min}(A|B)_{\rhot} + \log t
  \end{align}
  for any $\rho_{AB} \in \cSsub(AB)$ and $\rhot_{AB} \in \cSnorm(AB)$ with $\rho_{AB} = t \rhot_{AB}$.
  \qed
\end{proof}
\end{petit}

\subsection{Classical Information and Guessing Probability}

First, let us specialize some of the results in Proposition~\ref{pr:rcond-class} to the min- and max-entropy. In the limit $\alpha \to \infty$ and at $\alpha = \frac12$, we find that
\begin{align}
  H_{\min}(A|BY)_{\rho} &= - \log \bigg( \sum_y \rho(y) \exp \Big( - H_{\min}(A|B)_{\rhoh(y)} \Big) \bigg), \quad \textrm{and} \\
  H_{\max}(A|BY)_{\rho} &= \log \bigg( \sum_y \rho(y) \exp \Big( H_{\max}(A|B)_{\rhoh(y)} \Big) \bigg) \,.
\end{align}

\subsubsection*{Guessing Probability}

The classical min-entropy $H_{\min}(X|Y)_{\rho}$ can be interpreted as a \emph{guessing probability}. Consider an observer with access to $Y$. What is the probability that this observer guesses $X$ correctly, using his optimal strategy? The optimal strategy of the observer is clearly to guess that the event with the highest probability (conditioned on his observation) will occur. As before, we denote the probability distribution of $x$ conditioned on a fixed $y$ by $\rho(x|y)$.
Then, the guessing probability (averaged over the random variable $Y$) is given by 
\begin{align}
  \sum_y \rho(y)\, \max_x \rho(x|y) = \exp \big(-H_{\min}(X|Y)_{\rho} \big) \,.
\end{align}

It was shown by K\"onig \emph{et.\ al.}~\cite{koenig08} that this interpretation of the min-entropy extends to the case where $Y$ is replaced by a quantum system $B$ and the allowed strategies include arbitrary measurements of $B$.

Consider a classical-quantum state $\rho_{XB} = \sum_x \proj{x} \otimes \rho_B(x)$. For states of this form, the min-entropy simplifies to 
\begin{align}
  \exp\big( -H_{\min}(X|B)_{\rho} \big) &= 
    \max_{\sE \in \cptp(B,X')} \bracketB{\Psi}{ \sum_x \proj{x}_X \otimes \sE\big(\rho_B(x)\big) } { \Psi }_{XX'} \\
    &= \max_{\sE  \in \cptp(B,X')} \sum_x \bracket{x}{ \sE\big(\rho_B(x)\big) }{x}_{X'} \,.
\end{align}
The latter expression clearly reaches its maximum when $\sE$ has classical output in the basis $\{ \ket{x}_{X'} \}_x$, or in other words, when $\sE$ is a measurement map of the form $\sE: \rho_B \mapsto \sum_y \tr(\rho_B M_y) \proj{y}$ for a POVM $\{ M_y \}_y$. We can thus equivalently write
\begin{align}
  \exp\big( -H_{\min}(X|B)_{\rho} \big) &= \max_{ \{M_y\}_y \textrm{ a POVM}} \ \sum_y \tr( M_y \rho_B(y) ) \,. \label{eq:guessing}
\end{align}

Moreover, let $\{ \tilde{M}_y \} $ be a measurement that achieves the maximum in the above expression and define $\tau(x, y) = \tr(\tilde{M}_y \rho_B(x))$
as the probability that the true value is $x$ and the observer's guess is $y$. Then,
\begin{align}
  \exp\big( -H_{\min}(X|B)_{\rho} \big) &= \sum_y \tr (\tilde{M}_y \rho_B(y) ) \\
  &\leq \sum_y \max_x \tr (\tilde{M}_y \rho_B(x) ) = \exp\big( -H_{\min}(X|Y)_{\tau} \big) \,,
\end{align}
and this is in fact an equality by the data-processing inequality.
Thus, it is evident that
$H_{\min}(X|B)_{\rho} = H_{\min}(X|Y)_{\tau}$ can be achieved by a measurement on $B$.

%%%%%%%%%%%%%%%%%%%%%%%%

\section{Smooth Entropies}
\label{sc:smooth}

The smooth entropies of a state $\rho$ are defined as optimizations over the min- and max-entropies
of states $\rhot$ that are close to $\rho$ in \emph{purified distance}.
Here, we define the purified distance and the smooth min- and max-entropies and explore some properties 
of the smoothing.

\subsection[Definition of the Smoothing Ball]{Definition of the $\eps$-Ball}
\label{sc:eps-ball}

We introduce sets of $\eps$-close states that will be used to define the smooth entropies.
\begin{definition}
\begin{svgraybox}
  \label{df:ball}
  Let $\rho \in \cSsub(A)$ and $0 \leq \eps < \sqrt{\tr(\rho)}$. We define the 
  \textbf{$\eps$-ball} of states in $\cSsub(A)$ around $\rho$ as
  \begin{align}
    \cB^{\eps}(A; \rho) := \{ \tau \in \cSsub(A) : P(\tau, \rho) \leq \eps \} \,.
  \end{align}
  Furthermore, we define the $\eps$-ball of pure states around $\rho$ as $\cB_{*}^{\eps}(A; \rho) := \{ \tau \in \cB^{\eps}(A; \rho) : \rank(\tau) = 1 \}$.
\end{svgraybox}
\end{definition}

For the remainder of this chapter, we will assume that $\eps$ is sufficiently small so that 
$\eps < \sqrt{\tr{\rho}}$ is always satisfied. Furthermore, if it is clear from the 
context which system is meant, we will omit it and simply use the notation $\cB(\rho)$. We now list some properties of 
this $\eps$-ball, in addition to the properties of the underlying purified distance metric.

\begin{enumerate}

\item[i.] The set $\cB^{\eps}(A; \rho)$ is compact and convex.

\item[ii.] The ball grows monotonically in the smoothing parameter
  $\eps$, namely $\eps < \eps' \implies \cB^{\eps}(A; \rho) \subset \cB^{\eps'}(A; \rho)$. 
  Furthermore, $\cB^0(A; \rho) = \{ \rho \}$.

\end{enumerate}

\subsection{Definition of Smooth Entropies}

The \emph{smooth entropies} are now defined as follows. 
\begin{definition}
\begin{svgraybox}
  \label{df:smooth}
  Let $\rho_{AB} \in \cSsub(AB)$ and $\eps \geq 0$. Then, we define the \textbf{$\eps$-smooth min- and 
  max-entropy} of $A$ conditioned on $B$ of the state $\rho_{AB}$ as
  \begin{align}
    H_{\min}^{\eps}(A|B)_{\rho} &:= \max_{\rhot_{AB} \in \cB^{\eps}(\rho_{AB})} H_{\min}(A|B)_{\rhot} \quad \textrm{and} \\
    H_{\max}^{\eps}(A|B)_{\rho} &:= \min_{\rhot_{AB} \in \cB^{\eps}(\rho_{AB})} H_{\max}(A|B)_{\rhot} \,.
  \end{align}
  \vspace{-0.5cm}
\end{svgraybox}
\end{definition}

Note that the extrema can be achieved due to the compactness of the $\eps$-ball (cf. Property~i.).
We usually use $\rhot$ to denote the state that achieves the extremum.
%For the min-entropy, there exists a state $\rhot_{AB} \in \cB^{\eps}(\rho_{AB})$ such that 
%$H_{\min}^{\eps}(A|B)_{\rho} = H_{\min}(A|B)_{\rhot}$.
%The state $\rhot$ is $\eps$-indistinguishable from $\rho$ in the sense described
%in Property~ii.
Moreover, the smooth min-entropy is monotonically increasing in $\eps$ and the 
smooth max-entropy is monotonically decreasing in $\eps$ (cf.\
Property~ii.). Furthermore,
\begin{align}
  H_{\min}^0(A|B)_{\rho} = H_{\min}(A|B)_{\rho} \quad \textrm{and} \quad
  H_{\max}^0(A|B)_{\rho} = H_{\max}(A|B)_{\rho} \,.
\end{align}

If $\rho_{AB}$ is normalized, the optimization problems defining the smooth min- and max-entropies can be formulated as SDPs. To see this, note that the restrictions
on the smoothed state $\rhot$ are linear in the purification $\rho_{ABC}$ of $\rho_{AB}$. In 
particular, consider the condition $P(\rho, \rhot) \leq \eps$ on $\rhot$, or, 
equivalently, $F_*^2(\rho, \rhot) \geq 1 - \eps^2$. If $\rho_{ABC}$ is normalized, then the squared fidelity can be expressed 
as $F_*^2(\rho, \rhot) = \tr{\rho_{ABC}\,\rhot_{ABC}}$. 

We give the primal of the SDP for $\exp( -H_{\min}^{\eps}(A|B)_{\rho} )$ as an example. This SDP is parametrized by 
an (arbitrary) purification $\rho_{ABC} \in \cSnorm(ABC)$.

\begin{align}
  \begin{array}{rlcrl}
    \multicolumn{2}{c}{\underline{\textrm{primal problem}}} \vspace{0.2cm} \\
    \textrm{minimize}: & \tr(\sigma_B) \\
    \textrm{subject to}: & \id_A \otimes \sigma_B \geq \tr_C(\rhot_{ABC}) \\
    & \tr(\rhot_{ABC}) \leq 1 \\
    & \tr(\rhot_{ABC} \rho_{ABC}) \geq 1 - \eps^2 \\
    & \rhot_{ABC} \in \cS(ABC),\, \sigma_B \in \cP(B)
  \end{array} 
\end{align}
This program allows us to efficiently compute the smooth min-entropy as long as the involved Hilbert space dimensions are small.

\subsection{Remarks on Smoothing}

For both the smooth min- and max-entropy, 
we can restrict the optimization in Definition~\ref{df:smooth} 
to states in the support of $\rho_A \otimes \rho_B$.
\begin{proposition}
  \label{pr:smoothing/support}
  Let $\rho_{AB} \in \cSsub(AB)$ and $0 \leq \eps < \sqrt{\tr(\rho_{AB})}$. Then, there exist respective states  
  $\rhot_{AB} \in \cB^{\eps}(\rho_{AB})$ in the support of $\rho_A \otimes \rho_B$ such that
  \begin{align}
    H_{\min}^{\eps}(A|B)_{\rho} = H_{\min}(A|B)_{\rhot} \quad \textrm{or} \quad
    H_{\max}^{\eps}(A|B)_{\rho} = H_{\max}(A|B)_{\rhot} \,.
  \end{align}
\end{proposition}
\begin{petit}
\begin{proof}
  Let $\rho_{ABC}$ be any purification of $\rho_{AB}$. Moreover, let 
  $\Pi_{AB} = \{\rho_A > 0\} \otimes \{ \rho_B > 0\}$ be the projector 
  onto the support of $\rho_A \otimes \rho_B$. 
    
  For the min-entropy, first consider any state $\rhot_{AB}' \in \cB^{\eps}(\rho_{AB})$ 
  that achieves the maximum in Definition~\ref{df:smooth}. Then, there exists a $\sigma_B' \in \cSnorm(B)$ with $H_{\min}^{\eps}(A|B)_{\rho} = -\log \tr(\sigma_B')$ such that
  \begin{align}
    \rhot_{AB}' \leq \id_A \otimes \sigma_B' \implies \underbrace{\Pi_{AB} \rhot_{AB}' \Pi_{AB} }_{=:\, \rhot_{AB}} \leq \{\rho_A > 0\} \otimes \underbrace{\{\rho_B > 0\} \sigma_B' \{\rho_B > 0\} }_{=:\, \sigma_B } \,.
  \end{align}
  Moreover, $\rhot_{AB} \in \cB^{\eps}(\rho_{AB})$ since the purified distance contracts under trace non-increasing maps, and $\tr(\sigma_B) \leq \tr(\sigma_B')$. We conclude that $\rhot_{AB}$ must be optimal.
  
  For the max-entropy, again we start with any state $\rhot_{AB}' \in \cB^{\eps}(\rho_{AB})$ 
  that achieves the minimum in Definition~\ref{df:smooth}. Then, using $\rhot_{AB}$ as defined above
  %, and its purification $\rhot_{ABC}' \in \cB_*^{\eps}(\rho_{ABC})$. 
  \begin{align}
    \max_{\sigma_B' \in \cSnorm(B)} F(  \rhot_{AB} , \id_A \otimes \sigma_B') 
    &=  \max_{\sigma_B' \in \cSnorm(B)} F\big(  \Pi_{AB} \rhot_{AB}' \Pi_{AB}, \id_A \otimes \sigma_B' \big) \\
    &= \max_{\sigma_B' \in \cSnorm(B)} F\big( \rhot_{AB}', \{ \rho_A > 0 \} \otimes \{\rho_B > 0\} \sigma_B' \{\rho_B > 0\} \big) \\
    &\leq \max_{\sigma_B \in \cSsub(B)} F( \rhot_{AB}', \id_A \otimes \sigma_B ) \,.
  \end{align}
  Hence, $H_{\max}(A|B)_{\rhot} \leq H_{\max}(A|B)_{\rhot'}$, concluding the proof.
   \qed
\end{proof}
\end{petit}
Note that these optimal states are not necessarily normalized. In fact, it is in general not possible to find a normalized state in the support of $\rho_A \otimes \rho_B$ that achieves the optimum. 
Allowing sub-normalized states, we avoid this problem and as a consequence the smooth entropies are invariant under embeddings into a larger space.
\begin{corollary}
\label{co:smooth-iso}
\begin{svgraybox}
  For any state $\rho_{AB} \in \cSsub(AB)$ and isometries $U: A \to A'$ and $V: B\to B'$, we have
  \begin{align}
    H_{\min}^{\eps}(A|B)_{\rho} = H_{\min}^{\eps}(A'|B')_{\tau}, \quad H_{\max}^{\eps}(A|B)_{\rho} = H_{\max}^{\eps}(A'|B')_{\tau}
  \end{align}
  where $\tau_{A'B'} = (U \otimes V) \rho_{AB} (U \otimes V)^{\dagger}$.
\end{svgraybox}
\end{corollary}

On the other hand, if $\rho$ is normalized, we can always find normalized optimal states 
if we embed the systems $A$ and $B$
into large enough Hilbert spaces that allow 
smoothing outside the support of $\rho_A \otimes \rho_B$.
For the min-entropy, this is intuitively true since adding weight in a space orthogonal to $A$, 
if sufficiently diluted, will neither affect the min-entropy nor the purified distance. 

\begin{lemma}
  \label{lm:smoothing-normalized}
  There exists an embedding from $A$ to $A'$ 
  and a normalized state $\rhoh_{A'B} \in \cB^{\eps}(\rho_{A'B})$ such that 
  $H_{\min}(A'|B)_{\rhoh} = H_{\min}^{\eps}(A|B)_{\rho}$.
\end{lemma}

\begin{petit}
\begin{proof}
  Let $\{ \rhot_{AB}, \sigma_B \}$ be such that 
  they maximize the smooth min-entropy $\lambda = H_{\min}^{\eps}(A|B)_{\rho}$, i.e.\ we have $\rhot_{AB} \leq 
  \exp(-\lambda) \id_A \otimes \sigma_B$. Then we  embed $A$ into an auxiliary system $A'$ with
  dimension $d_A + d_{\bar{A}}$ to be defined below.
  The state $\rhoh_{A'B} = \rhot_{AB} \oplus (1-\tr(\rhot)) \pi_{\bar{A}} \otimes \sigma_B$,
  satisfies
  \begin{align}
    \rhoh_{A'B} = \rhot_{AB} \oplus (1-\tr(\rhot))\, \pi_{\bar{A}} 
      \otimes \sigma_B \leq \exp(-\lambda) (\id_{A} \oplus \id_{\bar{A}}) \otimes \sigma_B  
  \end{align}
  if $\exp(\lambda) (1 - \tr(\rhot)) \leq \exp(\lambda) \leq d_{\bar{A}}$. Hence, if $d_{\bar{A}}$ is chosen 
  large enough, we have $H_{\min}(A'|B)_{\rhoh} \geq \lambda$. Moreover, $\Fg(\rhoh, \rho) = \Fg(\rhot, \rho)$ is not affected by adding
  the orthogonal subspace. \qed
\end{proof}
\end{petit}

For the max-entropy, a similar statement can be derived using the duality of the smooth entropies.

\subsubsection*{Smoothing Classical States}

Finally, smoothing respects the structure of the state $\rho$, in particular if some subsystems are classical then the optimal state $\rhot$ will also be classical on these systems.
\begin{lemma}
  \label{lm:class-smooth}
  For both $H_{\min}^{\eps}(AX|BY)_{\rho}$ and $H_{\max}^{\eps}(AX|BY)_{\rho}$, there exist an optimizer  
  $\rhot_{AXBY} \in \cB^{\eps}(\rho_{AXBY})$ that is classical on $X$ and $Y$.
\end{lemma}
\begin{petit}
\begin{proof}
  Consider the pinching maps $\sP_{X}(\cdot) = \sum_x \proj{x} \cdot \proj{x}$ and $\sP_{Y}$ defined analogously. Since these are $\cptp$ and unital, the data-processing inequality yields
  $H_{\min}(AX|BY)_{\rhot'} \leq H_{\min}(AX|BY)_{\rhot}$ for any state $\rhot_{AXBY}'$ and $\rhot_{AXBY} = \sP_X \otimes \sP_Y(\rhot_{AXBY}')$ of the desired form, and this is in particular true when we choose $\rhot'$ to be a state that achieves the maximum in the definition of $H_{\min}^{\eps}(AX|BY)_{\rho}$. Furthermore, since $\rho_{AXBY}$ is invariant under this pinching, $\rhot' \in \cB^{\eps}(\rho)$ implies that $\rhot \in \cB^{\eps}(\rho)$ as well, and hence, we conclude that $\rhot$ must achieve the maximum too (and is of the desired form by construction).
  
  For the max-entropy, we first note that the data-processing inequality now goes in the wrong direction, so the above argument needs to be adapted. Our way out is to consider the purification $\rho_{XX'YY'ABC}$ of $\rho_{XYAB}$ that is of the following classical-coherent form:
  \begin{align}
    \ket{\rho}_{XX'YY'ABC} = \sum_{x,y} \ket{x}_X \otimes \ket{x}_{X'} \otimes \ket{y}_Y \otimes \ket{y}_{Y'} \ket{\rho^{x,y}}_{ABC}
  \end{align}
  where $\rho^{x,y}_{ABC}$ purifies the (unnormalized) state $\rho_{AXBY}$ conditioned on $x$ and $y$ and then consider the dual problem for the smooth min-entropy.
  
  For this purpose, let us introduce
  the maximizer for the min-entropy $H_{\min}^{\eps}(AX|CX'Y')_{\rho}$, which we denote by $\rhot_{XX'Y'AC}'$, and the classical-coherent state
\begin{align}
  \rhot_{XX'Y'AC} = \Pi_{XX'} \big( \sM_{Y'} ( \rhot_{XX'Y'AC}' ) \big) \Pi_{XX'}
\end{align}
  with $\Pi_{XX'} = \sum_x \proj{x}_X \otimes \proj{x}_{X'}$ and $\sM_{Y'}$ a pinching in the standard basis of $Y'$. Leveraging the fact that purified distance contracts under projections we can conclude that $\rhot \in \cB^{\eps}(\rho)$. We want to show that $H_{\min}(AX|CX'Y')_{\rhot'} \leq H_{\min}(AX|CX'Y')_{\rhot}$, establishing that $\rhot$ is indeed a maximizer for this min-entropy as well. To do this, consider that for some choice of $\sigma_{CX'Y'}$ we have
  \begin{align}
     \rhot_{XX'Y'AC}' \leq \exp \left( -H_{\min}(AX|CX'Y')_{\rho'} \right) \id_{AX} \otimes \sigma_{CX'Y'},
  \end{align}
  and thus applying the projection $\Pi_{XX'}$ and the measurement $\sM_{Y'}$ on both sides, we find
  \begin{align}
       \rhot_{XX'Y'AC} &\leq \exp \left( -H_{\min}(AX|CX'Y')_{\rho'} \right) \sum_x \proj{x}_X \otimes \id_{A} \otimes \proj{x}_{X'} \big( \sM_{Y'}(\sigma_{CX'Y'}) \big) \proj{x}_{X'} \\
       &\leq \exp \left( -H_{\min}(AX|CX'Y')_{\rho'} \right) \id_{AX} \otimes (\sM_{X'} \otimes \sM_{Y'})(\sigma_{CX'Y'}) \,.
  \end{align}
  This establishes the desired inequality since $(\sM_{X'} \otimes \sM_{Y'})(\sigma_{CX'Y'})$ is a valid state.
  
  Corollary~\ref{cor:extension} now yields a purification $\rhot_{XX'YY'ABC}$ or $\rhot_{XX'Y'AC}$ that is in the $\eps$-ball around $\rho_{XX'YY'ABC}$ and since $Y'$ is classical we can construct this purification in such a way that $YY'$ is classical-coherent as well. Using this, we can finally conclude that
  \begin{align}
  	H_{\max}^{\eps}(AX|BY)_{\rho} = -H_{\min}^{\eps}(AX|CX'Y')_{\rho} = -H_{\min}(AX|CX'Y')_{\rhot} = H_{\max}(AX|BY)_{\rhot} \,,
  \end{align}
  with $\rhot_{AXBY}$ having the desired properties.
  \qed
\end{proof}
\end{petit}

%%%%%%%%

\section{Properties of the Smooth Entropies}
\label{sc:smooth-properties}

The smooth entropies inherit many properties of the respective underlying unsmoothed R\'enyi entropies, including data-processing inequalities, duality relations and chain rules.

%%%%

\subsection{Duality Relation and Beyond}

The duality relation in Lemma~\ref{lm:min-max/dual} extends to smooth entropies.

\begin{proposition}
\begin{svgraybox}
  \label{pr:smooth-dual}
  Let $\rho \in \cSsub(ABC)$ be pure and $0 \leq \eps < \sqrt{\tr(\rho)}$. Then,
  \begin{align}
    H_{\max}^{\eps}(A|B)_{\rho} = -H_{\min}^{\eps}(A|C)_{\rho} \,.
  \end{align}
  \vspace{-0.5cm}
\end{svgraybox}
\end{proposition}

\begin{petit}
\begin{proof}
  According to Corollary~\ref{co:smooth-iso}, the smooth entropies are invariant under embeddings,
  and we can thus assume without loss of generality that the spaces $B$ and $C$ are large enough to entertain purifications of the optimal smoothed states, which are in the support of $\rho_A \otimes \rho_B$ and $\rho_A \otimes \rho_C$, respectively. 
  Let $\rhot_{AB}$ be optimal for the max-entropy, then
  \begin{align}
    H_{\max}^{\eps}(A|B)_{\rho} &= H_{\max}(A|B)_{\rhot} 
        \geq \min_{\rhot \in \cBp^{\eps}(\rho_{ABC})} H_{\max}(A|B)_{\rhot} \\
      &= \min_{\rhot \in \cBp^{\eps}(\rho_{ABC})} - H_{\min}(A|C)_{\rhot} 
      \geq \min_{\rhot \in \cB^{\eps}(\rho_{AC})} - H_{\min}(A|C)_{\rhot} = - H_{\min}^{\eps}(A|C)_{\rho} \,.
  \end{align}
  And, using the same argument starting with $H_{\min}^{\eps}(A|C)_{\rho}$, we can show the opposite inequality.
  \qed
\end{proof}
\end{petit}

Due to the monotonicity in $\alpha$ of the R\'enyi entropies the min-entropy cannot exceed
the max-entropy for normalized states. This result extends to smooth entropies~\cite{vitanov12,morgan13}.
\begin{proposition}
\begin{svgraybox}
  \label{pr:min-max-smooth}
  Let $\rho \in \cSnorm(AB)$ and $\varphi, \vartheta \geq 0$ such that $\varphi + \vartheta < \frac{\pi}{2}$. 
  Then,
  \begin{align}
    H_{\min}^{\sin(\varphi)}(A|B)_{\rho} \leq H_{\max}^{\sin(\vartheta)}(A|B)_{\rho} + 2 \log \frac{1}{\cos(\varphi+\vartheta)} \,.
  \end{align}
  \vspace{-0.5cm}
\end{svgraybox}
\end{proposition}
\begin{petit}
\begin{proof}
  Set $\eps = \sin(\varphi)$. According to Lemma~\ref{lm:smoothing-normalized}, there exists an embedding $A'$ of $A$ and a 
  normalized state $\rhot_{A'B} \in \cB^{\eps}(\rho_{A'B})$
  such that $H_{\min}(A'|B)_{\rhot} = H_{\min}^{\eps}(A|B)_{\rho}$.
  In particular, there exists a state 
  $\sigma_B \in \cSnorm(B)$ such that 
  $\rhot_{A'B} \leq \exp(-\lambda) \id_{A'} \otimes \sigma_{B}$ with $\lambda = H_{\min}^{\eps}(A|B)_{\rho}$. 
  Thus, letting $\rhob_{A'B} \in \cB^{\sin(\vartheta)}(\rho_{A'B})$ be a state that minimizes the smooth max-entropy, we find 
  \begin{align}
    H_{\max}^{\eps'}(A|B)_{\rho} &= H_{\max}(A'|B)_{\rhob} 
    \geq - D_{\nicefrac12} \big( \rhob_{A'B} \big\| \id_{A'} \otimes \sigma_{B} \big) \\
      &\geq \lambda - D_{\nicefrac12} (\rhob_{A'B} \| \rhot_{A'B} ) 
      = \lambda + \log \big (1 - P(\rhob_{A'B}, \rhot_{A'B})^2 \big) \\
      &\geq H_{\min}^{\eps}(A|B)_{\rho} + \log \big( 1 - \sin(\varphi+\vartheta)^2 \big) \,.
  \end{align}
  In the final step we used the triangle inequality in~\eqref{eq:pd-triangle-eps} to find  
  $P(\rhob_{A'B}, \rhot_{A'B}) \leq \sin(\varphi+\vartheta)$. \qed
\end{proof}
\end{petit}

%\begin{remark}
%  \label{rm:min-max-smooth}  
%  The term $-\log \big( 1 - (\eps + \eps')^2 \big)$ can be 
%  reduced by the use of Eq.~§[pd/tight-triangle] instead of the triangle inequality for the 
%  purified distance. Hence,
%  \begin{align}
%    H_{\min}^{\eps}(A|B)_{\rho} \leq H_{\\max}^{\eps'}(A|B)_{\rho} + 
%      \log \frac{1}{1 - \left( \eps √{1-\eps'^2} + \eps' √{1-\eps^2} \right)^2} ¶.
%  \end{align}
%  The range of allowed pairs $\{\eps, \eps'\}$ is extended to those satisfying
%    $\arcsin(\eps) + \arcsin(\eps') < \frac{\pi}{2}$.
%  In particular, this means that the term is finite if 
%  we choose $\eps' = 1 - \eps$, for any $0 < \eps < 1$. (See also Figure~\ref{fg:tight}.)
%\end{remark}

Proposition~\ref{pr:min-max-smooth} implies that smoothing states that have similar min- and max-entropies has almost 
no effect. In particular, let $\rho_{AB} \in \cSnorm(AB)$ with $H_{\min}(A|B)_{\rho} = H_{\max}(A|B)_{\rho}$. Then, 
\begin{align}
  H_{\min}^{\eps}(A|B)_{\rho} \leq H_{\max}(A|B)_{\rho} - \log ( 1 - \eps^2 ) 
  = H_{\min}(A|B)_{\rho} - \log ( 1 - \eps^2 ) \,.
\end{align}
This inequality is tight and the smoothed state
$\rhot = (1-\eps^2)\rho$ reaches equality.
An analogous relation can be derived for the smooth max-entropy.

%%%%%%%

\subsection{Chain Rules}

Similar to the conditional R\'enyi entropies, we also provide a collection of inequalities that replace the chain rule of the von Neumann entropy. 
These chain rules are different in that they introduce an additional correction term in $O\big(\log \frac{1}{\eps}\big)$ that does not appear in the results of the previous chapter.
\begin{theorem}
\begin{svgraybox}
  \label{th:chain-rules}
  Let $\rho \in \cSsub(ABC)$ and $\eps, \eps',\eps'' \in [0,1)$ with $\eps > \eps' + 2\eps''$. Then,
  \begin{align}
    H_{\min}^{\eps}(AB|C)_{\rho} &\geq H_{\min}^{\eps'}(A|BC)_{\rho} + H_{\min}^{\eps''}(B|C)_{\rho} - g(\delta) \label{eq:cr/min-ab},\\
    H_{\min}^{\eps'}(AB|C)_{\rho} &\leq H_{\min}^{\eps}(A|BC)_{\rho} + H_{\max}^{\eps''}(B|C)_{\rho} + 2 g(\delta) \label{eq:cr/min-a},\\
    H_{\min}^{\eps'}(AB|C)_{\rho}  &\leq  H_{\max}^{\eps''}(A|BC)_{\rho} + H_{\min}^{\eps}(B|C)_{\rho} + 3 g(\delta) \label{eq:cr/min-b},
  \end{align}
  where $g(\delta) = -\log \big(1 - \sqrt{1-\delta^2}\big)$ and $\delta = \eps - \eps' - 2\eps''$.
\end{svgraybox}
\end{theorem}

See~\cite{vitanov12} for a proof. Using the duality relation for smooth entropies on \eqref{eq:cr/min-ab}, \eqref{eq:cr/min-a} and \eqref{eq:cr/min-b}, we also find
the chain rules
\begin{align}
    H_{\max}^{\eps}(AB|C)_{\rho} &\leq H_{\max}^{\eps'}(A|BC)_{\rho} + H_{\max}^{\eps''}(B|C)_{\rho} + g(\delta) \label{eq:cr/max-ab},\\
    H_{\max}^{\eps'}(AB|C)_{\rho} &\geq H_{\min}^{\eps''}(A|BC)_{\rho} + H_{\max}^{\eps}(B|C)_{\rho} - 2 g(\delta) \label{eq:cr/max-b},\\     H_{\max}^{\eps'}(AB|C)_{\rho}  &\geq H_{\max}^{\eps}(A|BC)_{\rho} + H_{\min}^{\eps''}(B|C)_{\rho} - 3 g(\delta) \label{eq:cr/max-a}
\,.
\end{align}

\subsubsection*{Classical Information}

Sometimes the following alternative bounds restricted to classical information are very useful. The first result asserts that the entropy of a classical register is always non-negative and bounds how much entropy it can contain.

\begin{lemma}
  \label{pr:class/bounds-1}
  Let $\eps \in [0,1)$ and $\rho \in \cSsub(XAB)$ be classical on $X$. Then,
  \begin{align}
     H_{\min}^{\eps}(A|B)_{\rho} &\leq H_{\min}^{\eps}(XA|B)_{\rho} \leq H_{\min}^{\eps}(A|B)_{\rho} + \log d_{X} \qquad \textrm{and} 
      \label{eq:class/bounds-1/min} \\
     H_{\max}^{\eps}(A|B)_{\rho} &\leq H_{\max}^{\eps}(XA|B)_{\rho} \leq H_{\max}^{\eps}(A|B)_{\rho} + \log d_{X}
      \label{eq:class/bounds-1/max} \,.
  \end{align}
\end{lemma}
We are also concerned with the maximum amount of information a classical register $X$ can contain about
a quantum state $A$. 
\begin{lemma}
  \label{pr:class/bounds-2}
  Let $\eps \in [0,1)$ and $\rho \in \cSsub(AYB)$ be classical on $Y$. Then,
  \begin{align}
     H_{\min}^{\eps}(A|YB)_{\rho} &\geq H_{\min}^{\eps}(A|B)_{\rho} - \log d_{Y} \qquad \textrm{and} 
      \label{eq:class/bounds-2/min} \\
     H_{\max}^{\eps}(A|YB)_{\rho} &\geq H_{\max}^{\eps}(A|B)_{\rho} - \log d_{Y} 
      \label{eq:class/bounds-2/max} \,.
  \end{align}
\end{lemma}
We omit the proofs of the above statements, but note that they can be derived from~\eqref{eq:ua-class-ineq} together with the fact that the states achieving the optimum for the smooth entropies retain the classical-quantum structure (cf.~Lemma~\ref{lm:class-smooth}).

\subsection{Data-Processing Inequalities}

We expect measures of uncertainty of the system $A$ given side
information $B$ to be non-decreasing under local physical operations 
(e.g.\ measurements or unitary evolutions) applied to the $B$ system.
Furthermore, in analogy to the conditional R\'enyi entropies, we expect 
that the uncertainty of the system $A$ does not decrease when a sub-unital map 
is executed on the $A$ system.

\begin{theorem}
\begin{svgraybox}
  \label{th:data-proc} 
  Let $\rho_{AB} \in \cSsub(AB)$ and $0 \leq \eps < \sqrt{\tr(\rho)}$. 
  Moreover, let $\sE \in \cptp(A,A')$ be sub-unital, and let $\sF \in \cptp(B,B')$. 
  Then, the state $\tau_{A'B'} = (\sE \otimes \sF)(\rho_{AB})$ satisfies
  \begin{align}
    H_{\min}^{\eps}(A|B)_{\rho} \leq H_{\min}^{\eps}(A'|B')_{\tau} \quad \textrm{and} \quad
    H_{\max}^{\eps}(A|B)_{\rho} \leq H_{\max}^{\eps}(A'|B')_{\tau} \,.
  \end{align}
    \vspace{-0.5cm}
\end{svgraybox}
\end{theorem}

\begin{petit}
\begin{proof}
  The data-processing inequality for the min-entropy follows from the respective property of the unsmoothed conditional R\'enyi entropy. We have
  \begin{align}
    H_{\min}^{\eps}(A|B)_{\rho} = H_{\infty}^{\uparrow}(A|B)_{\rhot} \leq  
    H_{\infty}^{\uparrow}(A'|B')_{\taut} \leq H_{\min}^{\eps}(A'|B')_{\tau} \,.
  \end{align}
  Here, $\rhot_{AB}$ is a state maximizing the smooth min-entropy 
  and $\taut_{AB} = (\sE \otimes \sF)(\rhot_{AB})$ lies in $\cB^{\eps}(\tau_{A'B'})$.
  
  To prove the result for the max-entropy,  we take advantage of the Stinespring dilation
 of $\sE$ and $\sF$. Namely, we introduce the isometries $U: A \to A'A''$ and
  $V: B \to B'B''$ and the state $\tau_{A'A'B'B''} = (U \otimes V) \rho_{AB} (U^{\dag} \otimes V^{\dag})$ of 
  which $\tau_{A'B'}$ is a marginal. Let $\taut \in \cB^{\eps}(\tau_{A'A''B'B''})$ be the state that minimizes
  the smooth max-entropy $H_{\max}^{\eps}(A'|B')_{\tau}$. Then,
  \begin{align}
    H_{\max}^{\eps}(A'|B')_{\tau} &= \max_{\sigma_{B'} \in \cSnorm(B')} \log F\big(\taut_{A'B'}, \id_{A'} \otimes \sigma_{B'}\big) \\
      &\geq \max_{\sigma_{B'} \in \cSnorm(B')} \log F\big( \taut_{A'B'}, \tr_{A''}{\Pi_{A'A''}} \otimes \sigma_{B'} \big) 
      \label{eq:data-proc1}\,.
  \end{align}
  We introduced the projector $\Pi_{A'A''} = U U^{\dag}$ onto the image of $U$, which exhibits
  the following property  due to the fact that $\sE$ is sub-unital:
  \begin{align}
    \tr_{A''}(\Pi_{A'A''}) = \tr_{A''}\big(U \id_A U^{\dag}\big) = \sE(\id_A) \leq \id_{A'} \,.
  \end{align}
  The inequality in~\eqref{eq:data-proc1} is then a result of the fact that the fidelity
  is non-increasing when an argument $\tau$ is replaced by a smaller argument $\sigma \leq \tau$.
  Next, we use the monotonicity of the fidelity under
  partial trace to bound~\eqref{eq:data-proc1} further.
  \begin{align}
    H_{\max}^{\eps}(A'|B')_{\tau} &\geq \max_{\sigma_{B'B''} \in \cSnorm(B'B')} \log F\big( \taut_{A'A''B'B''}, \Pi_{A'A''} 
        \otimes \sigma_{B'B''} \big) \\
      &= \max_{\sigma_{B'B''} \in \cSnorm(B'B')} \log F \big( \Pi_{A'A''} \taut_{A'A''B'B''} \Pi_{A'A''}, \id_{A'A''} 
        \otimes \sigma_{B'B''} \big) \\
      &= H_{\max}(A'A''|B'B'')_{\tauh} \,.
  \end{align}
  Finally, we note that $\tauh_{A'A''B'B''} = \Pi_{A'A''} \taut_{A'A''B'B''} \Pi_{A'A''} \in \cB^{\eps}(\tau_{A'A''B'B''})$ due
  to the monotonicity of the purified distance under trace non-increasing maps. Hence,
  we established $H_{\max}^{\eps}(A'|B')_{\tau} \geq H_{\max}^{\eps}(A'A''|B'B'')_{\tau} = H_{\max}^{\eps}(A|B)_{\rho}$, where the
  last equality follows due to the invariance of the max-entropy under local isometries.
  \qed
\end{proof}
\end{petit}

\subsubsection*{Functions on Classical Registers}

Let us now consider a state $\rho_{XAB}$ that is classical on $X$. We aim to show that applying
a classical function on the register $X$ cannot increase the smooth entropies $AX$ given $B$, even if this operation is not necessarily sub-unital. In particular, for
the min-entropy this corresponds to the intuitive statement that it is
always at least as hard to guess the input of a function than it is to guess its output.

\begin{proposition}
\begin{svgraybox}
  \label{pr:func}
  Let $\rho_{XAB} = \sum_x p_x\, \proj{x}_X \otimes \hat{\rho}_{AB}(x)$ be classical on $X$. Furthermore, let $\eps \in [0,1)$ and let $f: X \to Z$ be
  a function. Then, the state $\tau_{ZAB} = \sum_x p_x\, \proj{f(x)}_Z \otimes \hat{\rho}_{AB}(x)$ satisfies
  \begin{align}
    H_{\min}^{\eps}(ZA|B)_{\tau} \leq H_{\min}^{\eps}(XA|B)_{\rho} \quad \textrm{and} \quad
    H_{\max}^{\eps}(ZA|B)_{\tau} \leq H_{\max}^{\eps}(XA|B)_{\rho} \,.
  \end{align}
\vspace{-0.5cm}
\end{svgraybox}
\end{proposition}

\begin{proof}
  A possible Stinespring dilation of $f$ is given by the isometry 
  $U: \ket{x}_X \mapsto \ket{x}_{X'} \otimes \ket{f(x)}_{Z}$ followed by a partial trace over $X'$.
  Applying $U$ on $\rho_{XAB}$, we get
  \begin{align}
    \tau_{X'ZAB} := U \rho_{XAB} U^{\dag} = \sum_x p_x\, \proj{x}_{X'} \otimes \proj{f(x)}_{Z} \otimes \hat{\rho}_{AB}(x) 
  \end{align}
  which is classical on $X'$ and $Z$ and an extension of $\tau_{ZAB}$. Hence, the invariance under
  isometries of the smooth entropies (cf.\ Corollary~\ref{co:smooth-iso}) in conjunction with 
  Proposition~\ref{pr:class/bounds-1} implies
  \begin{align}
    H_{\min}^{\eps}(XA|B)_{\rho} = H_{\min}^{\eps}(X'ZA|B)_{\tau} \geq H_{\min}^{\eps}(ZA|B)_{\tau} \,.
  \end{align}
  An analogous argument applies for the smooth max-entropy.
\end{proof}

%%%%%%%%%%%%%%

\section{Fully Quantum Asymptotic Equipartition Property}
\label{sc:qaep}

Smooth entropies give rise to an entropic (and fully quantum) version of the asymptotic equipartition property (AEP), which states that both the (regularized) smooth min- and max-entropies converge to the conditional von Neumann entropy for iid product states. The classical special case of this, which is usually not expressed in terms of entropies (see, e.g.,~\cite{cover91}), is a workhorse of classical information theory and similarly the quantum AEP has already found many applications.

The entropic form of the AEP explains the crucial role of the von Neumann
entropy to describe information theoretic tasks. 
While operational quantities in information theory (such as the amount of extractable randomness, the minimal
length of compressed data and channel capacities) can naturally be expressed in
terms of smooth entropies in the one-shot setting, the von Neumann entropy is 
recovered if we consider a large number of independent repetitions of the
task.

Moreover, the entropic approach to asymptotic equipartition lends itself to a
generalization to the quantum setting. Note that the traditional approach, 
which considers the AEP 
as a statement about (conditional) probabilities, does not have a natural quantum 
generalization due to the fact that
we do not know a suitable generalization of conditional probabilities to quantum
side information. Figure~\ref{fg:aep} visualizes the intuitive idea behind the entropic AEP.

\begin{figure}
\begin{flushleft}
  \hspace{1.5cm}
  \begin{overpic}[width =.75\columnwidth]{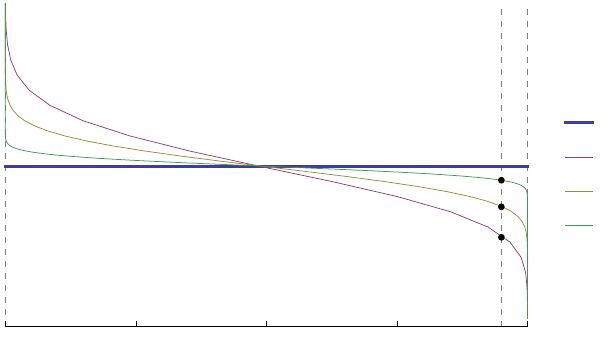}
    % legend
    \put(99,37){$n \to \infty$}
    \put(99,31.5){$n=50$}
    \put(99,26){$n=150$}
    \put(99,20.5){$n=1250$}
    % points
    \put(-8,30){$H(X)$}
    \put(-11,5){$H_{\min}(X)$}
    \put(67,15){$\frac{1}{n} H_{\min}^{\eps}(X^n)$}
    %\put(17,44){$\hat{D}(\rho\|\sigma)$}
    %\put(60,40){$\Dn_2(\rho\|\sigma)$}
    %\put(17,17){$\Dn_{1/2}(\rho\|\sigma)$}
    % axis label
     %\put(90,4){$\alpha$}
    % ticks
    \put(-1,0.5){\it 0.0}
    \put(20,0.5){\it 0.25}
    \put(41,0.5){\it 0.5}
    \put(62,0.5){\it 0.75}
    \put(83,0.5){\it 1.0}
  \end{overpic}
\end{flushleft}
  \caption[Emergence of the Typical Set.]{\textbf{Emergence of Typical Set.}
  We consider $n$ independent Bernoulli trials with $p = 0.2$ and denote the probability that an 
  event $x^n$ (a bit string of length $n$) occurs by $P_n(x^n)$.
  The plot shows the suprisal rate, $- \frac{1}{n} \log P_n(x^n)$, over the cumulated 
  probability of the events sorted such that events with high surprisal are on the left. 
  The curves for $n = \{ 50, 150, 1250 \}$ converge to the von Neumann entropy, 
  $H(X) \approx 0.72$ as $n$ increases. This indicates that, for large $n$, most (in probability) 
  events are close to typical (i.e.\ they have surprisal rate close to $H(X)$).
  
  The min-entropy, $H_{\min}(X) \approx 0.32$, constitutes the minimum of the curves while the max-entropy, $H_{\max}(X) \approx 0.85$, is upper bounded by their maximum. Moreover, the respective $\eps$-smooth entropies, $\frac{1}{n} H_{\min}^{\eps}(X^n)$ and $\frac{1}{n} H_{\max}^{\eps}(X^n)$, can be approximately obtained by cutting off a probability $\eps$ from each side of the $x$-axis and taking the minima or maxima of the remaining curve. Clearly, the $\eps$-smooth entropies converge to the von Neumann entropy as $n$ increases.} 
  \label{fg:aep}
\end{figure}
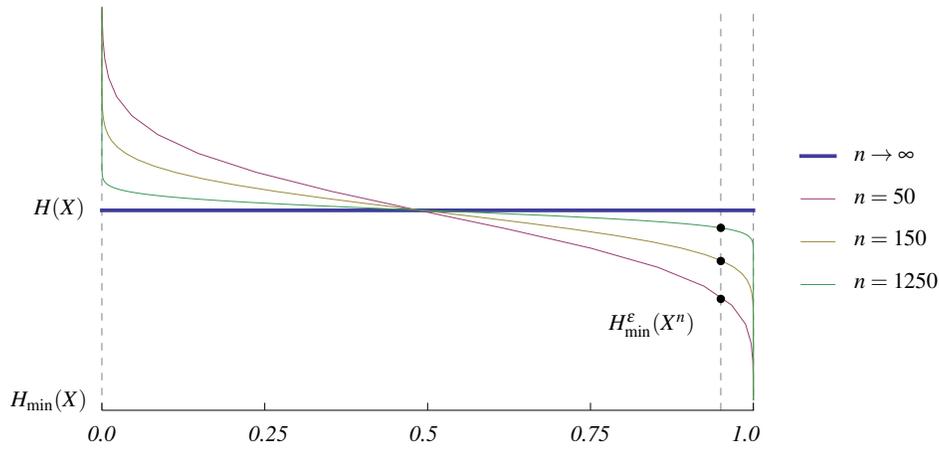

%%%

\subsection{Lower Bounds on the Smooth Min-Entropy}
\label{se:aep/min-renyi}

For the sake of generality, we state our results here in terms of the \emph{$\eps$-smooth relative max-divergence}, which we define for 
any $\rho \in \cSsub(A)$, $\sigma \in \cS(A)$ and $\eps \in [0, \sqrt{\tr(\rho)}]$ as
\begin{align}
  D_{\max}^{\eps}(\rho \| \sigma) := \min_{\rhot \in \cB^{\eps}(\rho)} D_{\max}(\rhot \| \sigma) \,. 
    \label{eq:smooth/rel} ,
\end{align}
where we used the shorthand $D_{\max} \equiv \Dn_{\infty}$.
The following gives an upper bound on the smooth relative max-entropy~\cite{tomamichel08,dattarenner08}. 

\begin{lemma}
\begin{svgraybox}
  \label{lm:aep/smooth-bound}
  Let $\rho \in \cSsub(A), \sigma \in \cS(A)$ and $\lambda \in \big(-\infty,\, D_{\max}(\rho \| \sigma)\big]$. Then,
  \begin{align}
     D_{\max}^{\eps}( \rho \| \sigma) \leq \lambda, \quad \textrm{where} \quad 
    \eps = \sqrt{2 \tr(\Sigma) - \tr(\Sigma)^2}
  \end{align}
  and $\Sigma = \{ \rho > \exp(\lambda) \sigma \} (\rho - \exp(\lambda) \sigma)$, i.e.\ the positive part of $\rho - \exp(\lambda) \sigma$.
\end{svgraybox}
\end{lemma}

The proof constructs a smoothed state $\rhot$ that reduces the smooth relative max-divergence relative to $\sigma$ by removing the subspace where $\rho$ exceeds 
$\exp(\lambda) \sigma$.

\begin{petit}
\begin{proof}
  We first choose $\rhot$, bound $D_{\max}^{\eps}(\rhot \| \sigma)$, and then show that 
  $\rhot \in \cB^{\eps}(\rho)$. We use the abbreviated notation $\Lambda := \exp(\lambda) \sigma$ and set
  \begin{align}
    \rhot := G \rho G^{\dag},\quad \textrm{where}\quad 
    G := \Lambda^{\nicefrac{1}{2}} (\Lambda + \Sigma)^{-\nicefrac{1}{2}} \,,
  \end{align}
  where we use the generalized inverse. From the definition of $\Sigma$, we have $\rho \leq \Lambda + \Sigma$;
  hence, $\rhot \leq \Lambda$ and $D_{\max}(\rhot \| \sigma) \leq \lambda$.

  Let $\ket{\rho}$ be a purification of $\rho$, then $(G \otimes \id) \ket{\rho}$ is a purification of $\rhot$ 
  and, using Uhlmann's theorem, we find a bound on the (generalized) fidelity:
  \begin{align}
    \sqrt{\Fg(\rhot, \rho)} &\geq \abs{\bracket{\rho}{G}{\rho}} + \sqrt{(1 - \tr(\rho))(1 - \tr(\rhot))} \\
      &\geq \Re \big( \tr(G \rho) \big) + 1 - \tr(\rho) = 1 - \tr\big((\id - \bar{G}) \rho\big) \,,
  \end{align}
  where we introduced $\bar{G} = \frac{1}{2}(G + G^{\dag})$ and $\Re$ denotes the real part. This can be simplified further 
  by noting that $G$ is a contraction. To see this, we multiply $\Lambda \leq \Lambda + \Sigma$ with 
  $(\Lambda + \Sigma)^{-\nicefrac{1}{2}}$ from left and right to get
  \begin{align}
    G^{\dag} G = (\Lambda + \Sigma)^{-\nicefrac{1}{2}} \Lambda (\Lambda + \Sigma)^{-\nicefrac{1}{2}} \leq \id.
  \end{align}
  Furthermore, $\bar{G} \leq \id$, since $\|\bar{G}\| \leq 1$ by the triangle inequality and 
  $\|{G}\| = \|{G^{\dag}}\| \leq 1$. 
  %Thus, $\tr{\bar{G} \rho} \leq 1$.  % it apperas we don't need this anymore
  Moreover,
  \begin{align}
    \tr\big((\id - \bar{G}) \rho\big) &\leq \tr(\Lambda + \Sigma) - \tr\big( \bar{G} (\Lambda + \Sigma) \big) \\
      &= \tr(\Lambda + \Sigma) - \tr\big( (\Lambda + \Sigma)^{\nicefrac{1}{2}} {\Lambda}^{\nicefrac{1}{2}} \big) 
      \leq \tr(\Sigma) \,,
  \end{align}
  where we used $\rho \leq \Lambda + \Sigma$ and $\sqrt{\Lambda + \Sigma} \geq \sqrt{\Lambda}$. The latter inequality follows from the 
  operator monotonicity of the square root function. Finally, 
  using the above bounds, the purified distance between $\rhot$ and $\rho$ is bounded by
  \begin{align}
    P(\rhot, \rho) = \sqrt{1 - \Fg(\rhot, \rho) \big)} \leq \sqrt{ 1 - \big( 1 - \tr(\Sigma) \big)^2 } = 
      \sqrt{2 \tr(\Sigma) - \tr(\Sigma)^2} \,.
  \end{align}
  Hence, we verified that $\rhot \in \cB^{\eps}(\rho)$, which concludes the proof.
\end{proof}
\end{petit}

In particular, this means that for a fixed $\eps \in \big[0, \sqrt{\tr(\rho)}\big)$ and $\rho \ll \sigma$, 
we can always find a finite~$\lambda$ such that Lemma~\ref{lm:aep/smooth-bound} holds. 
To see this, note that $\eps(\lambda) = \sqrt{ 2 \tr(\Sigma) - \tr(\Sigma)^2 }$
is continuous in $\lambda$ with $\eps(D_{\max}(\rho \| \sigma)) = 0$ and $\lim_{\lambda \to -\infty} \eps(\lambda) = \sqrt{2\tr(\rho) - \tr(\rho)^2} \geq \sqrt{\tr(\rho)}$.

\bigskip

Our main tool for proving the fully quantum AEP is a family of inequalities that 
relate the smooth max-divergence to quantum R\'enyi divergences for $\alpha \in (1, \infty)$.
\begin{proposition}
\begin{svgraybox}
  \label{pr:min-renyi}
  Let $\rho \in \cSnorm(A), \sigma \in \cS(A)$, $0 < \eps < 1$ and $\alpha \in  (1, \infty)$. Then, 
  \begin{align}
    D_{\max}^{\eps}(\rho \| \sigma) \leq \DD_{\alpha}(\rho \| \sigma) + \frac{g(\eps)}{\alpha - 1} \,, 
     \label{eq:min-renyi-bound} 
  \end{align}
  where $g(\eps) = -\log \big(1 - \sqrt{1 - \eps^2}\big)$ and $\DD_{\alpha}$ is any quantum R\'enyi divergence.
\end{svgraybox}
\end{proposition}

\begin{petit}
\begin{proof}
  If $\rho \not\ll \sigma$ the bound holds trivially, so for the following we have $\rho \ll \sigma$.
  Furthermore, since the divergences are invariant under isometries we can assume that $\sigma > 0$
  is invertible.

  We then choose $\lambda$ such that Lemma~\ref{lm:aep/smooth-bound} holds for the $\eps$ specified above.
  Next, we introduce the operator $X = \rho - \exp(\lambda) \sigma$ with eigenbasis 
  $\{ \ket{e_i} \}_{i \in S}$. The set $S^+ \subseteq S$ contains the indices $i$ 
  corresponding to positive eigenvalues of $X$. Hence, $\{X \geq 0\} X \{X \geq 0\} = \Sigma$ as defined in Lemma~\ref{lm:aep/smooth-bound}. 
  Furthermore, let $r_i = \bracket{e_i}{ \rho }{ e_i} \geq 0$ and $s_i = \bracket{e_i }{ \sigma }{ e_i} > 0$. 
  It follows that
  \begin{align}
    \forall\, i \in S^+ :\ r_i - \exp(\lambda) s_i \geq 0 \quad \textrm{and, thus,} \quad 
    \frac{r_i}{s_i} \exp(-\lambda) \geq 1 \,.
  \end{align}
  For any $\alpha \in (1, \infty)$, we bound $\tr(\Sigma) = 1 - \sqrt{1-\eps^2}$ as follows:
  \begin{align}
    1 - \sqrt{1-\eps^2} &= \tr(\Sigma) = \sum_{i \in S^+} r_i - \exp(\lambda) s_i 
    \leq \sum_{i \in S^+} r_i \\
    &\leq \sum_{i \in S^+} r_i \left( \frac{r_i}{s_i} \exp(-\lambda) \right)^{\alpha - 1} 
    \leq \exp\big(-\lambda (\alpha - 1)\big) \sum_{i \in S} r_i^\alpha\, s_i^{1 - \alpha} \,. 
  \end{align}
  Hence, taking the logarithm and dividing by $\alpha - 1 > 0$, we get
  \begin{align}
    \lambda \leq \frac{1}{\alpha-1} \log \bigg( \sum_{i \in S} r_i^\alpha\, s_i^{1 - \alpha} \bigg) + \frac{1}{\alpha - 1} 
      \log \frac{1}{1 - \sqrt{1 - \eps^2}} \label{eq:aep/bound1} \,.
  \end{align}

  Next, we use the data-processing inequality of the R\'enyi divergences.
  We use the measurement $\cptp$ map $\sM : X \mapsto \sum_{i \in S} \proj{e_i} X \proj{e_i}$ to obtain
  \begin{align}
     \DD_{\alpha}(\rho \| \sigma) \geq D_{\alpha}\big( \sM(\rho) \| \sM(\sigma) \big) 
      = \frac{1}{\alpha-1} \log \bigg( \sum_{i \in S} r_i^\alpha s_i^{1 - \alpha} \bigg) \,.
  \end{align}
  We conclude the proof by substituting this into~\eqref{eq:aep/bound1} and applying Lemma~\ref{lm:aep/smooth-bound}.
  \qed
\end{proof}
\end{petit}

We also note here that $g(\eps)$ can be bounded by simpler expressions. For example, $1 - \sqrt{1 - \eps^2} \geq \frac{1}{2} \eps^2$ using a second order Taylor expansion of the expression
around $\eps = 0$ and the fact that the third derivative is non-negative. This is a very good
approximation for small $\eps$.
Hence, \eqref{eq:min-renyi-bound} can be simplified to~\cite{tomamichel08}
\begin{align}
    D_{\max}^{\eps}(\rho \| \sigma) \leq \DD_{\alpha}(\rho \| \sigma) + \frac{1}{\alpha - 1} 
      \log \frac{2}{\eps^2} \label{eq:min-renyi-bound-untight} \,.
\end{align}
%This form of the inequality has been reported previously~\cite{tomamichel08}. 

Proposition~\ref{pr:min-renyi} is of particular interest when applied to the smooth conditional min-entropy.
In this case, let $\rho_{AB} \in \cSsub(AB)$ and $\sigma_B$ be of the form $\id_A \otimes \sigma_B$. Then,
for any $\alpha \in (1, \infty)$, we have
\begin{align}
  H_{\min}^{\eps}(A|B)_{\rho} &\geq \HH_{\alpha}(A|B)_{\rho} - \frac{g(\eps)}{\alpha-1} \,
      \label{eq:aep/smooth-alpha-bound/min} ,
\end{align}
where we again take $\HH_{\alpha}$ to be any conditional R\'enyi entropy whose underlying divergence satisfies the data-processing inequality.
The duality relation for the smooth min- and max-entropies (cf.\ Proposition~\ref{pr:smooth-dual}) and the
R\'enyi entropies (cf.\ Sec.~\ref{sc:rdual}) yield a corresponding dual relation 
for the max-entropy. 
%For example, for any $\alpha \in [0, 1)$, we have
%\begin{align}
%  H_{\max}^{\eps}(A|B)_{\rho} &\leq \Ho_{\alpha}^{\downarrow}(A|B)_{\rho} - \frac{g(\eps)}{1-\alpha} \label{eq:aep/smooth-alpha-bound/max} \,.
%\end{align}

%%%%%%%%%%

\subsection{The Asymptotic Equipartition Property}

In this section we now apply Proposition~\ref{pr:min-renyi} to two sequences $\{ \rho^n \}_{n}$ and $\{ \sigma^n \}_{n}$ of product states of the form 
\begin{align}
  \rho^n = \bigotimes_{i=1}^n \rho_i , \quad \sigma^n = \bigotimes_{i=1}^n \sigma_i,  \quad \textrm{with} \quad \rho_i, \sigma_i \in \cSnorm(A)
\end{align}
where we assume for mathematical simplicity that the marginal states $\rho_i$ and $\sigma_i$ are taken from a finite subset of $\cSnorm(A)$. 
Proposition~\ref{pr:min-renyi} then yields
\begin{align}
  \frac{1}{n} D_{\max}^{\eps}\big(\rho^n \big\| \sigma^n\big) \leq \frac{1}{n} \sum_{i=1}^n \Dn_{\alpha}(\rho_i \| \sigma_i) + \frac{g(\eps)}{n(\alpha - 1)} \,. \label{eq:aep1}
\end{align}

We can further bound the smooth max-divergence in Proposition~\ref{pr:min-renyi} using the Taylor series expansion for the R\'enyi divergence in~\eqref{eq:taylor}. This means that there exists a constant $C$ such that, for all $\alpha \in (1, 2]$ and all $\rho_i$ and $\sigma_i$, we have\footnote{Here we use that $\rho_i$ and $\sigma_i$ are taken from a finite set, so that we can choose $C$ uniformly.}
\begin{align}
    \Dn_{\alpha}(\rho_i\|\sigma_i) \leq D(\rho_i \| \sigma_i) + (\alpha-1) \frac{1}{2 \log(e)} V(\rho_i\|\sigma_i) + (\alpha-1)^2 C \,, 
\end{align}
It is often not necessary to specify the constant $C$ in the above expression. However, it is possible to give explicit bounds, which is done, for example, in~\cite{tomamichel08}. Substituting the above into~\eqref{eq:aep1} and setting $\alpha = 1 + \frac{1}{\sqrt{n}}$ yields
\begin{align}
  \frac{1}{n} D_{\max}^{\eps}(\rho^n \| \sigma^n) \leq \frac{1}{n} \sum_{i=1}^n D(\rho_i\|\sigma_i) + \frac{1}{\sqrt{n}} \bigg( g(\eps) + \frac{1}{2\log(e)} \frac{1}{n} \sum_{i=1}^n V(\rho_i\|\sigma_i)  \bigg) + \frac{C}{n} \,.
\end{align}
Hence, in particular for the iid case where $\rho_i = \rho$ and $\sigma_i = \sigma$ for all $i$, we find:
\begin{theorem}
\begin{svgraybox}
\label{th:min-vn}
  Let $\rho \in \cSnorm(A)$ and $\sigma \in\cS(B)$ and $\eps \in (0,1)$. Then,
\begin{align}
  \lim_{n \to \infty} \bigg\{ \frac{1}{n} D_{\max}^{\eps}\big( \rho^{\otimes n} \big\| \sigma^{\otimes n} \big) \bigg\} \leq  D(\rho\|\sigma) \,.  \label{eq:min-vn}
\end{align}
\vspace{-0.5cm}
\end{svgraybox}
\end{theorem}

This is the main ingredient of our proof of the AEP below.

%%%%%%%

\subsubsection*{Direct Part}

In this section, we are mostly interested in the application of~Theorem~\ref{th:min-vn} to
conditional min- and max-entropies. Here, for any state $\rho_{AB} \in \cSnorm(AB)$, we choose
$\sigma_{AB} = \id_A \otimes \rho_B$. Clearly,
\begin{align}
  H_{\min}^{\eps}(A^n|B^n)_{\rho^{\otimes n}} \geq - D_{\max}^{\eps} \big(\rho_{AB}^{\otimes n} \big\| 
      \sigma_{AB}^{\otimes n} \big)
\end{align}
Thus, by Theorem~\ref{th:min-vn}, we have
\begin{align}
    \lim_{n \to \infty} \bigg\{ \frac{1}{n} H_{\min}^{\eps}(A^n|B^n)_{\rho^{\otimes n}} \bigg\}
    &\geq \lim_{n \to \infty} \bigg\{ - \frac{1}{n} D_{\max}^{\eps} \big(\rho_{AB}^{\otimes n} \big\| 
      \sigma_{AB}^{\otimes n} \big) \bigg\} \\
      &\geq- D(\rho_{AB} \| \sigma_{AB}) = H(A|B)_{\rho} \,.
\end{align}

This and the dual of this relation leads to the following corollary, which is the \emph{direct part} of the AEP.
\begin{corollary}
\begin{svgraybox}
  \label{co:aep/cond}
  Let $\rho_{AB} \in \cSnorm(AB)$ and $0 < \eps < 1$. Then, the smooth entropies of the
   i.i.d.\ product state $\rho_{A^n B^n} = \rho_{AB}^{\otimes n}$ satisfy
\begin{align}
    \lim_{n\to \infty} \bigg\{ \frac{1}{n} H_{\min}^{\eps}(A^n|B^n)_{\rho} \bigg\} &\geq H(A|B)_{\rho}
   \quad \textrm{and} \label{eq:aep/min-vn} \\
    \lim_{n\to \infty} \bigg\{ \frac{1}{n} H_{\max}^{\eps}(A^n|B^n)_{\rho} \bigg\} &\leq H(A|B)_{\rho}
   \label{eq:aep/max-vn} \,.
\end{align}
\vspace{-0.5cm}
\end{svgraybox}
\end{corollary}

\subsubsection*{Converse Part}

To prove asymptotic convergence, we will also need converse bounds. For $\eps = 0$,
the converse bounds are a consequence of the monotonicity of the conditional R\'enyi entropies in $\alpha$, i.e.\ $H_{\min}(A|B)_{\rho} \leq
H(A|B)_{\rho} \leq H_{\max}(A|B)_{\rho}$ for normalized states $\rho_{AB} \in \cSnorm(AB)$. For $\eps > 0$, similar bounds can be derived based on the
continuity of the conditional von Neumann entropy in the state~\cite{alicki03}.
However, such bounds do not allow a statement of the form of
Corollary~\ref{co:aep/cond} as the deviation from the von Neumann entropy
scales as $n f(\eps)$, where $f(\eps) \to 0$ only for $\eps \to 0$. (See, for
example,~\cite{tomamichel08} for such a weak converse bound.) This is not sufficient for some
applications of the asymptotic equipartition property.

Here, we prove a tighter bound, which relies on the
bound between smooth max-entropy and smooth min-entropy established in 
Proposition~\ref{pr:min-max-smooth}.
Employing this in conjunction with~\eqref{eq:aep/min-vn}
and~\eqref{eq:aep/max-vn} establishes the converse AEP bounds. Let $0 < \eps < 1$. Then,
using any smoothing parameter 
$0 < \eps' < 1 - \eps$, we bound
\begin{align}
  \frac{1}{n} H_{\min}^{\eps}(A^n|B^n)_{\rho} &\leq \frac{1}{n} H_{\max}^{\eps'}(A^n|B^n)_{\rho} + 
      \frac{1}{n} \log \frac{1}{1 - (\eps + \eps')^2} \,. \label{eq:conv1}
\end{align}
The corresponding statement for the smooth max-entropy follows analogously. Starting from~\eqref{eq:conv1} we then apply the same argument that led to Corollary~\ref{co:aep/cond} in order to establish the following \emph{converse part} of the AEP.

\begin{corollary}
\begin{svgraybox}
  \label{co:aep-converse}
  Let $\rho_{AB} \in \cSnorm(AB)$ and $0 \leq \eps < 1$. Then, the smooth entropies of the
   i.i.d.\ product state $\rho_{A^n B^n} = \rho_{AB}^{\otimes n}$ satisfy
  \begin{align}
    \lim_{n\to \infty} \bigg\{ \frac{1}{n} H_{\min}^{\eps}(A^n|B^n)_{\rho} \bigg\} &\leq H(A|B)_{\rho}
   \quad \textrm{and} \label{eq:aep/conv/min-vn} \\
    \lim_{n\to \infty} \bigg\{ \frac{1}{n} H_{\max}^{\eps}(A^n|B^n)_{\rho} \bigg\} &\geq H(A|B)_{\rho}
   \label{eq:aep/conv/max-vn} \,.
  \end{align}  
  \vspace{-0.5cm}
\end{svgraybox}
\end{corollary}

These converse bounds are particularly important to bound the smooth entropies for large smoothing parameters. In this form, the AEP implies strong converse statements for many information theoretic tasks that can be characterized by smooth entropies in the one-shot setting.

\subsubsection*{Second Order}

It is in fact possible to derive more refined bounds here, in analogy with the second-order refinement for Stein's lemma encountered in Sec.~\ref{sc:app-hypo}. First we note that from the above arguments we can deduce that the second-order term scales as
\begin{align}
  D_{\max}^{\eps}\big(\rho^{\otimes n} \big\| \sigma^{\otimes n}\big) = n D(\rho\|\sigma) + O\big(\sqrt{n}\big) \,.
\end{align}
and thus it suggests itselfs to try to find an exact expression for the $O(\sqrt{n})$ term.\footnote{Analytic Bounds on the second-order term were also investigated in~\cite{audenaert12}.} One finds that the second-order expansion of $D_{\max}^{\eps}(\rho^{\otimes n} \| \sigma^{\otimes n})$ is given as~\cite{tomamichel12}
\begin{align}
  D_{\max}^{\eps}\big(\rho^{\otimes n} \big\| \sigma^{\otimes n}\big) = n D(\rho\|\sigma) - \sqrt{n V(\rho\|\sigma)} \, \Phi^{-1}(\eps^2) + O(\log n) \,, \label{eq:aep-second-order}
\end{align}
where $\Phi$ is the cumulative (normal) Gaussian distribution function. A more detailed discussion of this is outside the scope of this book and we defer to~\cite{tomamichel12} instead.

%%%%%%%%%%

\section{Background and Further Reading}

This chapter is largely based on~\cite[Chap.~4--5]{mythesis}. The exposition here is more condensed compared to~\cite{mythesis}. On the other hand, 
some results are revisited and generalized in light of a better understanding of the underlying conditional R\'enyi entropies.

The origins of the smooth entropy calculus can be found in classical cryptography, for example the work of Cachin~\cite{cachin97}. Renner and Wolf~\cite{renner04} first introduced the classical special case of the formalism used in this book. 
The formalism was then generalized to the quantum setting by Renner and K\"onig~\cite{rennerkoenig05} in order to investigate randomness extraction against quantum adversaries in cryptography~\cite{maurer05}. Based on this initial work, Renner~\cite{renner05} then defined conditional smooth entropies in the quantum setting. He chose $\Hn_{\infty}^{\ua}$ as the min-entropy (as we do here as well) and he chose $\Ho_{\scriptscriptstyle 0}^{\ua}$ as the max entropy. 
Later K\"onig, Renner and Schaffner~\cite{koenig08} discovered that $\Hn_{\scriptscriptstyle \nicefrac12}^{\ua}$ naturally complements the min-entropy due to the duality relation between the two quantities. Consequently, the max-entropy is defined as $\Hn_{\scriptscriptstyle\nicefrac12}^{\ua}$ in most recent work. (Notably, at the time the structure of conditional R\'enyi entropies as discussed in this book, in particular the duality relation, was only known in special cases.) 
Moreover, Renner~\cite{renner05} initially used a metric based on the trace distance to define the $\eps$-ball of close states. However, in order for the duality relation to hold for smooth min- and max-entropies, it was later found that the purified distance~\cite{tomamichel09} is more appropriate.

The chain rules were derived by Vitanov \emph{et al.}~\cite{vitanov11,vitanov12}, based on preliminary results in~\cite{tomamichel10,berta10}. The specialized chain rules for classical information in Lemmas~\ref{pr:class/bounds-1} and~\ref{pr:class/bounds-2} were partially developed in~\cite{renesrenner10} and~\cite{winkler11}, and extended in~\cite{mythesis}.

A first achievability bound for the quantum AEP for the smooth min-entropy was established in Renner's thesis~\cite{renner05}. However, the quantum AEP presented here is due to~\cite{tomamichel08} and~\cite{mythesis}; it is conceptually simpler and leads to tighter bounds as well as a strong converse statement. 
It is also noteworthy that a hallmark result of quantum information theory, the strong sub-additivity of the von Neumann entropy~\eqref{eq:strong-sub-add}, can be derived from elementary principles using the AEP~\cite{beaudry11}.

% other apps of smooth entropy framework

The smooth min-entropy of classical-quantum states has operational meaning in randomness extraction, as will be discussed in some detail in Section~\ref{sc:rand-ext}. Decoupling is a natural generalization of randomness extraction to the fully quantum setting (see Dupuis' thesis~\cite{dupuis09} for a comprehensive overview), and was initially studied in the context of state merging by Horodecki, Oppenheim and Winter~\cite{horodecki05}. %,S-horodecki06}. 
Decoupling theorems can also be expressed in the one-shot setting, where the (fully quantum) smooth min-entropy $H_{\min}^{\eps}(A|B)$ attains operational significance~\cite{berta08,szehr13,dupuis14}. 
Smooth entropies have been used to characterize various information theoretic tasks in the one-shot setting, for example in~\cite{renesrenner10} and~\cite{datta11a,datta11b,datta13}. The framework has also been used to investigate the relation between randomness extraction and data compression with side information~\cite{renes10}.
Smooth entropies have also found various applications in quantum thermodynamics, for example they are used to derive a thermodynamical interpretation of negative conditional entropy~\cite{delrio11}.

% infinite-dim

We have restricted our attention to finite-dimensional quantum systems here, but it is worth noting that the definitions of the smooth min- and max-entropies can be extended without much trouble to the case where the side information is modeled by an infinite-dimensional Hilbert space~\cite{furrer10} or a general von Neumann algebra~\cite{furrer11}. Many of the properties discussed here extend to these strictly more general settings. However, general chain rules and an entropic asymptotic equipartition property are not yet established in the most general algebraic setting~\cite{furrer11}.

%%%%%%%%%

%!TEX root = book.tex

\chapter{Selected Applications}
\label{ch:app} 

% use \chaptermark{}
% to alter or adjust the chapter heading in the running head

\abstract*{This chapter gives a taste of the applications of the mathematical toolbox discussed in this book. 
The discussion of binary hypothesis testing is crucial because it provides an operational interpretation for the two quantum generalizations of the R\'enyi divergence we treated in this book. This belatedly motivates our specific choice. Entropic uncertainty relations provide a compelling application of conditional R\'enyi entropies and their properties, in particular the duality relation. Finally, smooth entropies were originally invented in the context of cryptography, and the Leftover Hashing Lemma reveals why this definition has proven so useful.}

This chapter gives a taste of the applications of the mathematical toolbox discussed in this book, biased by the author's own interests. 

The discussion of binary hypothesis testing is crucial because it provides an operational interpretation for the two quantum generalizations of the R\'enyi divergence we treated in this book. This belatedly motivates our specific choice. Entropic uncertainty relations provide a compelling application of conditional R\'enyi entropies and their properties, in particular the duality relation. Finally, smooth entropies were originally invented in the context of cryptography, and the Leftover Hashing Lemma reveals why this definition has proven so useful.

%%%%%%%%%%%%%%%%%%%%%%%%%%%%

\section{Binary Quantum Hypothesis Testing}
\label{sc:app-hypo}

As mentioned before, the Petz and the minimal quantum R\'enyi divergence both find operational significance in binary quantum hypothesis testing.
We thus start by surveying binary hypothesis testing for quantum states. However, the proofs of the statements in this section are outside the scope of this book, and we will refer to the published primary literature instead.

Let us consider the following binary hypothesis testing problem. Let $\rho, \sigma \in \cSnorm(A)$ be two states. The \emph{null-hypothesis} is that a certain preparation procedure leaves system $A$ in the state $\rho$, whereas the \emph{alternate hypothesis} is that it leaves it in the state $\sigma$. If this preparation is repeated independently $n \in \mathbb{N}$ times, we consider the following two hypotheses.
\begin{description}[itemsep=0mm]
  \item[Null Hypothesis:] The state of $A^n$ is $\rho^{\otimes n}$.
  \item[Alternate Hypothesis:] The state of $A^n$ is $\sigma^{\otimes n}$.
\end{description}
A \emph{hypothesis test} for this setup is an event $T_n \in \cPsub(A^n)$ that indicates that the null-hypothesis is correct.
The \emph{error of the first kind}, $\alpha_n(T_n)$, is defined as the probability that we wrongly conclude that the alternate hypothesis is correct even if the state is $\rho^{\otimes n}$. It is given by
\begin{align}
  \alpha_n(T_n;\rho) := \tr \big( \rho^{\otimes n} (\id_{A^n} - T_n)\big) .
\end{align}
Conversely, \emph{the error of the second kind}, $\beta_n(T_n)$, is defined as the probability that we wrongly conclude that the null hypothesis is correct even if the state is $\sigma^{\otimes n}$. It is given by
\begin{align}
  \beta_n(T_n;\sigma) := \tr\big( \sigma^{\otimes n} \, T_{n} \big) .
\end{align}

%\subsubsection*{Quantum Chernoff Bound}
\subsection{Chernoff Bound}

We now want to understand how these errors behave for large $n$ if we choose on optimal test.
Let us first minimize the average of these two errors (assuming equal priors) over all hypothesis tests, which leads us to the well known distinguishing advantage (cf.~Section~\ref{sc:gtd}).
\begin{align}
  \min_{T_n \in \cPsub(A^n)} \frac12\Big( \alpha_n(T_n;\rho) + \beta_n(T_n;\sigma) \Big) 
  &= \frac12 + \frac12 \min_{T_n \in \cSsub(A^n)} \tr\big( T_n (\sigma^{\otimes n} - \rho^{\otimes n} ) \big) \nonumber\\
  &= \frac12 \big( 1 - \Delta(\rho^{\otimes n}, \sigma^{\otimes n}) \big) \,. \label{eq:chernoff1}
\end{align}

However, this expression is often not very useful in itself since we do not know how $\Delta(\rho^{\otimes n}, \sigma^{\otimes n})$ behaves as $n$ gets large. This is answered by the quantum Chernoff bound which states that the expression in~\eqref{eq:chernoff1} drops exponentially fast in $n$ (unless $\rho = \sigma$, of course). The exponent is given by the quantum Chernoff bound~\cite{nussbaum09,audenaert07}:

\begin{theorem}
\begin{svgraybox}
  Let $\rho,\sigma \in \cSnorm(A)$. Then,
  \begin{align}
    \lim_{n \to \infty} -\frac{1}{n} \log \min_{T_n \in \cPsub(A^n)} \frac12 \Big( \alpha_n(T_n;\rho) + \beta_n(T_n;\sigma) \Big)  = \max_{0 \leq s \leq 1} - \log \Qo_{s}(\rho\|\sigma) \,.
  \end{align}
    \vspace{-0.5cm}
  \end{svgraybox}
\end{theorem}
\noindent This gives a first operational interpretation of the Petz quantum R\'enyi divergence for $\alpha \in (0,1)$.

Note that the exponent on the right-hand side is negative and symmetric in $\rho$ and $\sigma$. The objective function is also strictly convex in $s$ and 
hence the minimum is unique unless $\rho = \sigma$.
The negative exponent is also called the
\emph{Chernoff distance} between $\rho$ and $\sigma$, defined as
\begin{align}
  \xi_{C}(\rho,\sigma) :=  -\min_{0 \leq s \leq 1} \log \Qo_{s}(\rho\|\sigma) = \max_{0 \leq s \leq 1}\, (1-s)\, \Do_{s}(\rho\|\sigma) \,. \label{eq:chernoff-dist}
\end{align}
In particular, we have $\xi_C(\rho,\sigma) \leq D(\rho\|\sigma)$ since $(1-s) \leq 1$ in~\eqref{eq:chernoff-dist}.

%\subsubsection*{Quantum Stein's Lemma}
\subsection{Stein's Lemma}

In the Chernoff bound we treated the two kind of errors (of the first and second kind) symmetrically, but this is not always desirable. Let us thus in the following consider sequences of tests $\{ T_n \}_n$ such that $\beta_n(T_n;\sigma) \leq \eps_n$ for some sequence of $\{ \eps_n \}_n$ with $\eps_n \in [0, 1]$. We are then interested in the quantities
\begin{align}
  \alpha_n^*(\eps_n;\rho,\sigma) := \min \Big\{ \alpha_n(T_n;\sigma) : T_n \in \cPsub(A^n) \land \beta_n(T_n,\rho) \leq \eps_n \Big\} \,.
\end{align}

Let us first consider the sequence $\eps_n = \exp(-n R)$.
Quantum Stein's lemma now tells us that $D(\rho\|\sigma)$ is a critical rate for $R$ in the following sense~\cite{hiai91,ogawa00}.

\begin{theorem}
\begin{svgraybox}
  Let $\rho,\sigma \in \cSnorm(A)$ with $\rho \ll \sigma$. Then,
  \begin{align}
    \lim_{n \to \infty}\ \alpha_n^*( \exp(-n R) ;\rho,\sigma) = 
    \begin{cases} 
       0 & \textrm{if } R < D(\rho\|\sigma) \\ 
       1 & \textrm{if } R > D(\rho\|\sigma) \label{eq:stein}
    \end{cases} \,.
  \end{align}  
  \vspace{-0.5cm}
\end{svgraybox}
\end{theorem}

\noindent This establishes the operational interpretation of Umegaki's quantum relative entropy.
In fact, the respective convergence to $0$ and $1$ is exponential in $n$, as we will see below.
  An alternative formulation of Stein's lemma states that, for any $\eps \in (0,1)$, we have
  \begin{align}
    \lim_{n \to \infty} - \frac{1}{n} \log \min \Big\{ \beta_n(T_n;\sigma) : T_n \in \cPsub(A^n) \land \alpha_n(T_n,\rho) \leq \eps \Big\}
    = D(\rho\|\sigma) \,. \label{eq:stein-alt}
  \end{align}

\subsubsection*{Second Order Refinements for Stein's Lemma}

A natural question then is to investigate what happens if $- \log \eps_n \approx n D(\rho\|\sigma)$ plus some small variation that grows slower than $n$. This is covered by the second order refinement of quantum Stein's lemma~\cite{li12,tomamichel12}.

\begin{theorem}
\begin{svgraybox}
  Let $\rho,\sigma \in \cSnorm(A)$ with $\rho \ll \sigma$ and $r \in \mathbb{R}$. Then,
  \begin{align}
    \lim_{n \to \infty}\ \alpha_n^*\big(\exp(- n D(\rho\|\sigma) - \sqrt{n} r); \rho,\sigma\big) 
    = \Phi\left( \frac{r}{\sqrt{V(\rho\|\sigma)}} \right) ,
  \end{align}
  where $\Phi$ is the cumulative (normal) Gaussian distribution function.
\end{svgraybox}
\end{theorem}

These works also consider a slightly different formulation of the problem in the spirit of~\eqref{eq:stein-alt}, and establish that
\begin{align}
   &-\log \min \Big\{ \beta_n(T_n;\sigma) : T_n \in \cPsub(A^n) \land \alpha_n(T_n,\rho) \leq \eps \Big\} \nonumber\\
   &\qquad \qquad \qquad \qquad =  n D(\rho\|\sigma) + \sqrt{n V(\rho\|\sigma)}\, \Phi^{-1}(\eps) + O(\log n) \,.
\end{align}

%\subsubsection*{Quantum Hoeffding Bound}
\subsection{Hoeffding Bound and Strong Converse Exponent}

Another refinement of quantum Stein's lemma concerns the speed with which the convergence to zero occurs in~\eqref{eq:stein} if $R < D(\rho\|\sigma)$.
The quantum Hoeffding bound shows that this convergence is exponentially fast in $n$, and reveals the optimal exponent~\cite{hayashi07,nagaoka06}:

\begin{theorem}
\begin{svgraybox}
  Let $\rho,\sigma \in \cSnorm(A)$ and $0 \leq R < D(\rho\|\sigma)$. Then,
  \begin{align}
    \lim_{n \to \infty} - \frac{1}{n} \log \alpha_n^*( \exp(-nR); \rho,\sigma) = \sup_{s \in (0,1)} \left\{ \frac{1-s}{s} \big( \Do_s(\rho\|\sigma) - R \big) \right\} .
  \end{align}
  \vspace{-0.5cm}
\end{svgraybox}
\end{theorem}

\noindent This yields a second operational interpretation of Petz' quantum R\'enyi divergence.

A similar investigation can be performed in the regime when $R > D(\rho\|\sigma)$, and this time we find that the convergence to one is exponentially fast in $n$. The strong converse exponent is given by~\cite{mosonyiogawa13}:

\begin{theorem}
\begin{svgraybox}
  Let $\rho,\sigma \in \cSnorm(A)$ with $\rho \ll \sigma$ and $R > D(\rho\|\sigma)$. Then,
  \begin{align}
    \lim_{n \to \infty} - \frac{1}{n} \log \Big( 1- \alpha_n^*( \exp(-nR); \rho,\sigma) \Big) = \sup_{s > 1} \left\{ \frac{s-1}{s} \big( R - \Dn_s(\rho\|\sigma)\big) \right\} .
  \end{align}
  \vspace{-0.5cm}
\end{svgraybox}
\end{theorem}

This establishes an operational interpretation of the minimal quantum R\'enyi divergence for $\alpha \in (1,\infty)$.

%%%%%%%%%%%%%%%

\section{Entropic Uncertainty Relations}
\label{sc:app-ucr}

The uncertainty principle~\cite{heisenberg27} is one of quantum physics' most intriguing phenomena. Here we are concerned with preparation uncertainty, which states that an observer who has only access to classical memory cannot predict the outcomes of two incompatible measurements with certainty. Uncertainty is naturally expressed in terms of entropies, and in fact entropic uncertainty relations (URs) have found many applications in quantum information theory, specifically in quantum cryptography.

Let us now formalize a first entropic UR. For this purpose, let $\{ \ketn{\phi_x} \}_x$ and $\{ \ketn{\vartheta_y} \}_y$ be two ONBs on a system $A$ and $\sM_X \in \cptp(A,X)$ and $\sM_Y \in \cptp(A,Y)$ the respective measurement maps. Then, Maassen and Uffink's entropic UR~\cite{maassen88} states that, for any initial state $\rho_A \in \cSnorm(A)$, we have
\begin{align}
  H_{\alpha}(X)_{\sM_X(\rho)} + H_{\beta}(Y)_{\sM_Y(\rho)} \geq - \log c\,, \quad \textrm{where} \quad c = \max_{x,y} \abs{ \braketn{\phi_x}{\vartheta_y} }^2 \label{eq:defc}
\end{align}
is the \emph{overlap} of the two ONBs and the parameters of the conditional R\'enyi entropy, $\alpha,\beta \in [\frac12, \infty)$, satisfy $\frac{1}{\alpha} + \frac{1}{\beta} = 2$. In the following we generalize this relation to conditional entropies and quantum side information. 

\subsubsection*{Tripartite Uncertainty Relation}

First, note that an observer with quantum side information that is maximally entangled with $A$ can predict the outcomes of both measurements perfectly (see, for instance, the discussion in~\cite{berta10}). This can be remedied by considering two different observers\,---\,in which case the monogamy of entanglement comes to our rescue. 
We find that the most natural generalization of the Maassen-Uffink relation is stated for a tripartite quantum system $ABC$ where $A$ is the system being measured and~$B$ and~$C$ are two systems containing side information~\cite{colbeck11,lennert13}.

\begin{theorem}
\begin{svgraybox}
\label{th:ur}
 Let $\rho_{ABC} \in \cS(ABC)$ and $\alpha, \beta \in [\frac12, \infty]$ with $\frac1{\alpha} + \frac{1}{\beta} = 2$. Then,
\begin{align}
    \Hn^{\ua}_{\alpha}(X|B)_{\sM_X(\rho)} + \Hn^{\ua}_{\beta}(Y|C)_{\sM_Y(\rho)} \geq -\log c\, ,
\end{align}
with $c$ defined in~\eqref{eq:defc}.
\end{svgraybox}
\end{theorem}

\begin{petit}
\begin{proof}
  We prove this statement for a pure state $\rho_{ABC}$ and the general statement then follows by the data-processing inequality. By the duality relation in Proposition~\ref{pr:dual-new}, it suffices to show that
  \begin{align}
    \Hn^{\ua}_{\alpha}(X|B)_{\sM_X(\rho)} \geq \Hn^{\ua}_{\alpha}(Y|Y'B)_{\sU_Y(\rho)}  - \log c\,, \label{eq:ucr-dual}
  \end{align}
  where $\cptp(A,YY') \ni \sU_Y: \rho_A \mapsto \sum_{y,y'} \bracketn{\vartheta_y}{\rho_A}{\vartheta_{y'}} \ketn{y}\!\!\bran{y'}_Y \otimes \ketn{y}\!\!\bran{y'}_{Y'}$ is the map corresponding to the Stinespring dilation unitary of $\sM_Y$. Let us now verify~\eqref{eq:ucr-dual}. We have
  \begin{align}
    \Hn^{\ua}_{\alpha}(Y|Y'B)_{\sU_Y(\rho)} &= \max_{\sigma_{Y'B} \in \cSnorm(Y'B)} -  \Dn_{\alpha}\big(\sU_Y(\rho_{AB}) \big\| \id_Y \otimes \sigma_{Y'B}\big) \\
    &\leq \max_{\sigma_{Y'B} \in \cSnorm(Y'B)} - \Dn_{\alpha}\big(\rho_{AB} \big\| \sU_Y^{-1}(\id_Y \otimes \sigma_{Y'B}) \big) \\
    &\leq \max_{\sigma_{Y'B} \in \cSnorm(Y'B)} - \Dn_{\alpha}\big(\sM_X(\rho_{AB}) \big\| \sM_X(\sU_Y(\id_Y \otimes \sigma_{Y'B})) \big) \,. \label{eq:ur-proof1}
  \end{align}
  The first inequality follows by the data-processing inequality pinching the states so that they are block-diagonal with regards to the image of $\sU_Y$ and its complement. We can then disregard the block outside the image since $\sU_Y(\rho_{AB})$ has no weight there using the mean Property (VI). 
   The second inequality is due to data-processing with $\sM_X$. 
  Now, note that for every $\sigma_{Y'B}$, we have
  \begin{align}
     \sM_X(\sU_Y(\id_Y \otimes \sigma_{Y'B})) &= \sum_y \sM_X \big( \proj{\vartheta_y}_A \big) \otimes  \bra{y}_{Y'} \sigma_{Y'B} \ket{y}_{Y'} \\
     &= \sum_{x,y} \big| \braket{\phi_x}{\vartheta_y} \big|^2\, \proj{x}_X \otimes \bra{y}_{Y'} \sigma_{Y'B} \ket{y}_{Y'} \\
     &\leq c \sum_{x,y} \proj{x}_X \otimes \bra{y}_{Y'} \sigma_{Y'B} \ket{y}_{Y'} = c \, \id_X \otimes \sigma_{B} \,.
  \end{align}
  Substituting this into~\eqref{eq:ur-proof1} yields the desired inequality.
\end{proof}
\end{petit}

\subsubsection*{Bipartite Uncertainty Relation}

Based on the tripartite UR in Theorem~\ref{th:ur}, we can now explore bipartite URs with only one side information system.
To establish such an UR, we start from~\eqref{eq:ucr-dual} and use the chain rule in Theorem~\ref{th:chains} to find
\begin{align}
  \Hn^{\ua}_{\alpha}(X|B)_{\sM_X(\rho)} \geq \Hn^{\ua}_{\gamma}(YY'|B)_{\sU_Y(\rho)} - H_{\beta}^{\ua}(Y'|B)_{\sU_Y(\rho)} - \log c \,,
\end{align}
where we chose $\beta, \gamma \geq \frac12$ such that
\begin{align}
 \frac{\gamma}{\gamma-1} = \frac{\alpha}{\alpha-1} + \frac{\beta}{\beta-1} \quad \textrm{and} \quad (\alpha-1)(\beta-1)(\gamma-1) < 0 \,.
\end{align}
Then, using the fact that the marginals on $YB$ and $Y'B$ of the state $\sU_Y(\rho_{AB}) \in \cSnorm(YY'B)$ are equivalent and that the conditional entropies are invariant under local isometries, we conclude that
\begin{align}
  \Hn^{\ua}_{\alpha}(X|B)_{\sM_X(\rho)} + \Hn_{\beta}^{\ua}(Y|B)_{\sM_Y(\rho)} \geq \Hn^{\ua}_{\gamma}(A|B)_{\rho} + \log \frac{1}{c} \,.
\end{align}
Interesting limiting cases include $\alpha=2$, $\beta \to \frac12$, and $\gamma \to \infty$ as well as $\alpha,\beta,\gamma \to 1$.

Clearly, variations of this relation can be shown using different conditional entropies or chain rules. However, all bipartite URs share the property that on the right-hand side of the inequality there appears a conditional entropy of the state $\rho_{AB}$ prior to measurement. This quantity can be negative in the presence of entanglement, and in particular for the case of a maximally entangled state the term on the right-hand side becomes negative or zero and the bound thus trivial.

%%%%%%%%%%

\section{Randomness Extraction}
\label{sc:rand-ext}

One of the main applications of the smooth entropy framework is in cryptography, in particular in \emph{randomness extraction}, the art of extracting uniform randomness from a biased source. Here the smooth min-entropy of a classical system characterizes the amount of uniformly random key that can be extracted such that it is independent of the side information.
More precisely, we consider a source that outputs a classical system
$Z$ about which there exists \emph{side information} $E$\,|\,potentially quantum\,|\,and ask how
much uniform randomness, $S$, can be extracted from $Z$ such that it is independent of the side information $E$. 
%This primitive
%is of crucial importance in many cryptographic tasks, for example in quantum 
%cryptography. There, we are interested to distill a secret key from a 
%raw key that is partially correlated with a quantum eavesdropper.

\subsection{Uniform and Independent Randomness}

The quality of the extracted randomness is measured using the trace distance to a perfect secret key, which is uniform on $S$ and product with $E$. Namely, we consider the distance
\begin{align}
  \Delta(S|E)_{\rho} := \Delta( \rho_{SE} , \pi_S \otimes \rho_E ),
\end{align}
where $\pi_S$ is the maximally mixed state. Due to the operational interpretation
of the trace distance as a {distinguishing advantage}, a small $\Delta$ implies that
the extracted random variable cannot be distinguished from a uniform and independent random variable with probability 
more than $\frac{1}{2} (1 + \Delta)$. This viewpoint is at the root
of universally composable security frameworks (see, e.g.,~\cite{canetti01,unruh10}), which ensure that a secret key satisfying the above property can safely be employed in any (composable secure) protocol requiring a secret key.

A probabilistic protocol $\sF$ extracting a key $S$ from $Z$ using a random seed $F$ 
is comprised of the following:
\begin{itemize}
  \item A set $\mathcal{F} = \{ f \}$ of functions $f: Z \to S$ which are in one-to-one 
  correspondence with the standard basis elements $\ket{f}$ of $F$.
  \item A probability mass function $\tau \in \cSnorm(F)$. 
\end{itemize}
The protocol then applies a function $f \in \mathcal{F}$ at random (according to the value in $F$) on the input $Z$ to create the key $S$. Clearly, this process can be summarized by a classical channel $\sF \in \cptp(Z,SF)$. 
More explicitly, we start with a classical-quantum state $\rho_{ZE}$ of the form
\begin{align}
  \rho_{ZE} = \sum_z \proj{z}_Z \otimes \rho_E(z) = \sum_{z} \rho(z) \proj{z}_Z \otimes \hat{\rho}_{E}(z) , 
    \quad \hat{\rho}_E(z) \in \cSnorm(E) \,.
\end{align}
The protocol will transform this state into $\rho_{SEF} = (\sF_{Z\to SF} \otimes \sI_E) (\rho_{ZE})$, where 
\begin{align}
  \rho_{SEF} &= \sum_{f} \tau(f)  \hat{\rho}_{SE}(f) \otimes \proj{f}_{F} , \quad \textrm{and} \quad \\  
  \hat{\rho}_{SE}(f) &= \sum_{s} \proj{s}_{S} \otimes  \sum_{z} \delta_{s,f(z)} \rho_E(z)  \label{eq:leftover1}
\end{align} 
is the state produced when $f$ is applied to the $Z$ system of $\rho_{ZE}$.

For such protocols, we then require that the average distance 
\begin{align}
  \sum_f \tau(f)\, \Delta(S|E)_{\rho^f} = \Delta(S|EF)_{\rho}
\end{align}
is small, or, equivalently, we require that the extracted randomness is independent of the seed $F$ as well as $E$. 
This is called the strong extractor regime in classical cryptography, and clearly independence of $F$ is crucial as otherwise the extractor could simply output the seed.
A randomness extractor of the above form that satisfies the security criterion $\Delta(S|EF)_{\rho} \leq \eps$ is said to be \emph{$\eps$-secret}.

Finally, the maximal number of bits of uniform and independent randomness that can be extracted from 
a state $\rho_{ZE}$ is then defined as $\log_2 \ell^{\eps}(Z|E)_{\rho}$, where
\begin{align}
  \ell^{\eps}(Z|E)_{\rho} := \max \big\{ \ell \in \mathbb{N} :\ \exists\, \sF \textnormal{ s.t. } 
    d_S = \ell \land \sF \textrm{ is $\eps$-secret}\, \big\} \,.
\end{align}

The classical Leftover Hash Lemma~\cite{mcinnes87,impagliazzo89,zuckerman89} states that the 
amount of extractable randomness is at least the min-entropy of $Z$ given $E$. In fact, since hashing is an entirely
classical process, one might expect that the physical nature of the side information is irrelevant and that a purely classical
treatment is sufficient. This is, however, not true in general. For example, the output of certain extractor functions may be
partially known if side information about their input is stored in a quantum device of a certain size, while the same output is
almost uniform conditioned on any side information stored in a classical system of the same size. (See~\cite{gavinsky07} for a
concrete example and~\cite{koenig07} for a more general discussion of this topic.)

\subsection{Direct Bound: Leftover Hash Lemma}

A particular class of protocols that can be used to extract uniform randomness is based on two-universal hashing~\cite{carter79}. A two-universal family of hash functions, in the language of the previous section, satisfies
\begin{align}
  \Pr_{F \leftarrow \tau}\big[F(z) = F(z')\big] = \sum_f \tau(f) \delta_{f(z),f(z')} = \frac{1}{d_S}  \qquad \forall\ z \neq z' \,.
\end{align}

Using two-universal hashing, Renner~\cite{renner05} established the following bound. 
\begin{proposition}
\begin{svgraybox}
\label{pr:redirect}
Let $\rho \in \cS(ZE)$. For every $\ell \in \mathbb{N}$, there exists a randomness extractor  as prescribed above such that
\begin{align}
  \Delta(S|EF)_{\rho} \leq \exp \bigg( \frac12 \big( \log \ell - H_{\min}(Z|E)_{\rho} \big) \bigg) \,. \label{eq:redirect}
\end{align}
\end{svgraybox}
\end{proposition}

We provide a proof that simplifies the original argument. We also note that instead of $H_{\min}$ one can write $\Ho_2^{\ua}$ to get a tighter bound in~\eqref{eq:redirect}.

\begin{petit}
\begin{proof}
  We set $d_S = \ell$. Using the notation of the previous section, we have
  \begin{align}
    \Delta(S|EF)_{\rho} = \sum_f \tau(f)\, \big\| \hat{\rho}_{SE}(f) - \pi_S \otimes \rho_E \big\|_1 \,.
  \end{align}
  We note that $\hat{\rho}_E(f) = \rho_E$ does not depend on $f$.
  Then, by H\"older's inequality, for any $\sigma \in \cSnorm(E)$ such that $\sigma_E \gg \rho_E^f$ for all $f$, we have
  \begin{align}
    \big\| \hat{\rho}_{SE}(f) - \pi_S \otimes {\rho}_E \big\|_1 
   % &= \Big\| \sigma_E^{\frac14} \sigma_E^{-\frac14} \big( \rho_{SE}^f - \pi_S \otimes \rho_E \big) \sigma_E^{-\frac14} \sigma_E^{\frac14} \Big\|_1 \\
   % &\leq \big\| \id_S \otimes \sigma_E^{\frac14} \big\|_4 \cdot \Big\| \sigma_E^{-\frac14} \big( \rho_{SE}^f - \pi_S \otimes \rho_E \big) \sigma_E^{-\frac14} \Big\|_2 \cdot \big\| \id_S \otimes \sigma_E^{\frac14} \big\|_4 \\
   % &= \sqrt{ d_S \, \tr \Big( \big(  \sigma_E^{-\frac12} \big( \rho_{SE}^f - \pi_S \otimes \rho_E \big) \big)^2 \Big) }
    &= \Big\| \sigma_E^{\frac12} \sigma_E^{-\frac12} \big( \hat{\rho}_{SE}(f) - \pi_S \otimes \rho_E \big) \Big\|_1 \\
    &\leq \Big\| \id_S \otimes \sigma_E^{\frac12} \Big\|_2 \cdot \Big\| \sigma_E^{-\frac12} \big( \hat{\rho}_{SE}(f) - \pi_S \otimes \rho_E \big) \Big\|_2 \\
    &= \sqrt{ d_S \, \tr \Big(  \sigma_E^{-1} \big( \hat{\rho}_{SE}(f) - \pi_S \otimes \rho_E \big)^2 \Big) } .
  \end{align}
  Hence, Jensen's inequality applied to the square root function yields
  \begin{align}
%    \frac{\Delta(S|EF)_{\rho}^2}{d_S} &\leq \sum_f \tau(f) \tr \Big( \sigma_E^{-\frac12} \big( \rho_{SE}^f - \pi_S \otimes \rho_E \big) 
%    \sigma_E^{-\frac12} \big( \rho_{SE}^f - \pi_S \otimes \rho_E \big) \Big) \\
%    &= \sum_{f} \tau(f) \tr \Big( \sigma_E^{-\frac12} \rho_{SE}^{f} \sigma_E^{-\frac12} \rho_{SE}^{f} \Big) 
%    - \frac{1}{d_S} \tr \Big( \sigma_E^{-\frac12} \rho_{E} \sigma_E^{-\frac12} \rho_{E} \Big) \label{eq:leftover3}
     \big( \Delta(S|EF)_{\rho} \big)^2 &\leq d_S \sum_f \tau(f) \tr \Big( \sigma_E^{-1} \big( \hat{\rho}_{SE}(f) - \pi_S \otimes \rho_E \big) 
     \big( \hat{\rho}_{SE}(f) - \pi_S \otimes \rho_E \big) \Big) \\
    &= d_S \sum_{f} \tau(f) \tr \Big( \sigma_E^{-1} \hat{\rho}_{SE}(f) \hat{\rho}_{SE}(f) \Big) 
    - \tr \Big( \sigma_E^{-1} \rho_{E}^2 \Big) \,, \label{eq:leftover3}
 \end{align}
  where we used that $\pi_S = \frac{1}{d_S} \id_S$. Next, by the definition of $\hat{\rho}_{SE}(f)$ in~\eqref{eq:leftover1}, we find 
  \begin{align}
%    &\sum_{f} \tau(f) \tr \Big( \sigma_E^{-\frac12} \rho_{SE}^{f} \sigma_E^{-\frac12} \rho_{SE}^{f} \Big) \\
%    &\qquad=\sum_{f,z,z'} \tau(f) \rho(z) \rho(z') \delta_{f(z),f(z')} \tr \Big( \sigma_E^{-\frac12} \rho_{E}^{z} \sigma_E^{-\frac12} \rho_{E}^{z'} \Big)\\
%    &\qquad= \sum_{z \neq z'} \frac{1}{d_S} \rho(z) \rho(z') \tr \Big( \sigma_E^{-\frac12} \rho_{E}^{z} \sigma_E^{-\frac12} \rho_{E}^{z'} \Big)
%      + \sum_z \rho(z)^2 \tr \Big( \sigma_E^{-\frac12} \rho_{E}^{z} \sigma_E^{-\frac12} \rho_{E}^{z} \Big) \\
%    &\qquad\leq \frac{1}{d_S} \tr \Big( \sigma_E^{-\frac12} \rho_{E} \sigma_E^{-\frac12} \rho_{E} \Big) + \tr \Big( \sigma_E^{-\frac12} \rho_{ZE} \sigma_E^{-\frac12} \rho_{ZE} \Big)
    \sum_{f} \tau(f) & \tr \Big( \sigma_E^{-1} \hat{\rho}_{SE}(f) \hat{\rho}_{SE}(f) \Big)
    \\ &=\sum_{f,z,z'} \tau(f) \delta_{f(z),f(z')} \tr \Big( \sigma_E^{-1} \rho_{E}(z) \rho_{E}(z') \Big)\\
    &= \sum_{z \neq z'} \frac{1}{d_S} \tr \Big( \sigma_E^{-1} \rho_{E}(z) \rho_{E}(z') \Big)
      + \sum_z \tr \Big( \sigma_E^{-1} \rho_{E}(z) \rho_{E}(z) \Big) \\
    &= \frac{1}{d_S} \tr \Big( \sigma_E^{-1} \rho_{E}^2 \Big) + \left( 1 - \frac{1}{d_S} \right) \tr \Big( \sigma_E^{-1} \rho_{ZE}^2 \Big) \,.
  \end{align}
  Substituting this into~\eqref{eq:leftover3}, we observe that two terms cancel, and maximizing over $\sigma_{E}$ we find
  \begin{align}
    \Delta(S|EF)_{\rho} \leq \sqrt{ (d_S - 1) \exp \big( - \Ho_2^{\ua}(Z|E)_{\rho} \big) } ,
  \end{align}
  where we used the definition of $\Ho_2^{\ua}(Z|E)_{\rho}$. The desired bound then follows since
  $\Ho_2^{\ua}(Z|E)_{\rho} \geq H_{\min}(Z|E)_{\rho}$ according to Corollary~\ref{cor:dual-ineq}.
\qed
\end{proof}
\end{petit}

From the definition of $\ell^{\eps}(Z|E)_{\rho}$ we can then directly deduce that  
\begin{align}
  \log \ell^{\eps}(Z|E)_{\rho} \geq \Hn_{2}^{\ua}(Z|E)_{\rho} - 2 \log \frac{1}{\eps} \geq H_{\min}(Z|E)_{\rho} - 2 \log \frac{1}{\eps} \,.
\end{align}
This can then be generalized using the smoothing technique as follows: 

\begin{corollary}
    The same statement as in Proposition~\ref{pr:redirect} holds with
    \begin{align}
    \Delta(S|EF)_{\rho} \leq \exp \Big( \frac12 \big( \log \ell - H_{\min}^{\eps}(Z|E)_{\rho} \big) \Big) + 2 \eps \,.
    \end{align}
\end{corollary}
\begin{petit}
\begin{proof}
  Let $\rhot_{ZE}$ be a state maximizing $H_{\min}^{\eps}(Z|E)_{\rho} = H_{\min}(Z|E)_{\rhot}$. Then, Proposition~\ref{pr:redirect} yields
  \begin{align}
    \Delta(S|EF)_{\rhot} \leq \exp \Big( \frac12 ( \log \ell - H_{\min}^{\eps}(Z|E)_{\rho} ) \Big) \,.
  \end{align} 
  Moreover, employing the triangle inequality twice, we find that $\Delta(S|EF)_{\rho} \leq \Delta(S|EF)_{\rhot} + 2\eps$.
\qed
\end{proof}
\end{petit}
This result can also be written in the following form:
\begin{align}
  \log \ell^{\eps}(Z|E)_{\rho} \geq H_{\min}^{\eps_1}(Z|E)_{\rho} - 2 \log \frac{1}{\eps_2} , \quad \textrm{where} \quad \eps = 2\eps_1 + \eps_2 
  \label{eq:ext/direct}.  
\end{align}

Note that the protocol families discussed above work on any state $\rho_{ZE}$ with sufficiently high min-entropy, i.e.\ they do not take into account other properties of the state. Next, we will see that these protocols are essentially optimal.

\subsection{Converse Bound}

We prove a {converse bound} by contradiction. Assume for the sake of the argument that we have 
an $\eps$-good protocol that extracts $\log \ell > H_{\min}^{\eps'}(Z|E)_{\rho}$ bits of randomness, where 
$\eps' = \sqrt{2\eps - \eps^2}$. Then, due to Proposition~\ref{pr:func}
we know that applying a function on $Z$ cannot increase the smooth min-entropy, thus
\begin{align}
  \forall\ f \in F : \quad H_{\min}^{\eps'}(S|E)_{\rho^f} \leq H_{\min}^{\eps'}(Z|E)_{\rho} < \log \ell \,. \label{eq:ext/converse-1}
\end{align}
This in turn implies that $\sum \tau(f)\, \Delta(S|E)_{\rho^f} > \eps$ as the following argument shows. The above inequality
as well as the definition of the smooth min-entropy implies that
all states $\rhot$ with 
\begin{align}
 P(\rhot_{SE}, \rho_{SE}^f) \leq \eps' \quad \textrm{or} \quad \Delta(\rhot_{SE}, \rho_{SE}^f) \leq \eps 
\end{align}
necessarily satisfy $H_{\min}(S|E)_{\rhot} < \log \ell$. (The latter statement follows from the Fuchs--van de Graaf inequalities in Lemma~\ref{lm:pd-gtd-bounds}.)
In particular, these close states can thus 
not be of the form $\pi_S \otimes \rho_E$, because such states have min-entropy $\log \ell$. Thus, $\Delta(S|E)_{\rho^f} > \eps$.

Since this contradicts our
initial assumption that the protocol is $\eps$-good, we have established the following
converse bound:
\begin{align}
  \log \ell^{\eps}(Z|E)_{\rho} \leq H_{\min}^{\eps'}(Z|E)_{\rho} \label{eq:ext/converse}.
\end{align}

Collecting \eqref{eq:ext/direct} and~\eqref{eq:ext/converse}, we arrive at the following theorem.
\begin{theorem}
\begin{svgraybox}
  \label{th:rand-ext}
  Let $\rho_{ZE} \in \cSsub(ZE)$ be classical on $Z$ and let $\eps \in (0,1)$. Then,
  \begin{align}
     H_{\min}^{\eps'}(Z|E)_{\rho} - 2 \log \frac{1}{\delta}  \leq \log \ell^{\eps}(Z|E)_{\rho} 
       \leq H_{\min}^{\eps''}(Z|E)_{\rho} ,
  \end{align}
  for any $\delta \in (0,\eps)$, $\eps' = \frac{\eps-\delta}{2}$, and $\eps'' = \sqrt{2\eps - \eps^2}$.
\end{svgraybox}
\end{theorem}

We have thus established that the extractable uniform and independent randomness is characterized by the smooth min-entropy, in the above sense. 
One could now analyze this bound further by choosing an $n$-fold iid product state and then apply the AEP to find the asymptotics of $\frac{1}{n} \log \ell^{\eps}(Z^n|E^n)_{\rho^{\otimes n}}$ for large~$n$. More precisely, using~\eqref{eq:aep-second-order} we can verify that the upper and lower bounds on this quantity agree in the first order but disagree in the second order. In particular, the dependence on $\eps$ is qualitatively different in the upper and lower bound. Thus, one could certainly argue that the bounds in Theorem~\ref{th:rand-ext} are not as tight as they should be in the asymptotic limit. We omit a more detailed discussion of this here (see~\cite{tomamichel12} instead) since most applications consider the task of randomness extraction only in the {one-shot} setting where the resource state is unstructured.

\section{Background and Further Reading}

The quantum Chernoff bound has been established by Nussbaum and Szkola~\cite{nussbaum09} (converse) and Audenaert \emph{et al.}~\cite{audenaert07} (achievability).
Quantum Stein's Lemma was shown by Hiai and Petz~\cite{hiai91} (achievability and weak converse) and Ogawa and Nagaoka~\cite{ogawa00} (strong converse). Its second order refinement was proven independently by Li~\cite{li12} and in~\cite{tomamichel12}. 
The quantum Hoeffding bound was established by Hayashi~\cite{hayashi07} (achievability) and Nagaoka~\cite{nagaoka06} (converse).
Audenaert \emph{et al.}~\cite{audenaert07-3} provide a good review of these results.
The optimal strong converse exponent was recently established by Mosonyi and Ogawa~\cite{mosonyiogawa13}.

The limiting cases $\alpha = \beta = 1$ and $\alpha \to \infty$, $\beta \to \frac12$ of the tripartite Maassen-Uffink entropic UR in Theorem~\ref{th:ur} were first shown by Berta \emph{et al.}~\cite{berta10} and in~\cite{tomamichel11}, respectively. The former was first conjectured and proven in a special case by Renes and Boileau~\cite{renes08} and extended to infinite-dimensional systems~\cite{furrer13,frank12}. Here we follow a simplified proof strategy due to Coles \emph{et al.}~\cite{colbeck11}. The exact result presented here can be found in~\cite{lennert13}.
Tripartite URs in the spirit of Section~\ref{sc:app-ucr} can also be shown for smooth min- and max-entropies, both for the case of discrete observables in~\cite{tomamichel11}, and for the case of continuous observables (e.g.\ position and momentum) by Furrer \emph{et al.}~\cite{furrer13}. These entropic URs lie at the core of security proofs for quantum key distribution~\cite{tomamichellim11,furrer12}.

There exist other protocol families that extract the min-entropy against quantum adversaries, for example based on almost two-universal hashing~\cite{tomamichel10} or Trevisan's extractors~\cite{portmann09}. These families are considered mainly because they need a smaller seed or can be implemented more efficiently than two-universal hashing.

%%%%%%%%%%%%%%%%

\appendix
%!TEX root = book.tex

\chapter{Some Fundamental Results in Matrix Analysis}
\label{app:analysis} 

\abstract*{One of the main technical ingredients of our derivations are the properties of operator monotone and concave functions. While a comprehensive discussion of their properties is outside the scope of this book, we will provide an elementary proof of the Lieb--Ando Theorem and the joint convexity of relative entropy, which lie at the heart of our derivations.}

One of the main technical ingredients of our derivations are the properties of operator monotone and concave functions. While a comprehensive discussion of their properties is outside the scope of this book, we will provide an elementary proof of the Lieb--Ando Theorem in~\eqref{eq:lieb-ando} and the joint convexity of relative entropy, which lie at the heart of our derivations.

%%%%%%%%%

\subsubsection*{Preparatory Lemmas}

We follow the proof strategy of Ando~\cite{ando79}, although highly specialized to the problem at hand. We restrict our attention to finite-dimensional positive definite matrices here and start with the following well-known result:% (see, e.g.,~\cite[Thm.~1.3.3]{bhatia07}):

\begin{lemma}
\label{lm:nicetool}
  Let $A, B$ be positive definite, and $X$ linear. We have
  \begin{align}
    \left(\begin{matrix} A & X \\ X^{\dag} & B \end{matrix}\right) \geq 0 \quad \iff \quad  A \geq X B^{-1} X^{\dag} \,.
  \end{align}
\end{lemma}

\begin{petit}
\begin{proof}
  Since the matrix $\left(\begin{matrix} \id & - X B^{-1} \\ 0 & \id \end{matrix}\right)$ is invertible, we find that $\left(\begin{matrix} A & X \\ X^{\dag} & B \end{matrix}\right) \geq 0$ holds iff and only iff
  \begin{align}
     0 \leq \left(\begin{matrix} \id & - X B^{-1} \\ 0 & \id \end{matrix}\right) \left(\begin{matrix} A & X \\ X^{\dag} & B \end{matrix}\right)  \left(\begin{matrix} \id & 0 \\ - B^{-1} X^{\dag} & \id \end{matrix}\right) = \left(\begin{matrix} A - X B^{-1} X^{\dag} & 0 \\ 0 & B \end{matrix}\right),
  \end{align}
  from which the assertion follows. \qed
\end{proof}
\end{petit}

From this we can then derive two elementary results:
\begin{lemma}
  \label{lm:jointly}
  The map $(A, B) \mapsto B A^{-1} B$ is jointly convex and the map $(A,B) \mapsto (A^{-1} + B^{-1})^{-1}$ is jointly concave.
\end{lemma}
The latter expression is proportional to the matrix harmonic mean $A!B = 2(A^{-1} + B^{-1})^{-1}$, and its joint concavity was first shown in~\cite{anderson69}. 

\begin{petit}
\begin{proof}
   Let $A_1, A_2, B_1, B_2$ be positive definite. Then, by Lemma~\ref{lm:nicetool}, for any $\lambda \in [0,1]$, we have
   \begin{align}
     0 &\leq \lambda \left(\begin{matrix} A_1 & B_1 \\ B_1 & B_1 A_1^{-1} B_1 \end{matrix}\right) + 
     (1-\lambda) \left(\begin{matrix} A_2 & B_2 \\ B_2 & B_2 A_2^{-1} B_2 \end{matrix}\right) \\
     &= \left(\begin{matrix} \lambda A_1 + (1-\lambda) A_2 & \lambda B_1 + (1-\lambda) B_2 \\ \lambda B_1 + (1-\lambda) B_2 & \ \lambda B_1 A_1^{-1} B_1 + (1-\lambda) B_2 A_2^{-1} B_2 \end{matrix}\right)
   \end{align}
   and, invoking Lemma~\ref{lm:nicetool} once again, we conclude that
   \begin{align}
     &\lambda B_1 A_1^{-1} B_1 + (1-\lambda) B_2 A_2^{-1} B_2 \nonumber\\
     &\qquad \quad \geq \big( \lambda B_1 + (1-\lambda) B_2 \big)
     \big( \lambda A_1 + (1-\lambda) A_2\big)^{-1} \big( \lambda B_1 + (1-\lambda) B_2 \big) ,
   \end{align}
   establishing joint convexity of the first map.
   
   To investigate the second map, we use a Woodbury matrix identity,
   \begin{align}
     \big(A^{-1} + B^{-1}\big)^{-1} = B - B (A + B)^{-1} B \,, \label{eq:matrixid}
   \end{align}
   which can be verified by multiplying both sides with $A^{-1} + B^{-1}$ from either side and simplifying the resulting expression. To conclude the proof, we note that $B (A + B)^{-1} B$ is jointly convex due to the first statement and the fact that $A + B$ is linear in $A$ and $B$.
\qed
\end{proof}
\end{petit}

As a simple corollary of this we find that $A \mapsto A^{-1}$ and $B \mapsto B^2$ are convex.

%%%%%%%%%

\subsubsection*{Proof of Lieb--Ando Theorem}

Let us now state Lieb and Ando's results~\cite{lieb73a,ando79}.

\begin{theorem}
\begin{svgraybox}
   The map $(A, B) \mapsto A^{\alpha} \otimes B^{1-\alpha}$ on positive definite operators is jointly concave for $\alpha \in (0,1)$ and jointly convex for $\alpha \in (-1,0) \cup (1, 2)$. 
\end{svgraybox}
\end{theorem}

\begin{petit}
\begin{proof}

Using contour integration one can verify that
$\int_{0}^{\infty} (1+\lambda)^{-1} \lambda^{\alpha-1} \mathrm{d} \lambda = \pi \sin(\alpha \pi)^{-1}$
for $\alpha \in (0,1)$. By the change of variable $\lambda \to \mu = t \lambda$, we then find the following integral representation for all $\alpha \in (0,1)$ and $t > 0$:
\begin{align}
  t^{\alpha} = \frac{\sin (\alpha \pi)}{\pi} \int_{0}^{\infty} \frac{t}{\mu + t} \mu^{\alpha-1}\; \mathrm{d} \mu \,. \label{eq:int-rep}
\end{align}

  Let us now first consider the case $\alpha \in (0,1)$. Using~\eqref{eq:int-rep}, we write
\begin{align}
  A^{\alpha} \otimes B^{1-\alpha} 
  &= \big( A \otimes B^{-1} \big)^{\alpha-1} \cdot A \otimes \id \\
  &= \frac{\sin (\alpha \pi)}{\pi} \int_{0}^{\infty} \big( \mu \id \otimes \id + A \otimes B^{-1} \big)^{-1} A \otimes \id\ \mu^{\alpha-1}\; \mathrm{d} \mu \,.
\end{align}
  Thus, it suffices to show joint concavity for every term in the integrand, i.e.\ for the map 
  \begin{align}
    (A, B) \mapsto  \big( \mu \id \otimes \id + A \otimes B^{-1} \big)^{-1} A \otimes \id = \big( \mu A^{-1} \otimes \id + \id \otimes B^{-1} \big)^{-1}
  \end{align}
  and all $\mu \geq 0$. This is a direct consequence of the second statement of Lemma~\ref{lm:jointly}.
  
  Next, we consider the case $\alpha \in (1,2)$. We again write this as
  \begin{align}
    A^{\alpha} \otimes B^{1-\alpha} 
  &= \big( A \otimes B^{-1} \big)^{\alpha-1} \cdot A \otimes \id \\
  &= \frac{\sin ( (\alpha-1) \pi)}{\pi} \int_{0}^{\infty} (A \otimes B^{-1}) \big( \mu \id \otimes \id + A \otimes B^{-1} \big)^{-1} A \otimes \id\ \mu^{\alpha-2}\; \mathrm{d} \mu \,.
  \end{align}
  The integrand here simplifies to
  \begin{align}
    (A \otimes B^{-1}) \big( \mu \id \otimes \id + A \otimes B^{-1} \big)^{-1} A \otimes \id = A \otimes \id \big( \mu \id \otimes B + A \otimes \id \big)^{-1} A \otimes \id \,.
  \end{align}
  However, the first statement of Lemma~\ref{lm:jointly} asserts that the latter expression is jointly convex in the arguments $A \otimes \id$ and $\mu \id \otimes B + A \otimes \id$. And, moreover, since they are linear in $A$ and $B$, it follows that the integrand is jointly convex for all $\mu \geq 0$.
  The remaining case follows by symmetry.
\qed
\end{proof}
\end{petit}

The joint convexity and concavity of the trace functional in~\eqref{eq:lieb-ando} now follows by the argument presented in Section~\ref{sc:functions}, which allows to write
\begin{align}
  \tr(A^{\alpha} K\,B^{1-\alpha} K^{\dag} ) = \bra{\Psi} K^{\dag} A^{\alpha} \otimes \big(B^T\big)^{1-\alpha} K \ket{\Psi} \,.
\end{align}
This thus gives us a compact proof of Lieb's Concavity Theorem and Ando's Convexity Theorem. Finally, we can also relax the condition that $A$ and $B$ are positive definite by choosing $A' = A + \eps \id$ and $B' = B + \eps \id$ and taking the limit $\eps \to 0$. Choosing $K = \id$, we find that this limit exists as long as we require that $B \gg A$
if $\alpha > 1$.
%%%%%%%%%

\subsubsection*{Joint Convexity of Relative Entropy}

As a bonus we will use the above techniques to show that the relative entropy is jointly convex, thereby providing a compact proof of strong sub-additivity.

\begin{theorem}
\begin{svgraybox}
  The map $(A, B) \mapsto A \log(A) \otimes \id - A \otimes \log(B)$ is jointly convex.
\end{svgraybox}
\end{theorem}

\begin{petit}
\begin{proof}
It suffices to prove this statement for the natural logarithm.
We will use the representation
\begin{align}
  \ln (t) = \int_{0}^{\infty} \frac{1}{\mu+1} - \frac{1}{\mu+t} \mathrm{d} \mu
\end{align}
Using this integral representation, we then write
\begin{align}
  A \ln(A) \otimes \id - A \otimes \ln(B) &= \ln \big(A \otimes B^{-1} \big) \cdot A \otimes \id \\
  &= \int_{0}^{\infty} \frac{A \otimes \id}{\mu+1} - \big(\mu \id \otimes \id + A \otimes B^{-1}\big)^{-1} \cdot A \otimes \id\; \mathrm{d} \mu \\
  &= \int_{0}^{\infty} \frac{A \otimes \id}{\mu+1} - \big(\mu A^{-1} \otimes \id + \id \otimes B^{-1}\big)^{-1}\;  \mathrm{d} \mu \,.
\end{align}
Invoking Lemma~\ref{lm:jointly}, we can check that the integrand is jointly convex for all $\mu \geq 0$.
\qed
\end{proof}
\end{petit}

As an immediate corollary, we find that
\begin{align} 
  D(\rho\|\sigma) = \tr(\rho \log \rho - \rho \log \sigma) = \bra{\Psi} \rho \log \rho \otimes \id - \rho \otimes \log \sigma^T \ket{\Psi}
\end{align}
is jointly convex in $\rho$ and $\sigma$. This in turn implies the data-processing inequality for the relative entropy using Uhlmann's trick as discussed in Proposition~\ref{pr:dp-jc}. In particular, we find strong subadditivity if we apply the data-processing inequality for the partial trace:
\begin{align}
  H(ABC)_{\rho} - H(BC)_{\rho} &= D(\rho_{ABC} \| \id_A \otimes \rho_{BC}) \\
  &\leq D(\rho_{AB} \| \id_A \otimes \rho_{B}) = H(AB)_{\rho} - H(B)_{\rho} \,.
\end{align}

\backmatter%%%%%%%%%%%%%%%%%%%%%%%%%%%%%%%%%%%%%%%%%%%%%%%%%%%%%%%

%\printindex

%\bibliographystyle{spphys}
%\bibliography{library}

\fussy

%%%%%%%%%%%%%%%%%%%%%%%%%%%%%%%%%%%%%%%%%%%%%%%%%%%%%%%%%%%%%%%%%%%%%%

\end{document}